\documentclass[a4paper,11pt]{article}

\usepackage[utf8]{inputenc}
\usepackage{graphicx}        
\usepackage[dvipsnames]{xcolor}
\usepackage{color}
\usepackage{amsfonts}
\usepackage{amsmath}
\usepackage{amssymb}
\usepackage{amsthm}
\usepackage{authblk}
\reversemarginpar

\def\finfty{f_{\infty}}
\def\AsHone{$\mathrm{H}_1$}
\def\AsHtwo{$\mathrm{H}_2$}
\def\AsHonep{$\mathrm{H}_1^{\prime}$}

\def\normal{\mathrm{n}}

\def\q{\zeta}
\def\inhomo{\mathcal{F}}
\def\inhomop{\widehat{\mathcal{F}}}
\def\Qf{{\mathcal Q}}

\def\Ia{I_a}
\def\Du{\widehat{D}}

\def\Jumpg{l_0}
\def\hone{l_1}
\def\htwo{l_2}

\def\czero{d_0}
\def\cone{d_1}

\def\vzero{V_0}
\def\vinfty{V_\infty}

\def\aone{a_1}
\def\atwo{a_2}

\def\valA{a}
\def\Af{{\mathcal A}}
\def\Bf{{\mathcal B}}
\def\S{S}

\def\Czo{C_0^{(1)}}
\def\L{\Theta}
\def\cA{A_0}

\usepackage{bm}

\def\baseback{{$\mathrm{B1}$}}
\def\baseKper{{$\mathrm{B2}$}}
\def\baseKperper{{$\mathrm{B3}$}}
\def\basematch{{$\mathrm{B4}$}}
\def\basebaro{{$\mathrm{B5}$}}

\def\chiefes{\kappa}

\def\Y{\mathcal Y}
\def\killback{\zeta}
\def\gaugefinal{\varphi}
\def\gauge{\Psi}
\def\gaugeCABYa{{\gauge(C;A,B,\Y,\alpha){}}}
\def\gaugeABYa{{\gauge(A,B,\Y,\alpha){}}}
\def\gaugeCABY{{\gauge(C;A,B,\Y){}}}
\def\gaugeABY{{\gauge(A,B,\Y){}}}

\def\gaugeAY{{\gauge(A,\Y){}}}

\def\newvhat{\delta}
\def\newsigma{\varsigma}
\def\newsigmaz{\varsigma^0_P}
\def\keygamma{\Gamma}

\def\rad{a_0}

\def\volform{\bm{\eta}_{\mathbb{S}^2}}

\def\trho{\eta}

\def\Phiq1{\Phi^{*1}}

\def\r{r}
\def\P{\mathfrak{P}}
\def\ff{F^{(2)}_P}
\def\gg{F^{(2)}_E}

\def\hatQtwo{\widehat{Q}_2}
\def\ftwo{f^{(2)}}
\def\qtwo{q^{(2)}}

\def\A{\Lambda_2}
\def\B{\Lambda_1}

\def\fin{\hfill \rule{2.5mm}{2.5mm}\\ \vspace{0mm}} 
\def\finn{\hfill \rule{2.5mm}{2.5mm}} 

\def\Eq{\mbox{Eq}}

\def \ge {g_\pertp}
\def \gp {K_1}
\def \gpp {K_2}

\def \hp {h^{(1)}}
\def \hpp {h^{(2)}}
\def \kp {\kappa^{(1)}}
\def \kpp {\kappa^{(2)}}

\def\mmm{M}

\def\lie{\mathcal{L}}

\def\riccifunc{{\mathfrak{R}}}
\def\riccixi{\lambda_{\xi}}
\def\riccieta{\lambda_{\eta}}

\def\pertp{\varepsilon}

\def\gfam{\hat g}

\def\gfamp{g}
\def\Supfamp{\Sigma}
\def\hfamp{h}
\def\kfamp{\kappa}

\def\gback{g}
\def\fpt{\gp}
\def\spt{\gpp}

\def\axial{\eta}
\def\stat{\xi}
\def\axialfSph{\bar{\bm\eta}}

\newcommand{\ot}[1]{\mathfrak{#1}}

 \def \tp {t_+}
\def \phip {\phi_+}
\def \rrp {r_+}
\def \thetap {\theta_+}
 \def \tm {t_-}
 \def \phim {\phi_-}
 \def \rrm {r_-}
 \def \thetam {\theta_-}

\def\ro{a} 

\def\opertbase{\varpi}
\def\opert{\omega}

\def \Wtwo {{\mathcal{W}}}

\def\defi{:=}

\def\Qtwo{Q_2}
\def\Qt2{\hat{Q}_2}

\def\defor{\Xi}

\def\energy{E}
\def\pressure{P}

 \def\Eb{\energy}
 \def\Ep{\energy^{(1)}}
 \def\Ebc{\Eb_c}
 \def\Pbc{\Pb_c}

\def\Epp{\energy^{(2)}}

\def\Pb{\pressure}
\def\Pp{\pressure^{(1)}}
\def\Ppp{\pressure^{(2)}}

\def \fintegral{\mathcal{I}}
\def \fintegralz{\mathcal{I}_0}

\def\hhat{\widehat{h}}
\def\vhat{\widehat{v}}
\def\uhat{\widehat{u}}

\def\qhat{\widehat{q}}

\def\qper{{\hfamp^{(1)}}}
\def\qperper{{\hfamp^{(2)}}}

\def\Tone{{T_1}}
\def\vecTone{{T_1}}

\def\Vtwo{\mathcal{J}_2}
\def\vecVtwo{{\Vtwo}}
\def\Ttwo{{T_2}}
\def\vecTtwo{{T_2}}

\def\Qone{Q_1}
\def\Qtwo{Q_2}

\def\lieQone{\,\zeta(\Qone)}

\def\Kper{\fpt}
\def\Kperper{\spt}

\def\kappaper{{\kfamp^{(1)}}}
\def\kappaperper{{\kfamp^{(2)}}}

\def\Kpernornor{{\fpt^{\perp}}}
\def\Kperpernornor{\spt^{\perp}}

\def\Kpernortan{{\tau^{(1)}}}
\def\Kperpernortan{{\tau^{(2)}}}

\def\Sper{S^{(1)}{}}
\def\Sperper{S^{(2)}{}}
\def\grad{\mathrm{\scriptscriptstyle grad}}

\def\BB{B}
\def\BBp{B^{+}}
\def\Bp{\BBp}

\def\domain{U}

\def\D{\overline{D}{}}

\def\gsph{{g_{\mathbb{S}^2}}}

\def\centresph{\mathcal{C}}

\def\rsdomain{U^3}

\def\axis{\mathcal{A}}

\def\Gg{\mathcal{G}}

\def\grad{\mathrm{\scriptscriptstyle grad}}

\def\foh{h^{(1)}}
\def\fom{m^{(1)}}
\def\fok{k^{(1)}}
\def\fof{f^{(1)}}

\def\sowf{\mathcal{W}}

\def\sper{V_1}
\def\sperper{V_2}

\def\lam{\lambda}

\def\la{\langle}
\def\ra{\rangle}

\def\Ric{\mbox{Ric}}
\def\Ricepsilon{\mbox{Ric}_{\pertp}}
\def\Repsilon{R_{\pertp}{}}
\def\F{f_\omega}

\def\b{b}
\def\c{c}

\newcommand{\trace}[1]{{}^{tr}#1}
\def\Q{\mathcal{R}}
\def\VR{\upsilon}
\def\UR{\chi}

\def\y{w}
\def\Sep{S_{\pertp}{}}
\def\Vep{H_{\pertp}{}}

\def\Sig{\mathfrak{S}}
\def\PP{{\mathcal P}}

\def\Mext{M_{{}_\mathrm{T}}}

\def\den{\energy}%
\def\pre{\pressure}%
\def\denepsilon{\den_{\pertp}}
\def\preepsilon{\pre_{\pertp}}

\def\uepsilon{u_{\pertp}}
\def\udepsilon{\bm{U}_{\pertp}}

\def\denper{\Ep}%
\def\preper{\Pp}%
\def\denperper{\Epp}%
\def\preperper{\Ppp}%

\def\uper{u^{(1)}{}}
\def\uperper{u^{(2)}{}}
\def\gepsilon{\ge{}}%
\def\kapepsilon{\Omega_{\pertp}}
\def\Nepsilon{N_{\pertp}}
\def\Omegaper{\Omega^{(1)}}
\def\Omegaperbase{\Pi^{(1)}}
\def\Omegaperbaseint{\Pi^{(1)}_+}
\def\Omegaperbaseext{\Pi^{(1)}_-}
\def\Omegaperper{\Omega^{(2)}}
\def\Omegaperperbase{\Pi^{(2)}}
\def\Omegaperperbaseint{\Pi^{(2)}_+}
\def\Omegaperperbaseext{\Pi^{(2)}_-}
\def\Tepsilon{T_{\pertp}}
\def\guepsilon{g^{\sharp}_{\pertp}{}}

\newtheorem{theorem}{Theorem}[section]
\newtheorem{proposition}[theorem]{Proposition}
\newtheorem{lemma}[theorem]{Lemma}
\newtheorem{corollary}[theorem]{Corollary}

\theoremstyle{definition}
\newtheorem{definition}[theorem]{Definition}
\newtheorem{notation}[theorem]{Notation}

\newtheorem{remark}[theorem]{Remark}
\numberwithin{equation}{section}

\begin{document}
\title{Existence and uniqueness of compact rotating configurations in GR in second order
perturbation theory}
\author[1]{Marc Mars}
\author[2]{Borja Reina}
\author[2]{Ra\"ul Vera}
\affil[1]{Department of Fundamental Physics and\protect\\
Institute of Fundamental Physics and Mathematics, University of Salamanca}
\affil[2]{Department of Theoretical Physics and History of Science,\protect\\ University of the Basque Country UPV/EHU}
\date{}
\maketitle

\begin{abstract}
  Existence and uniqueness of rotating fluid bodies in equilibrium is
  still poorly understood in General Relativity (GR). Apart from the limiting case of infinitely thin disks, the only known global results  in the stationary rotating case (Heilig \cite{heilig1995} and Makino \cite{makino2018})
  show existence in GR nearby a Newtonian configuration (under suitable additional restrictions).
  In this work we prove existence and uniqueness of rigidly (slowly) rotating
  fluid bodies in equilibrium
  to second order in perturbation theory in GR.
  The most widely used
  perturbation framework to describe
  slowly rigidly rotating stars in the strong field regime
  is the Hartle-Thorne model.
  The model involves a
  number of hypotheses, some explicit, like equatorial symmetry or that
  the perturbation parameter is proportional to the rotation, but some implicit,
  particularly on the structure and regularity of the perturbation tensors
  and the conditions of their matching at the surface.
  In this work, with basis on the  gauge  results obtained in \cite{paper1},
  the Hartle-Thorne model is fully derived from first principles
  and only assuming that the perturbations describe
  a rigidly rotating finite perfect fluid ball (with no layer at the surface)
  with the same barotropic equation of state
  as the static ball.
  Rigidly rotating fluid balls are analyzed consistently in second order perturbation
  theory by imposing only basic differentiability requirements and boundedness.
  Our results prove in particular that, at this level of approximation, the spacetime must be
  indeed equatorially symmetric and is fully determined by two parameters, namely the central
  pressure and the uniform angular velocity of the fluid.
\end{abstract}

\tableofcontents

\section{Introduction}

  Equilibrium configurations of self-gravitating rotating fluid bodies is an important and difficult subject in Einstein's theory of general relativity. From a physical perspective, they model astrophysical objects of finite size with strong gravitational fields, such as compact stars. One aspect of the problem is to construct and study physically realistic examples. Here the main tools are numerical methods and perturbation approaches. Another important aspect is to understand  structural issues such existence and uniqueness properties of the model.
Several approaches have considered the exterior and
interior problems separately.
Existence results for the  Dirichlet problem for both the interior and
the exterior problems on fixed boundaries have been estalished in
\cite{Schaudt, SchaudtPfister}, while the geometric uniqueness of
the exterior problem given an interior metric has been proved in
\cite{mars_seno_uniqueness,raul:convective}
(see \cite{vera_uniqueness} for the Einstein-Maxwell case).
In the perturbative setting (to second order)
the constraints on the Cauchy data coming from the interior problem that need to be imposed to guarantee existence and uniqueness of the exterior are known 
\cite{MaMaVe2007}.

Concernig the global problem,
static (non rotating) and spherically symmetric perfect fluid bodies in General Relativity  (GR) are known to
exist and be unique given an equation of state satisfying some mild conditions and the value
of the central pressure \cite{ADRendall1991} (see also \cite{Pfister2011}).
The solution is either of infinite extent (and then the energy density vanishes at infinity),
or of finite extent so that it can be matched to Schwarzschild.
Much less is known in the rotating case, for which we do
not even have a single explicit solution describing a rotating finite object with its corresponding
asymptotically flat exterior, except in the limiting case of infinitely thin disks \cite{NeugebauerMeinel, KleinRichter}.
In the rotating global problem we only have results on existence of solutions
sufficiently close to Newtonian configurations \cite{heilig1995,makino2018}.
In fact, even in the simpler Newtonian context the problem is highly non-trivial and a subject of active current research \cite{JangMakino17,JangMakino19,StraussWu17,StraussWu19}. Existence results for rotating configurations of other matter models
  have  been established for Vlasov matter \cite{AndreassonKunzeRein}
  and elastic bodies \cite{AnderssonBeigSchmidt}.

As a step forward towards establishing an existence and uniqueness proof of rotating configurations in GR far away from Newtonian regimes we analyze the  problem for ``slowly'' rotating perfect fluid bodies in the context of second order perturbation theory in GR. In informal terms
our main result is (see Theorem \ref{res:second_order} for a precise version)

\begin{theorem}
    Given a static and spherically symmetric perfect fluid body of finite extent in General Relativity with central pressuce $p_c$, 
    there exists a  solution of the second order perturbed field equations in General Relatity satisfying:
    \begin{itemize}
    \item[(a)] The perturbation is stationary and axially symmetric.
    \item[(b)] The interior is a rigidly rotating perfect fluid with central pressure $p_c$ and with the same barotropic equation of state as the spherical body. 
    \item[(c)] The exterior is vacuum (without cosmological constant) and bounded at infinity.
    \item[(d)] The matching conditions (with absence of surface layers)
      are fulfilled at the boundary of the body.
    \end{itemize}
    Moreover, the solution is uniquely determined by the angular velocity
    of the fluid, the configuration is equatorially symmetric and the boundary of the body is stationary and axially symmetric.
    In addition, if the first order perturbation is non zero,
    then the perturbation parameter can be taken to be proportional to the angular velocity.
  \end{theorem}

    This result is interesting in two respects. Firstly, we hope it can pave the way for applying an implicit function method to show existence of rotating configurations near any static spherical model in the fully
    non-linear theory and in the strong field regime. Secondly, as already mentioned, perturbation methods are widely used to study slowly rotating fluids. The literature in the subject is vast and, to a large extent, is based on the perturbation framework put forward by Hartle \cite{Hartle1967} and Hartle and Thorne \cite{HartleThorne1968} in the 60's. Under suitable extra assumptions (some explicit and some implicit) these works provided plausibility arguments towards the validity of the theorem above and, in fact, this validity has been taken for granted in the literature since then (not only in the Hartle-Thorne approach,
    but also in related or other  perturbation methods, e.g. \cite{Bradley_etal2007,CMMR}).
    Given the importance of the Hartle-Thorne approach we prove our theorem in their setup  and therefore
    provide a rigorous proof for its validity, once the relevant correction
    found in \cite{ReinaVera2015}
    is incorporated.
    Thus, our theorem  provides a rigorous and firm basis for all the results based on perturbations \`a la Hartle-Thorne where either the correction
    in \cite{ReinaVera2015} is irrelevant, or has already been taken into account. This applies in particular to the well known scalability property
      of the perturbative models widely used in astrophysics
      (see e.g. \cite{Berti2005}).

We have just mentioned  explicit and implicit extra assumptions, as well as
  plausibility arguments, in the Hartle-Thorne approach. Let us be more specific on this. By extra explicit assumptions we mean equatorial symmetry 
and that the perturbation parameter is proportional to the angular velocity. To discuss the extra implicit assumptions, let us review some basic facts about perturbation theory in GR.  Perturbations to second order around a background spacetime $(M,g)$
are described by two (symmetric and 2-covariant)
tensors on $M$, $\Kper$ and $\Kperper$, for the first and second order respectively.
The perturbed metric corresponds to the one-parameter family $g_\pertp$
given by
\[
g_\pertp=g+\pertp \Kper+ \frac{1}{2} \pertp^2 \Kperper + O(\pertp^3),
\]
where $\epsilon$ is a small parameter, called ``perturbation parameter''.
Given a static and spherically symmetric background configuration
\[
g= -e^{\nu(r)} dt^2+e^{\lambda(r)} dr^2 + r^2 (d\theta^2 + \sin^2 \theta d\phi^2)
\]
in static and spherical coordinates,
the first part of the classical studies consisted in restricting a
  priori the form of the perturbation tensors. Besides on the condition
  of stationarity and axial symmetry, this step was based primarily on  physical arguments relying on how different metric components should be excited
  at different orders in perturbation theory.
  This, combined with a convenient  gauge choice (based on
    a suitable form of stationary and axisymmetric metrics \cite{HartleSharp1967}),
  was used
  to write
  perturbation tensors
(see e.g. \cite{Hartle1967}, \cite{Bradley_etal2007}) as\footnote{The function $m$ here corresponds to $e^\lambda m/r$ in \cite{Hartle1967,HartleThorne1968}.}
\begin{align}
\Kper^{H}&= 2 \omega(r,\theta) r^2\sin^2\theta dt d\phi,\nonumber\\
\Kperper^{H}&=\left ( -4 e^{\nu(r)} h(r,\theta) + 2r^2 \sin^2 \theta \omega^2(r,\theta) \right ) dt^2
+ 4 e^{\lambda(r)} m(r,\theta) dr^2
+ 4 k(r,\theta) r^2\left ( d \theta^2 + \sin^2 \theta
d \phi^2 \right ) \nonumber
\end{align}
in terms of four functions $\{\omega,h,m,k\}$. 
Moreover, these four ``perturbation functions'' were assumed
to cover
both the interior and the exterior of the fluid ball  and to be
\emph{continuous (and $\omega$ also with continuous first derivatives)
}
across the boundary of the fluid,
located at $r=\ro, \ro\in \mathbb{R}$.
This implicit
assumption was combined with some plausibility arguments based on the field equations for a rigidly rotating perfect fluid and vacuum in order to achieve an important simplification of the \textit{angular structure} of the functions, namely, that $\omega=\omega(r)$
and that the expansion of the other functions in terms of Legendre polynomials
contains only the $\ell = 0,1, 2$ components.
The $\ell=1$ components were made to vanish by assuming equatorial symmetry.

Given this setting, the (perturbed) field equations for a rigidly rotating
perfect fluid with the barotropic equation of state of the background
in the interior
and vacuum in the exterior
were studied assuming (again implicitly) that
(i) the functions in $\Kper^H$ and $\Kperper^H$
are bounded at the origin $r=0$,
(ii) satisfy the aforementioned continuity conditions at $r=\ro$, and 
(iii) are zero at infinity. Under these assumptions, plus a fixed value of the
central pressure, it is argued that the field equations yield unique solutions
depending on a single scaling parameter that can thus be absorved
in the perturbation parameter $\pertp$.

While  point (iii) is justified quite directly
by demanding an asymptotically flat metric $g_\pertp$, the rest of implicit assumptions
were not rigorously established as necessary consequences of the problem under
consideration. In addition, the physical and plausibility arguments must be replaced by rigorous arguments. These were the tasks we set up ourselves to do. Actually, our results hold under weaker requirements, since
\emph{the only global assumption we shall need is that the perturbation tensors stay bounded}.
Asymptotic flatness turns out to be a consequence of boundedness and the field equations.

In a first approach to this problem \cite{ReinaVera2015} two of us
dropped assumption (ii)  regarding
the ``matching'' of the functions at $r=\ro$, by resorting to the general perturbed matching theory to second order developed in \cite{Mars2005}.
It was found that the point (ii) is inconsistent with the rest of the setting (this is the correction alluded to above).
More precisely, there is a gauge in which  $\omega$ is indeed $C^1$
and $h,k$ are $C^0$
at the surface, but  the function $m$
presents a jump proportional to the value of the energy density at the surface. This fact has
consequences in the computation of the mass in terms of the radius
(see \cite{Borja_constant_rho,reina-sanchis-vera-font2017}).\footnote{It is worth mentioning that
this correction is present, although it was somehow forgotten, in the original Newtonian approach
by Chandrasekhar \cite{Chandra_Poly_Newton_1933}, see \cite{ReinaVeranote}.}

Point (i) above or, more specifically, the issue of existence of a suitable gauge that transforms
a general stationary, axially symmetric and orthogonally transtive (see below for definitions) first and second order perturbation tensor into a suitable canonical form, while keeping under control their differentiability and boundedness properties, turned out to be a much harder task than originally expected. This problem has been solved
in \cite{paper1}, where we prove that the canonical form can be achieved with the loss of only one derivative and keeping all the
relevant quantities bounded near the origin. This is the content of Theorem 6.3 in \cite{paper1} and its Corollary 6.4, which here we collect together as Theorem \ref{theo:paper1}.
This result is of a purely geometric nature (i.e. independent of any field equations) and yields a  ``canonical form'' that
is still more general than the form of $\Kper^H$ and $\Kperper^H$ above.

This ``canonical form''  carries an associated gauge freedom,
  which is identified in Proposition
6.9 in \cite{paper1} and recovered here as Proposition \ref{prop:full_class_gauges}

The results in \cite{paper1} are the starting point of the present paper, where we
  derive rigorously the Hartle-Thorne model 
  without not only any ad-hoc or implicit assumption, but also
  without assuming \emph{equatorial symmetry} nor
  any \emph{a priori} relationship whatsoever between
  the perturbation parameter and the angular velocity.

The proof starts with three basic steps, 
  at each first and second order:
  (1) obtain the field equations in terms of a set of convenient functions
that encode  all the necessary information to solve the interior and exterior problems,
(2) solve the perturbed matching conditions for the perturbation tensors
to first and second order 
in terms of those functions together with the functions that
    describe the deformation of the surface of the star,
and
(3) join the interior and exterior problems at 
  the common boundary $\Sigma$.   Then, using elliptic methods that exploit
  the regularity and boundedness properties of the ``canonical form'', the analysis of
  the interior and exterior problems at each order
  follows by (a) proving that a number of relevant homogeneous problems only accept the trivial solution, (b) using the remaining gauge freedom left (at each stage) to get rid of spurious solutions, and finally, (c) proving
  existence and uniqueness of the remaining problems.
We stress that the perturbed matching problem in step (b) is solved
    without imposing any a priori condition.
    In particular, we allow the deformation of the surface to be non-axially symmetric
    and time-dependent. It is the global problem itself that, a posteriori, forces the
    deformation of the body to be stationary and axially symmetric.

Although this procedure needs to be applied firstly to the first order problem  and then to the second order, we follow a strategy that allows to treat both cases at once.
This strategy is based on a bootstrap-type argument based on the fact that
a second order perturbation  problem with identically vanishing
first order perturbation tensor
is formally  equivalent to a first order problem. We thus set up without a priori justification a very specific form for the first order perturbation tensor
  (which in fact corresponds to $\Kper^H$ with $\omega(r)$)
  and
  solve the second order problem under this assumption. We call this the
\emph{base global perturbation scheme}.
The bootstrap argument closes by showing that this problem, when restricted to an identically
vanishing $\Kper^H$ imples that the second order perturbation tensor must necessarily take the
form assumed in the base perturbation scheme. In other words, the
first order global problem is a particular case of the bootstrap argument,
applied with a vanishing first order tensor, while the second order global problem
becomes then the bootstrap argument itself.

\subsection{Plan of the paper}
The paper is structured as follows.

In Section \ref{sec:stat_axis_pert_scheme} we set up the stage by
recalling the definition of static and spherically symmetric spacetime and
establishing our basic set of 
global and differentiability assumptions. 
Next, we state the two main results of our previous paper \cite{paper1}
concerning the structure of 
stationary and axially symmetric perturbations
\cite{paper1}.
Theorem \ref{theo:paper1}, establishes the regularity and differentiability
properties of the functions 
when the perturbation tensors are cast in ``canonical form'',
while Proposition \ref{prop:full_class_gauges} provides
the full class of
gauge transformations that preserve the form of the perturbation tensors
in the later \emph{base perturbation scheme}. 

In Section \ref{sec:background_configuration} we establish the background
spacetime;
a static and spherically symmetric spacetime containing two regions matched across a
hypersurface that preserves the symmetries. One of the regions solves the field equations for a perfect fluid with a barotropic equation of state
and non-negative energy-density and pressure,
and the other one is just Schwarzschild. Such a background is called
\textit{perfect fluid ball configuration} (Definition \ref{SpherRegion}).

Section \ref{PertEFEs} is devoted to writing down the second order
perturbed field
equations, derived from the Einstein field equations, first for a general fluid and then particularizing to the rigidly rotating case. We also recall a well-known result on the relationship between rigid rotation and
orthogonal transtivity of the group action, whis is needed to make contact with the geometric results in \cite{paper1}. Finally we find
the consecuences of the imposition of a barotropic equation of state (that of
the background) at the level of the perturbed field equations.

In Section \ref{sec:Base_model} we set up and ellaborate the
global \textit{base perturbation scheme}, which lies at the basis of
the boostrap-type argument described above.  The section starts with a detailed description of the a priori assumptions that define the base scheme. All these assumptions are justified later as part of the bootstrap argument. We split the assumptions into five blocks, B1 to B5, because several intermediate
results only require a subset thereof. The next step, developed in Subsection  \ref{sec:2nd_order_EFEs}, is to write down the explicit form of the field equations, both in the interior and in the exterior domains, for the perturbation tensors of the
base scheme. Part of the computation, which may be of independent interest,
is postponed to Appendix \ref{tfi_sector} where a fully covariant expression for the first order perturbations of the Ricci tensor is obtained (actually for more general background spacetimes). A key step in this subsection is the introduction
of functions $\hat{h}$, $\hat{q}$ and $\hat{v}$ which are nearly
gauge invariant (Lemma \ref{res:gauges_hat}) and in terms of which the field equations simplify. In particular, this allows us to prove
that part of the gauge freedom
can be used to eliminate the $\ell=1$ Legendre sector of the functions
(Proposition \ref{prop:field_equations}). 
At this point, we have isolated
a set of functions and the corresponding equations that fully
characterise the base scheme
in the interior and the exterior regions:
$\{\Wtwo(r,\theta), \qhat_0(r), \vhat_0(r), \sigma(r), \vhat_2(r),
\vhat_\perp(r,\theta), f(r,\theta) \}$, where $ \qhat_0(r)$ and $\sigma(r)$
are free.
The equation of state of the background
is imposed in Subsection \ref{sec:barotropic}
to provide an algebraic expression for $\qhat_0$ in terms of the rest.
So far, no connection between the interior and exterior problems has been made. Subsection  \ref{sec:base_matching} is devoted to do this. The geometric  matching problem, which is technically rather involved, is left to Appendix \ref{app:staxi_perturbed_matching}.
The results of this Appendix are independent of any
  field equations and hence may find applications in other situations.
In particular no symmetry assumptions are made on how the matching surface gets perturbed, so they generalise the matching conditions
obtained in \cite{ReinaVera2015}, where axial symmetry was imposed. The
geometric matching results of the appendix are specialized to our specific fluid problem in Proposition \ref{Prop:matching}. 

In Section \ref{sec:e_u_global} we tackle the global problem of
existence of uniqueness 
of the base scheme.
The section starts with a core result (Proposition \ref{PDEresult})
that provides existence of a decomposition in terms of Legendre
polynomials of functions satisfying a sufficiently general elliptic global problem. 
This,
together with the existence and uniqueness results 
shown in Appendix \ref{App:bocher-like}, are the basic
ingredients for this section.
Subsection \ref{subsection_Wtwo} is
devoted to showing existence and uniqueness of the angular component $\Wtwo$.
In Subsection \ref{sec:vhat_l_2} we prove that $\vhat_\perp(r,\theta)$ must vanish everywhere
and that $\vhat_2(r)$ is unique and vanishes if the first order perturbation of the
base scheme is zero. The existence of a barotropic equation of state together
with the use of (most of) the remaining gauge freedom
is finally used in Subsection \ref{sec:vhat_l_0} to settle the $\ell=0$ sector
and find the exitence and uniqueness
result of the base global scheme (Proposition \ref{res:hkmf_uniqueness}).
This result is the basis of the bootstrap argument.

In Section \ref{sec:existence_and_uniqueness} we use 
the bootstrap argument in terms of the base global scheme
to obtain the main result of the paper.
After discussing the gauge behaviour and physical meaning of the integration parameter
($\Ppp_c$) introduced in the previous sections,
a first use of the bootstrap argument provides Proposition \ref{res:first_order},
which states the result for the first order problem for the first order perturbation
in ``canonical form''.
A second use of the bootstrap argument for the second order problem
provides the final and main result of the paper,
 Theorem \ref{res:second_order}.
In the accompanying Remark \ref{remark:global_solution}
we provide the explicit procedure for the calculation of the global
unique solution in a fully fixed gauge.
We stress \cite{ReinaVera2015} that when
the energy density
of the star does not vanish at the boundary, these expressions correct
    the standard formulae used in the literature. In addition, our gauge fixing respects the condition that the perturbation tensors stay bounded at infinity, something which has been often overlooked in applications of the Hartle-Thorne model.
Finally we exploit the freedom of re-defining the perturbation parameter in order to write down the one parameter family of metrics in the familiar form used in the literature, and discuss the physical interpretation of the only free parameter in the model.

\subsection{Notation}

Most of the notation used in this paper will be specified along the way. Here we only fix the basic objects.

A $C^{n+1}$ spacetime $(\mmm, g)$ is a four-dimensional (we never consider other dimensions in this paper)
   orientable $C^{n+2}$ manifold $\mmm$ endowed with a time-oriented Lorentzian metric $g$ of class $C^{n+1}$ and
    signature $+2$. Our sign conventions for the Riemann and Ricci tensor follow e.g \cite{Wald}. Scalar products of two vector
  fields $X$, $Y$ with the metric $g$ will be denoted by
    $\la X, Y\ra$. The covector metrically related to a vector $X$ is denoted with boldface,  $\bm{X} := g(X,\cdot)$.
Throughout the paper, for functions of one argument, a prime means derivative with respect to the argument.

\section{Stationary and axisymmetric perturbation scheme}
\label{sec:stat_axis_pert_scheme}

In this Section we summarize the results in Paper 1 needed below. Specifically we quote a theorem on the existence of a canonical form for the perturbation metric tensors to first and second order and the regularity of the corresponding coefficient functions, as well as the most general gauge transformation that respects this form (for a special form first order tensor, since this is all we shall need).

The background is spherically symmetric and satisfying appropriate global conditions. The definitions are as in \cite{paper1}

\begin{definition}
  A spacetime $(\mmm,g)$ is static and spherically symmetric
  if it admits an $SO(3)$ group of isometries
  acting transitively on spacelike surfaces (which may degenerate to points), and 
  a Killing vector $\stat$ which is timelike everywhere, commutes with the
  generators of $SO(3)$ and is orthogonal to the $SO(3)$ orbits.
\end{definition}

Our global and differentiability requirements on the spacetime are as follows:

\vspace{2mm}

\noindent {\bf Assumption {\AsHone}:}
$\mmm \simeq \rsdomain \times I$
where $I\subset\mathbb{R}$ is an open interval and
$\rsdomain$ is a radially symmetric domain of $\mathbb{R}^3$ 
with the orbits of the Killing $\xi$ along the 
$I$ factor and $SO(3)$ acting in the standard way on $\rsdomain$.
Moreover, in the cartesian coordinates $\{x,y,z,t\}$ of
  $\rsdomain \times I$, the metric $g$ is
    \begin{align*}
    g= - e^{\nu} dt^2 + \VR (x_i dx^i)^2 + \UR \delta_{ij} dx^i dx^j
\end{align*}
with $\nu,  \VR, \UR$ being $C^{n+1}$ radially symmetric functions
of $x,y,z$.

\vspace{2mm}

We note that the Killing vector $\xi = \partial_t$ is hypersurface orthogonal,
hence the name ``static and spherically symmetric''. The centre
of symmetry $\centresph_0 \subset M$ is by definition the set of points invariant under $SO(3)$. By the global diffeomorphism $\simeq$ in assumption {\AsHone} we have
$\centresph_0 \simeq  \{ 0_3 \cap \rsdomain\}  \times I  $, so $\centresph_0$ is
non-empty if and only if $\rsdomain$ is a ball. 

All geometric objects in $M$ will be identified with their image
by $\simeq$  and viceversa. This applies for instance to $\centresph_0$, or to the function $|x| := \sqrt{x^2 +y^2 +z^2}$ on $\rsdomain$,  which
also defines a function on $M$. The orbits of the $SO(3)$ action are the spheres
$S_r := \{ |x| = r\}$, which we view again  as subsets of $\rsdomain \times I$ or of $M$ depending on the context.

Define two functions $\lambda, \Q: M \rightarrow \mathbb{R}$ by
\begin{align*}
  e^{\lambda} := \UR + \VR |x|^2, \quad \quad \Q^2 := \UR |x|^2, \quad \quad  \Q \geq 0.
\end{align*}
Both are well defined because $\UR$ and $\UR+ \VR |x|^2$ are positive everywhere
(otherwise $g$ is not a Lorentzian metric). It is clear that
$\lambda \in C^{n+1}(M)$ and $\Q \in C^{n+1} (M \setminus \centresph_0) \cap
C^0(M)$, and that both are radially symmetric when expressed in
$\{x,y,z,t\}$ coordinates.

We shall mostly work in spherical coordinates $\{r,\theta,\phi,t\}$ defined
from $\{x,y,z,t\}$ in the standard way. This coordinate system covers
$M \setminus \axis$, where $\axis = \{x =0, y=0\}$ is the axis of the Killing vector $\eta= \partial_{\phi}$. On this domain the metric $g$ takes the form
\begin{align}
  g= - e^{\nu(r)} dt^2 + e^{\lambda(r)} dr^2 + \Q^2(r) \left (d \theta^2
  + \sin^2 \theta d \phi^2 \right ), \quad \quad \stat = \partial_t.
  \label{metric}
\end{align}
We make the usual abuse of notation of writing functions in different coordinate
systems with the same symbol (the meaning should  be clear from the context).
Nevertheless we write explicitly the arguments when we want to make clear which
representation is being used (we have already followed this convention in
\eqref{metric} when writing $\nu(r)$ etc.)

We can now quote the main theorem in \cite{paper1}.
To fix the basic notation, we recall that perturbation tensors are
defined  through a family of $C^{n+1}$ spacetimes
$(\mmm_\pertp,\gfam_\pertp)$, that includes the background $(\mmm,g)$ for $\pertp=0$,
diffeomorphically identified through some gauge $\psi_\pertp$ ($C^{n+2}$ for each $\pertp$). To first and to second order, the respective perturbation tensors
$\Kper$ and $\Kperper$ are defined as
\begin{equation}
  \Kper=\left.\frac{d g_\pertp}{d\pertp}\right|_{\pertp=0},\quad
  \Kperper=\left.\frac{d^2 g_\pertp}{d\pertp^2}\right|_{\pertp=0},
  \label{def:metricperturbations}
\end{equation}
where $g_\pertp\defi \psi^*_\pertp(\gfam_\pertp)$ on $(\mmm,g)$. For the precise notion of
``perturbation scheme'', ``inheritance of an orthogonally transitive
isometry group action'', as well as ``gauge transformations'' and our notation for gauge vectors we refer to \cite{paper1}.
For completeness, though, we recall here
  that a perturbation scheme
    is said to be of class $C^{n+1}$ when the family $\hat{g}_{\pertp}$  is $C^{n+1}$ and  the perturbation tensors $\Kper$, $\Kperper$ are,
    respectively, $C^n$ and $C^{n-1}$.
        We also recall
  that a two-dimensional group of isometries,
generated by say $\{\stat,\axial\}$,
acts orthogonally transitivelly when the $2$-planes orthogonal to the group
orbits generate surfaces. In four dimensions, this happens if and only if the two scalars
$\star(\bm{\stat}\wedge \bm{\axial}\wedge d \bm{\stat})$
and  $\star(\bm{\stat}\wedge \bm{\axial}\wedge d \bm{\axial})$
where we use $\star$ for the Hodge dual operation, vanish identically.

\begin{theorem}[\textbf{Canonical form
    \cite{paper1}}]
  \label{theo:paper1}
  Let $(\mmm,g)$ be a static and spherically symmetric background
  satisfying assumption {\AsHone},
  with $g$ of class $C^{n+1}$ with $n\geq 2$, given in spherical coordinates by \eqref{metric}.
  Let us be given a $C^{n+1}$ maximal
  perturbation scheme $(\mmm_\pertp, \gfam_\pertp,\{\psi_\pertp\})$
  inheriting the orthogonal transitive stationary and axisymmetric action
  generated by $\{ \xi=\partial_t, \eta=\partial_\phi\}$.
  Then, there exists gauge vectors $\sper$ and $\sperper$, that
    commute with $\axial$, are tangent to $S_r$ as well as orthogonal to $\axial$,
    and extend continuously to zero
    at $\centresph_0$, such that the gauge transformed tensors
  $\fpt^\Psi$ and $\spt^\Psi$  are of class
  $C^{n-1}(M \setminus \centresph_0)$ and $C^{n-2}(\mmm\setminus \centresph_0)$
  respectively,
       and such that the functions defined on $\mmm \setminus \centresph_0$ by
  \begin{align}
    &\foh\defi -\frac{1}{4} e^{-\nu} \fpt^\Psi(\partial_t, \partial_t)\label{def:foh}\\
    &\fok\defi \frac{1}{4\trho^2}\Kper^\Psi(\axial,\axial)\label{def:fok}\\
    & -  x \UR\opert=\fpt^\Psi(\partial_t, \partial_{y} ) , \qquad
      y \UR\opert=\fpt^\Psi(\partial_t, \partial_{x} ),\label{def:opert}\\
    &\fom \defi\frac{1}{4}\left\{\fpt^\Psi{}^\alpha_\alpha+e^{-\nu}\fpt^\Psi(\partial_t, \partial_t) - 8\fok\right\}=\frac{1}{4}\left\{\fpt^\Psi{}^\alpha_\alpha-4\foh - 8\fok\right\}\label{def:fom_0}\\
    &h\defi -\frac{1}{4} e^{-\nu} \left(\spt^\Psi(\partial_t, \partial_t)
      -2\trho^2\opert^2\right)\label{def:h}\\
    & k\defi \frac{1}{4\trho^2}\spt^\Psi(\axial,\axial),\label{def:k}\\
    & -  x \UR \sowf =\spt^\Psi(\partial_t,\partial_y), \qquad
      y \UR \sowf =\spt^\Psi(\partial_t, \partial_{x} ),\label{def:sow}\\
    &m \defi\frac{1}{4}\left\{\spt^\Psi{}^\alpha_\alpha+ e^{-\nu}\spt^\Psi(\partial_t,\partial_t) - 8 k\right\}\label{def:m}
  \end{align}
  have the following properties:
  \begin{itemize}
  \item[(a.1)]  $\foh$ extends to a $C^{n}(\mmm)$ function.
  \item[(a.2)] $\opert$ extends to a  $C^{n-1}(\mmm)$ function.
  \item[(a.3)] The vector field $\opert \eta$ is $C^n(\mmm\setminus\centresph_0)$.     \item[(a.4)] $\fom$ and  $\fok$ are
    $C^{n}(\mmm\setminus\centresph_0)$ and bounded near $\centresph_0$.
  \item[(b.1)] $h$ is $C^{n-1}(\mmm\setminus\centresph_0)$ and bounded near $\centresph_0$.
  \item[(b.2)] $\sowf$ is $C^{n-2}(\mmm\setminus\centresph_0)$ 
    and bounded near $\centresph_0$.
  \item[(b.3)] The vector field $\sowf \eta$ is $C^{n-1}(\mmm\setminus\centresph_0)$. 
  \item[(b.4)] $m$ and $k$ are $C^{n-1}(\mmm\setminus\centresph_0)$
    and bounded near $\centresph_0$.
  \end{itemize}
  Moreover, there exist two functions $\fof$ and $f$ defined
  on $\mmm\setminus\centresph_0$, invariant under
  $\stat$ and $\eta$ 
  and satisfying
  \begin{itemize}
  \item[(a.5)] $\fof$  is $C^{n-1}(\mmm\setminus\centresph_0)$,
    bounded near $\centresph_0$, $C^{n}(S_r)$ on all spheres $S_r$,
    $\partial_\theta\fof$ is $C^{n-1}$ outside the axis and extends continuously
    to $\axis\setminus\centresph_0$, where it vanishes, and both
    $\partial_r\fof$ and $\partial_t\fof$ are $C^{n-1}(S_r)$ on all spheres $S_r$,
  \item[(b.5)] $f$  is $C^{n-2}(\mmm\setminus\centresph_0)$,
    bounded near the origin, $C^{n-1}(S_r)$ on all spheres $S_r$,
    $\partial_\theta f$ is $C^{n-2}$ outside the axis and extends continuously
    to $\axis\setminus\centresph_0$, where it vanishes, and both
    $\partial_rf$ and $\partial_tf$ are $C^{n-2}(S_r)$ on all spheres $S_r$,
  \end{itemize}
  so that $\fpt^\Psi$ and $\spt^\Psi$ take the following form on $M\setminus \axis$ 
  \begin{align}
    \fpt^\Psi =& -4 e^{\nu(r)} \foh(r, \theta) dt^2
                 -2  \opert(r,\theta) \Q^2(r)\sin^2\theta dtd\phi
                 + 4 e^{\lambda(r)} \fom(r, \theta) dr^2 \nonumber\\
               &+4   \fok(r, \theta)\Q^2(r)(d\theta^2+ \sin ^2 \theta  d\phi^2)
                 + 4e^{\lambda(r)}\partial_\theta \fof(r,\theta)\Q(r) dr d\theta ,
                 \label{res:fopert_tensor}
    \\
    \spt^\Psi =& \left(-4 e^{\nu(r)} h(r, \theta) + 2{\opert}^2(r, \theta)\Q^2(r) \sin ^2 \theta \right)dt^2+ 4 e^{\lambda(r)} m(r, \theta) dr^2 \nonumber\\
               &+4  k(r, \theta) \Q^2(r) (d\theta^2+\sin ^2 \theta  d\phi^2)
                 + 4e^{\lambda(r)}\partial_\theta f(r,\theta)\Q(r)dr d\theta\nonumber\\
               &-2 \sowf(r,\theta) \Q^2(r) \sin^2\theta dt d\phi.\label{res:sopert_tensor}
  \end{align}\fin
\end{theorem}

\begin{remark}[\cite{paper1}]
  \label{remark:main}
  In the setup of Theorem \ref{theo:paper1} let
  $\Kper$ and $\Kperper$
  be perturbation tensors defined
  by the perturbation scheme
  $(\mmm_\pertp, \gfam_\pertp,\{\psi_\pertp\})$ and
  $\Kper^{\Psi}$, $\Kperper^{\Psi}$ be the corresponding tensors in canonical form. If the
  background admits no further local isometries and the
  perturbation scheme
  is restricted so that
  the inherited axial Killing vector  $\hat \axial_\pertp=d\psi_\pertp(\axial)$
  is independent of the choice $\psi_{\pertp} \in
  \{ \psi_{\pertp} \}$, then
  the gauge vectors $\sper$ and $\sperper$ transforming
  $\Kper$ and $\Kperper$ into \emph{fixed} $\Kper^{\Psi}$ and
  $\Kperper^{\Psi}$  are unique up to
  the addition of a Killing vector of the background that commutes with $\axial$.
  We emphasize that the condition on $\hat \axial_\pertp$ is no restriction at all if $\gfam_\pertp$, $\pertp \neq 0$, admits only one axial symmetry.
  \end{remark}

Since the form of the perturbation tensors in Theorem \ref{theo:paper1} is used
repeatedly in the paper, we put
forward the
following definition:
\begin{definition}
  \label{def:PFSTAXpert}
  First and second order perturbation tensors on
    a static and spherically symmetric background
  that have the structure and regularity properties given in Theorem \ref{theo:paper1}
  are said to be in \emph{canonical 
      form}.
  \end{definition}

As we shall see, the field equations for perturbed fluid balls restrict
strongly the first order metric perturbation tensor. It is an essential ingredient of this paper to understand the full gauge freedom  that respects this restricted form.  The following result, proved in \cite{paper1},
achieves this.

\begin{proposition}[\textbf{Gauge freedom \cite{paper1}}]
  \label{prop:full_class_gauges}
  Let $(M,g)$ be a static and spherically symmetric spacetime
  as in Theorem \ref{theo:paper1}. 
  Assume that $\Q'(r)$ and $\nu'(r)$
  do not vanish identically on open sets and consider the following first
  and second order perturbation tensors 
\begin{align}
\Kper = &  -2 \opert(r,\theta) \Q^2(r)  \sin^2 \theta dt d \phi,
\label{Kper_gauge} \\
\Kperper = &\left ( -4 e^{\nu(r)} h(r,\theta) +
2 \opert^2(r,\theta)  \Q^2(r) \sin^2 \theta \right ) dt^2
+ 4 e^{\lam(r)} m(r,\theta) dr^2 \nonumber \\
&+ 4  k(r,\theta) \Q^2(r)\left ( d \theta^2 + \sin^2 \theta
d \phi^2 \right )
+ 4  e^{\lam(r)} \partial_{\theta} f (r,\theta) \Q(r) dr d \theta\nonumber\\
&-2  \Wtwo (r,\theta)\Q^2(r)\sin^2\theta  dt d\phi. \label{Kperper_gauge}
\end{align}
Then a first order gauge vector $\sper$ 
preserves the form of $\Kper$ (i.e. there is $\opert^{g}$
such  that $\Kper^{g} := \Kper + \lie_{\sper} g$ is given by \eqref{Kper_gauge}
with $\opert \longrightarrow \opert^{g}$) if and only if, up to the
addition of a Killing vector
of the background,
\begin{align}
  \sper = C t \partial_{\phi}, \quad \quad C \in \mathbb{R},
  \quad \quad  \mbox{and then} \quad \quad
\opert^{g} = \opert - C.
\label{gauge_first}
\end{align}
For $\sper$ as in \eqref{gauge_first}, the second order
gauge vector $\sperper$ preserves the  form of $\Kperper$
if and only if 
\begin{align}
\sperper = A t \partial_t + B t \partial_{\phi}
+ 2 \Y(r,\theta) \partial_r + 2 \alpha(r) \sin \theta \partial_{\theta}
+ \killback,
\quad \quad A,B \in \mathbb{R}, \quad 
\killback \mbox{ Killing vector of } g,
\label{gauge_second}
\end{align}
and $\Kperper^{g}:= {\spt} + \lie_{\sperper} 
g + 2\lie_{\sper}\fpt^g
- \lie_{\sper} \lie_{\sper} g$  takes the form 
(\ref{Kperper_gauge}) with the coefficients $h,m,k,f$ transformed to
\begin{align}
h^{g} &= h + \frac{1}{2} A  + \frac{1}{2} \Y  \nu', \label{hg} \\
k^{g} &= k + \Y \frac{\Q'}{\Q} + \alpha(r) \cos \theta, \label{kg} \\
m^{g} &= m + \Y_{,r} + \frac{1}{2} \Y \lam', \label{mg} \\
f^{g} &= f + \frac{\Y}{\Q} - \Q e^{-\lam} \alpha' \cos \theta + \beta(r), \label{fg}\\
\Wtwo {}^{g} &= \Wtwo -B, \label{Wtwog}
\end{align}
where the arbitrary function $\beta(r)$ arises because
$\Kperper^{g}$ only involves $\partial_{\theta} f^{g}$.
\end{proposition}

\begin{remark}
  It is important to stress that this proposition includes in particular the {\it full} gauge freedom that preserves the first order metric perturbation tensor in canonical 
    form.
  Indeed, by setting $\Kper=0$ and $\sper=0$, second order metric perturbation tensors
  transform under a gauge change in exactly the same way as the first order
  perturbation tensors do. Since the tensor
  $\Kperper$ in \eqref{Kperper_gauge} is fully general (in the canonical
   form) it follows that the most general transformation vector that respects a general $\Kper$ in canonical 
  form
  is
  given by $\sper = \sperper$, with $\sperper$ as given in \eqref{gauge_second} and
  $\{\foh,\fom,\fok,\fof, \opert\}$ transform exactly as
  the corresponding \eqref{hg}-\eqref{Wtwog}.
\end{remark}

Exploiting the gauge freedom to first and second order in
Proposition \ref{prop:full_class_gauges} will be
an important tool to prove the results of this paper. We will use the following notation for it.

\begin{notation}
\label{def:gauge_classes}
We will denote by $\{\gaugeCABYa\}$ the family of gauges
described to second order by the gauge vectors \eqref{gauge_first} and
  \eqref{gauge_second} and such that the gauged functions satisfy the regularity properties
of the corresponding functions in Theorem \ref{theo:paper1}..
When e.g. $\alpha(r)$ has already been fixed, so that the gauge vectors are restricted to the form \eqref{gauge_first}-\eqref{gauge_second} with
$\alpha(r)=0$, the corresponding family will be denoted by
$\{\gaugeCABY\}\subset\{\gaugeCABYa\} $. This notation extends naturally
to any subset of gauge parameters in the family.
\end{notation}

\section{Background spherically symmetric global model}
\label{sec:background_configuration}
In this section we recall the basic construction of a spherically
symmetric spacetime consisting of two regions matched across a
hypersurface that preserves the symmetries.
We distinguish the two regions as ``interior'' and ``exterior'', but at
this point this is merely a convention.
We use $(+)$ to label objects in the interior, and a $(-)$
for the exterior. 
We denote by $(\mmm,\gback)$ the static and spherically symmetric
spacetime resulting from the matching 
$\mmm=\mmm^+\cup \mmm^-$ of two 
$C^{n+1}$  ($n \geq 4$)
static and spherically symmetric spacetimes 
$(M^{\pm}, g^\pm)$ with boundaries $\Supfamp^\pm$.
The matching hypersurface is $\mmm^-\cap \mmm^+ \simeq
\Supfamp^+  \simeq \Supfamp^-$.
We will use coordinates $\{ t_{\pm}, r_{\pm}, \theta_{\pm},
\phi_{\pm}\}$ on $(M^{\pm},g^{\pm})$ covering a neighbourhood of 
  the boundaries $\Supfamp^{\pm}$, such that the metrics read 
\begin{align*}
  g^{\pm} = -e^{\nu_{\pm}(r_{\pm})} dt_{\pm}^2 + e^{\lambda_{\pm}(r_{\pm})} dr_{\pm}^2
  +  \Q_{\pm}{}^2(r_\pm) \left (d \theta_{\pm}^2 + \sin^2 \theta_{\pm} d \phi_{\pm}^2 
  \right ).
\end{align*}
By spherically symmetry and staticity,
the hypersurfaces $\Supfamp^{\pm}$ can be described
by embeddings from an abstract manifold $\Supfamp$ (called 
\emph{the boundary}), coordinated by
$\{\tau,\vartheta,\varphi\}$, by means of
\begin{eqnarray}
\Supfamp^+ &=& \{\tp= \tau,\; \rrp=\ro_+,\; \thetap = \vartheta,\; \phip = \varphi\},\label{sigma0+}\\
\Supfamp^- &=& \{\tm= \tau,\; \rrm=\ro_-,\; \thetam = \vartheta,\; \phim = \varphi\},\label{sigma0-}
\end{eqnarray}
where $\ro_{\pm}$ are constants. We may choose $r_{\pm}$  so that
$r_{+}$ takes values to the left of $\ro_{+}$ in the real line
and $r_{-}$ to the right of $\ro_{-}$. Clearly,
  $\Q_{\pm} (\ro_{\pm}) > 0$ (the boundary is a hypersurface). We fix uniquely the unit normals $ \normal^{\pm}$ so that
$ \normal^+$ points $\mmm^+$ inwards
and $ \normal^-$ points $\mmm^-$ outwards.  Thus
\begin{equation}
\normal^+= -e^{-\frac{\lambda_+(\ro_+)}{2}}\partial_{\rrp}|_{\Supfamp^+},\qquad
\normal^-= -e^{-\frac{\lambda_-(\ro_-)}{2}}\partial_{\rrm}|_{\Supfamp^-}.\label{normal_vector}
\end{equation}
$\Supfamp^\pm$ are obviousy timelike everywhere and their
first and second fundamental forms read
\begin{align}
\hfamp^\pm_{ij}dx^idx^j=&-e^{\nu_\pm(\ro_\pm)}d\tau^2 +\Q_\pm^2(\ro_\pm)(d\vartheta^2+\sin^2\vartheta d\varphi^2),\label{h0ij}\\
\kfamp^\pm_{ij}dx^idx^j=&e^{-\frac{\lambda_\pm(\ro_\pm)}{2}}\left(\frac{1}{2}e^{\nu_\pm(\ro_\pm)}\nu'_\pm(\ro_\pm)d\tau^2-\Q_\pm(\ro_\pm)\Q'(\ro_\pm) (d\vartheta^2+\sin^2\vartheta d\varphi^2)\right).\label{k0ij}
\end{align}
The matching conditions across $\Supfamp$ require that the first and
second fundamental forms on both sides agree i.e.
$[\hfamp]=[\kfamp]=0$, where for any object $[f] := f^+ - f^-$. When  a
quantity $f$ satisfies $[f]=0$ we write $f^{+} = f^{-} := f$ on
$\Supfamp$. From \eqref{h0ij}-\eqref{k0ij} the matching conditions are equivalent to
\begin{equation}
[\Q]=0,\qquad [\nu]=0,\qquad [e^{-\lambda/2} \nu']=0,\qquad [e^{-\lambda/2} \Q']=0. \label{background_matching_R}
\end{equation}
The last two can also be written as 
$[ \normal(\nu)] = [  \normal(\Q)] =0$. So far, no field equations have been imposed. We summarize
the construction  with the following definition.
\begin{definition}
\label{SpherRegion}
A spacetime $(M,g)$ si called {\bf static and spherically symmetric
with two regions} if it is composed by $(M^{\pm}, g^{\pm}, \Supfamp^{\pm})$
as described in this section and satisfies the matching conditions
\eqref{background_matching_R}.
\end{definition}

\subsection{Background field equations}
\label{Sect:BackFE}

Our background spacetime 
describes a non-rotating self-gravitating fluid
of finite extent. Thus, it consists of two
regions,  one solving the gravitational field equations
for a perfect fluid and the other for vacuum. In the context 
of General Relativity without cosmological constant, which we assume from now on,
the field equations are  $\mbox{Ein}_{g}{}_{\alpha\beta}=\chiefes
T_{\alpha\beta}$,  where $\mbox{Ein}_g$ is the Einstein tensor of $g$, $\chiefes$ is the gravitational coupling constant and $T_{\alpha\beta}$ is the
energy-momentum tensor of the matter. For a perfect fluid
\begin{align*}
  T_{\mu\nu} = (\den + \pre) u_{\mu} u_{\nu} + \pre g_{\mu\nu},
\end{align*}
where $\pre$ is the pressure, $\den$ the density and
$u$ is the (unit timelike) four-velocity of the fluid.
For the metric 
(\ref{metric})
the perfect-fluid Einstein field equations hold if and only
if, in addition to  $u = e^{-\frac{\nu}{2}} \xi$ and
 \begin{align}
  \chiefes \pre & = e^{-\lam} \frac{\Q^{\prime}}{\Q}
              \left ( \frac{\Q^{\prime}}{\Q} + \nu^{\prime} \right ) -
              \frac{1}{\Q^2},
              \label{exprepre} \\
   \chiefes \den
                &= e^{-\lam} \left ( - 2 \frac{\Q^{\prime\prime}}{\Q} - \frac{\Q^{\prime}{}^2}{\Q^2}
                  + \frac{\Q^{\prime}}{\Q} \lam^{\prime} \right ) + \frac{1}{\Q^2},
              \label{expreden}
\end{align}
the following ODE is satisfied
\begin{align}
\nu^{\prime\prime} & =  - 2 \frac{\Q^{\prime\prime}}{\Q} + \frac{\Q^{\prime}}{\Q}
\left (   2 \frac{\Q^{\prime}}{\Q} + \lam^{\prime} + \nu^{\prime} \right )
+ \frac{1}{2} \nu^{\prime} \left ( \lam^{\prime} - \nu^{\prime} \right )
- \frac{2 e^{\lam}}{\Q^2} \quad \Longrightarrow  \quad 
\pre^{\prime} + \frac{\nu^{\prime}}{2} \left ( \den + \pre
    \right ) =0.
    \label{PDE_pf}
    \end{align}
The implication is in fact an equivalence wherever
$\Q' \neq 0$. From (\ref{exprepre}),
any critical value $r_{crit}$ of
$\Q(r)$ outside
the centre(s) of symmetry (i.e. satisfying 
$\Q'(r_{crit}) =0$, $\Q(r_{crit})\neq 0$) must have $\pre |_{r_{crit}} < 0$.
The boundary of the fluid ball (with vacuum exterior)
is located at $\pre=0$ (the fact that
$\pre |_{\Supfamp} =0$ is a general
consequence of the Israel  conditions and in our setup it follows
immediately from \eqref{exprepre} and \eqref{background_matching_R}). Thus, 
either $\pre >0$ or $\pre <0$ in the interior of the body, and the physical case
is $\pre>0$. Also on physical grounds it must be that the
energy density of the fluid is non-negative and positive somewhere.
We make this assumption explicit:

\vspace{3mm}

\noindent {\bf Assumption {\AsHtwo}:} The background spacetime has two 
non-empty regions, one vacuum and one covered
by a self-gravitating fluid satisfying
$\pre \geq 0$ and $\den \geq 0$.  Moreover, there is at least one point
in the fluid where $\den > 0$.

\vspace{3mm}

The condition $\pre \geq 0$ implies that $\Q(r)$ is
strictly monotonic and we can set $\Q(r)=r$, which we assume from now on.
The field equations (\ref{exprepre})-(\ref{PDE_pf}) become
\begin{align}
  &\lambda' = \frac{1}{r}(1-e^\lambda)+re^\lambda \chiefes \Eb,
    \label{eq:lambdaprime} \\
  &\nu' =\frac{1}{r}(e^\lambda-1)+re^\lambda \chiefes \Pb,\label{eq:nuprime} \\
  & \nu^{\prime\prime}  =  \frac{1}{r}
    \left (   \frac{2}{r} + \lam^{\prime} + \nu^{\prime} \right )
    + \frac{1}{2} \nu^{\prime} \left ( \lam^{\prime} - \nu^{\prime} \right )
    - \frac{2 e^{\lam}}{r^2} \quad \quad \Longleftrightarrow 
    \quad \quad
    \Pb' = -\frac{\nu'}{2}(\Eb+\Pb).\label{eq:pprimenuprime}
\end{align}
Consider the convenient and standard background quantities
\begin{eqnarray}
j(r) &\defi& e^{-(\lambda+\nu)/2},\label{defi:j} \\
1-\frac{ \chiefes M(r)}{ 4 \pi r}&\defi& e^{-\lambda}.\label{eq:Mass}
\end{eqnarray}
The former satisfies
\begin{equation}
\frac{j'}{j}= - \frac{1}{2} \left ( \lambda' + \nu' \right ) =
-\frac{1}{2}re^{\lambda}\chiefes (\Eb+\Pb),
\label{eq:jpoj}
\end{equation}
while the latter allows one to replace the variables
$\{\lambda,\nu\}$ by  $\{M,\Pb\}$ as follows:
(\ref{eq:Mass}) and (\ref{eq:nuprime}) give
\begin{align}
\nu'=\frac{\chiefes}{r(4\pi r-\chiefes M)}(M+4\pi r^3 \Pb),
\label{eqnup}
\end{align}
and the system (\ref{eq:lambdaprime})-(\ref{eq:pprimenuprime}) takes the
standard form (see (9) and (10) in \cite{OpenVolkoff1939})
\begin{align}
\quad 
M' =& 4 \pi r^2 \Eb , \label{eq:MTOV}\\
\pre'
=& - \frac{ \chiefes (\den+\pre)( M 
+ 4 \pi r^3 \pre)}{ 8 \pi r^2 (  1- \frac{\chiefes M}{4\pi r} )}
.\label{eq:PTOV}
\end{align}
These are the well-known TOV equations \cite{OpenVolkoff1939}.
These equations are usually suplemented with a barotropic equation
of state (EOS) $\den(\pre)$ which closes the system. 
A substantial portion of the paper does not rely on
the existence of a barotropic EOS.
We will 
make the assumption explicit when needed (in Section \ref{sec:barotropic}).

The vacuum case is obviously Schwarzschild, for which 
\begin{equation*}
M = \Mext \,\,\,\in \mathbb{R}
\quad \quad 
e^{-\lambda(r)} =  e^{\nu(r)} =1-\frac{\chiefes \Mext}{4 \pi r}.
\end{equation*}

The matching conditions (\ref{background_matching_R}) read,
after setting $\Q_\pm(r_\pm)=r_\pm$,
\begin{equation}
\ro_+=\ro_-(=\ro),\quad  [\nu]=0,\quad [\lambda]=0, \quad [\nu']=0, \label{background_matching}
\end{equation}
and are interpreted as follows: $[\lambda]=0$ is equivalent 
to the continuity of the mass $\Mext=M(\ro)$, $[\nu]=0$ fixes uniquely  
the additive integration constant that arises when solving
(\ref{eqnup}) and
$[\nu']=0$  corresponds to $[P]=0$, which, in principle, determines $\ro$.
Note that $[\nu']=0$ also provides
\begin{equation}
  \nu_{\pm}' (\ro)
    = \frac{1}{\ro}(e^{\lambda(\ro)}-1),
  \label{nup_value}
\end{equation}
where the equality follows directly from \eqref{eq:nuprime}.

Finally, the field equations combined with the matching conditions (\ref{background_matching}) allow us to express the jumps of higher order
derivatives  in terms of the fluid variables ($\cA$ is a constant whose
explicit form is not needed)
\begin{align}
  \left[\lambda'\right]  =&  \ro e^{\lambda(\ro)}\chiefes [\Eb],
\label{Bp} \\
  \left[\nu''\right] =&\frac{1}{\ro} \left( 1+\frac{\ro\nu'(\ro)}{2}\right)[\lambda'],
\label{nupp}\\
\left[\lambda''\right] =&  \ro e^{\lambda(\ro)}\chiefes [\Eb'] + [\lambda'{}^2],
\label{Bpp}\\
\left[ \nu''' \right] = &\frac{1}{\ro}\left(1+\frac{\ro \nu'(\ro)}{2}\right)[\lambda'']+ \cA[\lambda'].\label{nuppp}
\end{align}
Note  that the jumps of  $[\nu'']$ and  $[\lambda']$ are proportional.
All these expressions
are valid also when two perfect fluids are matched.

Everything we have said so far in this section holds locally near the boundaries.
We now make a global assumption similar in spirit to assumption {\AsHone}.
Since the spacetime $(\mmm,g)$ is now composed of two regions $(\mmm^{\pm}, g^{\pm})$
with boundaries $\Supfamp^{\pm}$,  we modify the assumption as follows

\vspace{2mm}

\noindent {\bf Assumption {\AsHonep}:} The interiors
$(\mbox{Int}(\mmm^{\pm}), g^{\pm})$ satisfy
assumption {\AsHone} with corresponding diffeomorphisms
\begin{align*}
  \mbox{Int} (\mmm^+)\simeq \BBp \times I, \quad \quad
  \mbox{Int} (\mmm^-) \simeq ( \mathbb{R}^3 \setminus \overline{\BB^+}) \times I,
\end{align*} where $\BBp$ is an open ball centered at the origin. Moreover
$\Supfamp^{\pm} \simeq (\partial \BBp) \times I$.

\vspace{2mm}

Under  assumptions {\AsHonep} and {\AsHtwo}, it must be the case
that the fluid lies in the interior $\mmm^+$. 
Indeed, if $\mmm^+$ were vacuum  then $\Mext=0$ and one easily concludes
from (\ref{eq:MTOV})-(\ref{eq:PTOV})  together with $\den \geq 0$ 
and  $\pre(\ro)=0$ that $\pre \leq 0$ in the fluid region, which is a
contradiction. Consequently
the coordinate $r$ takes values in $r \in (0,\ro]$ in the 
interior (fluid) region and $r \in [\ro,\infty)$ in the exterior (vacuum) domain.
The spacetime is $C^{n+1}$ (with $n\geq 4$) everywhere
except at $\Sigma$,
in particular in a neighbourhood of the centre  $r=0$.
Since the vacuum region is the exterior $\mmm^-$ can now write
\begin{equation}
e^{-\lambda_-(r)} =  e^{\nu_-(r)} =1-\frac{\chiefes \Mext}{4 \pi r} \qquad \Longrightarrow \qquad 
j_-(r)= 1,
\label{eqs_back_vacuum}
\end{equation}
where  an additive integration constant in $\nu$ has been adjusted
to zero. This choice fixes the (otherwise arbitrary) freedom in scaling
the static Killing $\xi$ by a positive constant.
Moreover,  one has, in addition,
$[E] = \Eb_+(\ro)$ and $[E'] = E'_+(\ro)$.

We now make use of the following result on
the differentiability of radially symmetric
functions (see e.g. \cite{paper1} or Lemma 3.1 in \cite{Andreasson})
\begin{lemma}
  \label{res:radial_functions_origin}
Let $q : \Bp \rightarrow \mathbb{R}$ be radially symmetric,
i.e. such that there exists $\trace{q}:  [0,\rad) \rightarrow \mathbb{R}$
(the trace of $q$) with $q(x)= \trace{q}(|x|) $. 
Then $q\in C^n(\Bp)$ ($n\geq 0$)
if and only if $\trace{q}$ is $C^n([0,\rad))$ (i.e.  up to the inner
boundary) and all its odd derivatives up to order $n$ vanish at zero. Equivalently, if and only if 
\begin{align}
  \trace{q}(r) = \mathcal{P}_n(r^2)  + \Phi^{(n)}(r), \label{polinom}
\end{align}
where $\mathcal{P}_n$ is a polynomial
of degree  $[\frac{n}{2}]$ and $\Phi^{(n)}$ is $C^n([0,\rad))$ and satisfies
$\Phi^{(n)}(r) = o (r^n)$.
\label{origin}
\end{lemma}

This Lemma implies that $\lambda(r)$ and $\nu(r)$ in the $+$ region
(as functions of one variable $r$)
are  $C^{n+1}$ up to boundary, and admit an expansion
\begin{align}
  \lambda(r) = \lambda_0 + \lambda_2 r^2 + \lambda_4 r^4 + \Phi^{(5)}_\lambda(r),
  \quad \quad
  \nu(r) = \nu_0 + \nu_2 r^2 + \nu_4 r^4+ \Phi^{(5)}_\nu (r),
  \label{background_expansion_origin}
\end{align}
with $\lambda_0,\lambda_2,\lambda_4, \nu_0, \nu_2,\nu_4 \in \mathbb {R}$,
and $\Phi^{(5)}_{\lambda}(r)$, $\Phi^{(5)}_{\nu}(r)$ are $C^{n+1}([0,\ro])$
and vanish, together with their derivatives up to order five, at $r=0$.
Combining this with the
field equations (\ref{eq:lambdaprime})-(\ref{eq:nuprime}) near $r=0$ one finds,
in particular,
\begin{equation}
\lambda_0=0,\quad \lambda_2=\frac{\chiefes}{3} \Ebc,\quad
\nu_2=\frac{\chiefes}{6}\left(\Ebc+3\Pbc\right),
\label{eq:l2nu2}
\end{equation}
where $\Ebc=\den(0)$ and $\Pbc=\pre(0)$ are the values of the energy density
and pressure at the origin,
while $\nu_0$ will be determined by the matching condition $[\nu]=0$.
Expressions \eqref{background_expansion_origin}-\eqref{eq:l2nu2} give
\begin{align}
e^{\lambda(r)} &= 1 + \frac{\chiefes}{3} \Ebc  r^2 + O(r^4), \label{expansionlam}\\
e^{\nu(r)} & = e^{\nu_0} \left (
1  +  \frac{\chiefes}{6}  (\Ebc+ 3 \Pbc)  r^2 + O(r^4) \right ).\nonumber
\end{align}
These expansions together with (\ref{eq:Mass}) imply that $M(r) \in O(r^3)$.

Another consequence of assumptions {\AsHonep} and {\AsHtwo} is that $\nu(r)$
is free of critical values outside the origin. First of all,
equation (\ref{eq:MTOV})  together with $M\in  O(r^3)$ and $\den(r) \geq 0$
implies $M(r) \geq 0$. Furthermore, the quantity
$4 \pi r - \chiefes M(r)$ is positive for $r$ sufficiently close to zero, so
regularity of the spacetime imposes (by \eqref{eq:PTOV})
that  $r > \chiefes M(r)/4\pi$ for all $r  \leq \ro $
and $r > \chiefes \Mext/4 \pi$ for $r \geq \ro$ (in fact this property holds in much more general circumstances \cite{marcsenor2m}). 
With these properties it is clear from  (\ref{eqnup}) that
$\nu' > 0$ away from the origin.

The setup described in this section is summarized in the following definition.
\begin{definition}
  \label{def:background}
A $C^{n+1}$ \textbf{perfect fluid ball configuration} is a static and 
spherically symmetric spacetime with two regions, c.f. Definition \ref{SpherRegion},
satisfying assumptions {\AsHonep} and {\AsHtwo}. 
\end{definition}

Whenever this definition is invoked, all the results and notation
introduced in this section will be understood.

\section{Perturbed Einstein's Field equations to second order}
\label{PertEFEs}
We review in this section the perturbations of the Ricci tensor in terms of
$\fpt$ and $\spt$, and the perturbations of the perfect fluid under the
assumption of rigid
rotation.  Recall that this means that the fluid $3$-velocity, as observed by the stationary observer, is uniform in both space and time and only has one component along the axial direction. This precludes, in particular, the presence
  of convective motions inside the fluid.
We then write down the first and second order
perturbed Einstein field equations under these conditions. This part,
just like the previous one,
is a reminder of known things and it is included to make the paper
as self-contained as possible and to fix some notation.

\subsection{First and second order perturbations of the Ricci tensor}
Given two metrics $g$ and $\gepsilon$, the respective Riemann tensors,
denoted by $R^{\mu}_{\phantom{\mu}\alpha\nu\beta}$
and $\Repsilon^{\mu}_{\phantom{\mu}\alpha\nu\beta}$ are related by
(e.g. \cite{Wald})
\begin{align}
\Repsilon^{\mu}_{\phantom{\mu}\alpha\nu\beta} 
= R^{\mu}_{\phantom{\mu}\alpha\nu\beta}  + 
\nabla_\nu \Sep^{\mu}_{\phantom{\mu}\alpha\beta} - \nabla_{\beta} 
\Sep^{\mu}_{\phantom{\mu}\alpha\nu}
+ \Sep^{\mu}_{\phantom{\mu}\nu\rho} \Sep^{\rho}_{\phantom{\rho}\alpha\beta}
- \Sep^{\mu}_{\phantom{\mu}\beta\rho} \Sep^{\rho}_{\phantom{\rho}\alpha\nu}
\label{Riemep}
\end{align}
where $\nabla$ is the Levi-Civita derivative of $g$ and
the tensor $\Sep$ is the difference of the respective
connections of $\gepsilon$ and $g$, explicitly
\begin{align*}
\Sep^{\mu}_{\phantom{\mu}\alpha\beta}=
\frac{1}{2} \gepsilon^{\sharp}{}^{\mu\nu} 
\left ( 
\nabla_{\alpha} \gepsilon_{\nu\beta}+
\nabla_{\beta} \gepsilon_{\nu\alpha}-
\nabla_{\nu} \gepsilon_{\alpha\beta}
\right )
:= \guepsilon
^{\mu\nu}  \Vep_{\nu\alpha\beta}
\end{align*}
where the last equality defines $\Vep$ and the tensor $\guepsilon^{\mu\nu}$
is the contravariant metric associated to  $\gepsilon$.
Recalling that $\gepsilon$
depends differentiably on $\pertp$, that $g_{\pertp=0}=g$, and the definitions
\eqref{def:metricperturbations}, it follows directly from
$\frac{ d \guepsilon ^{\alpha\beta}}{d\pertp}  = - 
\guepsilon^{\alpha\mu}  \guepsilon^{\beta\nu} 
\frac{ d \ge{}_{\mu\nu}}{d\pertp}$ that
\begin{align*}
\left . \frac{d \guepsilon^{\alpha\beta}}{d \pertp} \right |_{\pertp=0}
= - \Kper^{\alpha\beta}, \quad \quad
\left . \frac{d^2 \guepsilon^{\alpha\beta}}{d \pertp^2} \right |_{\pertp=0}
= - \Kperper^{\alpha\beta} + \Kper^{\alpha\mu} \Kper{}_{\mu}^{\phantom{\mu}\beta}.
\end{align*}
We emphasize that all objects are defined in $(\mmm,g)$ and that 
all indices are raised and lowered with the
background metric $g$. Define also
\begin{align}
\Sper_{\mu\alpha\beta} :=
\left . \frac{d \Vep_{\mu\alpha\beta}}{d \pertp} 
\right |_{\pertp=0} =
\frac{1}{2} \left ( \nabla_{\alpha} \Kper{}_{\mu\beta} +
\nabla_{\beta} \Kper{}_{\mu\alpha} - 
\nabla_{\mu} \Kper{}_{\alpha\beta}
\right ), \label{calSper} \\
\Sperper_{\mu\alpha\beta} :=
\left . \frac{d^2 \Vep_{\mu\alpha\beta}}{d \pertp^2} 
\right |_{\pertp=0} =
\frac{1}{2} \left ( \nabla_{\alpha} \Kperper{}_{\mu\beta} +
\nabla_{\beta} \Kperper{}_{\mu\alpha} - 
\nabla_{\mu} \Kperper{}_{\alpha\beta}
\right ), \label{calSperper}
\end{align}
from which it  follows directly
\begin{align*}
\left . \frac{d \Sep^{\mu}_{\phantom{\mu} \alpha\beta}}{d \pertp} 
\right |_{\pertp=0} =
\Sper^{\mu}_{\phantom{\mu}\alpha\beta}, 
\quad \quad
\left . \frac{d^2 \Sep^{\mu}_{\phantom{\mu} \alpha\beta}}{d \pertp^2} 
\right |_{\pertp=0} =
\Sperper^{\mu}_{\phantom{\mu}\alpha\beta} 
- 2 \Kper^{\mu\nu} \Sper_{\nu\alpha\beta}.
\end{align*}
Taking the first and second derivative of (\ref{Riemep}) with 
respect to $\pertp$ at $\pertp =0$, and using that 
$\Sep |_{\pertp=0}=0$, the following expressions are directly
obtained
\begin{align}
\left . \frac{ d \Repsilon^{\mu}_{\phantom{\mu}\alpha\nu\beta} }{d \pertp}
\right |_{\pertp=0} = & 
\nabla_\nu \Sper^{\mu}_{\phantom{\mu}\alpha\beta} 
- \nabla_{\beta}  \Sper^{\mu}_{\phantom{\mu}\alpha\nu} ,
\label{Riemann_first} \\
\left . \frac{ d^2 \Repsilon^{\mu}_{\phantom{\mu}\alpha\nu\beta} }{d \pertp^2}
\right |_{\pertp=0} = & 
\nabla_\nu \left ( \Sperper^{\mu}_{\phantom{\mu}\alpha\beta} 
- 2 \Kper^{\mu\rho} \Sper_{\rho\alpha\beta} \right )
-
\nabla_\beta \left ( \Sperper^{\mu}_{\phantom{\mu}\alpha\nu} 
- 2 \Kper^{\mu\rho} \Sper_{\rho\nu\alpha} \right ) \nonumber \\
& + 2 \Sper^{\mu}_{\phantom{\mu}\nu\rho} \Sper^{\rho}_{\phantom{\rho}\alpha\beta}
- 2 \Sper^{\mu}_{\phantom{\mu}\beta\rho} \Sper^{\rho}_{\phantom{\rho}\alpha\nu}.
\label{Riemann_interme}
\end{align}
We can elaborate (\ref{Riemann_interme}) by expanding the second terms
in the parentheses and inserting 
$\nabla_{\mu} \Kper{}_{\alpha\beta} = \Sper_{\alpha\beta\mu}
+ \Sper_{\beta\alpha\mu}$. The result is
\begin{align}
\left . \frac{ d^2 \Repsilon^{\mu}_{\phantom{\mu}\alpha\nu\beta} }{d \pertp^2}
\right |_{\pertp=0} = &
\nabla_\nu \Sperper^{\mu}_{\phantom{\mu}\alpha\beta} 
- \nabla_\beta \Sperper^{\mu}_{\phantom{\mu}\alpha\nu} 
+ 2 \Kper^{\mu\rho} \left ( \nabla_{\beta} \Sper_{\rho\nu\alpha} - 
\nabla_{\nu} \Sper_{\rho\alpha\beta} \right ) \nonumber \\
& + 2 \Sper^{\rho}_{\phantom{\rho}\nu\alpha} \Sper_{\rho}{}^{\mu}{}_{\nu} 
- 2 \Sper^{\rho}_{\phantom{\rho}\beta\alpha} \Sper_{\rho}{}^{\mu}{}_{\nu}.
\label{Riemann_second}
\end{align}
From (\ref{Riemann_first}) and (\ref{Riemann_second}), the first and second order
perturbations of the Ricci tensor are obtained by simply
contracting the $\mu$ and $\nu$ indices, namely
\begin{align}
  R^{(1)}_{\alpha\beta}   := 
  & \left . \frac{ d \Repsilon_{\alpha\beta} }{d \pertp}
    \right |_{\pertp=0}  = \nabla_\mu \Sper^{\mu}_{\phantom{\mu}\alpha\beta} 
    - \nabla_\beta \Sper^{\mu}_{\phantom{\mu}\alpha\mu} \label{Ricpera} \\ 
  & = \frac{1}{2}
    \left ( 
    \nabla_{\mu} \nabla_{\alpha} {\Kper}^{\mu}_{\phantom{\mu}\beta}
    + \nabla_{\mu} \nabla_{\beta} {\Kper}^{\mu}_{\phantom{\mu}\alpha}
    - \nabla_{\mu} \nabla^{\mu} {\Kper}_{\alpha\beta}
    - \nabla_{\alpha} \nabla_{\beta} {\Kper}^{\mu}_{\phantom{\mu}\mu}
    \right ), \label{Ricper} \\
  R^{(2)}_{\alpha\beta}  := 
  & \left . \frac{ d^2 \Repsilon_{\alpha\beta} }{d \pertp^2}
    \right |_{\pertp=0}  \nonumber \\
  = & \frac{1}{2}
    \left ( 
    \nabla_{\mu} \nabla_{\alpha} {\Kperper}^{\mu}_{\phantom{\mu}\beta}
    + \nabla_{\mu} \nabla_{\beta} {\Kperper}^{\mu}_{\phantom{\mu}\alpha}
    - \nabla_{\mu} \nabla^{\mu} {\Kperper}_{\alpha\beta}
    - \nabla_{\beta} \nabla_{\alpha} {\Kperper}^{\mu}_{\phantom{\mu}\mu}
    \right ) \nonumber  \\
  & +\frac{1}{2} \nabla_{\beta} \nabla_{\alpha} \left ( 
  {\Kper}^{\mu\rho} {\Kper}_{\mu\rho} \right )
  -  \left ( \nabla_{\beta} {\Kper}^{\mu\rho} \right )
  \left ( \nabla_{\alpha} {\Kper}_{\mu\rho} \right )\nonumber \\
  & + {\Kper}^{\mu\rho}
    \left (
    \nabla_{\mu} \nabla_{\rho} {\Kper}_{\alpha\beta}
    - \nabla_{\mu} \nabla_{\alpha} {\Kper}_{\rho\beta}
    - \nabla_{\mu} \nabla_{\beta} {\Kper}_{\rho\alpha} \right ) \nonumber \\
  & + 2 \Sper^{\rho}_{\phantom{\rho}\mu\alpha} \Sper_{\rho}{}^{\mu}{}_{\beta} 
    - 2 \Sper^{\rho}_{\phantom{\rho}\beta\alpha} \Sper_{\rho}{}^{\mu}{}_{\mu},
    \label{Ricperper}
\end{align}
where we have inserted (\ref{calSper})-(\ref{calSperper})
and in  the second  expression we have also used

\begin{align*}
2{\Kper}^{\mu\rho} \nabla_{\beta} \Sper_{\rho\mu\alpha} 
= &{\Kper}^{\mu\rho} 
\nabla_{\beta} \left ( \Sper_{\rho\mu\alpha} 
+ \Sper_{\mu\rho\alpha} \right )\\
  = &{\Kper}^{\mu\rho} \nabla_{\beta} \nabla_{\alpha} {\Kper}_{\mu\rho}
  = \frac{1}{2} \nabla_{\beta} \nabla_{\alpha} \left ( 
  {\Kper}^{\mu\rho} {\Kper}_{\mu\rho} \right )
  -\left ( \nabla_{\beta} {\Kper}^{\mu\rho} \right )
  \left ( \nabla_{\alpha} {\Kper}_{\mu\rho} \right ).
\end{align*}
Expression (\ref{Ricperper}) is advantageous over 
 alternative forms because it is manifestly symmetric in
 $\alpha, \beta$.

\subsection{Perfect fluid source}
Let us now assume that the matter content
of the perturbed scheme is a 
perfect fluid, that is,
the energy momentum tensor $ \widehat T_{\pertp}$ at each
$(\mmm_\pertp,\gfam_\pertp)$ has  the form 
\begin{equation}
  \widehat T_{\pertp} - \frac{1}{2} (\mbox{tr}_{\gfam_\pertp}  \widehat T_{\pertp}) \gfam_\pertp=
  ( \hat \energy_\pertp + \hat \pressure_\pertp) \bm{\hat u}_\pertp \otimes\bm{\hat u}_\pertp
  + \frac{1}{2}(\hat\energy_\pertp -\hat\pressure_\pertp) \gfam_{\pertp{}},
  \label{eq:hat_Tepsilon}
\end{equation}
where $\bm{\hat u}_\pertp$ is the ($\gfam_{\pertp{}}$-unit) one-form fluid flow,
and $\hat \energy_\pertp$ and $\hat \pressure_\pertp$
the mass-energy density and pressure. 
These expressions are pullbacked onto $(\mmm,g)$ as
\begin{align}
  \Tepsilon - \frac{1}{2} (\mbox{tr}_{\gepsilon} \Tepsilon) \gepsilon
  = ( \denepsilon + \preepsilon ) \udepsilon \otimes \udepsilon
  + \frac{1}{2} (\denepsilon - \preepsilon)  \gepsilon,
  \label{eq:Tepsilon}
\end{align}
where, in particular, $\udepsilon:=\psi^*_\pertp(\bm{\hat u}_\pertp)$.
The vectors (in contravariant form) $\hat u_\pertp$ are pushforwarded  through $\psi_\pertp^{-1}$
to a family of fluid vectors $u_\pertp:=d\psi_{\pertp}^{-1}(\hat u_\pertp)$.
It is immediate that $\udepsilon=\gepsilon(\uepsilon,\cdot)$,
$\udepsilon(\uepsilon)=-1$ hold.

The field equations of the perturbed scheme are 
$\widehat{\mbox{Ein}}_{\hat g_\pertp}{}_{\alpha\beta}=\chiefes \widehat T_\pertp{}_{\alpha\beta}$, 
and are pullbacked onto $(\mmm,g)$, and rearranged,  as
\begin{align}
\Ricepsilon = \chiefes (\Tepsilon
- \frac{1}{2} (\mbox{tr}_{\gepsilon} \Tepsilon) \gepsilon).
\label{EFEsepsilon}
\end{align}
Define
\begin{align*}
\denper := \left . \frac{d \denepsilon}{d \pertp}  \right |_{\pertp=0},
\quad
\preper := \left . \frac{d \preepsilon}{d \pertp}  \right |_{\pertp=0},
\quad
\uper :=
\left . \frac{d \uepsilon}{d \pertp}  \right |_{\pertp=0},
\quad \\
\denperper := \left . \frac{d^2 \denepsilon}{d \pertp^2}  \right |_{\pertp=0},
\quad
\preperper := \left . \frac{d^2 \preepsilon}{d \pertp^2}  \right |_{\pertp=0},
\quad
\uperper :=
\left . \frac{d^2 \uepsilon}{d \pertp^2}  \right |_{\pertp=0}.
\end{align*}
From $\udepsilon= \gepsilon (\uepsilon,\cdot)$ the
perturbations of the fluid velocity one-forms are
\begin{align}
\left . \frac{d \udepsilon}{d \pertp} \right |_{\pertp =0}
& = \Kper(u, \cdot) + \bm{\uper},  \label{oneformfir} \\
\left . \frac{d^2 \udepsilon}{d \pertp^2} \right |_{\pertp =0}
& = \Kperper(u, \cdot) + 2 \Kper(\uper, \cdot) + \bm{\uperper},
\label{oneformsec}
\end{align}
where $u := \uepsilon|_{\pertp=0}$ is the background fluid velocity vector.
The normalisation condition $\udepsilon (\uepsilon) =-1$ implies,
upon taking successive $\pertp$ derivatives at $\pertp=0$,
the two algebraic constraints
\begin{align}
& 2 \bm{u} (\uper)  + \Kper(u,u) =0, \label{algeb1} \\
& \Kperper(u, u) + 4 \Kper(\uper, u) 
+ 2 \bm{\uper}(\uper) + 2 \bm{u} (\uperper)
=0, \label{algeb2}
\end{align}
which determine the components of $\uper$ and $\uperper$ along $u$.
The perturbed Einstein field equations arise from the
$\pertp$ derivatives of (\ref{EFEsepsilon}) with (\ref{eq:Tepsilon}),
and yield
\begin{align}
R^{(1)}_{\alpha\beta} = &
\chiefes\left ( \denper + \preper \right ) u_{\alpha} u_{\beta}
+ 
\chiefes (\den + \pre) \left (  
\left ( {\Kper}_{\alpha\mu} u^{\mu} + \uper_{\alpha} \right )  u_{\beta}
+ \left ( {\Kper}_{\beta\mu} u^{\mu} + \uper_{\beta} \right )  u_{\alpha}
\right )  \nonumber \\
& + \frac{1}{2} \chiefes\left (  \denper- \preper \right ) g_{\alpha\beta} 
+ \frac{1}{2} \chiefes\left (\den- \pre  \right ) {\Kper}_{\alpha\beta},
\label{firstRic}
\end{align}
after using (\ref{oneformfir}) and (\ref{algeb1}).
The second order equations are similarly obtained from the second  derivative
of (\ref{EFEsepsilon}) and
using
(\ref{oneformfir})-(\ref{algeb2}),
\begin{align}
R&^{(2)}_{\alpha\beta} = 
\chiefes \left ( \denperper + \preperper \right ) u_{\alpha} u_{\beta}
+ 2 \chiefes(\denper + \preper) \left (  
\left ( {\Kper}_{\alpha\mu} u^{\mu} + \uper_{\alpha} \right )  u_{\beta}
+ \left ( {\Kper}_{\beta\mu} u^{\mu} + \uper_{\beta} \right )  u_{\alpha}
\right )  \nonumber \\
& + \chiefes (\den + \pre ) \left (
 \left ( {\Kperper}_{\alpha\mu} u^{\mu} + 2 {\Kper}_{\alpha\mu} \uper{}^{\mu}
+ \uperper_{\alpha} \right ) u_{\beta} + 
\left ( {\Kperper}_{\beta\mu} u^{\mu} + 2 {\Kper}_{\beta\mu} \uper{}^{\mu}
+ \uperper_{\beta} \right ) u_{\alpha}  \right .   \nonumber \\
& \hspace{2cm}  \left . 
+ 2 \left ( \uper_{\alpha} + {\Kper}_{\alpha\mu} u^{\mu} \right )
\left ( \uper_{\beta} + {\Kper}_{\beta\mu} u^{\mu} \right ) \right ) \nonumber \\
& + \frac{1}{2} \chiefes\left ( \denperper - \preperper \right ) g_{\alpha\beta} 
+\chiefes \left ( \denper - \preper \right ) {\Kper}_{\alpha\beta} 
+ \frac{1}{2} \chiefes\left ( \den - \pre \right ) {\Kperper}_{\alpha\beta}.
\label{SecOrderfield}
\end{align}
Let us now assume that the spacetime $(\mmm,g)$ admits
a hypersurface orthogonal timelike Killing vector $\stat$ and an axial
Killing vector $\axial$. We assume further
that the perturbation scheme inherits the local
symmetry generated by $\stat$ and $\axial$
and that for each $\pertp$,
the spacetime
$(\mmm_\pertp,\gfam_\pertp)$ is a solution of the
Einstein's field equations for a rigidly rotating perfect fluid, i.e.
that there exists a constant (on each $\pertp$)
$\kapepsilon$ and a positive
function $\hat N_\pertp\in C^{n+1}(\mmm_\pertp)$ such that
\begin{equation}
\hat u_\pertp=\hat N_\pertp(\hat\stat_\pertp + \kapepsilon \hat\axial_\pertp),
\label{eq:hat_rigid_rotation_epsilon}
\end{equation}
where $\hat\stat_\pertp\defi d\psi_\pertp(\stat)$ and
$\hat\axial_\pertp\defi d\psi_\pertp(\axial)$.
The pullback of the equations and this relation is equivalent to assume that
the spacetime
$(\mmm,\gepsilon)$ is a solution of the
Einstein's field equations \eqref{EFEsepsilon} with \eqref{eq:Tepsilon} and
\begin{align*}
\uepsilon = \Nepsilon \left ( \xi + \kapepsilon \eta \right )
\end{align*}
for some positive function $\Nepsilon \in C^{n+1}(\mmm)$.

Staticity of the background imposes that $u$ is parallel to $\xi$ and then
\eqref{metric} implies
\begin{align}
u = e^{-\frac{\nu}{2}}  \xi, \quad \quad \quad
\Longleftrightarrow \quad \quad \quad
\Nepsilon |_{\pertp=0} = e^{-\frac{\nu}{2}}, \quad \quad 
\kapepsilon|_{\pertp=0} =0.
\label{back}
\end{align}
In terms of the following quantities
\begin{align*}
  \Omegaper := \left . \frac{d \kapepsilon}{d \pertp} \right |_{\pertp=0},
\quad \quad
\Omegaperper := \left . \frac{d^2 \kapepsilon}{d \pertp^2} \right |_{\pertp=0},
\quad \quad
\uper^0 := \left  .\frac{d \Nepsilon}{d \pertp} \right |_{\pertp=0},
\quad \quad
\uperper^0 := \left . \frac{d^2 \Nepsilon}{d \pertp^2} \right |_{\pertp=0},
\end{align*}
the first and second order pertubation fluid velocity vectors are, 
using (\ref{back}),
\begin{align*}
\uper & = \uper^0 \xi + e^{-\frac{\nu}{2}}  \Omegaper \eta, \\
\uperper & = \uperper^0 \xi + \left (
2 \uper^0 \Omegaper + e^{-\frac{\nu}{2}}  \Omegaperper \right ) \eta.
\end{align*}
As already mentioned, the components $\uper^0$ and $\uperper^0$ are determined by
the algebraic constraints (\ref{algeb1})-(\ref{algeb2}) and 
$\uper$ and $\uperper$ are written in terms of
the metric perturbation
tensors and the constants $\Omegaper$ and $\Omegaperper$ as
\begin{align}
\uper  = & \frac{1}{2} e^{-\frac{3 \nu}{2}} \Kper(\xi,\xi) \xi 
+  e^{-\frac{\nu}{2}} \Omegaper \eta, \label{uper} \\
\uperper = &  e^{-\frac{3\nu}{2}} 
\left ( \frac{1}{2} \Kperper(\xi,\xi) + \frac{3}{4}
e^{-\nu} \Kper(\xi,\xi)^2
+ 2 \Omegaper  \Kper(\xi,\eta)
+ \Omegaper{}^2 \la \eta,\eta \ra \right ) \xi \nonumber \\
& + e^{- \frac{\nu}{2}}
\left ( e^{-\nu} \Kper(\xi,\xi) \Omegaper + \Omegaperper \right ) \eta.
\label{uperper}
\end{align}

We now exploit a well-known relation between orthogonal
transitivity of the Abelian group action and rigid rotation of the
self-gravitating fluid, forced upon by the Einstein field equations
(see e.g. \cite[Chapter 19.2]{Exact_solutions}). In
our present set up, the specific result we need is as follows.
\begin{proposition}[\textbf{Rigid rotation and orthogonal transitivity}]
  \label{prop:rigidrot_OT}
  Let $(\mmm,g)$ be a spacetime with $C^{n+1}$ ($n\geq 3$) metric that admits an Abelian
  $G_2$ group of isometries 
  generated by $\{\stat,\axial\}$. Assume also that $\mmm$ is simply connected
  and that $\axial$ is an axial symmetry with a non-empty
  set of fixed points.
  Let $(\mmm_\pertp, \gfam_\pertp,\{\psi_\pertp\})$ be 
  a $C^{n+1}$ maximal  perturbation scheme 
  inheriting this group (c.f.  Definition 2.1 in \cite{paper1}).
  If the matter content of the perturbation scheme is that of a
  rigidly rotating perfect fluid (or vacuum),
  then the background is orthogonally transitive
  and the perturbation scheme inherits 
  this property.
\end{proposition}

\begin{proof}
By assumption, for each $\pertp$ the spacetime
$(\mmm_\pertp,\gfam_\pertp)$ is a solution of the
Einstein's field equations with \eqref{eq:hat_Tepsilon} and
\eqref{eq:hat_rigid_rotation_epsilon}. Therefore it is straigforward to check that
the one-forms
$\widehat\Ricepsilon(\hat\stat_\pertp,\cdot)$ and
$\widehat\Ricepsilon(\hat\axial_\pertp,\cdot)$ lie in
$span\{\hat{\bm\stat}_\pertp,\hat{\bm\axial}_\pertp\}$,
where  $\hat{\bm{\axial}}_\pertp=\gfam_\pertp(\hat\axial_\pertp,\cdot)$
and  $\hat{\bm{\stat}}_\pertp=\gfam_\pertp(\hat\stat_\pertp,\cdot)$.
Standard curvature identities 
(see e.g. \cite[Chapter 19.2]{Exact_solutions}) imply that this fact,
in addition to the commutation of $\stat$ and $\axial$, implies that
the functions $O_1{}_\pertp\defi\star (\hat{\bm{\axial}}_\pertp\wedge \hat{\bm{\stat}}_\pertp\wedge d \hat{\bm{\stat}}_\pertp)$ and
$O_2{}_\pertp\defi\star (\hat{\bm{\axial}}_\pertp\wedge \hat{\bm{\stat}}_\pertp\wedge d \hat{\bm{\axial}}_\pertp)$ satisfy \cite{Kundt-Truper}
\[
d O_1{}_\pertp=0,\qquad
d O_2{}_\pertp=0.
\]
Simply connectedness implies that $O_1{}_\pertp$ and $O_2{}_\pertp$ are constant (for each $\pertp$),
while the existence of points where $\axial$ vanishes (the axis) implies they must, in fact,
be zero. Therefore the group generated by $\{\hat\stat_\pertp,\hat\axial_\pertp\}$
on $(\mmm_\pertp, \gfam_\pertp)$ for each $\pertp$ is orthogonal transitive. The result follows, in particular for $\pertp=0$.\fin
\end{proof}

Observe that the fluid quantities are gauge dependent in general.
  In particular, the perturbed pressure at first and second order
  transforms as  (see e.g. \cite{Bruni_et_al_1997})
  \begin{equation}
    \label{eq:gauge_pressure}
          \Pp{}^g=\Pp+\sper(\Pb),\qquad
      \Ppp{}^g=\Ppp+\sperper(\Pb)+ \sper(\Pp + \Pp{}^{g} ).
        \end{equation}
  Eventually, uniqueness of the solutions will rely on one free integration constant associated to the value of the perturbed pressure at the origin. In order
  to assign a clear physical meaning to this parameter it is necesary to show that the perturbed central pressure is a gauge invariant quantity.  
    In the following result
  we prove this fact, to second order, within the  
  perturbation scheme of Theorem \ref{theo:paper1} (which is
    differentiable everywhere in $\mmm$) provided the configuration is 
    equatorially symmetric. Later on we will show that this symmetry
    is a necessary consequence of the field equations.
    \begin{lemma}
  \label{lemma:central_P}
  Within the $C^{n+1}$ 
    ($n\geq 3$) perturbation scheme introduced
      in Theorem \ref{theo:paper1} and with the definitions above,  
  (i) the value $\Pp(0)=\Pp|_{r=0}$ is gauge  invariant and
  (ii) if  the configuration has equatorial symmetry
  then $\Ppp(0)=\Ppp|_{r=0}$ is also gauge invariant.
  \end{lemma}

\def\equator{\mathcal{E}}

\begin{proof}
  The first order $\sper$ and second order $\sperper$
  gauge vectors within the
  perturbation scheme are $C^{n+1}(\mmm)$ and $C^{n}(\mmm)$ respectively
  (see \cite{paper1}).
  On the other hand, $\Pb$ is  $C^{n-1}(\mmm)$.
  Note $\Pp$ and $\Ppp$ are $C^{n-2}(\mmm)$ and $C^{n-3}(\mmm)$ respectively
  and therefore both functions are continuous, in particular, at the origin.
  First, Lemma \ref{res:radial_functions_origin} implies
  $d\Pb$ vanishes at $r=0$.
  Therefore, the term $\sper(\Pb)\equiv d\Pb(\sper)$ vanishes at $r=0$, and
    the claim (i) follows from \eqref{eq:gauge_pressure}.

  The same argument applies to $\sperper(P)$, so in order to show
    gauge invariance of $\Ppp$  \eqref{eq:gauge_pressure} it suffices to prove that
    $\sper(\Pp + \Pp{}^{g} )=0$ at the origin.
        Under the stationary and axisymmetric perturbation scheme
    of Theorem    \ref{theo:paper1}, the functions
    $\Pp$ and $\Pp{}^{g}$ are time independent and axially symmetric. Moreover, by assumption they are equatorially symmetric, i.e. invariant under $z \rightarrow
    -z$. Let $f\in C^{1}(\mmm)$ be any function with these properties
    and consider the equatorially invariant hypersurface
    $\equator := \{ z=0\}$. The restriction $f|_{\equator}$ is independent of $t$ and radially symmetric in $\{x,y\}$. Hence Lemma \ref{res:radial_functions_origin}
    implies that $d(f|_{\equator})$ vanishes at the origin. We can decompose
    uniquely    $\sper=\sper^z+\sper^{\equator}$, with 
  $\sper^z$ along $\partial_z$ and $\sper^\equator$
  tangent to $\equator$.  By equatorial symmetry
  $\sper^z(f) |_{\equator}=0$. Therefore
  $\sper(f)|_{\equator}=\sper^\equator(f)|_{\equator}=\sper^\equator(f|_{\equator})$,
  and thus 
  $\sper(f)|_{\equator}=d (f|_{\equator}) (\sper^\equator)$ vanishes at the origin.
  Applying this fact to $f= \Pp + \Pp{}^{g}$, the claim (ii) follows.\fin
   \end{proof}

So far no equation of state for the perfect fluid has been imposed. 
We shall say that the perturbation scheme satisfies a barotropic
equation of state if there exists a $C^{2}$ function of one variable
$P (E)$ such that, for each value of $\pertp$, the 
pressure and density of the fluid are related by
$\Pb_{\pertp} = P(\Eb_{\pertp})$. Note that we assume that
the equations of state do not depend on $\pertp$, and thus
the barotropic EOS is that of the background. Taking
$\pertp$-derivatives at $\varepsilon = 0$, the perturbed pressures
are written in terms of the perturbed densitities as
\begin{align}
&\Pp - \frac{d P }{d E} \Ep = 0,\nonumber\\
&\Ppp - \frac{d P }{d E} \Epp - \frac{d^2 P}{d E^2}{\Ep}^2  = 0\label{eq:2ndbarotropic}.
\end{align}
where  the derivatives $\frac{dP}{dE}$ etc. are evaluated at the background
density.

\section{``Base'' global perturbation scheme} 
\label{sec:Base_model}

In order to tackle the first and second order problems we have to deal with the first and second order perturbation tensors  as given in (2.11) and (2.12). It is obvious that the problem involves two steps, namely addressing the first order problem first and dealing with the second order one afterwards. However, as mentioned in the Introduction, there is a strategy that allows one to treat both cases at the same time, and this entails a considerable simplification of the proof.

The underlying idea behind our strategy is that a second order perturbation 
problem under the assumption that the first order perturbation tensor vanishes identically is
completely equivalent to a first order problem. This fact is both physically
and geometrically clear, and can be checked explicitly: all
the equations, matching conditions, etc. are identical for
the first order perturbation tensor and for the second order
perturbation tensor after setting the first order tensor to zero
(the only difference  is in the differentiability class, which to second
order is one order lower than to first order, but this poses no problem
as the differentiability assumed at second order suffices
for the argument).

The idea is then to use a bootstrap type of argument. We assume a specific
form for the first order perturbation tensor which includes zero as a 
particular case, and leave the second order tensor completely free
(within our perturbation scheme, naturally). We then analyze the second 
order problem in full detail. After this has been done, we set
the first order perturbation tensor to zero and all the conclusions that 
we find apply immediately to the first order problem. This will allow
us to show that our assumption on the first order perturbation tensor
is in fact a consequence of the first order problem. This will close
the bootstrap and hence all the results obtained under the scheme will be fully
general.
We call the restricted second order problem
the ``base'' perturbation scheme and is built as follows:

\vspace{3mm}

\noindent
\emph{The base perturbation scheme:}
The global manifold consists of an interior region ($M^+$) and an exterior
region ($M^-$) separated by a hypersurface $\Supfamp$. 
We construct a  second order perturbation 
for each region
around a background configuration (we drop $\pm$ indices) satisfying
items {\baseback}-{\baseKperper} below.
We then construct the global model
 by solving the most general perturbed matching problem (item {\basematch}).
We finally consider a \emph{barotropic base scheme}
(and will explicitly specify \emph{barotropic}) when
the barotropic EOS of the background is assumed in the interior,
i.e. items {\baseback}-{\basebaro} below hold.

\begin{itemize}
\item[\baseback:] \emph{The background corresponds to a
    finite perfect fluid ball configuration
    according to Definition \ref{def:background}
    with $\Ebc+\Pbc\neq 0$}.
  
  \emph{Remark:} In particular,
  the metric, which we take to be $C^{n+1}$ on each region with $n\geq 4$, 
  is given by (\ref{metric}) with $\Q(r)=r$, and
  $\stat := \partial_t$, $\axial:= \partial_{\phi}$. Define
  $\Sig_{\alpha\beta}\defi \stat_\alpha \axial_\beta+\stat_\beta\axial_\alpha $,
  so that
  $\Sig= -2e^{\nu} r^2\sin^2\theta dtd\phi$
  in those coordinates.
\item[\baseKper:] \emph{The first order metric perturbation tensors $\Kper^\pm$ are bounded and}
  \begin{itemize}
  \item[\baseKper.1:] \emph{read (dropping $\pm$ indices)
      \begin{equation}
        \Kper =  \opertbase e^{-\nu}\Sig, 
        \label{Kper_cov}
      \end{equation}
      for given functions
      $\opertbase_+\in C^{n+2}(\mmm^+\setminus \centresph_0)\cap C^{2}(\mmm ^+)$ and
      $\opertbase_-\in C^{\infty}(\mmm^-)$,
      both radially symmetric and bounded.} 
    
    \emph{Remark:} In spherical coordinates, this assumption translates onto
    \begin{equation}
      \Kper = -2 \opertbase(r) r^2  \sin^2 \theta dt d \phi,
      \label{Kper}
    \end{equation}
    where $\opertbase_-(r)$ is $C^{\infty}([a,\infty))$,
    and, by virtue of  Lemma \ref{origin},
    $\opertbase_+(r)$ is $C^{2}([0,a]) \cap C^{n+2}((0,a])$, and admits
    the decomposition
    \begin{equation}
      \opertbase_+(r)=\opertbase_0+\opertbase_2r^2+\Phi_\opertbase^{(2)}(r),
      \label{omega_origin}
    \end{equation}
    where $\opertbase_0,\opertbase_2\in\mathbb{R}$ and $\Phi_\opertbase^{(2)}(r)\in C^2([0,\ro])$
    and $o(r^2)$.
    
  \item[\baseKper.2:] \emph{The functions $\opertbase_\pm(r)$
      satisfy the equation (we drop the $\pm$ signs)
      \begin{align}
        & \frac{1}{r^3}\frac{d}{dr}\left( r^4 j \frac{d\opertbase}{dr}\right) + 4 j' (\opertbase - \Omegaperbase) = 0 \quad  \Longleftrightarrow \nonumber \\
        &
          \opertbase'' = ( \lambda' + \nu') \left ( \frac{1}{2} \opertbase'
          + \frac{2}{r} (\opertbase - \Omegaperbase) \right ) - \frac{4}{r} \opertbase'
          \label{eqfo}
      \end{align}
      where $\Omegaperbaseint\in\mathbb{R}$ in $M^+$ and $\Omegaperbaseext=0$ in $M^-$.}

    \emph{Remark:} In $\mmm^-$ we have $j'=0$ and thus the value of $\Omegaperbaseext$
    is irrelevant, so we fix it to zero for definiteness. The general solution on $\mmm^-$ is given by
    \begin{equation*}
      \opertbase_-(r)=\frac{2J_\opertbase}{r^3} +\opertbase_\infty,
      \quad \mbox{ with } \quad  J_\opertbase ,\opertbase_\infty
      \in \mathbb{R}.
          \end{equation*}
    The vector fields $r^{-1} \eta$ and $\xi$ are
        smooth and bounded in $\mmm^-$. Thus, boundedness of $\Kper$  demands that
      $\Kper(\stat, r^{-1}\axial)=-\opertbase(r)e^{-\nu}r \sin^2\theta$ is  also bounded. This condition clearly requires $\opertbase_\infty=0$.
      and the function $\opertbase$ is
      \begin{equation}
\opertbase_-(r)=\frac{2J_\opertbase}{r^3}, \qquad J_{\opertbase} \in \mathbb{R}.        \label{eq:opert_ext_sol}
\end{equation}
       
    By Proposition \ref{prop:full_class_gauges}, a first order gauge
    transformation with $\sper^- = C^- t \partial_{\phi}$ changes
    $\opertbase_-\to \opertbase_- -C^-$. Only $C^-=0$ respects the condition  $\opertbase_\infty=0$ and we conclude that boundedness
      of $\Kper$ fixes the first order gauge freedom in the exterior completely.
  
    In the interior region $\mmm^+$, equation \eqref{eqfo} combined with
      \eqref{background_expansion_origin} determines
    \begin{equation}
      \opertbase_2=\frac{2}{5}(\opertbase_0-\Omegaperbaseint)(\lambda_2+\nu_2).
      \label{eq:sol_for_w2}
    \end{equation}
    Clearly, if $\opertbase_+=0$ then  \eqref{eqfo} implies $\Omegaperbaseint=0$.
    
    \end{itemize}
  We introduce an auxiliary number\footnote{This parameter is not to
be confused with the function $m$. The context will clarify the intended meaning.} $m$, restricted to $n\geq m \geq 2$,
    that prescribes the differentiability of $\Kperper$, and assume the outcome of Theorem \ref{theo:paper1}.

\item[\baseKperper:] \emph{The second order metric perturbation tensors $\Kperper^\pm$ satisfy:}
  \begin{itemize}
  \item[\baseKperper.1:]  $\Kperper^\pm$ \emph{are $C^{m}$  (in $\mmm^+\setminus\centresph_0$ and $\mmm^-$, respectively) and bounded everywhere,
      and in spherical coordinates (dropping $\pm$ indices) have the form
      \begin{align}
        \Kperper = &\left ( -4 e^{\nu(r)} h(r,\theta) +
                     2 \opertbase^2(r) r^2 \sin^2 \theta  \right ) dt^2
                     + 4 e^{\lam(r)} m(r,\theta) dr^2 \nonumber \\
                   &+ 4  k(r,\theta) r^2 \left ( d \theta^2 + \sin^2 \theta
                     d \phi^2 \right )
                     + 4  e^{\lam(r)} \partial_{\theta} f (r,\theta) r dr d \theta\nonumber\\
                   &-2 \Wtwo (r,\theta) r^2 \sin^2\theta  dt d\phi\label{Kperper}\\
        =:&\spt^H+ \Wtwo e^{-\nu}\Sig,\label{KperperH}
      \end{align}
      where the functions in \eqref{Kperper} correspond to the traces (in $\{x,y\}$) of 
      the axially symmetric functions with same name, which satisfy}
    \textit{
      \begin{itemize}
      \item $h_+,m_+,k_+\in C^{m+1}(\mmm^+\setminus\centresph_0)$ and bounded near $\centresph_0$,\quad $h_-,m_-,k_-\in C^{m+1}(\mmm^-)$
      \item $\Wtwo_+ \in C^{m}(\mmm^+\setminus\centresph_0)$ and bounded near $\centresph_0$,
        $\Wtwo_-\in  C^{m}(\mmm^-)$,
      \item the vectors $\Wtwo \axial_\pm $ are $C^{m+1}(\mmm^+\setminus\centresph_0)$ and $C^{m+1}(\mmm^-)$ respectively.
      \item $f_+\in  C^{m}(\mmm^+\setminus\centresph_0)$ and bounded near $\centresph_0$,
        $f_-\in  C^{m}(\mmm^-)$, both $f_\pm$ are $C^{m+1}(S_r)$ on all spheres $S_r$,
        all $\partial_ rf_\pm$ and $\partial_t f_\pm$ are $C^{m}(S_r)$ on all spheres $S_r$,
        and finally $\partial_\theta f_\pm$ are $C^{m}$ outside the axis $\axis$ and
        extend continuously to  $\axis\setminus\centresph_0$ where they vanish.
      \end{itemize}
    }
  \item[\baseKperper.2:] \emph{$\Kperper^-$ solves the second order field equations
      \eqref{SecOrderfield} for vacuum and
      $\Kperper^+$ solves the second order field equations
      \eqref{SecOrderfield}
      for a 
      rigidly rotating perfect fluid
      \eqref{uperper} with $\Omegaper=\Omegaperbaseint$ and $\Omegaperper=\Omegaperperbaseint \in \mathbb{R}$}.
  \end{itemize}
\item[\basematch:] \emph{The first and second order perturbed matching conditions (c.f. Appendix \ref{app:staxi_perturbed_matching}) 
    hold on $\Supfamp$.
  Moreover we assume $[\opertbase]=0$.}

\emph{Remark:} The first order perturbed matching conditions demand $[\opertbase]=b_1 \in\mathbb{R}$ (see Proposition \ref{teo:first_order_matching}).
  By Proposition \ref{prop:full_class_gauges} a first order gauge
  transformation with $\sper^+ = C^+ t \partial_{\phi}$ changes
  $\opertbase_+\to \opertbase_+ -C^+$. Given that $C^-$ has already been fixed to zero, the condition $[\opertbase]=0$ is fulfilled if and only if
  $C^+$ is chosen to be $C^+=b_1$. We conclude that (i) the assumption $[\opertbase]=0$ entails no loss of generality, and (ii) that this condition together with boundedness at infinity fixes completely the gauge freedom at first order, c.f. Proposition \ref{prop:full_class_gauges}.

\item[\basebaro:] \emph{The perfect fluid satisfies the background
      barotropic equation of state.}
\end{itemize}

Some additional remarks are in order.
We first stress that we \emph{ do not assume equatorial symmetry}.
The assumption of boundedness on $\Kperper$
will not play any role until we tackle the global problems in Section \ref{sec:e_u_global}.
Finally,
proving that \baseKper.1{} and \baseKper.2{} hold necessarily,
and that there exist indeed solutions of \eqref{eqfo} in
$C^2([0,\ro))\cap C^{n+2}((0,\ro))$,
will be part of the bootstrap argument.

Observe that the full set of gauge transformations compatible with
        the base scheme is given by Proposition \ref{prop:full_class_gauges}, restricted to $C^+=C^-=0$, or, in the Notation \ref{def:gauge_classes}, the class $\{\gaugeABYa\}$.

We introduce at this point some relevant
definitions and notation that refer to the existence of two problems,
the interior $(+)$ and the exterior $(-)$.
Let $D = \mathbb{R}^3\setminus \{0\}$, $D^+ = \overline{B_a} \setminus \{0\}$ and
$D^- = D \setminus B_a$, where $B_a$ is the ball of radius $a>0$ centered
at the origin. Since we deal with interior and exterior functions that take different values at the boundary we also introduce the
  disjoint union $\Du := D^+ \sqcup D^-$ endowed with the disjoint union topology. Let $\{r,\theta,\phi\}$ be standard spherical coordinates on $\Du$.
For each $r>0$, we let $S_r := \{ |x| = r \}$.
Observe both $D^\pm$ contain $S_a$ and  that $\Du$ constains two
  copies thereof. We shall use the following notation for the geometry
  of $S_r$.

\begin{notation}[Notation in $S_r$]
  \label{Notation:Sr}
  We endow $S_r$ with the standard metric of radius one
  $\gsph = d \theta^2 + \sin^2 \theta d \phi^2$. We fix the orientation of $S_r$ so that 
$\{ \partial_{\theta}, \partial_{\phi} \}$ is
positively oriented and denote by 
$\volform$ the corresponding volume form. The Hodge dual
on $p$-forms of $(S_r,\gsph)$ is denoted by 
$\star_{\mathbb{S}^2}$ and we define
$\axialfSph := \gsph(\eta, \cdot)$ 
i.e. the metrically related one-form of 
$\eta= \partial_{\phi}$. The covariant derivative associated to
$\gsph$ is  $\D$, the corresponding Laplacian
is $\Delta_{\mathbb{S}^2}$ and tensors on $(S_r, \gsph)$ carry capital Latin indices $A,B, \cdots$.
\end{notation}

We define $P_{\ell}$ to be the Legendre polynomial of order $\ell$ on
$(S_r, \gsph)$. More precisely, $P_{\ell}$ is the only solution invariant under $\eta = \partial_{\phi} $
of the eigenvalue problem
$(\Delta_{\mathbb{S}^2}
+  \ell( \ell +1)) P_{\ell} =0$, $\ell \in \mathbb{N}$, 
with the
normalization choice
$\int_{\mathbb{S}^2} P_{\ell} P_{\ell'} \bm{\eta}_{\mathbb{S}^2} =
\frac{4\pi}{2 \ell+ 1} \delta_{\ell \ell'}$
and satisfying $P_{\ell} >0$ on the
north pole (defined by $\theta=0$). 
The first three Legendre polynomials are
$P_0=1$, $P_1 = \cos \theta$, $P_2 = \frac{1}{2} ( 3 \cos^2 \theta -1 )$.
Given any function $f : \Du \rightarrow \mathbb{R}$
we define the following `components'
\begin{align}
f_{\ell} (r)\defi \frac{2\ell+1}{4\pi}  \int_{S_r} f P_{\ell} \volform.
\label{fell}
\end{align}
We emphasize that the integration is on $S_r$ with the volume form of
the standard round metric of radius one (in particular $\int_{S_r} \volform
= 4 \pi$). For functions $f$ independent of $\phi$, 
an alternative equivalent definition is
$f_{\ell}(r) := \frac{2\ell+1}{2}\int_0^\pi f(r,\theta) P_\ell(\cos\theta)\sin\theta d\theta$.
Whenever the $\ell$ subindices can lead to confusion, we will also use
$f_{(0)}$, $f_{(1)}, f_{(2)}$ etc.

We will use the same name for objects $F_\pm$ defined respectively on 
$\mmm^{\pm}$ invariant under $\xi$  and the corresponding objects $F_\pm$ 
defined on $D^+$ and $D^-$. Any object
  $F$ defined on $\Du$ is said to be  composed by $\{F_+,F_-\}$ when $F|_{D^{\pm}} = F_\pm$. Viceversa, given $\{ F_{+}, F_{-}\}$ defined on $D^{\pm}$ we will use $F$ to refer to the object defined on $\Du$ whose restriction to $D^{\pm}$ is $F_{\pm}$. For scalar functions, whenever $F_+\in C^{m_1}(D^+)$ and $F_-\in C^{m_2}(D^-)$ we will equivalently write $F\in C^{m_1}(D^+)\cap C^{m_2}(D^-)$. Moreover, if $[F]=0$ then $F$ is  also well defined on $D$,
  and if $F\in C^0(D^+)\cap C^0(D^-)$ then $F\in C^0(D)$. Observe this extends
  to any function and any of its (partial) derivatives iterativelly.
This notation will also translate to the corresponding intervals on the real line
for the coordinate $r$.

\subsection{Field equations for the base perturbation scheme}
\label{sec:2nd_order_EFEs}
In this section we write
down explicitly the field equations for the second order of the
base perturbation scheme.
More specifically,
we develop the point \baseKperper.2 of the base perturbation scheme  under assumptions
\baseback{} and \baseKper{}. 
We use the results introduced in Section \ref{PertEFEs}
combined with Appendix \ref{tfi_sector},
where we derive with covariant methods
the first order perturbed Ricci tensor
for a first order perturbation of the form $\Kper = \y \Sig$, with $\y$
depending on $r,\theta$.

We start with the first order problem. Expression
(\ref{uper}) and $\Kper(\xi,\xi)=0$ impose
$\uper = \Omegaperbase e^{-\frac{\nu}{2}} \eta$, where 
the redefinition of constants 
$\Omegaper \rightarrow \Omegaperbase$ has been made.
 Thus
\begin{equation}
{\Kper}_{\alpha\mu} u^{\mu} + \uper_{\alpha} = e^{-\nu/2} ( \Omegaperbase - \opertbase ) \eta_{\alpha}.
\label{Comb1}
\end{equation}
By Proposition \ref{prop:R1} and the notation introduced in Remark
  \ref{rem:prop:R1}, the first order perturbed Ricci tensor
  of \eqref{Kper}
(defined in \eqref{Ricper})  has the form
$R^{(1)}_{\alpha\beta}=  \riccifunc(\opertbase e^{-\nu})\Sig_{\alpha\beta}$. Inserting into
\eqref{firstRic} shows that the first order Einstein field equations
require 
  \begin{equation}
    \denper=\preper=0
    \label{eq:EpPp_zero}.
  \end{equation}

For the second order problem, it is advantageous to split the
second order Ricci tensor  $ R^{(2)}_{\alpha\beta}$ 
of (\ref{Kper}) and (\ref{Kperper}) into two terms. Define
$R^{(2)H}_{\alpha\beta}$
 as $ R^{(2)}_{\alpha\beta}$ computed with 
$\spt^H$ in (\ref{KperperH}), i.e. 
$R^{(2)H}_{\alpha\beta}:=R^{(2)}_{\alpha\beta}(\Wtwo =0)$.
Then, by virtue of (\ref{Ricperper}) and (\ref{Ricper}) together with Proposition \ref{prop:R1}
we have
\begin{align}
&R^{(2)}_{\alpha\beta}= R^{(2)H}_{\alpha\beta}+\riccifunc(\Wtwo  e^{-\nu})\Sig_{\alpha\beta}.
\label{eq:R2_common}
\end{align}
As for the right hand side of (\ref{SecOrderfield}), we start by computing
the second order perturbation vector  $\uperper$. We use
$\Kper(\xi,\eta)= - \opertbase
\la \eta, \eta \ra$ and $\Kperper(\xi,\xi) =
- 4 e^{\nu} h + 2\opertbase^2 r^2 \sin^2 \theta  $, so that 
(\ref{uperper}) yields, after the redefinition of constants
$\Omegaperper \rightarrow \Omegaperperbase$,
\begin{align}
\uperper = e^{-\frac{3\nu}{2}} 
\left ( - 2 e^{\nu} h + ( \opertbase - \Omegaperbase)^2 r^2 \sin^2 \theta  \right )
\xi + e^{- \frac{\nu}{2}} \Omegaperperbase \eta.
\label{eq:u2}
\end{align}
One immediately finds
\begin{align}
\Kperper^H{}_{\alpha\mu} u^{\mu} + 2 {\Kper}_{\alpha\mu} \uper^{\mu}
+ \uperper_{\alpha} = \left ( 2 e^{\frac{-\nu}{2}} h - e^{- \frac{3 \nu}{2}} 
\left (\opertbase^2 - \Omegaperbase{}^2 \right ) r^2 \sin^2 \theta   
\right ) \xi_{\alpha}
+ e^{-\frac{\nu}{2}} \Omegaperperbase \eta_{\alpha}.
\label{Comb2}
\end{align}
Inserting (\ref{Comb1}), (\ref{Comb2}) and (\ref{eq:R2_common})
 into (\ref{SecOrderfield}) and using 
$\Sig_{\alpha\mu} u^\mu =-\Wtwo  e^{-\nu/2}\axial_\alpha$ the 
second order field equations 
take the form 
\begin{align}
&R^{(2)H}_{\alpha\beta}+ \riccifunc(\Wtwo  e^{-\nu}) \Sig_{\alpha\beta}=\nonumber\\
&e^{-\nu} \chiefes\left ( \denperper + \preperper \right ) 
\xi_{\alpha} \xi_{\beta}
+ \chiefes (\den + \pre ) \left \{
\left ( 4 e^{-\nu} h - 2 e^{-2 \nu} \left ( \opertbase^2 - \Omegaperbase{}^2 \right )
r^2 \sin^2 \theta \right ) \xi_{\alpha}\xi_{\beta}
\right . \nonumber\\
& \left . +  e^{-\nu} \Omegaperperbase \Sig_{\alpha\beta}
+ 2 e^{-\nu} ( \opertbase - \Omegaperbase )^2 \eta_{\alpha}\eta_{\beta}
\right \}
+ \frac{1}{2} \chiefes\left ( \denperper - \preperper \right ) g_{\alpha\beta} \nonumber\\&
+ \frac{1}{2} \chiefes \left ( \den - \pre \right ) \Kperper^{H}{}_{\alpha\beta}
- \frac{1}{2} \chiefes (\den+3\pre) \Wtwo  e^{-\nu} \Sig_{\alpha\beta}.
\label{eq:EFEs_common}
\end{align}
Now, an explicit computation shows that
$R^{(2)H}_{\alpha\beta}$ has vanishing $\stat_{(\alpha}\axial_{\beta)}$ component
so equation (\ref{eq:EFEs_common}) splits into
two, namely the component along $\Sig$, which is
\begin{align}
&- \riccifunc(\Wtwo  e^{-\nu}) + \chiefes (\den + \pre ) e^{-\nu} \Omegaperperbase 
- \frac{1}{2} \chiefes (\den+3\pre) \Wtwo  e^{-\nu}=0,
\label{eq:EFEs_second_w}
\end{align}
and the rest
\begin{align}
0 = (\Eq)_{\alpha\beta} := & -R^{(2)H}_{\alpha\beta}+ e^{-\nu} \chiefes \left ( \denperper + \preperper \right ) 
\xi_{\alpha} \xi_{\beta}\nonumber\\
&+ \chiefes (\den + \pre ) \left \{
\left ( 4 e^{-\nu} h - 2 e^{-2 \nu} \left ( \opertbase^2 - \Omegaperbase{}^2 \right )
r^2 \sin^2 \theta \right ) \xi_{\alpha}\xi_{\beta}
\right . \nonumber\\
& \left . + 2 e^{-\nu} ( \opertbase - \Omegaperbase )^2 \eta_{\alpha}\eta_{\beta}
\right \}
+ \frac{1}{2} \chiefes\left ( \denperper - \preperper \right ) g_{\alpha\beta}
+ \frac{1}{2} \chiefes \left ( \den - \pre \right ) \Kperper^{H}{}_{\alpha\beta}.
\label{eq:EFEs_common_H}
\end{align}
In principle, these are nine equations (note $(\Eq)_{t\phi}\equiv0$ holds
identically, by construction).
Two of them, $(\Eq_{tt})$ and $(\Eq_{rr})$, determine
$\denperper$ and $\preperper$ algebraically. 
For the moment we are interested in studying a subset of seven independent linear combinations which do not involve   $\denperper$ nor $\preperper$. 
Introducing the notation $\ot{a},\ot{b}, \cdots := \{ r, \theta\}$ and
$\ot{i},\ot{j}, \cdots := \{ t, \phi \}$, one convenient such subset is
\begin{align}
 (\Eq)_{\ot{a}\, \ot{i}}  = 0, \quad  
(\Eq)_{r \theta} =0, \quad 
 \frac{(\Eq)_{\phi\phi}}{g_{\phi\phi}} - 
\frac{(\Eq)_{\theta\theta}}{g_{\theta\theta}} =0, \quad 
\frac{(\Eq)_{\theta\theta}}{g_{\theta\theta}} - 
\frac{(\Eq)_{rr}}{g_{rr}} =0.
\label{Eqs}
\end{align}
Before writing them down explicitly, let us introduce 
the following scalar functions and discuss their properties,
\begin{align}
  \hhat & := h - \frac{1}{2}r \nu^{\prime} f, \nonumber\\
\vhat & := k+\hhat-f =k + h - f \left ( \frac{1}{2} r \nu^{\prime} +1 \right ), \label{def:hat_functions}\\
\qhat & := m+\hhat-e^{-\lambda/2}\left(e^{\lambda/2}r f\right){}_{,r} =m + h - \frac{1}{2} r f \left ( \lam' + \nu' \right )
- ( r f)_{,r}.\nonumber
\end{align}
Given the differentiability and boundedness properties
of the original set $\{h,m,k,f\}$ (point \baseKperper), and that $n\geq m$,
we have $\hhat,\vhat\in C^{m}(D^+)\cap C^{m}(D^-)$,
$C^{m+1}(S_r)$ on all spheres $S_r$, and bounded near $\centresph_0$, and 
$\qhat\in C^{m-1}(D^+)\cap C^{m-1}(D^-)$ is also $C^{m}(S_r)$
on all spheres $S_r$.

The motivation behind these definitions is their very special gauge behaviour,
as described in the following lemma. Its proof is by explicit
calculation using the results of Proposition \ref{prop:full_class_gauges}.
\begin{lemma}
  \label{res:gauges_hat}
Under the gauge vector $\sperper$ given by
(\ref{gauge_second}) with $\Q=r$ the functions
$\hhat$, $\vhat$, $\qhat$ transform as
\begin{align}
\hhat^{g} &= \hhat + \frac{1}{2} A + \frac{1}{2} r \nu' \left (
 r e^{-\lam} \alpha' \cos \theta - \beta \right ), 
            \label{hhatg} \\
\vhat^{g} &= \vhat + \frac{1}{2}A  + \alpha 
\cos \theta + \left (  r e^{-\lam} \alpha'\cos \theta- \beta \right ) 
\left ( 1+ \frac{1}{2} r \nu' \right ), \label{uhatg} \\
\qhat^g & = \qhat + \frac{1}{2} A 
+ \cos \theta \left [ 
\frac{1}{2} r^2 e^{-\lam} \alpha' \left ( \lam' + \nu' \right )
+ (r^2 e^{-\lam} \alpha')' \right ]
- \left [  \frac{1}{2} r \beta ( \lam' + \nu'  )+ (r \beta)_{,r}
\right ], \label{qhatg}
\end{align}
where $\beta(r)$ is the arbitrary function that enters $f$ in (\ref{fg}).\finn
\end{lemma}
Note that the gauge function $\Y$ has disappeared from these transformations
so that $\hhat$, $\vhat$ and $\qhat$ go a long way towards being gauge
invariant. The remaining gauge transformation is fully explicit
in the variable $\theta$. This will be important in the following.

We can now write down 
explicitly equations (\ref{eq:EFEs_second_w}) and (\ref{Eqs}).
The first set involves a long computation  which has been carried out with
the aid of computer algebra systems. As for the second,
its explicit form can be obtained
directly from (\ref{eq:ricci_functional}) after taking into
account that $\riccixi$ and $\riccieta$, as defined in Proposition \ref{prop:R1}
take the form
\begin{align*}
 R_{\alpha\beta} \xi^{\beta}  = - \frac{\chiefes}{2}  \left ( \den + 3 \pre \right )
\xi_{\alpha}, \quad R_{\alpha\beta} \eta^{\beta} =  \frac{\chiefes}{2}  \left ( \den - \pre \right )
\eta_{\alpha} & \Longrightarrow
\riccixi = -\frac{\chiefes}{2} \left ( \den + 3 \pre \right ),
\riccieta = \frac{\chiefes}{2} \left ( \den - \pre \right ) \\
&  \Longrightarrow   \quad
\riccixi +\riccieta = - 2 \chiefes \pre.
\end{align*}
The result is
\begin{lemma}
	\label{lemma:second_order_efes}
In the setup described above, the field equations (\ref{eq:EFEs_second_w})
and (\ref{Eqs}) take, respectively, the following explicit form
\begin{align}
\frac{\partial}{\partial r}\left( r^4 j \frac{\partial\Wtwo }{\partial r}\right)+ \frac{r^2 j e ^\lambda}{\sin ^3 \theta}\frac{\partial}{\partial \theta}\left( \sin^ 3 \theta \frac{\partial \Wtwo  }{\partial \theta}\right)+4r^3j'(\Wtwo  - \Omegaperperbase) =0,\label{Eqtphi_full}
\end{align}
and
\begin{align}
0  = (\Eq)_{\ot{a}\ot{i}} \equiv & 0,  \label{Eqtphi} \\
0  = (\Eq)_{r \theta}  \equiv &
\frac{\partial}{\partial \theta} \left (
\frac{2}{r} \qhat  - 2 \vhat_{,r} 
 + (\qhat-2 \hhat) \nu' \right ), \label{Eqrth} \\
0  =  \frac{(\Eq)_{\phi\phi}}{g_{\phi\phi}} - 
\frac{(\Eq)_{\theta\theta}}{g_{\theta\theta}}  \equiv &
-e^{-(\lam+\nu)} r \sin^2\theta \left ( 
r \opertbase'{}^2 +
2 \left ( \lam^{\prime} + \nu^{\prime} \right ) (\opertbase-\Omegaperbase)^2
\right ) \nonumber \\
& + \frac{2}{r^2} \left ( \qhat_{,\theta\theta}
- \frac{\cos \theta}{\sin \theta} \qhat_{,\theta} \right ),  \label{Eqfifithth} \\
0 = \frac{(\Eq)_{\theta\theta}}{g_{\theta\theta}} -   \frac{(\Eq)_{rr}}{g_{rr}} 
\equiv &
2 e^{-\lam} \vhat_{,rr}
- \frac{2}{r^2} \Delta_{\mathbb{S}^2}\vhat
-e^{-\lam} (\nu'+\lam') \vhat_{,r}
- \frac{4}{r^2}  \vhat \nonumber \\
& +4 e^{-\lam} \nu^{\prime} \, \hhat_{,r} 
- e^{-\lam} \left (\nu'+\frac{2}{r} \right ) \qhat_{,r} 
+ \frac{2}{r^2} \left (\frac{\cos \theta}{\sin \theta} \qhat_{,\theta}
+ 2 \qhat \right ) \nonumber \\
& - e^{-(\lam+\nu)} r^2 \sin^2 \theta \,
\opertbase'{}^2. \label{Eqththrr}
\end{align}\finn
\end{lemma}
\begin{remark}
Equation (\ref{Eqtphi_full}) 
includes vacuum as a particular case. As in the remark after
{\baseKper.2},
the constant $\Omegaperperbase$ is irrelevant in the vacuum region $\mmm^-$ (where
$j'=0$). Without loss of generality,
we shall set $\Omegaperperbaseext=0$ in $\mmm^-$.
\end{remark}
\begin{remark}
The fundamental underlying reason that will allow us to 
close the bootstrap argument below is the decoupling of equation
(\ref{Eqtphi_full}), which only involves $\Wtwo$, and
the system (\ref{Eqtphi})-(\ref{Eqththrr}),
which only involves the rest of terms in $\Kperper$, that is $\Kperper^H$.
\end{remark}
\begin{remark}
  \label{rem:first_FE}
  Since the second order Einstein field equations reduce to the first order equations when
  $\Kper =0$, it follows from this lemma that the field equations (\ref{firstRic}) for \eqref{Kper} are equivalent to \eqref{eqfo} plus \eqref{eq:EpPp_zero}.
\end{remark}
\begin{remark}
  The requirement $m\geq 2$ in the bootstrap hypothesis
  {\baseKperper.1}
  allows us to work with classical solutions of
  (\ref{Eqtphi})-(\ref{Eqththrr}) because, as discussed after \eqref{def:hat_functions}, 
  $\hhat_{\pm},\vhat_{\pm}$ are $C^{2}(D^{\pm})$,
    $\qhat_\pm$ is $C^{1}(D^{\pm})$ and all of them
  are $C^{2}(S_r)$ on each sphere $S_r$.
\end{remark}
\begin{remark}
\label{remark:fw_invariant}
For later use we introduce the notation
\begin{align}
\F(r)  := \frac{1}{6} e^{-(\lam+\nu)} r^3 \left ( 
r \opertbase'{}^2 +
2 \left ( \lam^{\prime} + \nu^{\prime} \right ) (\opertbase-\Omegaperbase)^2
\right ) , \label{eq:f_omega}
\end{align}
which arises in the inhomogeneous term in (\ref{Eqfifithth}).
\end{remark}

Equation (\ref{Eqfifithth}) is a second order linear ODE in $\theta$
for the function $\qhat$ and can be explicitly integrated. Its general solution is 
\begin{align}
  \qhat(r,\theta) & = \qhat_0(r) +\qhat_1(r) P_1 (\cos \theta)
                    + \F(r) P_2 (\cos \theta),\nonumber
                    \end{align}
where $\qhat_0(r)$ and $\qhat_1(r)$ are free functions of $r$.
We show next that $\qhat_1(r)$ is pure gauge, i.e.
that it can be set to zero by a suitable choice of 
gauge transformation (\ref{qhatg}).
In terms of the functions
\begin{align}
\b(r):= r^2 e^{-\lam} \alpha', \quad \quad \c(r):=r \beta,
\label{bc}
\end{align}
the gauge transformation law of $\qhat$ (\ref{qhatg}) takes the form
\[
\qhat^g  = 
\qhat + \frac{1}{2} A 
+ P_1 (\cos \theta) \left  ( \frac{\b}{2}  
\left ( \lam' + \nu' \right ) + \b' \right )
- \left [ \frac{\c}{2}  ( \lam' + \nu'  ) + \c'
\right ].
\]
We impose that $b(r)$ solves the ODE 
\begin{equation}
\frac{\b}{2} (\lam' + \nu' ) + b' =-\qhat_1,
\label{Eqb}
\end{equation}
so that $\qhat^{g}$ becomes
$\qhat^{g} = \qhat_0 + \F P_2 (\cos \theta)$. 
Dropping the superindex ${}^g$ one has
\begin{equation}
\qhat = \qhat_0(r) + \F(r) P_2 (\cos \theta),
\label{expressionqhat}
\end{equation}
and we have proved that $\qhat_1(r)$ can be gauged away, as claimed.
The remaining 
gauge freedom is given by
the general solution of the homogeneous part of (\ref{Eqb}), i.e.
\[
\b = \b_0 e^{-\frac{1}{2} (\lam+\nu)}, \quad \quad
b_0 \in \mathbb{R}.
\]
which, in terms of $\alpha(r)$ is (by the definition of $b(r)$ in (\ref{bc}))
\begin{equation}
\alpha' = \frac{\b_0}{r^2} e^{\frac{1}{2}(\lam-\nu)}.
\label{eq:for_alpha}
\end{equation}
This residual gauge will be used later to simplify $\vhat$.

We  now insert $\qhat$ from \eqref{expressionqhat} into equations (\ref{Eqrth}) and 
(\ref{Eqththrr}) and perform a trivial integration in $\theta$ in the first one, which 
introduces an arbitrary function $\sigma(r)$,
\begin{align}
0 = & \left ( \frac{2}{r} + \nu' \right )
\left ( \qhat_0 + \F P_2 (\cos \theta) \right )- 2 \vhat_{,r} - 2 \hhat \nu' 
+ \sigma(r) \nu', \label{equation1} \\
0 = & 2 e^{-\lam} \vhat_{,rr}
- \frac{2}{r^2} \Delta_{\mathbb{S}^2}\vhat
-e^{-\lam} (\nu'+\lam') \vhat_{,r}
- \frac{4}{r^2}  \vhat  
+4 e^{-\lam} \nu^{\prime} \, \hhat_{,r}\nonumber \\
& - e^{-\lam} \left (\nu'+\frac{2}{r} \right ) 
\left (\qhat_0{}^{\prime} + \F{}^{\prime} P_2 (\cos \theta) \right )
+ \frac{2}{r^2} \left ( 2 \qhat_0 - \F \right )
- \frac{2}{3} e^{-(\lam+\nu)} r^2 \opertbase'{}^2 
( 1 - P_2(\cos\theta) ). \label{equation2}
\end{align}
Equation (\ref{equation1}) determines $\hhat$ algebraically (recall that
$\nu'$ is nowhere zero outside the centre). Inserting the result into
(\ref{equation2}) yields
\begin{align}
  r^2 \vhat_{,rr}
  + \frac{1}{2} r^2  \left ( \lambda'+\nu'  - 4\frac{\nu''}{\nu'} \right )
  \vhat_{,r}  +  e^{\lambda} \Delta_{\mathbb{S}^2}\vhat
  + 2 e^{\lambda} \vhat  
  = \inhomo_0(r) + \inhomo_2(r) P_2(\cos \theta),
  \label{Field}
\end{align}
where we have introduced
\begin{align}
  \inhomo_0(r):=
  & - 2 \qhat_0\left(\frac{r \nu''}{\nu'}+1- e^{\lambda}\right)
    + \qhat_0'\left ( r + \frac{1}{2} r^2 \nu' \right )
    -\frac{1}{3}e^{-\nu}r^4\opertbase'{}^2
    - e^{\lambda} \F
    + r^2 \sigma'\nu',\label{eq:mu0}\\
  \inhomo_2(r):=
  & \F{}'\left ( r + \frac{1}{2} r^2 \nu' \right )
    - 2\F\frac{r \nu''}{\nu'}
    +\frac{1}{3}e^{-\nu}r^4\opertbase'{}^2- 2 \F,
    \label{eq:mu2}
\end{align}
and recall $\F$ is defined in \eqref{eq:f_omega}.

For later use, let us find the most general solution of the homogenous part of
(\ref{Field}) (i.e. with $\inhomo_0 = \inhomo_2=0$)
with  the form $\vhat = W(r) P_1(\cos \theta)$. It is immediate that 
this will be a solution iff
\begin{equation}
2 W'' + \left (\lambda'+ \nu' - \frac{4 \nu''}{\nu'} \right ) W' =0.
\label{l=1Eq}
\end{equation}
This is a second order ODE that can be trivially integrated once.
However, finding the general
solution is a harder problem, which we address by exploiting the
gauge behaviour of $\vhat$ described in (\ref{uhatg}).
Consistency of the whole construction requires that the gauge transformation
(\ref{uhatg}) restricted to $A=\beta(r)=0$
and $\alpha(r)$
satisfying (\ref{eq:for_alpha}) must transform solutions of (\ref{Field})  
into solutions. The $\ell=1$ component of any such solution must solve the homogeneous PDE.
The gauge transformation above preserves the $\ell=1$ character of the function,
so it produces another solution of the same
homogeneous PDE. In other words, a solution
$W(r)$ of (\ref{l=1Eq}) transforms under this gauge into another
solution $W^g(r)$. This is true in particular for $W(r)=0$, which is an obvious solution.
Summing up, for any function $\gamma(r)$ satisfying
\begin{align}
\gamma'  = \frac{\b_0}{r^2} e^{\frac{1}{2}(\lam-\nu)}
\label{Eqgamma}
\end{align}
it must be the case that the function
\begin{align}
W(r) =
\left ( \gamma(r) + \frac{\b_0}{r} e^{-\frac{1}{2} (\lam+\nu)}
\left ( 1+ \frac{1}{2} r \nu' \right ) \right ) 
\label{sol:uhat_h}
\end{align}
solves (\ref{l=1Eq}). It is a matter of direct computation to confirm
that this is indeed the case. We still need to show that (\ref{sol:uhat_h})
is the general solution. Given that the expression
involves two arbitrary constants, namely an additive intregration constant $\gamma_0$
in \eqref{Eqgamma},
and $\b_0$, we need to make sure that
it contains two linearly independent solutions. One solution is 
$\gamma(r) = \gamma_0 \in \mathbb{R}$, $\b_0=0$, so it suffices
to check  that the solution with $\b_0 \neq 0$ is not constant.
Computing the derivative of  \eqref{sol:uhat_h} and using
the background field
equation \eqref{eq:pprimenuprime} yields 
\begin{align*}
W' = - \frac{\b_0}{2} e^{-\frac{1}{2} (\lam + \nu)} (\nu')^{2},
\end{align*}
which is not identically zero when $\b_0 \neq 0$. We conclude that
indeed (\ref{sol:uhat_h}) is 
the general solution of (\ref{l=1Eq}), and moreover, that it is regular everywhere.

Take now an arbitrary solution $\vhat$ of  (\ref{Field}).
We define  $\vhat_{\perp}$ by means of
\begin{align*}
  \vhat(r,\theta):=\vhat_0(r)+\vhat_1(r)P_1(\cos \theta)+\vhat_2(r)P_2(\cos\theta)+\vhat_\perp(r,\theta), 
\end{align*}
where $\vhat_0(r), \vhat_1(r), \vhat_2(r)$ are the components
defined as in (\ref{fell}).
As mentioned,  $\vhat_1(r) P_1(\cos\theta)$ necessarily satisfies the homogeneous part of  equation (\ref{Field}), so there must exist
$\gamma(r)$ solving (\ref{Eqgamma}) such that 
$\vhat_1(r) = W(r)$ as given in
(\ref{sol:uhat_h}). Apply now the gauge transformation (\ref{uhatg})  
with $\alpha (r)= - \gamma(r)$. The transformed function $\vhat^g$ reads
\begin{align}
\vhat^g(r,\theta)= &\vhat_0(r)+\vhat_1(r)\cos\theta+\vhat_\perp(r,\theta)+ \frac{1}{2}A -W(r)-\beta\left ( 1+ \frac{1}{2} r \nu' \right )\nonumber\\
=&\vhat_0(r)+ \frac{1}{2}A-\beta\left ( 1+ \frac{1}{2} r \nu' \right ) +\vhat_2(r)P_2(\cos\theta)+\vhat_\perp(r,\theta).
\end{align}

To sum up, we have found a gauge transformation of the form (\ref{gauge_second})
that gets rid of the $\ell=1$ terms of both $\qhat(r,\theta)$ and $\vhat(r,\theta)$.
This choice fixes 
completely the  function $\alpha(r)$, so the remaining gauge
freedom is encoded in the constants $A,B$ and the functions $\beta(r), \Y(r)$.  

So far, we have ignored the perturbed field equations involving $\Ppp$ and
$\Epp$.  The following proposition summarizes the previous results and
incorporates the information on $\Ppp$ and $\Epp$ needed later.
\begin{proposition}
  \label{prop:field_equations}
  Assume  {\baseback}-{\baseKperper} in the base perturbation scheme.
  Then, $\uperper$ is given by (\ref{eq:u2}) and equation
  $(\Eq)_{t\phi}=0$ is equivalent to (\ref{Eqtphi_full}). In addition,
  there is a class of gauges $\{\gaugeABY\}$, c.f. Notation \ref{def:gauge_classes},
  for which the remaining Einstein field equations for a perfect fluid (including vacuum)
  to second order
  are satisfied if only if, in terms of the functions defined in (\ref{def:hat_functions}),
\begin{align}
  \qhat(r,\theta)=&\qhat_0(r)+\F(r) P_2(\cos\theta),\label{eq:qhat}\\
  \vhat(r,\theta)=&\vhat_0(r)+\vhat_2(r) P_2(\cos\theta)+\vhat_\perp(r,\theta),\label{eq:uhat} \\
  \hhat(r,\theta)=& \frac{1}{2} \sigma(r)  - \frac{\vhat_{,r}}{\nu'}
+ \left ( \frac{1}{r \nu'} + \frac{1}{2} \right ) 
\left ( \qhat_0 + \F P_2 (\cos \theta) \right ), \label{eq:hhat} 
\end{align}
with $\F(r)$ given explicitly by (\ref{eq:f_omega}),
$\qhat_0(r)$, $\sigma(r)$ are free functions
and $\vhat_0$, $\vhat_2$, $\vhat_\perp$ satisfy
\begin{align}
  & r^2 \vhat_0{}'' +  \frac{1}{2} r^2 \left (\lambda'+ \nu' - 4\frac{\nu''}{\nu'} \right ) \vhat_0{}'+ 2 e^{\lambda}  \vhat_0  = \inhomo_0(r), \label{eq:uhat0}\\
  & r^2 \vhat_2{}'' + \frac{1}{2} r^2 \left (\lambda'+ \nu'- 4\frac{\nu''}{\nu'} \right ) \vhat_2{}'- 4 e^{\lambda}   \vhat_2  = \inhomo_2(r), \label{eq:uhat2}\\
  & r^2 \vhat_{\perp,rr}
    + \frac{1}{2} r^2 \left ( \lambda'+\nu' - 4\frac{\nu''}{\nu'} \right )
    \vhat_{\perp,r}+ e^{\lambda} \Delta_{\mathbb{S}^2}\vhat_\perp+ 2 e^{\lambda} \vhat_\perp = 0,\label{eq:uhatperp}
\end{align}
with $\inhomo_0$ and $\inhomo_2$ explicitly given by (\ref{eq:mu0})-(\ref{eq:mu2}).

Moreover, the second order perturbed 
pressure $\Ppp$ and energy-density $\Epp$ are determined algebraically 
from the previous functions and have the explicit forms
\begin{align}
  & \Ppp= \ff(r) + 
\frac{4}{\nu'} \Pb' \left(\frac{1}{2}fr\nu'+\hhat+\frac{1}{3}e^{-\nu}r^2(\opertbase-\Omegaperbase)^2P_2(\cos\theta)\right), \label{eq:Ppp_in} \\
  & \Epp= \gg(r) + \frac{4}{\nu'} \Eb' \left(\frac{1}{2}fr\nu'+\hhat+\frac{1}{3}e^{-\nu}r^2(\opertbase-\Omegaperbase)^2P_2(\cos\theta)\right), 
    \label{eq:Epp_in} 
\end{align}
where
\begin{align}
  \ff  & \defi       2(\Eb+\Pb)\left(\fintegral + \frac{1}{3}e^{-\nu} r^2(\opertbase-\Omegaperbase)^2\right), \label{eq:ff} \\
  \gg   & \defi -\frac{4}{\nu'} \left\{
          (\Eb+\Pb)\fintegral'+\Eb'\left(\fintegral + \frac{1}{3}e^{-\nu} r^2(\opertbase-\Omegaperbase)^2\right)
          \right\},  \label{eq:gg} \\
  2\chiefes(\Eb+\Pb)  \fintegral
       & \defi \frac{2}{r^2}\qhat_0+\chiefes(\Eb+3\Pb)\sigma+e^{-\lambda}\sigma'\nu' -\frac{4}{r^2} \F. \label{def:fintegral}
\end{align}
\end{proposition}

\begin{remark}
  \label{vacuum:l=0}
  Expressions \eqref{eq:Ppp_in}-\eqref{eq:gg} also hold in vacuum (i.e.
  when $\Eb = \Pb = \Epp = \Ppp=0$). In this case, 
  \eqref{eq:Ppp_in} with \eqref{eq:ff} and \eqref{def:fintegral} imply
  \begin{align}
    \frac{2}{r^2}\qhat_0+e^{-\lambda}\sigma'\nu'
    -\frac{4}{r^2} \F=0.  \label{integral:vacuum}
  \end{align}
  Conversely, one easily checks that if $\Pb=\Eb=0$ and
  \eqref{integral:vacuum} holds then $\Ppp = \Epp =0$.
\end{remark}

\begin{proof}
  All the statements not involving $\Ppp$ or $\Epp$ have  been
  already established except for (\ref{eq:hhat}), which is is a direct consequence of 
  (\ref{equation1}), and the equivalence of the PDE
  (\ref{Field})  with 
  (\ref{eq:uhat0})-(\ref{eq:uhatperp}), which  is a direct consequence of the
  splitting (\ref{eq:uhat}).
  
  The two remaining field equations in (\ref{eq:EFEs_common_H}),
  namely $(\Eq)_{rr}$ and $(\Eq)_{tt}$
  provide explicit algebraic expressions for $\Ppp$ and $\Epp$.
  The resulting expressions can be rewritten in the form given
  in \eqref{eq:Ppp_in}-\eqref{eq:Epp_in} with the definitions below them.
  \fin
\end{proof}

\begin{remark}\label{remark:barotropic}
    For this proposition the perturbation of the fluid has not been assumed to satisfy
any barotropic equation of state. 
\end{remark}

\subsubsection{Barotropic equation of state}
\label{sec:barotropic}
In this subsection we analyse assumption {\basebaro},
namely that the perfect fluid satisfies the equation of state of the background.
This assumption yields an additional constraint affecting only the $\ell=0$ sector
in the interior region.

The existence of a barotropic equation of state
is equivalent to demand (\ref{eq:2ndbarotropic}).
Since $\Pp=\Ep=0$ in the ``base'' scheme, this is simply
\begin{equation}
\Ppp\Eb'-\Epp\Pb'=0. \label{EOS_second_order}
\end{equation}
Given the expressions (\ref{eq:Ppp_in})-(\ref{eq:Epp_in}), this,
in turn, is equivalent to $\ff\Eb'-\gg\Pb'=0$. From the explicit forms
\eqref{eq:ff}-\eqref{eq:gg} and recalling that $\nu' = -2 \Pb' /(\Eb+\Pb)$, see \eqref{eq:pprimenuprime}, it follows that
\[
\ff \Eb'-\gg\Pb'= -2(\Eb+\Pb)^2\fintegral'.
\]
Thus, the barotropic equation of state yields a first integral
$\fintegral = \fintegralz \in \mathbb{R}$, or explicitly
\begin{align}
\frac{2}{r^2}\qhat_0+\chiefes(\Eb+3\Pb)\sigma+e^{-\lambda}\sigma'\nu'
-\frac{4}{r^2} \F = 2\fintegralz \chiefes(\Eb+\Pb),
\label{eq:fintegral}
\end{align}  
which provides an algebraic equation for $\qhat_0$ in terms of background and first order quantities as well as the free function $\sigma(r)$.
The derivation has been done in the interior domain $D^+$. However,
by Remark \ref{vacuum:l=0} this equation also holds in $D^-$ for any constant
$\fintegralz^-$. Following our convention, we set $\fintegralz := \{ \fintegralz^+,
\fintegralz^-\}$ and work with both domains at the same time.

In terms of $\fintegralz$, expressions \eqref{eq:ff}-\eqref{eq:gg} simplify to
\begin{align}
  \ff & =2(\Eb+\Pb)\left(\fintegralz+\frac{1}{3}e^{-\nu}r^2(\opertbase-\Omegaperbase)^2\right),\label{eq:Ppp}\\
  \gg & =-\frac{4}{\nu'}\Eb'\left(\fintegralz+\frac{1}{3}e^{-\nu}r^2(\opertbase-\Omegaperbase)^2\right)
\label{eq:Epp}.
\end{align}
Observe that these expressions only involve background and first order quantities. Replacing back into  (\ref{eq:ff}) and (\ref{eq:gg}), we can also simplify
$\Ppp$ and $\Epp$. It is convenient to write them
in terms of the original (non-hatted) function $h$ (see \eqref{def:hat_functions}). The result is
\begin{align}
  \Ppp
  &= - 2(\Eb+\Pb)\left(h - \fintegralz + \frac{1}{3}e^{-\nu}r^2(\opertbase-\Omegaperbase)^2 \left ( P_2(\cos\theta) - 1 \right )  \right), 
  \label{eq:Ppp_h}\\
  \Epp & =  \frac{4}{\nu'}\Eb'
         \left(h  - \fintegralz + \frac{1}{3}e^{-\nu}r^2(\opertbase-\Omegaperbase)^2 \left ( P_2(\cos\theta) - 1 \right ) \right). \label{eq:Epp_h}
\end{align}

  Under a change of gauge in $\{\gaugeCABYa\}$, i.e. \eqref{gauge_first} and \eqref{gauge_second},
  $h$ changes as \eqref{hg}, while $\Ppp$, taking into account that
  $\Pp=0$, does as $\Ppp{}^g=\Ppp-\Y\nu'(\Eb+\Pb)$, c.f. \eqref{eq:gauge_pressure}.
  Substracting equation \eqref{eq:Ppp_h} and its gauged counterpart we thus have
  \[
    \Y\nu'(\Eb+\Pb)=\Ppp-\Ppp{}^g=2(\Eb+\Pb)\left(h^g-h-\fintegralz^g+\fintegralz\right)
    =2(\Eb+\Pb)\left(\frac{1}{2}A+\frac{1}{2}\Y\nu'-\fintegralz^g+\fintegralz\right).
  \]
  Therefore, because $\Ebc+\Pbc\neq 0$ by assumption,
    the first integral in the interior $\fintegralz^+$ is gauged transformed by
  \begin{equation}
    \fintegralz^+{}^g=\fintegralz^++\frac{1}{2}A^+.
    \label{eq:gauge_fintegral}
  \end{equation}

Since $\opertbase_+(r)\in C^{2}([0,\ro])$ and  $h_+ \in C^{m+1}(\mmm^+\setminus\centresph_0)$
and bounded near $\centresph_0$, it follows that $\Ppp$ is also bounded near the centre and,
\emph{as long as the limit of $h$ at $r=0$ exists} it holds
\begin{equation}
  \lim_{r\to 0}\Ppp=2(\Ebc+\Pbc)\left(\fintegralz^+- \lim_{r\to 0}h\right)
                      \quad \quad \mbox{ if }  \lim_{r\to 0}h \mbox{ exists}. \label{eq:Ppp_c}
\end{equation}
Let us advance here the limit of $h$ will exist as a consequence of the field equations. This will be discussed in Section \ref{sec:vhat_l_0}.

Tackling the problem for the $\ell=0$ sector means taking care of
the functions $\vhat_0(r)$, $\qhat_0(r)$ and  $\sigma(r)$.
From Proposition \ref{prop:field_equations}, $\vhat_0, \sigma$ must satisfy
(\ref{eq:uhat0}) with (\ref{eq:mu0}) and
the barotropic EOS forces $\qhat_0$ to satisfy (\ref{eq:fintegral}) in
both domains $D^{\pm}$.
The key  
is to introduce a change of unknowns and
replace the pair $\{\vhat_0(r),\sigma(r)\}$ in terms of two functions
$\{ \newvhat(r), \newsigma(r)\}$ by means of
\begin{align}
  &\vhat_0=\frac{1}{2}\newvhat\left(2+r\nu'\right)+\fintegralz,\label{eq:newvhat}\\
  &\nu'r\sigma
    + 2e^{\lambda}\left(\newvhat+ \newsigma\frac{1}{2+r\nu'}\right)=\fintegralz(r\nu' -2).\label{eq:newsigma}
\end{align}
This change of functions is invertible  because 
$2+r \nu'=1+e^\lambda(1+r^2\chiefes\Pb)$ by \eqref{eq:nuprime}
and the right-hand side is 
everywhere positive.

The function $\qhat_0$ is obtained from  the barotropic EOS condition
(\ref{eq:fintegral}). In terms  of the new variables and replacing also $\Eb, \Pb$ from the background equations \eqref{eq:lambdaprime}-\eqref{eq:nuprime} this gives
\begin{equation}
  \qhat_0=\frac{1}{2}r\newvhat(\lambda'+\nu')+(r\newvhat)'+\fintegralz
  +\newsigma\frac{\nu'^2 r^2+2 e^{\lambda}}{(2+r\nu')^2}+\newsigma'\frac{r}{2+r\nu'}+2\F.\label{eq:qhat_zeta}
\end{equation}
We now insert this in (\ref{eq:uhat0}) and apply the change
\eqref{eq:newvhat} and \eqref{eq:newsigma}. A long but straightforward   calculation
that uses \eqref{eq:pprimenuprime} and \eqref{eq:f_omega}
gives
\begin{align}
  &r^2\newsigma{}''+ \frac{1}{2} r^2 \left ( \nu' + \lam' -4 \frac{ \nu''}{\nu'} \right )\newsigma{}'
    + 2e^{\lambda}  \newsigma  =
    - r^3e^{-\nu}\left(2(\lambda'+\nu')(\opertbase-\Omegaperbase)^2-\opertbase'{}^2r\right)
    -4\inhomo_2,
    \label{eq:for_zeta}
\end{align}
where $\inhomo_2$ was given in \eqref{eq:mu2}.

Equation \eqref{eq:for_zeta}
is remarkable for several reasons. First of all it
involves only the unknown function $\newsigma(r)$ (this means that the function  $\newvhat(r)$ is completely unrestricted).
Moreover, the homogeneous part of the equation is identical to the one in (\ref{eq:uhat0})
in terms of $\vhat_0$. This is quite unexpected, given the rather involved change of functions
\eqref{eq:newvhat}-\eqref{eq:newsigma}. Moreover, the inhomogeneous term
in \eqref{eq:for_zeta} involves {\it only} background and first order quantities, unlike (\ref{eq:uhat0}) which also involves unknowns.
It is also interesting that this inhomogeneous term is directly related to $\inhomo_2(r)$, which appeared in the  $\ell=2$ sector of the equations.

This equation for $\newsigma$ is the key object that will allow us in
Section \ref{sec:vhat_l_0}
to obtain existence and uniqueness of the barotropic base scheme.

\subsection{Matching conditions for the base perturbation scheme}
\label{sec:base_matching}
In this section we find the necessary and sufficient conditions
that the base perturbation scheme must satisfy so that 
the second order
perturbed matching at the boundary of the fluid ball are satisfied.
The perturbed matching conditions derived in 
\cite{Battye01,Mukohyama00} (first order)
and \cite{Mars2005}  (second order) are summarized and explained in Appendix 
\ref{app:staxi_perturbed_matching}, where we also
determine the most general matching conditions in a 
spherically symmetric static background with two regions,
as defined in \ref{SpherRegion},
for a first perturbation tensor $\Kper = -2 \Q^2 \omega(r,\theta) 
dt d \phi$ 
and $\Kperper$ of the general form (\ref{Kperper}).
The results are purely geometric and do not rely on any field equations.
Moreover,
they extend the matching conditions obtained in \cite{ReinaVera2015}
in that the matching hypersurface $\Supfamp$ is not assumed to be axially symmetric,
and are in turn generalised to a still unfixed function $\Q(r)$.
Throughout this section we use the notation of Appendix
\ref{app:staxi_perturbed_matching}. 

The base perturbation scheme fits into the setup of Appendix
\ref{app:staxi_perturbed_matching} as the particular case where
$\omega(r,\theta) = \opertbase(r)$, $\Q(r)=r$  and matching
hypersurface $\Supfamp$ located at $r_+ = r_- = a$.
Proposition \ref{teo:first_order_matching} states that 
$\opertbase$ must satisfy
\begin{align}
  &\left[\opertbase \right]=b_1, \quad b_1 \in \mathbb{R}, \quad \quad \quad 
\quad \quad
\left[\opertbase'\right]=0. \label{omega_matching}
\end{align}
Let us recall that in the base scheme we have further fixed the first order
gauges so that $b_1=0$.

In addition, the
first order deformation functions $\Qone^{\pm}(\tau,\vartheta,\varphi)$ 
on either side satisfy the conditions listed in (\ref{teo:Q1_matching_R}). The explicit forms
of $\B$ and $\A$, defined  in (\ref{DefUpsilon}), are
\begin{align*}
  \B = \frac{1}{2} e^{\nu- \lambda} \left ( \nu'' + \frac{1}{2} \nu'
  \left ( \nu' - \lambda' \right ) \right ), \quad \quad \quad
  \A = \frac{1}{2} r e^{-\lambda} \lambda'.
\end{align*}
Therefore the conditions (\ref{teo:Q1_matching_R}) in our present setup become
\begin{align}
\left[\Qone\right]=0,\qquad 
\Qone[\lambda']=0\qquad  \Qone[\nu'']=0, \label{Q1_matching}
\end{align}
after using that $\normal = - e^{- \frac{\lambda}{2}} \partial_r$ and that
the background matching conditions are $[r] = [ \nu] =
[\lambda ] = [ \nu'] =0 $ (see \eqref{background_matching}).
Moreover, by (\ref{Bp})-(\ref{nupp})) these equations are equivalent to
\begin{equation}
\left[\Qone\right]=0,\qquad \Qone[\Eb]=0.
\label{eq:mc:Q1}
\end{equation}

The second order matching conditions for the base perturbation scheme 
are obtained from
Proposition \ref{teo:second_order_matching} with
$\omega(r,\theta) = \opertbase(r)$ and $\Q(r)=r$. 
It follows
\begin{align*}
\Qone^2 [ \normal(\A) ] = \Qone^2 \left [ \frac{1}{2} e^{- 3 \lambda/2}
\left ( - r\lambda'' + \lambda' (r \lambda' - 1) \right ) \right ] 
= - \frac{\ro}{2} 
e^{-3\lambda(\ro)/2} \Qone^2 [ \lambda''], 
\end{align*}
where we used (\ref{Q1_matching}) (and the obvious fact that 
$\Qone^2 [a] = \Qone^2[b] =0 \Longrightarrow \Qone^2 [ab] =0$).
By a slightly longer, but analogous,  calculation one finds
\begin{align*}
\Qone^2 \left [ \frac{\nu'}{r} \normal(\A) 
+ \frac{2}{e^{\nu}} \normal(\B) \right ] = - e^{-3 \lambda(\ro)/2} \Qone^2 [ \nu'''].
\end{align*}
Using this and rewriting $\{h,k,m\}$ in terms 
of $\{\vhat,\hhat,\qhat\}$ as defined in
(\ref{def:hat_functions}), all the terms involving $f$ in 
the matching conditions (\ref{defor})-(\ref{Propnh})  drop out.
The result is 
\begin{align}
&  \left [ \defor \right ]   =  \ro e^{\lambda(\ro)/2} 
  \left (
  2 c_0 +
  (2 c_1 + H_1) \cos \vartheta  \right ),\label{common:mc:defor} \\
&  \left[\Wtwo \right]=  D_3,\label{common:mc:W2}\\
&   \left[\Wtwo_{,r} \right]= 2e^{-\lambda(\ro)/2} \Qone[\opertbase''],
\label{common:mc:W2p} \\
&[ \vhat - \hhat ] =  c_0 + c_1 \cos \vartheta,
\label{common:mc:vh} \\
&[ \hhat ] = 
\frac{1}{2} \left ( H_0 + \ro \nu'(\ro) c_0  \right )
+ \frac{1}{4} \ro \nu'(\ro) \left ( H_1 + 2 c_1 \right ) \cos \vartheta,
\label{common:mc:h} \\
&\left [ \qhat + \hhat - 2 \vhat - r ( \vhat - \hhat)_{,r} \right ]
=  \left ( H_1 - \frac{1}{2} e^{\lambda(\ro)} \left ( 2 c_1 + H_1
\right ) \right ) \cos \vartheta  \nonumber \\
& \quad \quad \quad \quad \quad \quad 
\quad \quad \quad \quad \quad \quad 
+ \frac{1}{2} \left [ 
\defor e^{-\lambda/2} \left ( \frac{\lambda'}{2} - \frac{1}{r}
\right ) \right ]
- \frac{1}{4} e^{- \lambda(\ro)} \Qone^2 \left [ \lambda'' \right ], 
\label{common:mc:q} \\
&- [ \hhat_{,r}] + \frac{\ro \nu'(\ro)}{2}
[ \vhat_{,r} - \hhat_{,r} ]
+ \nu'(\ro) \left ( 1 - \frac{\ro \nu'(\ro)}{2} \right )
[ \vhat - \hhat ] 
=  
- \frac{1}{2} \nu'(\ro) 
\left ( \left ( 1 - \frac{\ro \nu'(\ro)}{2} \right )
H_1 \right .  \nonumber \\
& \quad \quad \quad \quad \left .  - \frac{1}{2} e^{\lambda(\ro)} \left ( 2 c_1 + H_1 \right ) \right)
\cos \vartheta
- \frac{1}{4} \left [ \defor e^{- \lambda/2}
\left ( \nu'' + \nu'{}^2 - \frac{\nu'}{r}
\right ) \right ]
+ \frac{1}{4} e^{-\lambda(\ro)} \Qone^2 
\left [ \nu''' \right ].
\label{common:mc:hprime} 
\end{align}
Specifically, Proposition \ref{teo:second_order_matching} states that 
the matching conditions are satisfied if and only if
there exist constants $c_0$, $c_1$, $H_0$, $H_1$ and $D_3$
and functions $\defor^{\pm}$ on $\Supfamp$ such that
(\ref{common:mc:defor})-(\ref{common:mc:hprime}) hold. So far no field 
equations have been used. 
In the next proposition we determine
the matching conditions when the field equations hold. 
\begin{proposition}[\textbf{Perturbed matching}]
  \label{Prop:matching}
Assume the setup of the base perturbation scheme ({\baseback}-{\baseKperper}).
Restrict to the class of gauges $\{\gaugeABY\}$,
c.f. Notation \ref{def:gauge_classes}, at both
sides $\mmm^+$ and $\mmm^-$ so that
the results of Proposition \ref{prop:field_equations} hold.

Then the second order matching
conditions across $\Supfamp = \{r=\ro\}$ are satisfied  if and only if
there exists contants $D_3, c_0, H_0$ such that 
\begin{align}
  \left[\Wtwo \right]  = &  D_3,\label{final:W2}\\
   \left[\Wtwo_{,r} \right]  = & 0,
\label{final:W2p} \\
[ \vhat  ]  = &  \frac{H_0}{2} + \left ( 1 + \frac{\ro \nu'(\ro)}{2} \right )
c_0,
\label{final:uhat} \\
 [ \vhat_{,r} ]  = &  
\left ( \frac{1}{\ro} + \frac{\nu'(\ro)}{2} \right ) 
\left ( [\qhat_0] - \frac{H_0}{2} + 
\frac{1}{3} e^{-\nu(\ro)} \ro^4\chiefes \Eb_+(\ro) (\opertbase_+(\ro)-\Omegaperbase)^2
P_2 (\cos \vartheta) 
\right ) \nonumber \\
& - \left ( \frac{e^{\lambda(\ro)}}{\ro} + \frac{1}{2} \ro \nu'(\ro){}^2 \right ) 
c_0, \label{final:uhatp}
\\
[ \sigma ] = & \left ( \frac{1}{2} - \frac{1}{\ro \nu'(\ro)} \right )
H_0 - \frac{2 e^{\lambda(\ro)}}{\ro \nu'(\ro)} c_0, \label{final:sigma} \\
[ \sigma^{\prime}] =&  \frac{1}{2} \chiefes \Qone^2 
E^{\,\prime}_+ (\ro)  - \frac{2 e^{\lambda(\ro)}}{\ro^2 \nu'(\ro)} [\qhat_0]
\quad \quad  \quad \quad \quad
\mbox{\underline{provided}} \quad E_+(\ro) =0.
\label{final:sigma_r} 
\end{align}
\end{proposition}
\begin{remark}
Note that in the case $E_+(\ro)=0$ and $E_+'(\ro) \neq 0$,
 the matching condition (\ref{final:sigma_r}) forces the 
first order deformation function $\Qone$ to be a constant on 
$\Supfamp$.
\end{remark}
 
\begin{proof}
We start by computing the linear combination
$2 \ro (\mbox{\ref{common:mc:hprime}} )   + (2 + \ro \nu'(\ro) )
 (\mbox{\ref{common:mc:q}} )$. This  is advantageous because
the factor involving $\defor$ becomes, after inserting
$\nu''$ from the
background equation
(\ref{eq:pprimenuprime}),
\begin{align*}
\frac{\ro}{2} \left [ \defor e^{- \lambda/2} 
\left ( - \nu' \left ( \frac{1}{r} + \frac{\nu'}{2} \right ) 
+ \frac{2 e^{\lambda}- 4}{r^2}  \right ) \right ] 
= \left . e^{- \lambda/2}  \left (
- \frac{\nu'}{2} \left ( 1 +  \frac{r\nu'}{2} \right )
+ \frac{e^{\lambda}- 2}{r}  \right ) 
\right |_{r=\ro}  \left [ \defor \right ],
\end{align*}
the equality being true because the term in parenthesis is  continuous
across $\Supfamp$.
Another simplification occurs with  the terms involving $\Qone^2$, which
become
\begin{align*}
\frac{\ro}{2} e^{- \lambda(\ro)} \Qone^2 
\left [ \nu''' 
- \left ( \frac{1}{r}  + \frac{\nu'}{2} \right ) \lambda'' \right ],
\end{align*}
and this  is zero as a consequence of (\ref{nuppp}).
The explicit form of the linear combination is
\begin{align}
& ( 2 + \ro \nu'(\ro)) [ \qhat - \vhat] + 
\left ( 2 - \ro \nu'(\ro) + \ro^2 \nu'{}^2(\ro) \right )
[ \hhat - \vhat ]
- 2 \ro [ \vhat_{,r}] = \label{replacement}\\
& 
\left . \left (\left ( 2 + \frac{r^2 \nu'^2}{2}\right )
H_1 - e^{\lambda} ( 2 c_1 + H_1 ) \right ) \right |_{r=\ro} \cos \vartheta
+ \left . e^{- \lambda/2} 
\left ( 
- \frac{\nu'}{2} \left ( 1 + \frac{r \nu'}{2} \right ) 
+ \frac{e^{\lambda} -2}{r} 
\right )  \right |_{r=\ro}[ \defor]. \nonumber 
\end{align}
Thus, the matching conditions to be satisfied are 
(\ref{common:mc:defor})-(\ref{common:mc:q}) and (\ref{replacement}). In all
of them $\hhat$ and $\qhat$ are to be understood as short-hands of
the explicit expressions given in (\ref{eq:qhat}) and (\ref{eq:hhat}).

We next show the necessity of (\ref{final:W2})-(\ref{final:uhat}).
  There is nothing to prove in
  (\ref{final:W2}).  For (\ref{final:W2p}) we need to determine
  $[\opertbase'']$. Since $\lambda_-' + \nu'_{-}=0$ and
  $[\opertbase'] =0$ we compute from
  (\ref{eqfo})
  \begin{align*}
    [\opertbase''] = 
    ( \lambda'_{+} + \nu'_{+} )
    \left ( \frac{1}{2} \opertbase_+'
    + \frac{2}{\ro} (\opertbase_+ - \Omegaperbase) \right )
    = [\lambda ']     \left ( \frac{1}{2} \opertbase_+'
    + \frac{2}{\ro} (\opertbase_+ - \Omegaperbase) \right ).
  \end{align*}
  Since $\Qone [\lambda'] = 0$, (\ref{final:W2p}) follows at once from 
  (\ref{common:mc:W2p}).  For the rest of expressions, we first observe that neither $\vhat - \hhat$ nor $\hhat$ have terms
$\ell=1$ in the decompositions \eqref{eq:uhat} and \eqref{eq:hhat}. 
Thus 
(\ref{common:mc:vh}) and (\ref{common:mc:h}) force $c_1= H_1 =0$.
Expression (\ref{final:uhat}) 
is then an immediate consequence of (\ref{common:mc:vh})
and (\ref{common:mc:h}).
Concerning (\ref{final:uhatp}), we  substitute
(\ref{common:mc:defor}), (\ref{common:mc:vh}), 
(\ref{common:mc:h}) into  (\ref{replacement}) to find
\begin{align*}
[ \vhat_{,r} ] =
\left . \left ( \frac{1}{r} + \frac{\nu'}{2} \right ) \right |_{r=\ro}
\left ( [\qhat] - \frac{H_0}{2} \right )
- \left . \left ( \frac{e^{\lambda}}{r} + \frac{1}{2} r \nu'{}^2 \right ) \right |_{r=\ro}
c_0,
\end{align*}
and this transforms into (\ref{final:uhat}) after inserting 
$\qhat = \qhat_0 + \F P_2(\cos \theta))$ and computing  the jump
of $\F$  from its
explicit expression in (\ref{eq:f_omega}) as
\begin{align}
  [\F]= &\frac{1}{6} e^{-(\lam(\ro)+\nu(\ro))} \ro^3 \left( 
          \ro [\opertbase'{}^2] + 2\left[ ( \lam^{\prime} + \nu^{\prime}) (\opertbase-\Omegaperbase)^2\right] \right)\nonumber\\
  =&\frac{1}{3} e^{\lambda(\ro)} \ro^4\chiefes \Eb_+(\ro) (\opertbase_+(\ro)-\Omegaperbase)^2,\label{dif_qhat2}
\end{align}
where in the second equality we used $[\opertbase'] =0$
as well as the general identity $ [AB] = A^+ [B] + B^- [A]$
applied to $A = (\opertbase -  \Omegaperbase)^2$ and
$B = \lambda' + \nu'$, together with (\ref{Bp})
and $\lambda'_{-} + \nu'_{-} =0$. Expression (\ref{final:sigma})
is obtained directly from (\ref{common:mc:h})  after
taking into accound that $\hhat$
is given by (\ref{eq:hhat}) and $[\vhat_{,r}]$ has already been computed.

Finally, we establish (\ref{final:sigma_r}). First of all we take the radial 
derivative of $\hhat$ defined (\ref{eq:hhat}) and replace $\vhat_{,rr}$ from (\ref{equation2}) to obtain
\begin{align}
\hhat_{,r} = & - \frac{\sigma'}{2} +
\frac{e^{\lambda}}{r^2 \nu'} \left ( \Delta_{\mathbb{S}^2} \vhat + 2 
\vhat \right ) - \frac{\vhat_{,r}}{2 \nu'} \left ( \frac{2 \nu''}{\nu'}
- (\lambda' + \nu') \right )
+ \frac{1}{r \nu'} \left ( \frac{1}{r} + \frac{\nu''}{\nu'} \right )
(\qhat_0 + \F P_2 )  \nonumber \\& 
- \frac{e^{\lambda}}{r^2 \nu'} (2 q_0 - \F )  + \frac{e^{-\nu}}{3 \nu'} r^2 \opertbase_{r}^2 ( 1 - P_2  ).
\label{derhr}
\end{align}
Under the assumption $[E] =0$, we have $[\nu''] =0$ 
and $[\lambda']=0$ (cf. (\ref{Bp}) and (\ref{nupp})). This implies that no term
in  (\ref{common:mc:q})
 involves jumps of products with two or more
discontinuous factors. For instance, the term involving $\defor$ is
\begin{align*}
 \frac{1}{2} \left [ 
\defor e^{-\lambda/2} 
\left ( \frac{\lambda'}{2} - \frac{1}{r}
\right ) 
\right ]
= \left ( \frac{\ro \lambda'(\ro)}{2} - 1 \right ) c_0.
\end{align*}

Inserting (\ref{derhr}) into
(\ref{common:mc:q}) one can solve for $[\sigma']$. A straightforward, if somewhat long calculation, yields (\ref{final:sigma_r}) after using (\ref{final:uhat})-(\ref{final:sigma}), as well as
$[\lambda'']  = a e^{\lambda(a)} \chiefes E^{\, \prime}_{+}(\ro)$, cf. (\ref{Bpp}), and
\begin{equation*}
\lambda_{\pm}'(\ro) = - \nu'_{\pm}(\ro)= \ro^{-1} (1 - e^{\lambda(\ro)}),
\end{equation*}
which is a consequence of the  background field equations
(\ref{eq:lambdaprime})-(\ref{eq:nuprime}) under $E_{+}(\ro)=0$.

This proves the ``only if'' part of the proposition. To show suffiency we only need to 
care about (\ref{common:mc:q}) when $[E] \neq 0$, as this is the 
only equation
left out. We now have $[\lambda'] \neq 0$ and hence
the term involving $\defor$ in (\ref{common:mc:q}) becomes, after applying again the
identity $[AB] = A^+ [B] + B^- [A]$,
\begin{align*}
 \frac{1}{2} \left [ 
\defor e^{-\lambda/2} \left ( \frac{\lambda'}{2} - \frac{1}{r}
\right ) \right ]
= \frac{1}{2} [ \defor] e^{-\lambda(\ro)/2} \left ( \frac{\lambda'_{+}(\ro)}{2}
-\frac{1}{\ro} \right ) 
+ \frac{e^{-\lambda(\ro)/2}}{4} [ \lambda']  \defor_{-}.
\end{align*}
Thus, equation
(\ref{common:mc:q}) can be solved for $\defor_{-}$ and hence imposes
no additional restrictions on the matching. 
This concludes the ``if'' part of the 
proposition.
\finn
\end{proof}

\section{Existence and uniqueness results of the ``base'' second order global problem}
\label{sec:e_u_global}
We start by proving the following global decomposition result, for which we use
the analytic results discussed in Appendices \ref{App:Lemmas} and \ref{App:bocher-like}.
This proposition will be used later in several circumstances.

\begin{proposition}
\label{PDEresult}
Let $D = \mathbb{R}^3\setminus \{0\}$, $D^+ = \overline B_a \setminus \{0\}$ and
$D^- = D \setminus B_a$, where $B_a$ is the ball of radius $a>0$ centered
at the origin. Let $\{r,\theta,\phi\}$ be standard spherical coordinates on $D$.
Consider $\uhat \in C^2 (D^+) \cap C^2 (D^-) \cap C^1 (D)$,
invariant
under $\eta= \partial_{\phi}$ and satisfying the PDE
\begin{align}
r^2 \uhat_{,rr}  + r \Af(r) \uhat_{,r} + V(r) \left ( \Delta_{\mathbb{S}^2} \uhat + \gamma(r)  
\uhat \right )  =0
\label{PDEuhat}
\end{align}
on $D^+$ and $D^-$.
Assume that the functions $V(r)$, $\gamma(r)$, $\Af(r)$ satisfy
\begin{itemize}
\item[(i)] $V(r) \geq 0$, $\gamma(r)$ is bounded from above,
\item[(ii)] the parts $V^-(r)$, $\gamma^-(r)$, $\Af^-(r)$ are $C^1([\ro,\infty))$ functions and $V^+(r)$, $\gamma^+(r)$, $\Af^+(r)$ extend to the origin
as $C^1([0,\ro])$ functions.
\item[(iii)] the following limits exist and are finite
\begin{align}
& \lim_{r \rightarrow 0} V^+(r) = \vzero, \quad \quad
\lim_{r \rightarrow \infty} V^-(r) = \vinfty, \quad \quad 
\lim_{r \rightarrow 0} \Af^+(r) = a_0,  \quad \quad
\lim_{r \rightarrow \infty}\Af^-(r) = a_{\infty}, \nonumber \\
& \lim_{r \rightarrow 0} \gamma^+(r) = \gamma_{0}, \quad \quad
\lim_{r \rightarrow \infty} \gamma^-(r) = \gamma_{\infty},
\label{limits}
\end{align}
with $\vzero >0$ and $\vinfty >0$.
\end{itemize}  

 Suppose, in addition, that
$\uhat$  is bounded in $D$. Define $\gamma_{\mbox{\tiny max}} := \sup_D \gamma$.
Then the following holds:
\begin{itemize}
\item If $\gamma_{\mbox{\tiny max}} < 0$ then $\uhat =0$. 
\item
 If $\gamma_{\mbox{\tiny max}} \geq 0$ define $\ell_{\mbox{\tiny max}}$ as the largest natural number satisfying $\ell (\ell+1 ) \leq \gamma_{\mbox{\tiny max}}$. 
Then there exist functions
$\uhat_{\ell}(r) \in C^{2} ((0,a]) \cap C^2([a, \infty)) \cap C^1(0,\infty) \cap
L^{\infty}(0,\infty)$ with  $\ell \in \{ 0, \cdots, \ell_{\mbox{\tiny max}} \}$
 such that
\begin{align}
\uhat(r,\theta) = \sum_{\ell=0}^{\ell_{\mbox{\tiny max}}} \uhat_{\ell} (r) 
P_{\ell} (\cos \theta).
\label{expans}
\end{align}
\end{itemize}
\end{proposition}

\begin{proof} For all $\ell \in \mathbb{N} \cup \{ 0 \} $
define $\uhat_{\ell} (r) := \frac{2 \ell +1}{4\pi} 
\int_{\mathbb{S}^2} \uhat  P_{\ell} \volform$. It is clear that this function
is $C^2$ on $\Ia^{+} := (0,a]$ as well as on $\Ia^-:=[a, \infty)$.
It is also $C^1$ on $(0,\infty)$. On $\Ia^{\pm}$ we compute
\begin{align*}
 r^2\frac{d^2 \uhat_{\ell}}{dr^2}  + r\Af(r) \frac{d \uhat_{\ell}}{dr}
& + V(r) \left ( \gamma(r) - \ell (\ell+1) \right ) \uhat_{\ell} =  \\
& \frac{2 \ell+ 1}{4 \pi}
\int_{\mathbb{S}^2} \left (r^2 \uhat_{,rr}
+ r\Af(r) \uhat_{,r} + V(r) (\gamma(r) - \ell(\ell+1)  )
\uhat \right ) P_{\ell} \volform  = \\
&  \frac{2 \ell+ 1}{4 \pi}
\int_{\mathbb{S}^2} - V(r) \left ( \Delta_{\mathbb{S}^2} \uhat
+ \ell (\ell+1) \uhat \right ) P_{\ell} \volform =0,
\end{align*}
where in the second equality we used the
PDE (\ref{PDEuhat}) and in the last one we integrated by parts twice
and used $\Delta_{\mathbb{S}^2} P_{\ell} = - \ell(\ell+1) P_{\ell}$. Thus,
$\uhat_{\ell}$ satisfies the following ODE on $\Ia^{\pm}$
\begin{align}
 r^2\frac{d^2 \uhat_{\ell}}{dr^2} + r \Af(r) \frac{d \uhat_{\ell}}{dr}
+ V(r) \left ( \gamma(r) - \ell (\ell+1) \right ) \uhat_{\ell} = 0.
\label{ODEuhatell}
\end{align}
Boundedness of $\uhat$ on $D$ implies that $\uhat_{\ell}(r)$ is bounded
on $(0,\infty)$.

It suffices to apply Theorem \ref{theorem:global_r}
to the problem for
$\{\uhat^+_{\ell},\uhat^-_{\ell}\}$ with
\[
b_0=\vzero(\gamma_0-\ell(\ell+1),\qquad
b_\infty=\vinfty(\gamma_\infty-\ell(\ell+1)),
\]
$\inhomo^\pm=0$ and $\czero=\cone=0$
to ensure that 
if $b_0,b_\infty<0$ then 
the only bounded solution is the trivial one.
Restrict $\ell$ to satisfy $\ell > \ell_{\mbox{\tiny max}}$. Since, by definition
of supremum, $\gamma_{\mbox{\tiny max}} \geq \gamma_0$
and $\gamma_{\mbox{\tiny max}} \geq \gamma_\infty$ it follows
 $\ell (\ell +1 ) > \gamma_0$ and  $\ell (\ell +1 ) > \gamma_\infty$
and thus $b_0,b_\infty<0$ because $\vzero,\vinfty>0$ by assumption.
To sum up, if $\ell>\ell_{\mbox{\tiny max}}$ then $\uhat_{\ell}^\pm=0$.

At any $r >0$, the function $\uhat |_{S_{r}}$ is $C^2$ and invariant under 
$\partial_{\phi}$. Thus, it can uniquely decomposed as
\begin{align*}
\uhat |_{S_{r}} = \sum_{\ell=0}^{\infty} \uhat_{\ell}(r) P_{\ell} 
\end{align*}
where convergence is in $L^2$. All terms after $\ell_{\mbox{\tiny max}}$ are zero, so convergence is also 
pointwise and we conclude that $\uhat$ takes the form
\begin{align*}
\uhat (r, \theta):=  
\sum_{\ell=0}^{\ell_{\mbox{\tiny max}}} \uhat_{\ell}(r) P_{\ell} (\cos \theta)
\end{align*}
as claimed in the Proposition.
\fin
\end{proof}

\begin{remark}
It is clear from the proof that the condition that $\uhat$ is independent
of $\phi$ can be dropped, at the expense that the decomposition in this
case is in terms
of all spherical harmonics of order  $\ell \le \ell_{\mbox{\tiny max}}$
and not just the Legendre polynomials.
\end{remark}

\subsection{Global problem for $\Wtwo$: existence and uniqueness}
\label{subsection_Wtwo}
In this subsection we study
the  existence, uniqueness and structural properties of the function
$\Wtwo$, which  is restricted to be bounded and satisfy   
the PDE (\ref{Eqtphi_full}) on each side $\mmm^\pm$
together with the matching conditions \eqref{final:W2} and \eqref{final:W2p}.

By item \baseKperper .1 of the base perturbation scheme 
the one-form $\Wtwo \axialfSph \in C^{m+1} (D^+) \cap C^{m+1} (D^-)$ with $m\geq 2$ and hence on 
each $S_r$. We use the terminology introduced in Notation \ref{Notation:Sr}.
By the Hodge decomposition on the sphere, there exist two
functions $\tau, \Gg$ on each $S_r$ satisfying (recall that $\Omegaperperbase=0
$ on $D^-$) 
\begin{align}
(\Wtwo -\Omegaperperbase) \axialfSph = d \mathcal{H} + \star_{\mathbb{S}^2} d \Gg.
\label{HodgeDecom}
\end{align}
Since $\D_A ((\Wtwo - \Omegaperperbase) \eta^A)= \eta(\Wtwo) =0$, the function $\mathcal{H}$
is constant on each $S_r$, and can be set to zero without loss of generality.
The potential function $\Gg$ solves
\begin{align}
\Delta_{\gsph} \Gg = - \mbox{div}_{\mathbb{S}^2} \left ( 
( \Wtwo -  \Omegaperperbase)
\axialfSph\right ),
\label{DeltaGg}  
\end{align}
where, as usual, the divergence of a one-form $\bm{w}$ is 
$\mbox{div}_{\mathbb{S}^2} \bm{w} = \D_A w^A$ with indices  raised
with $\gsph$. The right-hand side of (\ref{DeltaGg}) is
$C^m$ on each $S_r$, so $\Gg$ is $C^{m+2}$ as a function on
$S_r$.\footnote{The problem is one-dimensional and therefore no Hölder
requirement is needed.}   The solution is unique up to an additive constant on each $S_r$, hence
up to a radially symmetric function. In the spherical coordinates 
$\{ \theta,\phi\}$, the Hodge decomposition (\ref{HodgeDecom})
takes the explicit form
\begin{align}
( \Wtwo -  \Omegaperperbase) \sin \theta = \partial_{\theta} \Gg.
\label{ExpForm}
\end{align}
Let $\widetilde{\Gg}$ be the unique solution of this PDE satisfying the 
boundary condition $\widetilde{\Gg}|_{\theta=0} = 0$, i.e. vanishing at the north pole of each sphere $S_r$. 
The right-hand side of (\ref{DeltaGg}) is $C^m$
as a function of $r$ 
both on $r \in (0,a]$ and on $r \in [a, \infty)$. Since the boundary condition
is differentiable in  $r$, the solution $\widetilde{\Gg}$
is $C^m(D^+) \cap C^m(D^-)$. Moreover, $\Wtwo$ is bounded
on $\Du$, so the same holds for $\widetilde{\Gg}$. It turns out to be
convenient to extract the $\ell=0$ component of $\widetilde{\Gg}$ and define
\begin{align*}
\Gg := \widetilde{\Gg} - \frac{1}{4\pi} \int_{S_r} \widetilde{\Gg}
\volform.
\end{align*}
It is clear from this definition that $\Gg \in C^m(D^+) \cap C^m(D^-)$,
bounded on $\Du$ and satisfies $\int_{S_r} \Gg \volform =0$. Next we obtain
the PDE that $\Gg$ must satisfy. We insert $\Wtwo -  \Omegaperperbase  
= (\sin \theta)^{-1} \partial_{\theta} \Gg$ into (\ref{Eqtphi_full}) and find,
after a direct calculation,
\begin{align*}
\frac{1}{\sin \theta} \frac{\partial}{\partial \theta}
\left ( \frac{\partial}{\partial r} \left ( r^4 j \frac{\partial \Gg}{\partial r}
\right )
+ r^2 j e^{\lambda} \left ( \Delta_{\mathbb{S}^2} \Gg + 2 \Gg \right )
+ 4 r^3 j' \Gg \right ) =0.
\end{align*}
Integrating in $\theta$ there appears an arbitrary integration function
of $r$ which is then uniquely fixed by the condition 
$\int_{S_r} \Gg \volform=0$. Thus,
\begin{align}
\frac{\partial}{\partial r} \left ( r^4 j \frac{\partial \Gg}{\partial r}\right )
+ r^2 j e^{\lambda} \left ( \Delta_{\mathbb{S}^2} \Gg + 2 \Gg \right )
+ 4 r^3 j' \Gg =0
\label{PDEGg}
\end{align}
on $\Ia^+=(0,\ro]$ and $\Ia^-=[\ro,\infty)$.
We also need to determine the matching conditions for $\Gg$.
From Proposition \ref{Prop:matching} (specifically from
$[\Wtwo] = D_3$, $[\Wtwo_{,r}]=0$) and (\ref{ExpForm}), the jump of $\Gg$ and $\partial_r \Gg$ satisfy
\begin{align}
\left ( D_3 - \Omegaperperbaseint \right ) \sin \theta = \partial_{\theta} [\Gg],
\quad 
0 = \partial_{\theta} [ \partial_r \Gg] \quad  \Longleftrightarrow 
  \quad
  [\Gg] = ( \Omegaperperbaseint - D_3) \cos \theta, 
\quad
[\partial_r \Gg] = 0
\label{jumpGg}
\end{align}
where in the integration we have imposed that neither $[\Gg]$ nor $[\partial_r \Gg]$ have $\ell=0$ term.

We start with a lemma on existence and uniqueness of  bounded solutions of \eqref{PDEGg}.

\begin{lemma}
  \label{lemma:g}
  Let $\Gg \in C^m (D^+) \cap C^m(D^-)$ ($m \geq 2$)
  be bounded and satisfy
  $\int_{S_r} \Gg \volform=0$
  together with \eqref{PDEGg} on $D^{+}$ and $D^-$ and the jumps
  \begin{align}
    [\Gg] = \Jumpg \cos \theta, 
    \qquad
    [\partial_r \Gg] = 0, \quad \quad \Jumpg \in \mathbb{R}.
    \label{jumpGgu0}
  \end{align}
  Assume that $\Eb_c + \Pb_c \neq 0$. Then, there exists a unique radially symmetric bounded function  $G \in C^2(\overline{D}^+) \cap
  C^{n+2} (D^+) \cap
  C^{\infty}(D^-) \cap C^{1} (D)$ satisfying $G(0)=1$
  and a constant $\Wtwo_c \in \mathbb{R}$   such that
  \begin{align}
    \Gg(r,\theta) = \left \{
    \begin{array}{ll}
      - \Wtwo_c G(r) P_1(\cos \theta) & \mbox{ on } D^+ \\
      - ( \Wtwo_c G(r) + \Jumpg   )  P_1(\cos \theta)  & \mbox{ on } D^-. \\
    \end{array}
    \right .
    \label{Gexp}
  \end{align}
  Moreover,
    the function $G(r)$ on $D^-$ is given by
      \begin{align}
        G^-(r) = - \frac{G'(\ro) \ro^4}{3 r^3} + G_{\infty}, \qquad
        G_\infty\defi G(\ro)+\ro G'(\ro)/3 \label{Gext}
      \end{align}
      and $G'(\ro)>0$, $G_{\infty}> 1$. Clearly also $\lim_{r\to\infty}G=G_{\infty}$.
\end{lemma}

\begin{proof}
Let us define $\Gg_{\perp} := \Gg - \frac{3}{4\pi} \int_{S_r} \Gg P_{1}
\volform$. We want to  apply Proposition \ref{PDEresult}, so we check
that all hypotheses are satisfied.
By construction $\Gg_{\perp}$ is $C^m(D^+)\cap C^m(D^-)$, bounded in $\Du$,
and has no $\ell=0$, or $\ell=1$ components.
By (\ref{jumpGgu0}), it satisfies  
$[ \Gg_{\perp} ] = [ \partial_{r} \Gg_{\perp} ] =0$, so that $\Gg_{\perp}\in C^1(D)$.
The PDE (\ref{PDEGg}) in expanded form is
\begin{align}
  \Gg_{\perp}{}_{,rr} + \left ( \frac{4}{r} + \frac{j'}{j} \right )
  \Gg_{\perp}{}_{,r} + \frac{e^{\lambda}}{r^2} \left ( \Delta_{\mathbb{S}^2} \Gg_{\perp} 
  + \left ( 2 + \frac{4 r e^{-\lambda} j'}{j} \right )
  \Gg_{\perp}  \right ) =0,
  \label{PDEGgExp}
\end{align}
so it fits into the general form (\ref{PDEuhat})  with
\begin{align}
  V(r) = e^{\lambda}, \quad 
  \gamma(r) = 2 + \frac{4 r e^{-\lambda}j'}{j} 
  = 2 - 2 r^2 \chiefes  (\den + \pre ),
  \quad 
  \Af(r) = 4 + r \frac{j'}{j} = 
  4 - \frac{1}{2} r^2 e^{\lambda} \chiefes (\den + \pre),
  \label{VgamA}
\end{align}
where (\ref{eq:jpoj}) has been substituted in the last two expressions.
These functions are all $C^{n}([0,\ro]) \cap C^{n}([\ro,\infty))$.
By assumption {\AsHtwo} we have
$\sup_{\Du} \gamma(r)= 2$, and the limit conditions (\ref{limits})  are all 
fullfilled (c.f. \eqref{expansionlam}) with
\begin{align*}
& 
\lim_{r \rightarrow 0} V^+(r) = 1, \quad \quad
\lim_{r \rightarrow \infty} V^-(r) = 1, \quad \quad 
\lim_{r \rightarrow 0} \Af^+(r) = 4,  \quad \quad
\lim_{r \rightarrow \infty} \Af^-(r) = 4, \nonumber \\
& \lim_{r \rightarrow 0} \gamma^+(r) = 2, \quad \quad
\lim_{r \rightarrow \infty} \gamma^-(r) = 2.
\end{align*}
All the conditions of Proposition \ref{PDEresult} are satisfied and
$\ell_{\mbox{\tiny max}} = 1$ so we conclude that
$\Gg_{\perp}$ must be of the form
$\Gg_{\perp} = \Gg_{\perp}^0 (r) + \Gg_{\perp}^1 (r) P_1 (\cos \theta)$. However,
 by construction $\Gg_{\perp}$ has no such components, hence it vanishes
identically. Consequently, $\Gg$ has only $\ell=1$ component, i.e. takes the 
form
\begin{align}
\Gg(r,\theta) = \Gg_1(r) P_1 (\cos \theta)
\label{Ggl=1}
\end{align}
for some radially symmetric function $\Gg_1$ at either $D^\pm$. From
\eqref{PDEGgExp} and \eqref{VgamA}, this function satisfies the ODE
\begin{align}
  \frac{1}{r^3}\frac{d}{dr}\left( r^4 j \frac{d\Gg_1}{dr}\right) + 4 j' \Gg_1 = 0
  \label{eq:Gg1_0}
\end{align}
or, in expanded form,
\begin{equation}
r^2\Gg_1''+r\left(4-\frac{1}{2}r^2 e^{\lambda}\chiefes(\Eb+\Pb)\right)\Gg_1'
-2r^2e^{\lambda}\chiefes(\Eb+\Pb)\Gg_1=0,
\label{eq:Gg1}
\end{equation}
on $\Ia^+ = (0,\ro]$ and $\Ia^-= [\ro,\infty)$ together with the jumps (from
\eqref{jumpGgu0})
\begin{equation}
[\Gg_1]= \Jumpg,\qquad [\Gg_1']=0.\label{eq:mc:Gg1}
\end{equation}
We now show that \eqref{eq:Gg1} admits a unique solution, up to scale,
which is $C^1(\ro,\infty)$  and bounded. We start with the
interior domain $\Ia^+$.  Equation (\ref{eq:Gg1}) satisfies the requirements of
items \emph{(ii)-(iii)} of  Lemma \ref{Lemma:behaviour} with 
\[
\Af(r)=4-\frac{1}{2}r^2 e^{\lambda}\chiefes(\Eb+\Pb),\qquad \Bf(r)=r^2\Qf(r) \quad \mbox{ with }\quad  \Qf(r)=-2e^{\lambda_+}\chiefes(\Eb+\Pb),
\]
so that $a_0=4$, 
and $\Qf(0)=-2\chiefes (\Ebc+\Pbc) \neq 0$ (by assumption of the lemma).
Therefore, Lemma \ref{Lemma:behaviour} ensures that
there exists a unique up to scaling function $g^+(r)$ that stays bounded
in $(0,\ro)$, and  extends to a function in $C^2([0,\ro])$
satisfying $g^+(0)\neq 0$ and $g^+{}'(0)=0$. It is clear that
$g^+(r)\in C^{n+2}((0,\ro])$ also. We fix the scale
by imposing $g^+(0)=1$.
In the exterior part $\Ia^-$, equation (\ref{eq:Gg1}) can be solved explicity. The solution is
\begin{equation}
  g_{\hone,\htwo}^-(r)=-\frac{2\hone}{r^3}-\htwo \qquad \hone, \htwo \in\mathbb{R}.\label{Ggext}
\end{equation}
Consider the function $\{ g^+(r), g^{-}_{\hone,\htwo}(r) \}$. This corresponds to a function $g(r) \in C^1(0,\infty)$ if and only if
\begin{align*}
  g^+(\ro) = g^-_{\hone,\htwo}(\ro) = -\frac{2\hone}{\ro^3}-\htwo, \quad \quad
  g^+{}^{\prime}(\ro) = g^{-}_{\hone,\htwo}{}^{\prime} (\ro) =
  \frac{6 \hone}{\ro^4}.
\end{align*}
It is clear that this system admits a unique solution $\{\hone, \htwo\}$ with corresponding $g^-(r):=g^-_{\hone,\htwo}(r)$ given by
\begin{align}
  g^-(r):= g^+(\ro) + \frac{\ro}{3} g^+{}^{\prime}(\ro) \left ( 1 - \frac{\ro^3}{r^3} \right ).
  \label{gminusexp}
\end{align}
This establishes the existence of a unique bounded function $g(r)
\in C^2([0,\ro]) \cap C^{n+2} ((0,\ro]) \cap C^1(0,\infty) \cap C^{\infty}([\ro,\infty)$ satisfying $g(0)=1$ and solving \eqref{eq:Gg1}  on $\Ia^{\pm}$.
This $g(r)$ is the trace of  a bounded radially symmetric function $G :
\overline{D} \longrightarrow \mathbb{R}$. It is immediate from the properties
of $g(r)$ that $G \in C^{n+2} (D^+) \cap
C^{\infty}(D^-) \cap C^{1} (D)$. Moreover, since $g(r) \in C^2([0,\ro])$ and satisfies $g'(0)=0$, Taylor's theorem gives $g(r)=1+g_2r^2+\Phi_g^{(2)}(r)$
where $g_2\in\mathbb{R}$ and $\Phi_g^{(2)}(r)$ is $C^2([0,\ro])$ and $o(r^2)$.
Using Lemma \ref{origin} it follows that $G \in C^2(\overline{D}^+)$.
This proves the first claim of the Lemma.

The explicit form \eqref{Gext} follows at once from
  \eqref{gminusexp}.
Given that $0\not\equiv\Eb+\Pb\geq 0$
and since $g^+(0)=1$ and $g^+{}'(0)=0$,
Lemma \ref{LemmaRot} establishes that $g^+(\ro)>1$ and $g^+{}'(\ro)>0$.
Therefore $G(\ro) >1$, $G'(\ro)>0$ and the claim $G_\infty>1$ also follows.

Concerning the function $\Gg_1(r)$, by the uniqueless up to scale
of bounded solutions of \eqref{eq:Gg1} on $\Ia^+$, there exists
a constant $\Wtwo_c$ such that $\Gg_1(r) = - \Wtwo_c g(r)$ on $\Ia^+$
(the choice of sign will be convenient later).
On $\Ia^-$, $\Gg_1^-(r)$ has the form (\ref{Ggext}) for some constants $\hone$ and $\htwo$. 
Imposing the jumps (\ref{jumpGgu0}) it is immediate that
$\Gg_1^-(r) = -( \Wtwo_c g(r) +\Jumpg)$ on $\Ia^-$.
Combining with \eqref{Ggl=1} concludes the proof.
\fin
\end{proof}

\begin{remark}
  \label{Gfunc}
  In the proof of this lemma it has been useful to distinguish the
  function $G$ from its trace $g(r)$. For the rest of the paper, this
  is no longer necessary, so we use the same
  symbol $G$ for both. This follows the general convention used
  throughout the paper.
\end{remark}

We can now prove the following result on existence and uniqueness of $\Wtwo$.

\begin{proposition}[\textbf{Existence and uniqueness of $\Wtwo$}]
\label{prop:problem_Wtwo}
Assume the setup of the base perturbation scheme ({\baseback}-{\basematch}).
Then
\begin{enumerate}
\item $\Wtwo$ is radially symmetric on $\overline{D}$, i.e. it is a function 
$\Wtwo(r)$.
\item There exists a (unique) choice of constants $B^\pm$
in the gauge freedom \eqref{gauge_second} such  that the transformed 
function (still denoted by $\Wtwo$) is continuous across $r = \ro$
and fulfils the property that  $\Kperper(\xi, r^{-1} \eta)$ is bounded at infinity. Moreover, this function
is given by
\begin{align*}
  \Wtwo = \Wtwo_c ( G(r) - G_{\infty}), \qquad \Wtwo_c \in \mathbb{R}
\end{align*}
where the function $G(r)$ and constant $G_{\infty}$ are defined in Lemma
\ref{lemma:g}. In particular
$\Wtwo \in C^{n+2}(\mmm^+\setminus\centresph_0) \cap C^2(\mmm^+) \cap
C^{\infty} (\mmm^-) \cap C^1 (\mmm \setminus \centresph_0)$, 
\begin{align*}
  \Wtwo_{-} = \frac{2J}{r^3}, \qquad  J := - \Wtwo_c G'(\ro) \frac{\ro^4}{6}
\end{align*}
and $J$ vanishes if only if $\Wtwo_c=0$.
The parameter $\Omegaperperbaseint$ corresponding to this function is
$\Omegaperperbaseint = - \Wtwo_c G_{\infty}$
 \end{enumerate}
\end{proposition}

\begin{proof}
Let   $\Gg$ be the unique function related to 
$\Wtwo$ by the Hodge
decomposition \eqref{ExpForm} and satisfying
$\int_{S_r} \Gg \volform=0$. This function satisfies the jumps 
\eqref{jumpGgu0}, so all the hypothesis of Lemma \ref{lemma:g} hold
and we
conclude that $\Gg$ is given by
\eqref{Gexp} with $\Jumpg = \Omegaperperbaseint-D_3$. Inserting back into
\eqref{ExpForm} and using that 
$\Omegaperperbaseext =0$ and \eqref{Gext}, $\Wtwo$ can be  written as
\begin{align}
    \Wtwo = \left \{
    \begin{array}{ll}
      \Wtwo_c G(r) + \Omegaperperbaseint   
      & \mbox{ on } D^+ \\
      \Wtwo_c G(r) + \Omegaperperbaseint-D_3  = \frac{2J}{r^3} + \Wtwo_c G_{\infty} + \Omegaperperbaseint-D_3                                        & \mbox{ on } D^-. \\
    \end{array}
    \right .
    \label{W2exp}
\end{align}
with $J := - \Wtwo_c G'(\ro) \frac{\ro^4}{6}$. This proves item \emph{1}. 
From \eqref{Kperper}, boundedness of the component 
$\Kperper(\xi,r^{-1} \eta)$ on $\mmm^{-}$ 
is equivalent  
to $\lim_{r \rightarrow \infty} \Wtwo =0$
(compare item \baseKper.2 in the base scheme). In addition, given that $G(r)$ is $C^1$ on $D$, the function $\Wtwo$ is continuous on $D$
if and only if the constant $D_3$ can be transformed away. To achieve both properties, and given the transformation law
\eqref{Wtwog}, we use the gauge transformation with vectors
$\sperper^- = (\Wtwo_c G_{\infty} + \Omegaperperbaseint - D_3 ) t \partial_{\phi}$
on $\mmm^-$ and
$\sperper^+ = (\Wtwo_c G_{\infty} + \Omegaperperbaseint ) t \partial_{\phi}$ on $\mmm^+$.
It is clear that no other possible choice of the constants $B^{\pm}$ in
\eqref{gauge_second} can accomplish this.  The gauge transformed $\Wtwo$ (which we still call
$\Wtwo$) is now given by $\Wtwo = \Wtwo_c (G(r) - G_{\infty})$ everywhere
 and the corresponding parameter $\Omegaperperbaseint = - \Wtwo_c G_{\infty}$,
as follows directly from \eqref{W2exp} on $D^+$.
All the properties claimed in the proposition are immediate consequences
of the corresponding properties for $G(r)$ obtained in Lemma
\ref{lemma:g}. In particular $J$ vanishes if and only if $\Wtwo_c$ does
because $G'(\ro) >0$.
\finn
\end{proof}

\subsection{Global problem for $\vhat$: existence and uniqueness of the $\ell\geq 2$ sector}
\label{sec:vhat_l_2}
In this subsection we apply Proposition \ref{PDEresult} to deal with
the global problem for $\vhat$, consisting of the PDE (\ref{Field})
at either side $D^\pm$ plus the matching conditions \eqref{final:uhat}-\eqref{final:sigma_r}
in Proposition \ref{Prop:matching}. Contrary to the problem for $\Wtwo$, however,
we cannot prove uniqueness of $\vhat$ yet, since the radially symmetric ($\ell=0$) part 
still contains one free function (the integrating factor $\sigma(r)$).
Adding the requirement of a barotropic equation of state (in the next subsection)
will allow us to tackle the existence and uniqueness of the $\ell=0$ sector of $\vhat$.

\begin{proposition}[\textbf{Global problem for $\vhat$}]
\label{prop:problem_vhat}
  Assume the setup of the base perturbation scheme ({\baseback}-{\basematch})
  and restrict to the class of gauges $\{\gaugeABY\}$ constructed in Proposition \ref{prop:field_equations} in both $M^\pm$. Then,
$\vhat(r,\theta)$ must have the form
  \begin{align}
    \vhat(r,\theta) = \vhat_0(r) + \vhat_2(r) P_2 (\cos \theta).
\label{decomvhat}  
  \end{align}
  Moreover,  the field equations and matching conditions
    for $\vhat_2(r)$ admit a unique bounded solution. This solution satisfies
    \begin{itemize}
    \item $\vhat_2^+(r) \in C^{n+1}((0,\ro])$ and
  $\vhat_2^+(r)$ is $O(r^4)$ and
    extends as a $C^1([0,\ro])$ function, and $\vhat_2^+{}'(r)$ is $O(r^3)$
   near $r=0$, 
  \item  $\vhat_2^-(r) \in C^{\infty}([\ro,\infty))$, $\vhat_2^-(r)$ is
    $O(r^{-4})$ and $\vhat_2^-{}'(r)$ is $O(r^{-5})$ near $r=\infty$.
\end{itemize}
    In particular if $\opertbase_{\pm}(r)=0$ then $\vhat_2(r)=0$.
  These results are independent of the function $\beta(r)$.
\end{proposition}

\begin{proof} 
  By Proposition \ref{prop:field_equations} the decomposition
  (\ref{eq:uhat}) holds on both regions,
  and the function $\vhat_{\perp}$ 
    satisfies, c.f. \eqref{eq:uhatperp},
\begin{align}
  r^2 \vhat_{\perp}{}_{,rr}+ r \Af(r)   \vhat_{\perp}{}_{\,r}
  + V(r) 
    \left ( \Delta_{\mathbb{S}^2} \vhat_{\perp} +
    2 \vhat_{\perp} \right )
     = 0
\label{eq:vhatbis}
    \end{align}
with
\begin{align*}
  V(r)=e^{\lambda}, \quad \quad
    \Af(r) =\frac{1}{2}r\left(\lambda'+\nu' -4\frac{\nu''}{\nu'}\right).
\end{align*}
Equation (\ref{eq:vhatbis}) is of the form (\ref{PDEuhat}) with
$\gamma(r) =2$.
Recall that $\lambda(r), \nu(r) \in C^{n+1}(\mmm^+) \cap C^{n+1} (\mmm^-)$, so the same holds for $V(r)$. The values of $V(r)$ at the origin
  and infinity are, respectively, $\vzero = 1$ (by \eqref{expansionlam})
  and $\vinfty=1$. Concerning $\Af(r)$,  we use the expansion at the origin
  for $\nu(r)$  in \eqref{background_expansion_origin}
  together with
  $\nu_2\neq 0$,
    which follows from \eqref{eq:l2nu2} because of assumption {\AsHtwo} 
  and the base perturbation scheme condition $\Eb_c + \Pb_c \neq 0$.
  With that,
\begin{align}
  \frac{r}{\nu'_+}  =\frac{1}{2 \nu_2} + \Phi^{(2)}, \qquad
\Phi^{(2)} \in C^{n-1} (\mmm^+)  \mbox{ and }O(r^2).
  \label{eq:for_nu''}
  \end{align}
Thus, $\Af^+ \in C^{n-1} (\mmm^+)$ and $\Af^+(0)= -2$.
Therefore \eqref{eq:vhatbis} satisfies the requirements of Proposition \ref{PDEresult} with $\vzero=1$, $\vinfty=1$,
$a_0=-2$, $a_\infty=4$.
Since $\gamma(r) =2$ we have
$\ell_{\mbox{\tiny max}} = 1$. We conclude that $\vhat_{\perp}$ must be of the form 
$\vhat_{\perp}^0(r) P_0 (\cos \theta) + 
\vhat_{\perp}^1(r) P_1 (\cos \theta)$ and hence identically zero since by construction $\vhat_{\perp}$ does not have such components. The decomposition
(\ref{eq:uhat})
gives 
(\ref{decomvhat}) at once. Furthermore, since the class of gauges
 is restricted to $\alpha(r)=0$,
Lemma \ref{res:gauges_hat} ensures that (\ref{decomvhat}) holds
in the  class of gauges $\{\gaugeABY\}$, i.e. for
arbitrary parameters $A$, $B$ and  free function $\Y(r)$
in \eqref{gauge_second}, as well as for any 
choice of $\beta(r)$.

It remains to show that $\vhat_2(r)$ exists and is unique, and obtain
its behavior around $r=0$ { and $r\to \infty$}.
The problem for $\vhat_2$ is given by
equation \eqref{eq:uhat2} at both $\pm$ sides,
  together with the matching conditions obtained from
(\ref{final:uhat})-(\ref{final:uhatp}) of Proposition \ref{Prop:matching}, which
explicitly read
\begin{align}
  [\vhat_2]  =0, \quad \quad [ \vhat_2' ]  = &  
\frac{1}{6} \left ( 2 + \ro\nu'(\ro) \right ) 
e^{-\nu(\ro)} \ro^3\chiefes \Eb_+(\ro) (\opertbase_+(\ro)-\Omegaperbaseint)^2.
\label{eq:mc:vhat2}
\end{align}
We want to apply  Theorem \ref{theorem:global_r} 
with $\inhomo= \inhomo_2$ given in \eqref{eq:mu2}, and
\begin{align*}
\Af^{\pm}(r)=\frac{1}{2}r\left ( \nu_{\pm}' + \lam_{\pm}' -4 \frac{\nu_{\pm}''}{\nu_{\pm}'} \right ),\qquad
\Bf^{\pm}(r)=-4 e^{\lambda_{\pm}}.
\end{align*}
Let us check that all the hypotheses are satisfied.
We have already seen that 
$\Af^+(r)\in C^{n-1}([0,\ro])$
and $a_0=\Af^+(0)=-2$,
while we have $b_0=\Bf^+(0)=-4$ (c.f. \eqref{expansionlam}).
In the exterior, we may write (by the background field equations)
 \begin{align}
  \Af^-(r)=2(1+e^{\lambda_-})    ,\qquad
\Bf^-(r)=-4 e^{\lambda_-} ,
\label{eq:A_exterior}
 \end{align}
where $\lambda_-(r)$ is given explicitly in \eqref{eqs_back_vacuum}. Consequently,
\begin{align*}
&a_\infty=\lim_{r\to+\infty}\Af^-(r)=4, &&b_\infty=\lim_{r\to+\infty}\Bf^-(r)=-4,\\
&\lim_{r\to+\infty}r^2\frac{d\Af^-(r)}{dr}= -\frac{\chiefes \Mext}{2\pi},
&&\lim_{r\to+\infty}r^2\frac{d\Bf^-(r)}{dr}= \frac{\chiefes \Mext}{\pi}.
\end{align*}
The function $\inhomo_2^+(r)$ is $C^{n-1}((0,\ro])$ and extends continuously
to the centre, where it vanishes. The
structure of $\inhomo_2^+(r)$ around $r=0$ is obtained
using \eqref{background_expansion_origin}, \eqref{omega_origin} and \eqref{eq:sol_for_w2},
and it is found to be of the form
$\inhomo_2^+=r^6(\sigma_6+\Phi_\inhomo^{(1)})$
where $\sigma_6\in\mathbb{R}$
and $\Phi_\inhomo^{(1)}$ is  $o(1)$. 
Concerning  $\inhomo_2^-$, inserting the background vacuum
field equations in \eqref{eq:mu2} gives
\begin{equation}
\inhomo_2^-=\frac{1}{3}r^4(e^\lambda-1)\opertbase_-'{}^2
 =\frac{3\chiefes \Mext J_\opertbase^2}{\pi r^5}\left(1+O\left(\frac{1}{r}\right)\right)
\label{eq:inhomo2_plus}
\end{equation}
where the second equality follows from the explicit form
\eqref{eq:opert_ext_sol}  of $\opertbase_{-}$.
Hence, $\inhomo_2^+$ and $\inhomo_2^-$ satisfy the requirements
of Theorem \ref{theorem:global_r}
with $\alpha_0=6$ and $\alpha_\infty=-5$. The quantities
$\lambda^0_{-}$ and $\lambda^\infty_{-}$ defined in Theorem \ref{theorem:global_r} take the values 
$\lambda^0_-=-4$ and $\lambda^\infty_-=-4$. All the hypothesis of Theorem
\ref{theorem:global_r} are satisfied, including \eqref{Assump},
as well as $\lambda^0_-+1\leq 0$ and $\alpha_0-1\geq 0$. Consequently,
there exists a unique 
solution $\{\vhat_2^+(r),\vhat_2^-(r)\}$
that stays bounded on $(0,\infty)$,
and, moreover, $\vhat^+_2(r)$ is  $O(r^4)$, extends as a $C^1([0,\ro])$ function,
and $\vhat^+_2{}'(r)$ is $O(r^3)$ because $\lambda^0_-=-4$ and $\alpha_0=6$.
The behaviour of the solution $\vhat^-$ and its derivative
  $\vhat^-{}'$ near $r=\infty$ is obtained from Theorem \ref{theorem:global_r} with the values
$\min\{|\lambda^\infty_-|,|\alpha_\infty|\}=4$
and $\min\{|\lambda^\infty_--1|,|\alpha_\infty-1|\}=5$ respectively.
The differentiability of the solutions
in $D^+$ and $D^-$ follow from the fact that the coefficients $\Af^+$, $\Bf^+$
and $\inhomo_2^+$ are
$C^{n-1}$, $C^{n}$ and $C^{n-1}$ on $(0,\ro)$
respectively, while $\Af^-$, $\Bf^-$ and $\inhomo_2^-$ are $C^{\infty}([\ro,\infty))$.

As above, Lemma \ref{res:gauges_hat} ensures that
the class of gauges $\{\gaugeABY\}$ given by
\eqref{gauge_second} with arbitrary
parameters $A$, $B$, free function $\Y(r)$ and $\alpha(r)=0$,
and also free choice of $\beta(r)$,
keeps $\vhat_2(r)$ invariant.

The final statement concerning the case
$\opertbase_{\pm}=0$ is immediate since $\vhat_2(r)=0$ solves the ODE
(\ref{eq:uhat2}) with $\inhomo_2(r)=0$.
\fin
\end{proof}

\subsection{Barotropic equation of state: existence and uniqueness of $\vhat$}
\label{sec:vhat_l_0}
Let us recapitulate. Propositions \ref{prop:field_equations} and \ref{prop:problem_vhat} have shown the existence of a
class of gauges $\{\gaugeABY\}$ and free $\beta(r)$ where
$\vhat$ only has $\ell=0$ and $\ell=2$ components.  Inverting
the definitions \eqref{def:hat_functions},  the original functions $\{h,m,k\}$ take the form (on either side $D^{\pm}$)
\begin{align}
  h &= \hhat + \frac{1}{2}r \nu^{\prime} f \nonumber\\
    &= \frac{1}{2} \sigma
      - \frac{1}{\nu'}(\vhat_0'+\vhat_2'P _2 (\cos \theta))
      + \left ( \frac{1}{r \nu'} + \frac{1}{2} \right ) 
      \left ( \qhat_0 + \F P_2 (\cos \theta) \right ) + \frac{1}{2}r \nu^{\prime} f,
      \label{def:unhat_function_h}\\
  k &= \vhat -\hhat+f \nonumber\\
    &= \vhat_0+\vhat_2 P_2(\cos\theta)-\frac{1}{2} \sigma
      + \frac{1}{\nu'}(\vhat_0'+\vhat_2'P _2 (\cos \theta))- \left ( \frac{1}{r \nu'}
      + \frac{1}{2} \right ) \left ( \qhat_0 + \F P_2 (\cos \theta) \right ) +f
      \label{def:unhat_function_k}\\
  m &= \qhat - h + \frac{1}{2} r f \left ( \lam' + \nu' \right)+ ( r f)_{,r}\nonumber\\
    &= \left (1- \frac{1}{r \nu'} - \frac{1}{2} \right ) \left ( \qhat_0 + \F P_2 (\cos \theta) \right )
      -\frac{1}{2} \sigma + \frac{1}{\nu'}(\vhat_0'+\vhat_2'P _2 (\cos \theta)) +\frac{1}{2}r\lambda' f+( r f)_{,r},
      \label{def:unhat_function_m}
\end{align}
where  $\sigma,\qhat_0,\qhat_2,\vhat_0,\vhat_2,\F$
are functions of $r$, while $f$ is still a free function depending
on $r,\theta$.

From the previous subsections, and leaving aside $f(r,\theta)$ (to be discussed later),
the only part of the solution where existence and uniqueness has not yet
been established is the $\ell=0$ sector, where the unknowns
are $\{ \vhat_0, \qhat_0, \sigma\}$.
In this section we accomplish this by imposing the background barotropic
EOS. As discussed  in subsection \ref{sec:barotropic},
$\qhat_0(r)$ is then given explicitly by (\ref{eq:fintegral}) and 
it is useful
to replace the unknowns $\{\vhat_0(r),\sigma(r)\}$
by $\{\newvhat(r),\newsigma(r)\}$. The main advantage
is that $\newsigma(r)$ decouples from $\newvhat(r)$ and
satisfies a global problem for which we can show existence and uniqueness, while
$\newvhat(r)$ will be shown later to be pure gauge.
Let us first focus on the problem for  $\newsigma(r)$.

We already have the equations that $\newsigma^+$ and $\newsigma^-$
satisfy in their respective domains,
i.e. \eqref{eq:for_zeta}. Let us determine the
jumps of $\newsigma^+$ and $\newsigma^-$ across $r=\ro$,
as well as the regularity
conditions of $\newsigma^+$ around $r=0$, both following
from assumptions {\baseKper} and {\baseKperper} of the base scheme.
Incidentally, no conditions at $r=\infty$ will be needed, since the field equations
will provide bounded solutions only.
We start with the regularity and observe,  first of all,
that  \eqref{eq:newvhat} already implies $\delta(r)$ is $C^m(0,\ro)$ and bounded near $r=0$.
Since we have a priori information on $\{h,m,k\}$, let us 
rewrite  \eqref{def:unhat_function_h}-\eqref{def:unhat_function_m} in
terms of $\{ \newvhat, \newsigma\}$.
After a straightforward calculation and introducing the auxiliary function
\begin{equation}
\keygamma(r)\defi \newsigma'\frac{1}{\nu'}+\newsigma\frac{ r\nu'}{2+r\nu'}
+\frac{2+r \nu'}{r \nu'}2\F
\label{def:keygamma}
\end{equation}
as a shorthand, 
we have
\begin{align}
h&=\frac{1}{2}r\nu'(\newvhat+f)+\fintegralz+\frac{1}{2}\keygamma
+\left ( \frac{2+r\nu'}{2r\nu'}\F -\frac{1}{\nu'}\vhat_2'  \right )P_2(\cos\theta),
\nonumber\\
k&=\newvhat+f-\frac{1}{2}\keygamma+
\left(\vhat_2+\frac{1}{\nu'}\vhat_2'-\frac{2+r\nu'}{2r\nu'}\F\right)P_2(\cos\theta),\label{eq:invert_hats_zeta}\\
m&=\frac{1}{2}r\lambda'(\newvhat+f)+(r(\newvhat+f))_{,r}
+\frac{r\nu'-2}{2(2+r\nu')}\keygamma
+\newsigma\frac{4e^\lambda}{2(2+r\nu')^2}
+\left(\frac{1}{\nu'}\vhat_2'-\frac{2-r\nu'}{2r\nu'}\F\right)P_2(\cos\theta).\nonumber
\end{align}
By 
  Proposition \ref{prop:problem_vhat},  $\vhat_2(r)$ is $C^{n+1}((0,\ro])$ and $O(r^4)$, extends $C^1$
at the origin and $\vhat_2'(r)$ is $O(r^3)$. Moreover,  $\newvhat(r)$ is $C^m(0,\ro)$
    and bounded near $r=0$ and $\F(r)$ is $O(r^4)$ as follows from its defining expression \eqref{eq:f_omega}
together with \eqref{background_expansion_origin} and \eqref{omega_origin}.
Consequently, the expression for $h$ (or that for $k$) 
forces $\keygamma(r)$  to be of class $C^m((0,\ro])$ and bounded near $r=0$.
This implies that $r\keygamma(r)$ must vanish as $r\to 0$.
From \eqref{def:keygamma}, this limit is
\begin{equation}
0=\lim_{r\to 0}r\keygamma(r) = \lim_{r\to 0}\left( \frac{1}{2\nu_2}\newsigma'(r)+r^3\nu_2\newsigma(r)\right),
\label{newsigma_origin}
\end{equation}
where in the second equality we used \eqref{eq:for_nu''}
and \eqref{background_expansion_origin}.
On the other hand, \emph{if the limit of $\keygamma(r)$ as $r\to 0$ exists} then so does
the limit of $h$, and therefore the expression of $h$ in \eqref{eq:invert_hats_zeta}
inserted in \eqref{eq:Ppp_c} provides
\begin{equation}
  \lim_{r\to 0}\Ppp=-(\Ebc+\Pbc)\lim_{r\to 0}\keygamma(r)
    \qquad\mbox{ if } \lim_{r\to 0}\keygamma(r)
\mbox{ exists.}
\label{eq:limit_keygamma}
\end{equation}

We next obtain
the jumps that $\newsigma(r)$ must satisfy on $r=\ro$.
The matching conditions (\ref{final:uhat})-(\ref{final:uhatp}) imply,
restricting to the $\ell=0$ sector,
\begin{align}
  [\vhat_0]  &=\frac{H_0}{2} + \left ( 1 + \frac{\ro \nu'(\ro)}{2} \right ) c_0,\label{eq:mc:vhat0}\\
  [\vhat_0'] &=\left ( \frac{1}{\ro} + \frac{\nu'(\ro)}{2} \right )\left ( [\qhat_0] - \frac{H_0}{2}\right ) - \left ( \frac{e^{\lambda(\ro)}}{\ro} + \frac{1}{2} \ro \nu'(\ro){}^2 \right ) 
c_0.\label{eq:mc:vhat0p}
\end{align}
Eliminating $\newvhat$ from \eqref{eq:newvhat} into  \eqref{eq:newsigma}
gives an expression relating $\vhat_0$ and $\newsigma$. Taking the diference at both sides  and inserting \eqref{eq:mc:vhat0} gives 
\begin{equation}
  [\newsigma]=\frac{1}{4}e^{-\lambda(\ro)}(e^{\lambda(\ro)}+3)(1-e^{\lambda(\ro)})
  (H_0-2[\fintegralz]),
\label{eq:mc:newsigma}
\end{equation}
where $\nu'(a)$ is substituted from  \eqref{nup_value}.
To obtain $[\newsigma']$ we make use of \eqref{eq:qhat_zeta}, after eliminating  $\newvhat$
with \eqref{eq:newvhat}, at both $\pm$ sides. The expression contains $[\vhat_0]$
and $[\vhat_0']$, which we substitute by their expressions in \eqref{eq:mc:vhat0} and
\eqref{eq:mc:vhat0p}. The terms containing $[\qhat_0]$ cancel. Inserting \eqref{eq:mc:newsigma}
and using 
\eqref{nup_value} we obtain
\begin{equation}
[\newsigma']=\frac{1}{4\ro}e^{-\lambda(\ro)}(e^{2\lambda(\ro)}+3)(e^{\lambda(\ro)}-1)(H_0-2[\fintegralz])
-\frac{2}{\ro}(e^{\lambda(\ro)}+1)[\F],
\label{eq:mc:newsigmap}
\end{equation}
keeping in mind that the explicit expression of $[\F]$ is given by \eqref{dif_qhat2}.

  Later in the paper we will face the issue of fixing the gauge completely. To do that it will be determinant to understand the role of the parameter
   $H_0-2[\fintegralz]$. In preparation for that, let us introduce
      $\newsigma_*(r)$ as the function that satisfies the same equation as $\newsigma$ 
  and shares its behaviour around $r=0$, namely
\begin{equation}
\lim_{r\to 0}\left( \frac{1}{2\nu_2}\newsigma_*'(r)+r^3\nu_2\newsigma_*(r)\right)=0,
\label{newsigma_star_origin}
\end{equation}
but with jumps given by
\begin{equation}
[\newsigma_*]=0,\quad[\newsigma_*']=-\frac{2}{\ro}(e^{\lambda(\ro)}+1)[\F].
\label{eq:mc_newsigma_star}
\end{equation}
We also introduce the corresponding function $\keygamma_{\star}$
\begin{equation}
  \keygamma_*\defi \newsigma_{*}'\frac{1}{\nu'}
  +\newsigma_{*}\frac{ r\nu'}{2+r\nu'} +\frac{2+r \nu'}{r \nu'}2\F.
  \label{def:keygamma_star_nss}
\end{equation}
and require
\begin{align}
  \mbox{ if } \quad \lim_{r \to 0} \keygamma_{\star} \quad \mbox{and} \quad   \lim_{r \to 0} \keygamma \quad \mbox{ exist } \qquad \Longrightarrow \qquad
  \lim_{r \to 0} \keygamma_{\star} =
  \lim_{r \to 0} \keygamma. \label{newsigma-newsigmastar}
\end{align}

In the next proposition we establish  existence and uniqueness of $\newsigma_*$. The corresponding result for 
the original $\newsigma$ is obtained as  a corollary.

\begin{proposition}[\textbf{Existence and uniqueness of $\newsigma_*$}]
  \label{prop:newsigma_star_unique}
  The problem for $\newsigma_*(r)$, namely
  equation \eqref{eq:for_zeta}  
    on $D^+$ and $D^-$ with matching
    conditions on $r=\ro$ given by 
    \eqref{eq:mc_newsigma_star}
    and such that the restriction around the origin
    \eqref{newsigma_star_origin} holds, admits a one-parameter family
      of solutions. In addition, the limits 
        $\lim_{r \to 0} \keygamma(r)$ and $\lim_{r \to 0} \keygamma_{\star}(r)$
        exist and the function $\newsigma_{\star}$ is uniquely determined by the value $\Ppp_c\defi  \lim_{r\to 0}\Ppp$ by means of
    \begin{align}
      \lim_{r \to 0} \keygamma_{\star} = - \frac{\Ppp}{\Ebc+\Pbc}.
      \label{boundary:newsigmastar}
           \end{align}
    This solution has  the following properties:
  \begin{enumerate}
  \item  $\newsigma_*^+(r)$ is of class $C^{n+1}((0,\ro])$, extends to a $C^1([0,\ro])$ function,
    and has the form
    \begin{equation}
      \newsigma_*^+(r)=-\chiefes\Ppp_c\frac{\Ebc+3\Pbc}{6(\Ebc+\Pbc)} \newsigma^+_-(r)+\newsigma_P^+(r),
      \label{eq:newsigma_plus_star_sol}
    \end{equation}
    where 
    $\newsigma^+_-(r)$ and $\newsigma_P^+(r)$ are unique:
    $\newsigma^+_-(r)$ solves the homogeneous part of \eqref{eq:for_zeta}
    and satisfies \eqref{behaviour_zeta}, while $\newsigma_P^+(r)$ is the
    only particular solution of \eqref{eq:for_zeta} that is $O(r^4)$.
      \item  $\newsigma_*^-(r)$ reads
    \begin{align}
            \newsigma^-_{*}(r)=
      & e^{\lambda_-(r)}\left(\frac{\newsigma_A}{r^2}
        +\frac{\newsigma_B}{r}\right)+ 2J^2_\opertbase\frac{1}{r^4}(2+e^{\lambda_-(r)}),
        \label{eq:newsigma_minus_star}
    \end{align}
    where  $\lambda_-(r)$ is given by \eqref{eqs_back_vacuum}
    and $\newsigma_A, \newsigma_B$ are constants fully determined by the matching conditions
    in terms of quantities of the background configuration, plus
    $J_\opertbase, \Omegaperbase$ and 
    $\{\Ppp_c, \newsigma^+_-(\ro), \newsigma^+_P(\ro),\newsigma^+_-{}'(\ro),\newsigma^+_P{}'(\ro)\}$.
  \item If $\opertbase_\pm=0$ then 
    
      \begin{align}
        \newsigma_*^+(r)=
        &-\chiefes\Ppp_c\frac{\Ebc+3\Pbc}{6(\Ebc+\Pbc)} \newsigma^+_-(r), \label{eq:newsigma_int_sol_no_w_1}\\
        \newsigma_*^-(r)=
        &
                    \chiefes\Ppp_c\newsigma^-_{**}(r), \label{eq:newsigma_ext_sol_no_w_1}
    \end{align}
    where
    \begin{eqnarray*}
      \newsigma^-_{**}(r) :=
       \frac{\Ebc+3 \Pbc}{6(\Ebc+\Pbc)} e^{\lambda_- (r)}
      \left(\left(\frac{\ro^2}{r^2}-\frac{\ro}{r}\right) e^{-\lambda(\ro)} \ro \newsigma^+_-{}'(\ro)
      +\left(\frac{\ro^2}{r^2} -\frac{\ro}{r}(1+e^{-\lambda(\ro)})\right) \newsigma^+_-(\ro)\right).
    \end{eqnarray*}
  \end{enumerate}
\end{proposition}
\begin{proof}
We first analize the equation for  $\newsigma_*^+$ in $D^+$.
To do that we make use of Lemma \ref{Lemma:behaviour}
for the homogeneous part of \eqref{eq:for_zeta}.
Lemma \ref{Lemma:behaviour} applies
with (changing $t$ for $r$)
\[
\Af^+(r)=\frac{1}{2}r\left ( \nu_+' + \lam_+' - 4\frac{\nu_+''}{\nu_+'} \right ),\qquad
\Bf^+(r)=2e^{\lambda_+},
\]
which have been already analised
(except for a different constant factor in $\Bf^+$)
in the proof of
Proposition \ref{prop:problem_vhat}. We showed $\Af^+(r)\in C^{n-1}([0,\ro])$ and $a_0=\Af^+(0)=-2$,
while  $\Bf^+(r)\in C^{n+1}([0,\ro])$ with $b_0=2$ (observe that $b_0\geq 0$, which prevents us from
using Theorem \ref{theorem:global_r}). As a result $\lambda_+=-1$ and $\lambda_-=-2$,
and therefore point \emph{(i)} of Lemma \ref{Lemma:behaviour} ensures there exist
two linearly independent solutions $\newsigma^+_\pm(r)$, which necessarily are of class $C^{n+1}(0,\ro)$,
with
\begin{align}
&\newsigma^+_+(r)=r (1+o(1)),\qquad \newsigma^+_+{}'(r)=1+o(1),\nonumber \\
&\newsigma^+_-(r)=r^2 (1+o(1)),\qquad \newsigma^+_-{}'(r)=r(2+o(1)).
\label{behaviour_zeta}
\end{align}

The inhomogeneous term of equation \eqref{eq:for_zeta} reads
\begin{align}
\inhomop_0^+\defi
-r^3e^{-\nu_+}\left(2(\lambda_+'+\nu_+')(\opertbase_+-\Omegaperbase)^2-\opertbase_+'{}^2r\right)
-4\inhomo_2^+. \label{eq:inhomop}
\end{align}

Although Theorem \ref{theorem:global_r} cannot be applied directly, we
  may still use several constructions introduced in its proof, specifically regarding the properties of the particular solution $U_p$  introduced there.
By direct inspection, the  function $\inhomop_0^+$ is $C^{n-1}([0,\ro])$,
just like $\inhomo_2^+$ (c.f. \eqref{eq:mu2}).
Its structure 
around $r=0$ is obtained
from \eqref{background_expansion_origin}, \eqref{omega_origin},
\eqref{eq:sol_for_w2}, plus the fact that
$\inhomo_2^+$ 
is 
$O(r^6)$, and turns out to be
$\inhomop_0^+= - r^4(4e^{-\nu_0}(\opertbase_0-\Omegaperbase)^2(\lambda_2+\nu_2)+o(1))$.
In the notation of  Theorem \ref{theorem:global_r} with 
$\inhomo^+=\inhomop^+_0$  we have $\alpha_0=4$,
so that $\alpha_0+\lambda_+=3$ and $\alpha_0+\lambda_-=2$
and therefore the general solution of the equation for $\newsigma^+(r)$ 
has the form
\[
\newsigma_*^+(r)=c_+\newsigma^+_+(r) + c_-\newsigma^+_-(r)+ \newsigma_P^+(r),
\]
where the particular solution satisfies, see Remark \ref{remark:diff_origin_particular_D1},
\begin{align}
\newsigma_P^+(r)=r^4(\newsigmaz+o(1)), \qquad
  \newsigma_P^+{}'(r)=r^3(4\newsigmaz+o(1))
  \label{exp:newsigma}
\end{align}
with $\newsigmaz$ a fixed number (see \eqref{Upr0}).
  From \eqref{behaviour_zeta}
    and \eqref{exp:newsigma}, the requirement \eqref{newsigma_star_origin}
  forces $c_+=0$.  This implies, taking into account \eqref{eq:for_nu''} and $\F \in O(r^4)$,
  that $\lim_{r\to 0 }\keygamma_*=c_-/\nu_2$.

So far we  only used
  equation \eqref{eq:for_zeta} and \eqref{newsigma_star_origin}. Both
  are satisfied by the original function  $\newsigma^+$, so it must also be that
    $\newsigma^+(r) = \hat{c}_- \newsigma^+_{-}(r) + \newsigma^+_P(r)$ for an,
    a priori, different integration constant $\hat{c}_{-}$. However, since 
  $\lim_{r \to 0} \keygamma = \hat{c}_-/\nu_2$, condition
  \eqref{newsigma-newsigmastar}
   implies  $\hat{c}_-= c_-$ and we conclude that $\newsigma^+(r) = \newsigma^+_{\star}(r)$.
Using  $\nu_2 =  \chiefes (\Ebc + 3\Pbc)/6$, and that $\Ebc+\Pbc\neq0$,
the relation \eqref{eq:limit_keygamma} fixes $c_-$ as
\[
  c_- = - \frac{\chiefes(\Ebc+3 \Pbc)}{6 (\Ebc + \Pbc)}  \Ppp_c.
\]
This proves \eqref{boundary:newsigmastar} and item \emph{1}.

Equation
(\ref{eq:for_zeta}) for $\newsigma_*$ in the exterior region $D^-$ reads 
\begin{equation*}
r^2 \newsigma_*^-{}''(r)+2r(e^{\lambda_-}+1)\newsigma_*^-{}'(r)
+2 e^{\lambda_-}\newsigma_*^-(r)= - 12 J^2_\opertbase \frac{1}{r^4}(e^{\lambda_-}-4),
\end{equation*}
after inserting  \eqref{eq:A_exterior} for $\Af^-(r)$, \eqref{eq:inhomo2_plus}
and \eqref{eq:opert_ext_sol}. We do not replace $\lambda_-(r)$ by its explicit form \eqref{eqs_back_vacuum} for conciseness.
The general solution is given by
\begin{align*}
\newsigma_*^-(r)=e^{\lambda_-(r)}\left(\frac{\newsigma_A}{r^2}+\frac{\newsigma_B}{r}\right) +  2J^2_\opertbase\frac{1}{r^4}(2+e^{\lambda_-(r)}),
\qquad \qquad \newsigma_A,\newsigma_B \in\mathbb{R}.
\end{align*}
The integration constants $\newsigma_A, \newsigma_B$ are restricted  to
  satisfy the jumps \eqref{eq:mc_newsigma_star},
which can be arranged in the form
\begin{align}
    \left (
      \begin{array}{cc}
        1 & \ro \\
        2 & \ro
      \end{array}
    \right )
    \left (
      \begin{array}{c}
        \newsigma_A \\
        \newsigma_B
      \end{array} \right )
    = \left (
      \begin{array}{c}
       \kappa_1\\
       \kappa_2
      \end{array} \right ),\label{eq:matrix_for_AB}
\end{align}
where $\kappa_1,\kappa_2$ depend on $\{\Ppp_c,\newsigma^+_-(\ro),\newsigma^+_P(\ro),J_\opertbase,\ro,\Mext,\newsigma^+_-{}'(\ro),\newsigma^+_P{}'(\ro),[\F]\}$,
with $[\F]$ given by \eqref{dif_qhat2}. 
The $2\times 2$ matrix has determinant $-\ro\neq 0$ and therefore there exist unique values of $\newsigma_A$, $\newsigma_B$ that fulfill these conditions.  This proves item \emph{2}, as well as the global existence and uniqueness claim.

Assume now $\opertbase_\pm=0$.  The inhomogeneous term vanishes
$\inhomop_0^+=0$ (see \eqref{eq:inhomop} and \eqref{eq:mu2})
and therefore $\newsigma_P^+(r)=0$,
so that \eqref{eq:newsigma_int_sol_no_w_1} follows. In addition,
$J_\opertbase=\Omegaperbase=[\F]=0$, and 
we can solve \eqref{eq:matrix_for_AB}, to  obtain  
\begin{align}
  \newsigma_A &=
                                \frac{\chiefes(\Ebc+3 \Pbc)}{6 (\Ebc + \Pbc)}  \Ppp_c \ro^2
               \left(\ro e^{-\lambda(\ro)}\newsigma^+_-{}'(\ro)+\newsigma^+_-(\ro)\right),\\
  \newsigma_B &=
              - \frac{\chiefes(\Ebc+3 \Pbc)}{6 (\Ebc + \Pbc)}  \Ppp_c \ro
               \left(\ro \newsigma^+_-{}'(\ro)+(e^{-\lambda(\ro)}+1)\newsigma^+_-(\ro)\right).
\end{align}
Inserting into \eqref{eq:newsigma_minus_star} gives
\eqref{eq:newsigma_ext_sol_no_w_1} 
after using the explicit form \eqref{eqs_back_vacuum} of $\lambda_{-}(r)$.\fin
\end{proof}

\begin{corollary}
 \label{coro:newsigma_unique}
The function $\newsigma$ is given by
\begin{align}
  \newsigma^+(r)=&\newsigma^+_*(r),\label{eq:newsigma_plus_new}\\
  \newsigma^-(r)=&
                   \left(H_0-2[\fintegralz]\right)e^{\lambda_-(r)}
                   \left(-\frac{3}{4}\left(\frac{\chiefes\Mext}{4\pi}\right)^2 \frac{1}{r^2}
                   +\left(\frac{\chiefes\Mext}{4\pi}\right) \frac{1}{r}\right)
                   +\newsigma^-_{*}(r).\label{eq:newsigma_minus_new}
\end{align}
\end{corollary}
\begin{proof}
As shown above, only \eqref{eq:newsigma_minus_new} needs attention.
The function
$\tilde\newsigma\defi \newsigma-\newsigma_*$ satisfies
the homogeneous part of equation \eqref{eq:for_zeta}.
The general solution in the exterior $D^-$
is thus given by
\[
\tilde\newsigma^-(r)=e^{\lambda_-(r)}\left(\frac{\tilde\newsigma_A}{r^2}+\frac{\tilde\newsigma_B}{r}\right).
\]
It suffices to obtain the constants $\tilde\newsigma_A$ and $\tilde\newsigma_B$
from the  matching conditions
\begin{align*}
\tilde\newsigma^-(\ro)=-[\tilde\newsigma]=&-\frac{1}{4}e^{-\lambda(\ro)}(e^{\lambda(\ro)}+3)(1-e^{\lambda(\ro)})(H_0-2[\fintegralz]),\\
\tilde\newsigma^-{}'(\ro)=-[\tilde\newsigma']=&-\frac{1}{4\ro}e^{-\lambda(\ro)}(e^{2\lambda(\ro)}+3)(e^{\lambda(\ro)}-1)(H_0-2[\fintegralz]),
\end{align*}
that follow from \eqref{eq:mc:newsigma}-\eqref{eq:mc:newsigmap} and
\eqref{eq:mc_newsigma_star}.
Using
the explicit form of $\lambda(\ro)$, c.f. \eqref{eqs_back_vacuum},
we obtain \eqref{eq:newsigma_minus_new}.
\fin
\end{proof}

\begin{remark}
\label{remark:keygamma_det}
As a consequence of the above results, the function
$\keygamma^+=\keygamma^+_*$ is determined
in terms of quantities of the background configuration plus
$J_\opertbase, \Omegaperbase$ and $\Ppp_c$,
and near $r=0$ it has the form
\begin{equation}
\keygamma^+(r)=-\frac{1}{\Ebc + \Pbc} \Ppp_c + o(1).
\label{eq:Gamma_origin}
\end{equation}
The function $\keygamma$ in $D^-$ takes the form
\begin{equation}
\keygamma^-(r)=-2\frac{4\pi r -\chiefes\Mext}{8\pi r -\chiefes\Mext}(H_0-2[\fintegralz])+
\keygamma^-_*(r),
\label{keygamma_-_determined}
\end{equation}
with $\keygamma^-_*$
being fully determined
in terms of quantities of the background configuration plus
$J_\opertbase, \Omegaperbase$ and 
$\{\Ppp_c, \newsigma^+_-(\ro), \newsigma^+_P(\ro),\newsigma^+_-{}'(\ro),\newsigma^+_P{}'(\ro)\}$.  
\end{remark}

Having established the existence result for $\newsigma$, we show that the
functions that remain undertermined, namely $\newvhat(r)$ and $f(r,\theta)$,
are pure gauge. Since we want to stay in the context where
Propositions \ref{prop:problem_Wtwo} (point \emph{3}) 
and Proposition \ref{prop:problem_vhat} can be applied, the 
available gauge freedom has already been restricted to the subset
$\{\gaugeAY\}$, namely a function $\Y(r,\theta)$ and
a constant $A$, on each side $\mmm^{\pm}$.
There is also the free integration function
$\beta(r)$ on each side. Our next result fixes this freedom and removes the functions $f$ and $\newvhat$ altogether. This concludes our analysis of the base perturbation scheme.

\begin{proposition}[\textbf{Existence and uniqueness of the barotropic base scheme}]
\label{res:hkmf_uniqueness}
  Assume the setup of the barotropic base perturbation scheme
({\baseback}-{\basebaro}).
Then, there exists a gauge on each region given by
$\gaugefinal^+\defi\{\gauge()^+\}$ and $\gaugefinal^-\defi\{\gauge()^-\}$
(no arguments left), c.f. Notation \ref{def:gauge_classes},
in which
\begin{itemize}
\item[1.] the two items of Proposition \ref{prop:problem_Wtwo}
  hold,
\item[2.] $\beta^\pm$ can be fixed such that $f^\gaugefinal_+=0$ and $f^\gaugefinal_-=0$,
\item[3.] The solutions of the field equations, denoted by
$\{h_+^\gaugefinal,k_+^\gaugefinal,m_+^\gaugefinal\}$ in $D^+$
and $\{h_-^\gaugefinal,k_-^\gaugefinal,m_-^\gaugefinal\}$ in $D^-$, exist
and for given $\Ppp_c\in \mathbb{R}$  are unique. 
  Moreover, the corresponding composed functions
     in $\Du$ take the form
    \begin{align*}
    &h^{\gaugefinal}(r,\theta)= h^{\gaugefinal}_0(r)+ h^{\gaugefinal}_2(r) P_2(\cos\theta),\\
      &m^{\gaugefinal}(r,\theta)= m^{\gaugefinal}_0(r)+ m^{\gaugefinal}_2(r) P_2(\cos\theta),\\
    &k^{\gaugefinal}(r,\theta)= k^{\gaugefinal}_2(r) P_2(\cos\theta),
  \end{align*}
with $h^{\gaugefinal}_0, h^{\gaugefinal}_2, m^{\gaugefinal}_0, m^{\gaugefinal}_2, k^{\gaugefinal}_2 \in C^{n}((0,\ro]) \cap C^{\infty}([\ro,\infty))$,
extend continuously to $r=0$ and are bounded.
\end{itemize}
Furthermore, if $\opertbase^\pm=0$ then 
$k^{\gaugefinal}_+=0$ and

\begin{align}
h^{\gaugefinal}_+=&-\chiefes\Ppp_c\frac{\Ebc+3\Pbc}{24(\Ebc+\Pbc)}
\left\{\frac{2+r\nu_+'}{\nu_+'}\newsigma^+_-{}'
          +r\nu'_+\newsigma_-^+\right\},\label{hg_+_no_w}\\
m^{\gaugefinal}_+=&-\chiefes\Ppp_c\frac{\Ebc+3\Pbc}{24(\Ebc+\Pbc)}\frac{1}{r\nu_+'^2}\nonumber\\
          &\times\left\{\left((2+r\nu_+')(2+r\lambda'_+)-4e^{\lambda_+}\right)\newsigma_-^+{}'+\left((2+r\lambda'_+)r\nu'_+-4e^{\lambda_+}\right)\newsigma_-^+\right\},\label{mg_+_no_w}
\end{align}
while $k^{\gaugefinal}_-=0$ and

\begin{align}
  h^{\gaugefinal}_-=
  &\chiefes\Ppp_c\left(\frac{8\pi r -\chiefes\Mext}{4\chiefes\Mext} r  \newsigma^-_{**}{}'
    +\frac{\chiefes\Mext}{4(4\pi r -\chiefes\Mext)} \newsigma^-_{**}\right),
    \label{hg_-_no_w}\\
      m^{\gaugefinal}_-=
  &\chiefes\Ppp_c\left\{\frac{16\pi r -3 \chiefes\Mext}{4\chiefes\Mext}r\newsigma^-_{**}{}'
    -\left(8\pi r +\frac{(16\pi r-3\chiefes\Mext)\chiefes^2\Mext^2}{4(4\pi r-\chiefes\Mext)^2}\right)
    \frac{4\pi r -\chiefes\Mext}{\chiefes\Mext(8\pi r -\chiefes\Mext)}\newsigma^-_{**}\right\},
    \label{mg_-_no_w}
\end{align}
where $\newsigma^+_-(r)$ and $\newsigma^-_{**}(r)$ are defined in Proposition \ref{prop:newsigma_star_unique}.
\end{proposition}

\begin{proof}
We start considering the classes $\{\gaugeABY^\pm\}$
in which Proposition \ref{prop:problem_vhat} holds, both on $\mmm^+$ and $\mmm^-$.
We apply a change of gauge \eqref{hg}-\eqref{fg} in
Proposition \ref{prop:full_class_gauges} with $\Q=r$, $\alpha(r)=0$ and
\begin{equation}
 \Y=-r\left(\newvhat+f-\frac{\keygamma}{2}\right).\label{gauge_Y_new}
\end{equation}
We also fix $\beta(r)$ by
\begin{equation}
\beta(r)=\newvhat-\frac{\keygamma}{2}.
\label{eq:beta_fixed_new}
\end{equation}
Applying this to \eqref{eq:invert_hats_zeta} one obtains
\begin{align}
  h^g=&\frac{1}{2}\left(A+2\fintegralz+\frac{\keygamma}{2}\left(2+r\nu'\right)\right)+\left ( \frac{2+r\nu'}{2r\nu'}\F -\frac{1}{\nu'}\vhat_2'  \right )P_2(\cos\theta),\label{hg_final_new}\\
  k^g=&\left(\vhat_2+\frac{1}{\nu'}\vhat_2'-\frac{2+r\nu'}{2r\nu'}\F\right)P_2(\cos\theta),
       \label{kg_final_new}\\
  m^g=&\frac{1}{r\nu'}\left\{\frac{\keygamma}{4}\left( (2+r\nu')(2+r\lambda')-4e^{\lambda}\right)
       -\frac{2 e^\lambda}{2+r\nu'}\newsigma-\frac{1}{6}r^3e^{-\nu}\left(2(\lambda'+\nu')(\opertbase-\Omegaperbase)^2-r\opertbase'{}^2\right)\right\}\nonumber\\
     &+\left(\frac{1}{\nu'}\vhat_2'-\frac{2-r\nu'}{2r\nu'}\F\right)P_2(\cos\theta),\label{mg_final_new}\\
  f^g=&0.
\end{align}
This already proves item \emph{2}.
Observe that this partial gauge fixing still leaves arbitrary
the constants $A,B$ on each side.
We now choose (uniquely) $B^+$ and $B^-$ so that the second point in
Proposition \ref{prop:problem_Wtwo} holds.

At the interior we choose 
$A^+=-2\fintegralz^+$,
so that \eqref{hg_final_new}-\eqref{mg_final_new} become
\begin{align}
  h^{\gaugefinal}_+
  =&\frac{\keygamma^+}{4}\left(2+r\nu'_+\right)+\left ( \frac{2+r\nu'_+}{2r\nu'_+}\F^+ -\frac{1}{\nu_+'}\vhat_2^+{}'  \right )P_2(\cos\theta),\label{hg_+_new}\\
  k^{\gaugefinal}_+
  =&\left(\vhat^+_2+\frac{1}{\nu'_+}\vhat^+_2{}'-\frac{2+r\nu'_+}{2r\nu_+'}\F^+\right)P_2(\cos\theta),
               \label{kg_+_new}\\
  m^{\gaugefinal}_+
  =&\frac{1}{r\nu_+'}\left\{\frac{\keygamma^+}{4}\left( (2+r\nu_+')(2+r\lambda_+')-4e^{\lambda_+}\right) 
    -\frac{2 e^{\lambda_+}}{2+r\nu_+'}\newsigma^+\right.\nonumber\\
  &\left.-\frac{1}{6}r^3e^{-\nu_+}\left(2(\lambda_+'+\nu'_+)(\opertbase_+-\Omegaperbase)^2-r\opertbase_+'{}^2\right)\right\}+\left(\frac{1}{\nu_+'}\vhat^+_2{}'-\frac{2-r\nu_+'}{2r\nu_+'}\F^+\right)P_2(\cos\theta).\label{mg_+_new}
\end{align}
The behaviour of these expressions as $r\to 0$, using
\eqref{background_expansion_origin}-\eqref{expansionlam},
\eqref{eq:for_nu''}, \eqref{eq:Gamma_origin},
that $\F\in O(r^4)$,
and Propositions \ref{prop:problem_vhat} and
\ref{prop:newsigma_star_unique} as well as Corollary
\ref{coro:newsigma_unique}, is given by
\begin{eqnarray*}
\lim_{r\to 0} h^\gaugefinal_+=-\frac{1}{2(\Ebc + \Pbc)} \Ppp_c,\qquad
\lim_{r\to 0} k^\gaugefinal_+=0,\qquad
\lim_{r\to 0} m^\gaugefinal_+=0,
\end{eqnarray*}
i.e. the limits exist. This shows in particular that the change of gauge
  defined by \eqref{gauge_Y_new} lies within the class
of gauges described in Notation \ref{def:gauge_classes}
and therefore that \eqref{hg_+_new}-\eqref{mg_+_new} are written
  in an admissible and fully fixed gauge $\{\gauge()^+\}$ (no arguments left), which we have denoted simply by $\gaugefinal^+$.
These expressions involve only functions whose existence, uniqueness
  and regularity properties have already been established in
  Propositions \ref{prop:problem_vhat} and
  \ref{prop:newsigma_star_unique} together with Corollary \ref{coro:newsigma_unique}.
  Specifically $\vhat^+_2, \newsigma^+$ are $C^{n+1}((0,a])$, $\vhat^+_2$
  is unique and $\newsigma^+$, and thus also $\keygamma^+$ by Remark \ref{remark:keygamma_det}, is unique up to the constant $\Ppp_c$.
The claim for $\{h_+^\gaugefinal,k_+^\gaugefinal,m_+^\gaugefinal\}$ follows.

Regarding $\{h_-^g,k_-^g,m_-^g\}$ in $D^-$ we choose
$A^-=-2\fintegralz^++H_0$.
Then \eqref{hg_final_new}-\eqref{mg_final_new} become, in the fixed gauge
  $\gaugefinal^-=\{\gauge()^-\}$ (no arguments left), 
\begin{align}
       h^{\gaugefinal}_-=
       &\frac{1}{4}(2+r\nu'_-)\keygamma^-_*          +\frac{1}{\chiefes\Mext}\left (3(8\pi r-\chiefes\Mext)\frac{J^2_\opertbase}{r^4} 
	-(4\pi r-\chiefes\Mext) r \vhat^-_2{}' \right )P_2(\cos\theta),\label{hg_-_new}\\
k^{\gaugefinal}_-=&\left(\vhat^-_2+\frac{4\pi r-\chiefes\Mext}{\chiefes\Mext}r\vhat^-_2{}'
	-\frac{8\pi r -\chiefes\Mext}{\chiefes\Mext}3\frac{J^2_\opertbase}{r^4}\right)P_2(\cos\theta),\label{kg_-_new}\\
m^{\gaugefinal}_-=&\left(8\pi r-\frac{(16\pi r-3\chiefes\Mext)(8\pi r-\chiefes\Mext)}{4\pi r -\chiefes\Mext}\right)
              \frac{3 J_\opertbase^2}{\chiefes\Mext r^4}\nonumber\\
       & -\left(8\pi r +\frac{(16\pi r-3\chiefes\Mext)\chiefes^2\Mext^2}{4(4\pi r-\chiefes\Mext)^2}\right)\frac{4\pi r -\chiefes\Mext}{\chiefes\Mext(8\pi r -\chiefes\Mext)}\newsigma^-_{*}
         -\frac{16\pi r -3 \chiefes\Mext}{4\chiefes\Mext}r\newsigma^-_{*}{}'\nonumber\\
&
+\left(\frac{4\pi r-\chiefes\Mext}{\chiefes\Mext}r \vhat_2^-{}'
	-\frac{8\pi r - 3\chiefes\Mext}{\chiefes\Mext}3\frac{J_\opertbase^2}{r^4}\right)P_2(\cos\theta)\label{mg_-_new}
\end{align}
after using \eqref{eqs_back_vacuum}, \eqref{eq:opert_ext_sol},
\eqref{eq:f_omega},  \eqref{eq:newsigma_minus_new}
and \eqref{keygamma_-_determined}. It is straigforward to check first that
given \eqref{eq:newsigma_minus_star}, $\newsigma^-_{*}$ is $O(1/r)$ and $\newsigma^-_{*}{}'$
is $O(1/r^2)$ and therefore $\keygamma^-_*$, c.f. \eqref{def:keygamma_star_nss},
is bounded near $r=\infty$.
Then, since $\vhat_2$ is $O(1/r^4)$ and $\vhat'_2$ is $O(1/r^5)$, c.f. Proposition
\ref{prop:problem_vhat}, the three functions in \eqref{hg_-_new}-\eqref{mg_-_new}
are bounded in $D^-$, which justifies the fact that the change of
gauge was indeed within the class $\{\gauge\}$ from Notation \ref{def:gauge_classes}.
Again,  Propositions \ref{prop:problem_vhat} 
  and \ref{prop:newsigma_star_unique} together with Corollary \ref{coro:newsigma_unique}
have established all the required existence, uniqueness and regularity properties: $\vhat_2^{-}, \newsigma^- $ are both
  $\in C^{\infty} ([\ro,\infty))$, $\vhat_2^{-}$ is uniquely determined and $\newsigma_{*}$ is explicitly given in
 \eqref{eq:newsigma_minus_star} with the constants  $\newsigma_A$ and $\newsigma_B$
 fully determined once $\Ppp_c$ is fixed,  and Remark \ref{remark:keygamma_det} establishes the same for $\keygamma_{*}$.
The claim for $\{h_-^\gaugefinal,k_-^\gaugefinal,m_-^\gaugefinal\}$ follows.

We now consider the particular case $\opertbase^\pm=0$,
so that, in particular, $\Omegaperbaseint=J_\opertbase=0$ and $\F(r)=0$.
First, Proposition \ref{prop:problem_vhat} gives $\vhat_2^\pm(r)=0$  
and therefore all terms in the $\ell=2$ sector in \eqref{hg_+_new}-\eqref{mg_-_new} vanish.
Thence we  already have that $k^\gaugefinal_+=k^\gaugefinal_-=0$.
On the other hand,
Proposition \ref{prop:newsigma_star_unique}  item \emph{3},
and Corollary \ref{coro:newsigma_unique},
provide the form of $\newsigma^+$ and the explicit expression of
$\newsigma^-$, which inserted in \eqref{hg_+_new},  \eqref{mg_+_new}, \eqref{hg_-_new},
 using \eqref{def:keygamma},
and \eqref{mg_-_new} yield \eqref{hg_+_no_w} and \eqref{mg_+_no_w} in $D^+$,
and 
\eqref{hg_-_no_w} and \eqref{mg_-_no_w} in $D^-$.
\fin
\end{proof}

\begin{remark}
  \label{res:final_Ppp}
  Combining $A^+=-2\fintegralz^+$ with \eqref{eq:gauge_fintegral}
  it follows that, in the gauge $\gaugefinal$,
  the constant $\fintegralz^\gaugefinal$ vanishes. Hence, the perturbed pressure and energy density
  \eqref{eq:Ppp_h}-\eqref{eq:Epp_h} in this gauge are
  \begin{align}
    \Ppp{}^\gaugefinal=& - 2(\Eb+\Pb)\left(h_+^\gaugefinal + \frac{1}{3}e^{-\nu}r^2(\opertbase-\Omegaperbase)^2 \left ( P_2(\cos\theta) - 1 \right )  \right),\nonumber\\
  \Epp{}^\gaugefinal=& \frac{4}{\nu'}\Eb'\left(h_+^\gaugefinal + \frac{1}{3}e^{-\nu}r^2(\opertbase-\Omegaperbase)^2 \left ( P_2(\cos\theta) - 1 \right )  \right).\nonumber                        
  \end{align}
\end{remark}

\section{Existence and uniqueness of the general set up}
\label{sec:existence_and_uniqueness}
We are now ready to apply the results obtained for the ``base'' perturbation scheme
to solve, using a bootstrap argument, the general first order and second order problems for the
perturbation scheme in the canonical 
form 
over a background
configuration for a rigidly rotating perfect fluid interior and vacuum exterior
following Definition \ref{def:PFSTAXpert}.

Before stating the main results of this paper, it is necessary
to discuss  the physical meaning of the constant $\Ppp_c$ that has been
introduced  in the previous section. We already know
that $\Ppp = \lim_{r\rightarrow 0} \Ppp$
(Proposition \ref{prop:newsigma_star_unique}),
so one might think that this parameter already has a clear meaning.
However, the point is more subtle than
one may think, as we discuss next.

\subsection{The perturbed central pressure}

 In this work, three different sets of gauge vectors  play a role.
  Theorem \ref{theo:paper1} assumes a $C^{n+1}$ perturbation scheme. So, in particular it assumes perturbation tensors  $\Kper^o$ and $\Kperper^o$
  that are $C^{n} (\mmm^{\pm})$ and $C^{n-1}(\mmm^{\pm})$ respectively. The first set of gauge vectors is the standard one, namely vector fields that respect this differentiability class everywhere. They correspond to gauge vectors 
  $\sper \in C^{n+1}(\mmm^{\pm})$ at first order
  and
  $\sperper \in C^{n}(\mmm^{\pm})$ at second order. The second set of gauge vectors is the one that transforms
  $\Kper^o$, $\Kperper^o$ into $\Kper^{\Psi}$, $\Kperper^{\Psi}$
  as given in Theorem
    \ref{theo:paper1}.
    These gauge vectors are no longer differentiable everywhere. However, it is part of the content of Theorem \ref{theo:paper1} that they can be chosen to have no radial component and to extend continuously at the origin. Moreover, under a very mild extra condition discussed in
    Remark \ref{remark:main}, when the target tensors
    $\Kper^{\Psi}$ and $\Kperper^{\Psi}$ are taken as fixed, these vectors are uniquely defined up to a linear
    combination of the background Killings $\xi$ and $\eta$.
    So, all such vectors
    extend continuously to the origin and have no radial component.
    We call this  ``canonical gauge transformation''.
  The third class is defined in Proposition \ref{prop:full_class_gauges} and has been called $\{\gaugeCABYa\}$
  in Notation \ref{def:gauge_classes}. This class has been extensively used in the analysis of the base perturbation scheme.

  Concerning the perturbed pressures $\Pp$ amd $\Ppp$, the field equations imply that these functions are of class $C^{n-2}(\mmm^+)$ and $C^{n-3}(\mmm^+)$ in the starting gauge $\Kper^o$ and $\Kperper^o$. We already know (Lemma \ref{lemma:central_P}) that, under this first set of gauge transformations,
$\Ppp_c$ is invariant provided the configuration
  has equatorial symmetry (and this property is true in the present setup, see below).

  In the base perturbation scheme we have only assumed the \emph{outcome}
  of Theorem \ref{theo:paper1}. While $\Pp=0$ followed directly from the field equations, at this level of generality we did not know a priori that
  $\Ppp$ is well-behaved at the centre (not even bounded). We prefered this route
  (instead of assuming the \emph{hypotheses} of Theorem \ref{theo:paper1},
  which would of course would have been justified) in order to emphasize that, even with this generality, imposing that the fluid satisfies a barotropic equation of state (independent of the perturbation parameter $\pertp$)
 already forces the continuity of $\Ppp$
 at the centre and hence the existence of the parameter $\Ppp_c$.

We now discuss the gauge invariance of $\Ppp_c$ under the canonical gauge transformation and under $\{\gaugeCABYa\}$ 
when the \emph{hypotheses}
  of Theorem \ref{theo:paper1} are assumed
  (instead of only its conclusions). We are only interested in the
  case when $\Pp=0$.

  \begin{lemma}
\label{res:Ppp_gauge_inv}
    Let $\Kper^o$ and $\Kperper^o$
    be perturbation tensors defined by the $C^{n+1}$ ($n \geq 3$) perturbation scheme
    $(\mmm_\pertp, \gfam_\pertp,\{\psi_\pertp\})$ assumed in Theorem
    \ref{theo:paper1}.    Assume further that  the perturbed field equations for a rigidly rotating
    perfect fluid hold on $\mmm^+$ with $\Pp=0$ and set $\Ppp_c :=
        \Ppp(0)$.  
    \begin{itemize}
    \item[(i)]  If $\gfam_{\pertp}$, $\pertp \neq 0$, only admits one axial Killing vector, then $\Ppp_c$ is gauge invariant under the
      canonical gauge transformation.
\item[(ii)] $\Ppp_c$ is gauge invariant under  $\{\gaugeCABYa\}$.
\end{itemize}
    \end{lemma}

\begin{proof} For the gauge vectors $\sper$ and $\sperper$ that transform
  $\Kper^o$ and $\Kperper^o$ into the form $\Kper^{\Psi}$ and
  $\Kperper^{\Psi}$ given in Theorem \ref{theo:paper1}, we know that
  $\sper$ and $\sperper$ have no radial component (because the condition described in Remark \ref{remark:main} is satisfied). Since
  the background pressure $P$ is radially symmetric we have
  $\sper(P) = \sperper(P) =0$  outside the centre, and in fact
  everywhere because
  $\sper$, $\sperper$ extend continuosly to the centre. It is now
  immediate from \eqref{eq:gauge_pressure} that
  $\Pp{}^g= \Pp =0$ and 
  $\Ppp{}^g= \Ppp$. In particular, the value at the centre
  $\Ppp_c = \Ppp(0)$ is gauge invariant. This proves item \emph{(i)}.

For item \emph{(ii)}, let $\sper$, $\sperper$ be gauge vectors
    given by \eqref{gauge_first}
    and \eqref{gauge_second} respectively. Since $\sper$ has again no
    radial component, the property $\Pb = \Pp =0$ follows as before. Now,
    \eqref{eq:gauge_pressure} gives
    $\Ppp{}^g=\Ppp+\sperper(\Pb)
=\Ppp + 2 \Y(r,\theta) \partial_r(\Pb)$. Boundedness of $\Y$
and the vanishing of $d\Pb$ at the centre, ensured by Lemma \ref{res:radial_functions_origin}, proves that $\Ppp_c$ is also gauge invariant in this case.
\fin
\end{proof}

    This result allows us to call $\Ppp_c$ ``the perturbed
    central pressure'' (to second order) in an unambiguous way.

\subsection{Existence and uniqueness of the first order problem}
We focus first on the first order problem, i.e. that for (\ref{res:fopert_tensor}).
To do that we only need to consider the \emph{base perturbation scheme}
with $\opertbase(r)=\Omegaperbase=0$ and identify
$\Kper^\Psi$ with $\Kperper$ in \eqref{Kperper} thanks to the substitutions
\begin{align}
&\Wtwo (r,\theta)\to \opert(r,\theta),\quad 
h \to \foh,\quad
m \to \fom,\quad
k \to \fok,\quad
f \to \fof,\label{eq:subs_2_to_1_K1}
\end{align}
at both sides $\pm$, while the perturbed density and pressure in the exterior $\mmm^-$
get substituted by
\begin{equation}
\Ppp \to \Pp, \qquad \Epp \to \Ep. \label{eq:subs_2_to_1_rho_p}
\end{equation}
Obviously, also
  \begin{equation}
\Wtwo_c \to \opert_c, \qquad \Omegaperperbaseint \to
  \Omegaperbaseint,  \qquad \Ppp_c \to \Pp_c,
    \label{eq:subs_2_to_1_param}
\end{equation}
    and the latter
  is called simply ``perturbed pressure at the origin''. We make the full argument precise in the proof of the following proposition.

\begin{proposition}[\textbf{Rotating stars to 1st order}]
\label{res:first_order}
Consider a $C^{n+1}$ perfect fluid ball configuration
according to Definition \ref{def:background}
with $n\geq 4$ and $\Ebc+\Pbc\neq 0$.
Let us be given a $C^{n+1}$ maximal
perturbation scheme $(\mmm_\pertp, \gfam_\pertp,\{\psi_\pertp\})$
inheriting the
Abelian
$G_2$ generated by $\{ \xi=\partial_t, \eta=\partial_\phi\}$.
Assume that the corresponding first order perturbation tensor $\Kper$
\begin{itemize}
\item solves the 1st order perturbed equations for a rigidly rotating perfect fluid
  with the barotropic equation of state of the background
  in $\mmm^+$ and for vacuum in $\mmm^-$,
\item satisfies the linearized matching conditions across the boundary of the fluid ball.
\item it is bounded,
\item the perturbed pressure vanishes at the origin, $\Pp_c=0$.
\end{itemize}
Then there exist first order gauge vectors
(at each interior and exterior regions $\mmm^\pm$)
such that the gauge transformed tensor $\Kper^{\gaugefinal}$
takes the form
\begin{equation}
\Kper^{\gaugefinal}=-2 \opert_c (G(r)-G_\infty) r^2 \sin^2\theta dt d\phi,
\label{res:K1_prop}
\end{equation}
where
    $\opert_c\in\mathbb{R}$,
  $G(r)$ is the unique $C^1(0,\infty)$
  solution of \eqref{eq:Gg1} in $D$ 
  with $G(0)=1$,  and $G_\infty=\lim_{r\to\infty}G(r)=G(\ro)+\ro G'(\ro)/3>1$.
  Moreover, $G^+$ extends to a 
    $C^2(\mmm^ +)\cap C^{n+2}(\mmm^+\setminus \centresph_0)$ function,
    while $G^-$ is $C^{\infty}([\ro,\infty))$ and given explicitly by \eqref{gminusexp}.
 
    Furthermore, the parameter      $\Omegaper= - \opert_c G_{\infty}$ and
      the first order perturbed pressure and density $\Pp$ and $\Ep$ vanish identically.
\end{proposition}

\begin{proof}
We start by setting the problem under the frame of the  \emph{base perturbation scheme}.
Point \baseback{} of the \emph{base perturbation scheme} is satisfied by assumption.
We set $\Kper^\pm=0$, i.e. $\opertbase_\pm=0$,
 and  the whole
  point \baseKper{} is trivially satisfied with $\Omegaperbase=J_\opertbase=0$, while the matching conditions
\eqref{omega_matching} and \eqref{Q1_matching} are also satisfied for $b_1=\Qone^\pm=0$.
By Proposition \ref{prop:rigidrot_OT} the
perturbation scheme inherits the ortogonal transitivity of the group generated
by $\{\stat,\axial\}$ and therefore Theorem \ref{theo:paper1} applies
to both the interior and exterior regions $\mmm^\pm$
ensuring that there exists a gauge transformation 
at each region $M^\pm$
for which the first order perturbation tensors
take the form \eqref{res:fopert_tensor} on each $M^\pm$.
With the identifications in \eqref{eq:subs_2_to_1_K1} and setting $\Q(r)=r$ (which is allowed by the assumptions in Definition \ref{def:background}), the properties of the functions in \eqref{res:fopert_tensor}
imply that point \baseKperper.1{} of the \emph{base perturbation scheme}
is satisfied for $\Kperper=\Kper^\Psi$  with $m=n-1$ (on both $M^+$ and $M^-$). Finally, the points  \baseKperper.2{}, \basematch{} and \basebaro{}
are incorporated as assumptions in the Proposition.

It suffices now to apply Proposition \ref{res:hkmf_uniqueness}
for the case $\opertbase^\pm=0$ and impose $\Ppp_c=0$, which is the translation of
$\Pp_c=0$ under \eqref{eq:subs_2_to_1_rho_p}.  Applying again the translation
    \eqref{eq:subs_2_to_1_rho_p} to the outcome of this Proposition
    gives \eqref{res:K1_prop} as well as all the listed properties of
    $G(r)$.

    The parameter $\Omegaper$ was called $\Omegaperbaseint$ in the \emph{base scheme} and takes the value (by Proposition \ref{prop:problem_Wtwo} after appyling the translation \eqref{eq:subs_2_to_1_param})
  $\Omegaperbaseint = - \opert_c  G_{\infty}$. 
  Moreover, Remark \ref{res:final_Ppp},
  applied to $\opertbase=0$ and
    $h^\gaugefinal=0$ provides $\Epp=\Ppp=0$.
The translation \eqref{eq:subs_2_to_1_rho_p} provides the final claim.
Calling $\Ppp_c$ the ``perturbed pressure at the origin''
  is justified by Lemma \ref{res:Ppp_gauge_inv}
\finn
\end{proof}

\subsection{Existence and uniqueness to second order}
Given the previous result for the first order problem, the second order problem
just follows the \emph{base  perturbation scheme}, and we just need to make direct use of
Proposition \ref{res:hkmf_uniqueness}. We make the result and the argument precise in the
following. 

\begin{theorem}[\textbf{Rotating stars to second order}]
  \label{res:second_order}
  Consider a $C^{n+1}$ perfect fluid ball configuration
  according to Definition \ref{def:background}
  with $n\geq 4$  and $\Ebc+\Pbc\neq 0$.
  Let us be given a $C^{n+1}$ maximal
  perturbation scheme $(\mmm_\pertp, \gfam_\pertp,\{\psi_\pertp\})$
  inheriting the 
  Abelian
  $G_2$ generated by $\{ \xi=\partial_t, \eta=\partial_\phi\}$.
  Assume that the corresponding first and second order perturbation tensors $\Kper$
  and $\Kperper$
  \begin{itemize}
  \item solve the perturbed equations to second order for a rigidly rotating perfect fluid
    with the barotropic equation of state of the background
    in $\mmm^+$  and for vacuum in $\mmm^-$,
  \item satisfy the matching conditions to second order across the boundary of the fluid ball,
  \item are bounded,
  \item the perturbed pressure at the origin vanishes, $\Pp_c=\Ppp_c=0$.
  \end{itemize}
  Then there exist  first and second order gauge vectors
  (at each interior and exterior regions $\mmm^\pm$)
  such that the gauge transformed tensors
    $\Kper^{\gaugefinal}$ and $\Kperper^{\gaugefinal}$ are of class
    $C^2(\mmm^+)\cap C^{n+2}(\mmm^+\setminus\centresph_0)\cap C^{\infty}(\mmm^-)\cap C^1(\mmm)$
    and $C^0(\mmm^+)\cap C^{n}(\mmm^+\setminus\centresph_0)\cap C^{\infty}(\mmm^-)$ respectively,
    and take the form
  \begin{align}
    &\Kper^{\gaugefinal}=-2\opert_c (G(r)-G_\infty) r^2 \sin^2\theta dt d\phi,\label{res:K1_theorem}\\
    &\Kperper^{\gaugefinal} = \left ( -4 e^{\nu(r)} h(r,\theta) +
      2 \opert^2_c (G(r)-G_\infty)^2 r^2 \sin^2 \theta  \right ) dt^2
      + 4 e^{\lam(r)} m(r,\theta) dr^2 \nonumber \\
    &+ 4 k(r,\theta) r^2 \left ( d \theta^2 + \sin^2 \theta d \phi^2 \right )
      - 2  \Wtwo_c (G(r)-G_\infty)r^2 \sin^2\theta  dt d\phi\label{res:K2_theorem},
  \end{align}
  with
  \begin{align}
    &h(r,\theta)= h_0(r)+ h_2(r) P_2(\cos\theta),\nonumber\\
    &m(r,\theta)= m_0(r)+ m_2(r) P_2(\cos\theta),\label{res:hmk_theorem}\\
    &k(r,\theta)= k_2(r) P_2(\cos\theta),\nonumber
  \end{align}
  where $\opert_c, \Wtwo_c \ \in \mathbb{R}$ are free parameters and
  \begin{itemize}
  \item[(i)]$G(r)$  is the unique $C^1(0,\infty)$
    solution of \eqref{eq:Gg1} in $D$ 
    with $G(0)=1$.
    Moreover,  $G_\infty=\lim_{r\to \infty}G(r)= G(\ro)+\ro G'(\ro)/3>1$, $G'(\ro)>0$ and
    $G(r)$ in $D^+$ extends to a
    $ C^2([0,\ro])\cap C^{n+2}((0,\ro])$ function.
  \item[(ii)] The functions $h_0(r),h_2(r),m_0(r),m_2(r),k_2(r)$
    are of class
    $C^{n}((0,\ro])\cap C^{\infty}([\ro,\infty))$, extend continuously to $r=0$,
    are bounded,  are 
    uniquely determined by $\opert_c$,
    and all vanish if $\opert_c=0$.
    \item[(iii)] The rotation parameters
      $\Omegaper = - \opert_c G_{\infty}$ and $\Omegaperper = - \Wtwo_c G_{\infty}$.  \end{itemize}
\end{theorem}

\begin{remark}\label{remark:global_solution}
    In all the expressions referred to in this remark
    the replacement $\opertbase_+-\Omegaperbaseint\to \opert_c G$
      is to be made.
      
      The global solution in the gauge $\gaugefinal$
    is obtained, apart from $G(r)$,
    in terms of two fully determined
    functions $\newsigma_*(r)$ and $\vhat_2(r)$.
        The function $\vhat_2(r)$ is the unique bounded solution
    of the equation \eqref{eq:uhat2} with \eqref{eq:mu2} and \eqref{eq:f_omega}
        that satisfies  the matching conditions
    \[
      [\vhat_2]=0,\qquad
      [\vhat_2']=\frac{1}{6}e^{\lambda(\ro)}(1+e^{\lambda(\ro)})\ro^3\chiefes\Eb_+(\ro) G^2(\ro),
    \]    
    c.f. Proposition \ref{prop:problem_vhat}.
    The function $\newsigma_*(r)$
    is the unique bounded solution
    of the equation \eqref{eq:for_zeta}  with \eqref{eq:mu2} and \eqref{eq:f_omega}
      that satisfies \eqref{newsigma_star_origin} and the
    matching conditions
    \[
      [\newsigma_*]=0,\qquad
      [\newsigma_*']=-\frac{2}{3}e^{\lambda(\ro)}(1+e^{\lambda(\ro)})\ro^3\chiefes\Eb_+(\ro) G^2(\ro),
    \]
    and is determined by the value $\Ppp_c=\lim_{r\to 0}\Ppp$
    through the relation
    \eqref{boundary:newsigmastar}, 
    c.f. Proposition \ref{prop:newsigma_star_unique}.
        Although in Theorem \ref{res:second_order}
      we have set $\Ppp_c=0$, the solution has been obtained
      for general values of the second order perturbed pressure $\Ppp_c$. This may be of independent interest.
  \begin{itemize}
  \item In the interior region we have $\Ep=\Pp=0$, and
      the functions $\{h,m,k\}$ and $\{\Epp,\Ppp\}$
    are given by the
    right-hand sides of \eqref{hg_+_new}-\eqref{mg_+_new} and Remark \ref{res:final_Ppp}, respectively, 
    with $\keygamma(r)$ as defined in \eqref{def:keygamma}, $\F$ as in 
    \eqref{eq:f_omega}
        and $\newsigma_*^+$, $\vhat^+_2$ being the interior parts of
      $\newsigma_*$, $\vhat_2$.
          
  \item In the exterior region the functions $\{h,m,k\}$ are
    given by the right-hand sides of \eqref{hg_-_new}-\eqref{mg_-_new}
    with  $\keygamma_{\star}$ as defined in \eqref{def:keygamma_star_nss}
    and $J_{\opertbase} \to - \opert_c G'(a) a^4/6$, with $\newsigma^-_*$
      and $\vhat_2^-$ being the exterior parts of $\newsigma_*$ and $\vhat_2$.
  \end{itemize}
\end{remark}

\begin{proof}
  All the hypotheses of Proposition \ref{res:first_order} are satisfied, so \eqref{res:K1_theorem} as well as item \emph{(i)} follow readily.
We now use the  \emph{base perturbation scheme}.
The first point \baseback{} is satisfied by assumption. Point \baseKper{}  holds with $\Kper=\Kper^\gaugefinal$ and
$\opertbase=\opert_c(G-G_\infty)$, $\Omegaperbaseint=-\opert_c G_\infty$.
Proposition \ref{prop:rigidrot_OT} and Theorem \ref{theo:paper1}
  imply the existence of a second order gauge vector that transforms the second
  order tensor onto the form $\Kperper^\Psi{}^\pm$ as given in \eqref{res:sopert_tensor}  on each $M^\pm$.
  By the same Theorem \ref{theo:paper1}, point \baseKperper.1 of the \emph{base perturbation scheme}
is satisfied with $m=n-2$. Since
$n\geq 4$ by assumption, the condition $m\geq 2$ holds.
Finally, \baseKperper.2, \basematch{} and \basebaro{} are satisfied by
assumption.

It suffices now to apply Proposition \ref{res:hkmf_uniqueness} with
  $\opertbase=\opert_c (G-G_\infty)$ to conclude
 \eqref{res:K2_theorem}, \eqref{res:hmk_theorem}, as well as
point \emph{(ii)}.
As before, calling $\Ppp_c$ the ``perturbed pressure at the origin''
  (at second order) is justified by Lemma
  \ref{res:Ppp_gauge_inv}. Item \emph{(iii)} follows because $\Omegaperper$ corresponds
  to $\Omegaperperbaseint$ in the  base scheme and its value was given in
Proposition \ref{prop:problem_Wtwo} (the value of
$\Omegaper$ already appears in Proposition \ref{res:first_order}).

The matching conditions for $\newsigma_*$ and $\vhat_2$ in the
Remark follow from \eqref{eq:mc_newsigma_star} and \eqref{eq:mc:vhat2}
respectively, for $\opertbase=\opert_c (G-G_\infty)$ and
$\Omegaperbaseint=-\opert_c G_\infty$, so that
$\opertbase_+(\ro)-\Omegaperbaseint=\opert_c G(\ro)$ and
therefore $[\F]=\frac{1}{3}e^{\lambda(\ro)}\ro^4\chiefes \Eb_+(\ro) G^2(\ro)$,
c.f. \eqref{dif_qhat2}. The expressions are simplified using
\eqref{nup_value} and that $\nu(\ro)=-\lambda(\ro)$,
c.f. \eqref{eqs_back_vacuum} and \eqref{background_matching}.\fin
\end{proof}

We finish the paper by writing down the family of metrics $\gfamp_\pertp$ to second order
for the gauge $\gaugefinal$ obtained in this theorem.
The final step will be to exploit the freedom in  redefining 
the perturbation parameter $\pertp \to \tilde\pertp=\tilde\pertp(\pertp)$,
inherent to any perturbation theory, as well as the
scalability of the perturbations,
to obtain the clasical
form of the stationary and axially symmetric  perturbations
around static balls.

From Theorem \ref{res:second_order} we have
\begin{align}
  g_\pertp=g&-\pertp(2\opert_c+\pertp \Wtwo_c)( G(r)-G_\infty) r^2\sin^2\theta dt d\phi\nonumber\\
            &+\pertp^2\left\{\left ( -2 e^{\nu(r)} h(r,\theta) +\opert^2_c (G(r)-G_\infty)^2 r^2 \sin^2 \theta  \right ) dt^2
              + 2 e^{\lam(r)} m(r,\theta) dr^2 \right.\nonumber \\
            &\left.+ 2 k(r,\theta) r^2\left ( d \theta^2 + \sin^2 \theta
              d \phi^2 \right ) \frac{}{}\right\}+ O(\pertp^3).\label{res:fam_g_e}
\end{align}
If $\opert_c=0$ then \eqref{res:fam_g_e} reduces to
$g_\pertp=g-\pertp^2 \Wtwo_c ( G(r)-G_\infty) r^2\sin^2\theta dt d\phi
+ O(\pertp^3)$, which simply means that the perturbation is set to start
at second order, which then becomes the first non-trivial order and takes exacly the
same form as the first order with $\opert_c=\Wtwo_c/2$. 
We can thus assume $\opert_c\neq 0$ without loss of generality.
The change $\tilde\pertp=\pertp(\opert_c+\pertp \Wtwo_c/2)$,
after a suitable rescaling
\[\{h,m,k\}\to \{h/\opert_c^2,m/\opert_c^2,k/\opert_c^2\}
\]
yields
\begin{align}
g_{\tilde\pertp}=g&-2\tilde\pertp ( G(r)-G_\infty) r^2\sin^2\theta dt d\phi\nonumber\\
                  &+\tilde\pertp^2\left\{\left ( -2 e^{\nu(r)} h(r,\theta)
                    +( G(r)-G_\infty)^2 r^2 \sin^2 \theta  \right ) dt^2
                    + 2 e^{\lam(r)} m(r,\theta) dr^2 \right.\nonumber \\
                  &\left.+ 2 k(r,\theta) r^2\left ( d \theta^2 + \sin^2 \theta
                    d \phi^2 \right ) \frac{}{}\right\}+ O(\tilde\pertp^3).\label{res:fam_g_te}
\end{align}
This redefinition amounts to setting $\opert_c=1$ and $\Wtwo_c=0$.
  Let us stress the fact that all functions in this last expression are unique,
and that the only free paramenter that enters
the first and second order perturbations, which we have taken to be $\opert_c$,
is now integrated into $\tilde\pertp$. This corresponds to the property, widely used in the literature, that stationary and axially symmetric perturbations
  of fluid balls to second order depend on a single free parameter,
  that  can be encoded
  in the perturbation parameter and which physically is related to the rotation of the fluid. As discussed in the introduction, establishing this fact rigorously was one of the aims of  this paper.
  
Let us finally stress that the relation of the parameters and functions
in \eqref{res:fam_g_te} with the rotation of the star, as measured by the static observer
$\stat$ at infinity, requires the full control of the gauges and the jump
at the surface of the relevant functions, and is of global nature and gauge invariant,
as discussed in \cite{ReinaVera2015}.
One detail that is missing in \cite{ReinaVera2015} is that boundedness of the perturbation
forces the gauges (at first order) to be fixed so that $\stat$
remains to be the stationary observer for $g_\pertp$ at infinity.
This fixing of gauges, as done in the present paper, yields the following fluid velocity
\begin{align}
  \label{velocity}
u= F_{\pertp}
(\stat -\tilde\pertp G_\infty\axial) +  O(\tilde\pertp^3), \qquad
F_{\pertp} : =   e^{-\frac{\nu}{2}} +
    \frac{1}{2}\tilde\pertp^2 e^{- \frac{3\nu}{2}}
\left ( -2 e^\nu h+G^2 r^2\sin^2\theta \right ).
\end{align}
Indeed, the redefinition $\{\opert_c \to 1, \Wtwo_c \to 0\}$
 gives (by Theorem \ref{res:second_order}, item (ii))
  $\{\Omegaper = - G_{\infty} = - (G(\ro)+\ro G'(\ro)/3) <-1 , \Omegaperper=0\}$
  and \eqref{velocity}  follows from \eqref{uper} and \eqref{uperper} applied to
\eqref{res:fam_g_te}.
The velocity of rotation of $u_\pertp$ along $\axial$ as
measured with respect to
$\stat$ at infinity is thus given by $-\tilde\pertp G_\infty$,
and vanishes iff $\tilde\pertp=0$.

\section*{Acknowledgements}

We thank Alfred Molina for very interesting discussions on this problem and
  related matters. M.M. acknowledges financial support under the projects
PGC2018-096038-B-I00
(Spanish Ministerio de Ciencia, Innovaci\'on y Universidades and FEDER)
and SA083P17 (JCyL).
B.R. and R.V. acknowledge financial support under
the projects FIS2017-85076-P (Spanish Ministerio de Ciencia, Innovación y Universidades and
FEDER) and IT-956-16 (Basque Government). B.R. was supported by the post-doctoral grant
POS-2016-1-0075 (Basque Government). Many calculations have been performed using the free
PSL version of the REDUCE computer algebra system.

\appendix

\section{First order perturbed Ricci tensor in covariant form}
\label{tfi_sector}
In this appendix we derive, in  a fully covariant manner, the first order perturbed Ricci tensor in backgrounds admitting
two Killing vectors $\xi$ and $\eta$ satisfying suitable conditions
(see Proposition \ref{prop:R1} below) and for metric perturbation tensors $\Kper$
with a single component along the $\xi,\eta$ direction. It turns out that
 the perturbed Ricci tensor preserves this structure, namely its only non-zero component is again along the $\xi,\eta$ direction. In practice
this implies a decoupling of the perturbed field equations.

In the main text we apply these results in the context 
of static and spherically symmetric backgrounds. However, the decoupling
holds in much more generality. In view of their potential interest for other
problems we present the general result.

We start by writing down an expression for the perturbed Ricci tensor in an arbitrary
background when the metric perturbation tensor 
$\Kper$ is splitted as
${\Kper}_{\alpha\beta} = \y \Sig_{\alpha\beta}$, where $\y$ and $\Sig_{\alpha\beta}$ are for the moment any $C^2$ scalar and symetric $(0,2)$-tensor, respectively.
Directly from the definition (\ref{calSper})
\begin{align}
&\Sper_{\mu\alpha\beta} =
\frac{1}{2} \left ( \nabla_{\alpha} \y \, \Sig_{\mu\beta} 
+ \nabla_{\beta} \y \, \Sig_{\mu\alpha} - \nabla_{\mu} \y \, \Sig_{\alpha\beta} \right )
+ \y \PP_{\mu\alpha\beta}, \label{calSperP} \\
&\PP_{\mu\alpha\beta} :=
 \frac{1}{2} \left ( \nabla_{\alpha} \Sig_{\mu\beta} + \nabla_{\beta} \Sig_{\mu\alpha}
- \nabla_{\mu} \Sig_{\alpha\beta} \right ).
\nonumber 
\end{align}
Taking the trace in $\mu\alpha$ in (\ref{calSper}) yields immediately
\begin{align}
\Sper^{\mu}_{\phantom{\mu}\mu\beta} = \frac{1}{2} \nabla_{\beta} {\Kper}^{\mu}_{\phantom{\mu}\mu}
  = \frac{1}{2} \nabla_{\beta} \left ( \y \Sig^{\mu}_{\phantom{\mu}\mu} \right ).
  \label{tr1}
\end{align}
We also need to compute
\begin{align}
\nabla_{\mu} \Sper^{\mu}_{\phantom{\mu}\alpha\beta}
= &  \phantom{+} \frac{1}{2} \left ( \nabla^{\mu} \nabla_{\alpha} \y \, \Sig_{\mu\beta}
+ \nabla^{\mu} \nabla_{\beta} \y \, \Sig_{\mu\alpha}
- \nabla^{\mu} \nabla_{\mu} \y\, \Sig_{\alpha\beta} \right ) \nonumber \\
& + \frac{1}{2} \left ( \nabla_{\alpha} \y \, \nabla^{\mu} \Sig_{\mu\beta}
+ \nabla_{\beta} \y \, \nabla^{\mu} \Sig_{\mu\alpha} 
- 2 \nabla^{\mu} \y \, \nabla_{\mu} \Sig_{\alpha\beta}
+ \nabla^{\mu} \y \, \nabla_{\alpha} \Sig_{\mu\beta}
+ \nabla^{\mu} \y \, \nabla_{\beta} \Sig_{\mu\alpha}
\right ) \nonumber \\
& + \y \nabla_{\mu} \PP^{\mu}_{\phantom{\mu}\alpha\beta}  \nonumber \\
= & \phantom{+}
\frac{1}{2} \Big ( \nabla_{\alpha} \left ( \nabla^{\mu} \y \Sig_{\mu\beta} \right )
+\nabla_{\beta} \left ( \nabla^{\mu} \y \Sig_{\mu\alpha} \right )
- \nabla^{\mu} \nabla_{\mu} \y \, \Sig_{\alpha\beta} \Big )
- \nabla_{\mu} \y \nabla^{\mu} \Sig_{\alpha\beta} \nonumber \\
& + \frac{1}{2} \left ( \nabla_{\alpha} \y \nabla^{\mu} \Sig_{\mu\beta} +
\nabla_{\beta} \y \nabla^{\mu} \Sig_{\mu\alpha} \right ) +  
 \y \nabla_{\mu} \PP^{\mu}_{\phantom{\mu}\alpha\beta}. \label{tr2}
\end{align}
Inserting (\ref{tr1})-(\ref{tr2})  into (\ref{Ricpera}), yields
the following (fully general) identity:
\begin{align*}
R^{(1)}_{\alpha\beta}
= & \phantom{+}
\frac{1}{2} \left ( \frac{}{}\nabla_{\alpha} \left ( \nabla^{\mu} \y \, \Sig_{\mu\beta} \right )
+\nabla_{\beta} \left ( \nabla^{\mu} \y \, \Sig_{\mu\alpha} \right )
- \nabla^{\mu} \nabla_{\mu} \y \, \Sig_{\alpha\beta} \right )
- \nabla_{\mu} \y \,\nabla^{\mu} \Sig_{\alpha\beta} \\
& + \frac{1}{2} \left ( \frac{}{}\nabla_{\alpha} \y \, \nabla^{\mu} \Sig_{\mu\beta} +
\nabla_{\beta} \y  \, \nabla^{\mu} \Sig_{\mu\alpha} \right ) +
 \y  \,\nabla_{\mu} \PP^{\mu}_{\phantom{\mu}\alpha\beta}
- \frac{1}{2} \nabla_{\alpha} \nabla_{\beta} \left ( \y \, \Sig^{\mu}_{\phantom{\mu}\mu} \right ).
\end{align*}
We now assume that $(\mmm,g)$ admits two Killing vectors $\xi$ and 
$\eta$ and that $\Sig_{\alpha\beta} = \xi_{\alpha} \eta_{\beta}
+ \xi_{\beta} \eta_{\alpha}$. Then
$\nabla^{\mu} \y \, \Sig_{\mu\alpha} = \xi(\y) \eta_{\alpha}
+ \eta(\y) \xi_{\alpha}$, 
$\Sig^{\mu}_{\phantom{\mu}\mu} =
2 \la \xi, \eta \ra$ and
\begin{align*}
\nabla^{\mu} \Sig_{\mu\alpha} = 
\xi_{\mu} \nabla^{\mu} \eta_{\alpha} 
+ \eta_{\mu} \nabla^{\mu} \xi_{\alpha}
= -\nabla_{\alpha} \la \xi, \eta \ra.
\end{align*}
Moreover, the Killing equations also imply
\begin{align}
\PP_{\mu\alpha\beta} = \nabla_{\alpha} \xi_{\mu} \, \eta_{\beta}
+ \nabla_{\beta} \xi_{\mu} \, \eta_{\alpha}
+ \nabla_{\alpha} \eta_{\mu} \, \xi_{\beta}
+ \nabla_{\beta} \eta_{\mu} \, \xi_{\alpha},
\label{PP}
 \end{align}
whose divergence is, after using the standard identity
$\nabla_{\mu} \nabla_{\alpha} \xi^{\mu} = R_{\alpha\mu} \xi^{\mu}$ (and similarly
for $\eta$),
\begin{align*}
\nabla_{\mu} \PP^{\mu}_{\phantom{\mu}\alpha\beta}
= R_{\alpha\mu} \xi^{\mu} \eta_{\beta} + 
R_{\beta\mu} \xi^{\mu} \eta_{\alpha}+ 
R_{\alpha\mu} \eta^{\mu} \xi _{\beta} + 
R_{\beta\mu} \eta^{\mu} \xi_{\alpha}  
- 2 \nabla_{\alpha} \xi^{\mu} \nabla_{\beta} \eta_{\mu}
- 2 \nabla_{\beta} \xi^{\mu} \nabla_{\alpha} \eta_{\mu}.
\end{align*}
Putting everything together, it follows that 
\begin{align}
R^{(1)}_{\alpha\beta} = &
\nabla_{(\alpha} \left ( \xi(\y) \eta_{\beta)} \right ) + 
\nabla_{(\alpha} \left ( \eta(\y) \xi_{\beta)} \right )
- \frac{1}{2} \left ( \nabla_{\mu} \nabla^{\mu}  \y \right ) \Sig_{\alpha\beta}
- \nabla^{\mu} \y \, \nabla_{\mu} \Sig_{\alpha\beta} - 
\nabla_{(\alpha} \y \nabla_{\beta)} \la \xi , \eta \ra \nonumber \\
& 
+ 2 \y \left ( R_{\mu(\alpha} \xi^{\mu} \eta_{\beta)} 
+ R_{\mu(\alpha} \eta^{\mu} \xi_{\beta)}  \right )
- 4 \y \nabla_{(\alpha} \xi^{\mu} \nabla_{\beta)} \eta_{\mu}- \nabla_{\alpha} \nabla_{\beta} \left (  \y \la \xi, \eta \ra \right ).
\label{Ric1prod}
\end{align}
where brackets denote symmetrization.
This is 
a general identity valid for a perturbation tensor ${\Kper}_{\alpha\beta}$ 
of the form
${\Kper}_{\alpha\beta} =
 2 \y \xi_{(\alpha} \eta_{\beta)}$ with $\xi$ and $\eta$
Killing vectors of the background.

This general identity may have applications in several contexts. For
the purposes of this paper we need the following particular case:
\begin{proposition}
\label{prop:R1}
Let $(\mmm,g)$ be a spacetime admitting two Killing vectors 
$\xi$ and $\eta$ satisfying the following three conditions:
\begin{itemize}
\item[(i)] $\xi$ and  $\eta$ are perpendicular,  i.e. $\la \xi,\eta
\ra =0$,
\item[(ii)] $[\xi,\eta]=0$,
\item[(iii)] both $\xi$ and $\eta$ are hypersurface orthogonal and
non-null on an open set $U$. 
\end{itemize}
Consider a first order perturbation tensor ${\Kper}_{\alpha\beta} = \y 
\Sig_{\alpha\beta}$ with
$\Sig_{\alpha\beta} := 2 \xi_{(\alpha} \eta_{\beta)}$
and
$\y \in C^{2}(U)$ satisfying $\xi(\y) = \eta(\y)=0$.
 Then, on $U$, the first
order perturbation of the Ricci tensor is
\begin{align}
R^{(1)}_{\alpha\beta} = &
\Sig_{\alpha\beta} \left (( \riccixi + \riccieta) \y
- \frac{1}{2 \la \xi, \xi \ra \la \eta , \eta \ra}
\nabla_{\mu} \Big ( 
\la \xi, \xi \ra \la \eta, \eta \ra \nabla^{\mu} \y \Big )
- \frac{\y}{2 \la \xi, \xi \ra \la \eta , \eta \ra}
\nabla_{\mu} \la \xi, \xi \ra \nabla^{\mu} \la \eta,\eta \ra 
\right ),
\label{Ricperetaxi}
\end{align}
where $\riccixi$ and $\riccieta$ are defined by
$\riccixi = \frac{1}{\la \xi, \xi \ra} \Ric(\xi,\xi)$ and 
$\riccieta = \frac{1}{\la \eta, \eta \ra} \Ric(\eta,\eta)$.
\label{FirstOrder}
\end{proposition}
\begin{remark}
  \label{rem:prop:R1}
In the main text we use this result several times. For notational
simplicity, it is convenient to define the second order
differential operator
$\riccifunc(f)$ 
\begin{align}
\riccifunc(f):= &
\left (( \riccixi + \riccieta) f
- \frac{1}{2 \la \xi, \xi \ra \la \eta , \eta \ra}
\nabla_{\mu} \Big ( 
\la \xi, \xi \ra \la \eta, \eta \ra \nabla^{\mu} f \Big )
- \frac{f}{2 \la \xi, \xi \ra \la \eta , \eta \ra}
\nabla_{\mu} \la \xi, \xi \ra \nabla^{\mu} \la \eta,\eta \ra
\right )
\label{eq:ricci_functional}
\end{align}
with $\riccixi$ and $\riccieta$
as above, so that (\ref{Ricperetaxi}) is simply
$R^{(1)}_{\alpha\beta} =\riccifunc(\y) \Sig_{\alpha\beta}$.
\end{remark}
\begin{proof}
We work on $U$. Being hypersurface orthogonal
and non-null, the derivatives of $\xi$  and $\eta$ are necessarily of the form
\begin{align}
\nabla_{\alpha} \xi_{\beta} & = \xi_{\alpha} H_{\beta} 
- \xi_{\beta} H_{\alpha} \quad \quad
H_{\alpha} := - \frac{1}{2 \la \xi, \xi \ra } \nabla_{\alpha} \la \xi,\xi
\ra \label{nabxi} \\
\nabla_{\alpha} \eta_{\beta} & = \eta_{\alpha} M_{\beta} 
- \eta_{\beta} M_{\alpha} \quad \quad
M_{\alpha} := - \frac{1}{2 \la \eta, \eta \ra } \nabla_{\alpha} \la \eta,\eta
\ra. \label{nabeta}
\end{align}
For any Killing field $\xi$ and vector $X$ it holds ($\lie$ denotes Lie derivative)
  \begin{align*}
    \xi^{\mu} \nabla_{\mu} \la X,X \ra  = \lie_{\xi} \la X, X \ra =
    2 \la \lie_{\xi} X, X \ra .
  \end{align*}
As a consequence,  the commutation property $[\xi,\eta] =0$ implies
$\xi^{\alpha} M_{\alpha} = \eta^{\alpha} H_{\alpha}=0$. We can now compute 
\begin{align*}
&  \nabla_{\alpha} \xi^{\mu} \nabla_{\beta} \eta_{\mu}
= \left (\xi_{\alpha} H^{\mu} - \xi^{\mu} H_{\alpha} \right )
\left (\eta_{\beta} M_{\mu} - \eta_{\mu} M_{\beta} \right ) =
\xi_{\alpha} \eta_{\beta} H_{\mu} M^{\mu} 
\quad \quad  \Longrightarrow 
\\
&  2 \nabla_{(\alpha} \xi^{\mu} \nabla_{\beta)} \eta_{\mu} = 
\Sig_{\alpha\beta} \left ( H_{\mu} M^{\mu} \right ) 
=  \Sig_{\alpha\beta}
\left ( \frac{1}{4 \la \xi, \xi \ra \la \eta,\eta \ra}
\nabla_{\mu} \la \xi,\xi \ra \nabla^{\mu} \la \eta,\eta \ra 
\right )
\end{align*}
and  also
\begin{align*}
\nabla_{\mu} \Sig_{\alpha\beta}
& = 2\nabla_{\mu} \left ( \xi_{(\alpha} \eta_{\beta)} \right )=
2 \left ( \nabla_{\mu} \xi_{(\alpha} \right ) \eta_{\beta)} 
+ 2 \left ( \nabla_{\mu} \eta_{(\alpha} \right ) \xi_{\beta)}  \\
 & = 
2 \xi_{\mu} H_{(\alpha} \eta_{\beta)} 
+ 2 \eta_{\mu} M_{(\alpha} \xi_{\beta)}
- \left ( H_{\mu} + M_{\mu} \right ) \Sig_{\alpha\beta}
\\
& = 
2 \xi_{\mu} H_{(\alpha} \eta_{\beta)} 
+ 2 \eta_{\mu} M_{(\alpha} \xi_{\beta)}
+\frac{1}{2 \la \xi,\xi \ra \la \eta, \eta \ra}
\nabla_{\mu} \Big ( \la \xi, \xi \ra \la \eta, \eta \ra \Big )
\Sig_{\alpha\beta} .
\end{align*}
Hypersurface orthogonality implies that both $\xi$ and
$\eta$ are eigenvectors of the Ricci tensor, so that 
\begin{align*}
R_{\alpha\mu} \xi^{\mu} = \riccixi \xi_{\alpha}, \quad \quad
R_{\alpha\mu} \eta^{\mu} = \riccieta \eta_{\alpha}, 
\end{align*}
with $\riccixi$ and $\riccieta$ as defined in the Proposition. Thus,
under assumptions (i),(ii) and (iii),
the first order perturbation of the Ricci tensor (\ref{Ric1prod}) simplifies
to
\begin{align*}
R^{(1)}_{\alpha\beta} = &
\nabla_{(\alpha} \left ( \xi(\y) \eta_{\beta)} \right ) + 
\nabla_{(\alpha} \left ( \eta(\y) \xi_{\beta)} \right )
+  \xi(\y)  
\frac{1}{\la \xi, \xi \ra} \nabla_{(\alpha} \la \xi, \xi \ra
\eta_{\beta)} 
+  \eta(\y) 
\frac{1}{\la \eta, \eta \ra} \nabla_{(\alpha} \la \eta, \eta \ra
\xi_{\beta)} \\
&\hspace{-1mm} + 
\Sig_{\alpha\beta} \left ( (  \riccixi + \riccieta) \y
- \frac{1}{2 \la \xi, \xi \ra \la \eta , \eta \ra}
\nabla_{\mu} \Big ( 
\la \xi, \xi \ra \la \eta, \eta \ra \nabla^{\mu} \y \Big )
 - \frac{\y}{2 \la \xi, \xi \ra \la \eta , \eta \ra}
\nabla_{\mu} \la \xi, \xi \ra \nabla^{\mu} \la \eta,\eta \ra 
\right ).
\end{align*}
When, in addition, $\y$ is invariant under $\xi$ and
$\eta$, the first four terms vanish and
the perturbed Ricci tensor is proportional
to the metric perturbation tensor $\Kper$, with explicit
expression given in (\ref{Ricperetaxi}).\fin
\end{proof}

\begin{remark}
\label{nSperee}
We note that
under the assumptions of this Proposition, we may insert (\ref{nabxi})
and (\ref{nabeta}) into the expression for $\PP$ in (\ref{PP}) to get
\begin{align*}
\PP_{\mu\alpha\beta} = 
- \frac{1}{2} \frac{\nabla_{\mu} \left ( \la \xi, \xi\ra \la \eta,\eta \right )}{
\la \xi, \xi\ra \la \eta,\eta \ra} \Sig_{\alpha\beta}
+ 
\frac{1}{\la \xi, \xi \ra}
\xi_{\mu} \eta_{(\alpha} \nabla_{\beta)} \la \xi,\xi \ra
+
 \frac{1}{\la \eta, \eta \ra}
\eta_{\mu} \xi_{(\alpha} \nabla_{\beta)} \la \eta,\eta \ra 
\end{align*}
and  the tensor $\Sper$ (\ref{calSperP})  takes the following form
\begin{align}
\Sper_{\mu\alpha\beta} = &
\frac{1}{2} \left ( \nabla_{\alpha} \y \, \Sig_{\mu\beta} 
+ \nabla_{\beta} \y \, \Sig_{\mu\alpha} - \nabla_{\mu} \y \, \Sig_{\alpha\beta} \right )
- \frac{1}{2} \frac{\nabla_{\mu} \left ( \la \xi, \xi\ra \la \eta,\eta \right )}{
\la \xi, \xi\ra \la \eta,\eta \ra} {\Kper}_{\alpha\beta} \nonumber \\
 & + 
\frac{\y}{\la \xi, \xi \ra}
\xi_{\mu} \eta_{(\alpha} \nabla_{\beta)} \la \xi,\xi \ra
+
 \frac{\y}{\la \eta, \eta \ra}
\eta_{\mu} \xi_{(\alpha} \nabla_{\beta)} \la \eta,\eta \ra.
\label{calSperP2}
\end{align}
\end{remark}

\section{Geometrical stationary and axisymmetric perturbed matching to second order}
\label{app:staxi_perturbed_matching}
The perturbed matching to second order for the Hartle setup
presented in \cite{ReinaVera2015} assumes that
the perturbed matching hypersurface is axially symmetric,
so that the interior and exterior regions are stationary and axially symmetric both
in structure and in shape.
In this appendix we revisit that framework by dropping any assumption on the perturbation
of the matching hypersurface, thus considering the general case.
Furthemore, for the sake of generality, we will also include the radial functions $\Q_\pm$ at either side
in the background configuration.
To be more precise, in Propositions 1 and 2 in \cite{ReinaVera2015}, apart from having set $\Q(r)=r$, all
four functions $\Qone$, $\Tone$, $\Qtwo$, $\Ttwo$ on $\Supfamp$ are assumed not to depend on $\varphi$. We present
in the following the corresponding general results.

We start by recalling the perturbed matching theory
to second order, as developed
in \cite{Mars2005}
(see \cite{Battye01,Mukohyama00} for the first order).
We do this for completeness and also because, following \cite{ReinaVera2015},
it allows us to introduce a quantity 
with better gauge behaviour that simplifies the expressions 
to some extent.

The first order matching conditions
require the equality of two pairs of symmetric tensors
$\hfamp^{(1)}_{\pm}$, $\kfamp^{(1)}_{\pm}$ defined
on the background matching hypersurface $\Supfamp$. Geometrically, these tensors
correspond, respectively, to the linear pertubations of the first and 
second fundamental forms of the matching hypersurfaces
$\Sigma_{\epsilon}$ in the one-paramenter family of spacetimes $(M,\ge^{\pm})$ defining the perturbation. They take explicit forms
in terms of background quantities, the metric perturbation tensor
$\Kper^{\pm}$ and a vector field $Z_1^\pm$ along $\Supfamp$ which encodes
the first order variation of the matching hypersurface with 
$\epsilon$. Its decomposition into normal and tangential components
$Z_1^\pm=\Qone^\pm \normal^\pm+\vecTone^\pm$, where $\normal^{\pm}$ is the unit normal
to $\Supfamp^{\pm}$, introduces two scalars  $\Qone^\pm$ which describe
the deformation of $\Supfamp^\pm$ as a set of points, and two tangential vectors
$\vecTone^{\pm}$ which determine how the different points within the sets are
identified. The construction to second order is analogous and involves
tensors 
$\hfamp^{(2)}_{\pm}$, $\kfamp^{(2)}_{\pm}$ and vector fields
$Z_2^\pm = \Qtwo^\pm \normal^\pm+\Ttwo^\pm$ along 
$\Supfamp$. 

We drop the $\pm$ indexes for simplicity. 
The matching problem involves two independent gauges, the usual
spacetime gauge and a hypersurface gauge. The former involves two
vectors $\sper$ and $\sperper$ called (spacetime) gauge vectors
and affect $Z_1$ and $Z_2$ as \cite{Mars2005}
\begin{equation}
{Z}^{g}_1={Z}_1-\sper |_{\Supfamp}  ,\qquad
{Z}^{g}_2={Z}_2-\sperper-2\nabla_{{Z}_1} \sper+\nabla_{\sper}  \sper |_{\Supfamp}.
\label{eq:Z_spt_gauges}
\end{equation}
The hypersurface gauge involves
two vector fields  $ U_1$ (first order) and $ U_2$ (second order) both
tangential to  $\Supfamp$ and 
transform $Z_1$ and $Z_2$ as \cite{Mars2005}
\begin{equation}
{Z}^{h}_1={Z}_1+ U_1,\qquad
{Z}^{h}_2={Z}_2+ U_2+2\nabla_{{U}_1} Z_1- \sigma \kappa( U_1, U_1)  \normal,
\label{eq:Z_hyp_gauges}
\end{equation}
where $\sigma = +1$ when $\Supfamp$ is timelike and
$\sigma = -1$ when $\Supfamp$ is spacelike\footnote{In this paper
  we only deal with $\sigma = +1$, but here we present the general expressions
  in terms of the new variables.}.
One possible use of the hypersurface gauge is setting to zero
the tangential parts at one side of  $\vecTone$  and $\vecTtwo$
(either side but not both sides simultaneously).
Concerning the  effect of (\ref{eq:Z_hyp_gauges})
on the normal components, we observe that  
the scalar $Q_1$  is not affected at all. This just reflects the fact
that the hypersurface gauge does not modify the matching hypersurfaces
as sets of points
and  only affects
how they are identified pointwise. This is no longer true at second order. The underlying reason is that ${Z}_2$ measures ``accelerations'' (in the sense of second order changes) and  this has 
the not so obvious consequence that $Q_2$ is affected by ${U}_1$. From the second in 
(\ref{eq:Z_hyp_gauges}) it follows
\[
Q_2^{h} = Q_2 + 2 {U}_1 (Q_1) - \sigma \kappa(U_1, U_1 + 
2{\Tone}). \nonumber
\] 
This suggests the construction 
of the {\it hypersurface gauge invariant} quantity (cf. \cite{ReinaVera2015})
\begin{equation}
\hatQtwo := Q_2 + \sigma \kappa (\vecTone, \vecTone) - 2 \vecTone (Q_1). \label{def:Q2hat}
\end{equation}
We thereore rewrite the explicit expressions
of $\hfamp^{(1)}$, $\kfamp^{(1)}$, $\hfamp^{(2)}$, $\kfamp^{(2)}$  given in Propositions 2 and 3 in \cite{Mars2005} in terms of this gauge invariant quantity
$\hatQtwo$. The first order objects 
$\hfamp^{(1)}$, $\kfamp^{(1)}$  are independent of $Q_2$, so we simply reproduce from \cite{Mars2005}:
\begin{align}
\qper_{ij} = & 
\pounds_{\vecTone } h_{ij}  + 2 \Qone \kappa_{ij} + {\Kper}_{\alpha\beta} e^{\alpha}_i e^{\beta}_j,
\label{FirstPert.1}
\\
\kappaper_{ij} = & 
\pounds_{\vecTone } \kappa_{ij} 
- \sigma D_{i} D_{j} \Qone 
+ \Qone \left ( - \normal^{\mu} \normal^{\nu} R_{\alpha\mu\beta\nu}
e^{\alpha}_i e^{\beta}_j + \kappa_{il} \kappa^{\,\,l}_{j} \right )
+ \frac{\sigma}{2}  \Kpernornor
\kappa_{ij} - \normal_{\mu} \Sper^{\mu}_{\phantom{\mu}\alpha\beta}e^{\alpha}_i e^{\beta}_j, 
\label{FirstPert.2}
\end{align}
where $D$ is the Levi-Civita covariant derivative of the (background)
induced metric $h$ on $\Supfamp$,
$\Sper$ is defined in (\ref{calSper}), $e^{\alpha}_i$ 
are tangent vectors to $\Supfamp$ 
and $\Kpernornor := \Kper( \normal, \normal)$.

For second order quantities, we replace $Q_2$ in terms of $\hatQtwo$ in 
the expressions in \cite[Proposition 3]{Mars2005}. The result is
\begin{align}
\qperper_{ij} & = \lie_{\vecTtwo} h_{ij} + 2 \hatQtwo \kappa_{ij} + {\Kperper}_{\alpha\beta} e^{\alpha}_i e^{\beta}_j
+ 2 \lie_{\vecTone} \qper_{ij} - \lie_{\vecTone} \lie_{\vecTone} h_{ij} + \nonumber \\
& + \lie_{2 \Qone \Kpernortan - 2 \Qone \kappa (\vecTone )   -  D_{\vecTone} 
\vecTone } h_{ij} 
+ 4 \sigma \Qone \Kpernornor 
\kappa_{ij} + \nonumber \\
& + 2 \Qone^2 \left ( -\normal^{\mu} \normal^{\nu} R_{\alpha\mu\beta\nu} e^{\alpha}_i
e^{\beta}_j + \kappa_{il} \kappa^{l}_{j} \right ) + 2 \sigma D_{i} \Qone D_{j} \Qone
- 4 \Qone \normal_{\mu}\Sper^{\mu}_{\phantom{\mu}\alpha\beta} e^{\alpha}_i e^{\beta}_j,
\label{SecondPert.1} \\
\kappaperper_{ij} & =
\lie_{\vecTtwo} \kappa_{ij} - \sigma D_{i} D_{j} \hatQtwo 
+  \left ( \hatQtwo 
+ \sigma \Qone \Kpernornor 
\right )
\left ( - \normal^{\mu} \normal^{\nu} R_{\alpha\mu\beta\nu} e^{\alpha}_i e^{\beta}_j 
+ \kappa_{il} \kappa^{l}_{j} \right )
- \normal_{\mu} \Sperper^{\mu}_{\phantom{\mu}\alpha\beta} e^{\alpha}_i e^{\beta}_j 
 \nonumber \\
& + 2 \lie_{\vecTone} \kappaper_{ij}\nonumber\\
& + \kappa_{ij} \left ( \frac{\sigma}{2} \Kperpernornor - \frac{1}{4} (\Kpernornor)^2 - \sigma 
\left ( \Kpernortan_{l} + \sigma D_{l}\Qone \right )
\left ( \Kpernortan{}^{l} + \sigma D^{l}\Qone \right ) 
+ 2 \sigma \Qone \normal_{\mu} \normal^{\rho} \normal^{\delta}\Sper^{\mu}_{\phantom{\mu}\rho\delta} \right )
\nonumber \\
& +  \left ( \sigma \Kpernornor \normal_{\mu} + 2 \Kpernortan_{\mu} + 2 \sigma D_{\mu} \Qone \right ) 
\Sper^{\mu}_{\phantom{\mu}\alpha\beta} e^{\alpha}_i e^{\beta}_j
- 2 \Qone \normal_{\mu} \normal^{\nu} ( \nabla_{\nu}\Sper^{\mu}_{\phantom{\mu}\alpha\beta} ) e^{\alpha}_i e^{\beta}_j \nonumber \\
& - 2 \normal_{\mu} \normal^{\nu} \Sper^{\mu}_{\phantom{\mu}\alpha \nu} \left ( e^{\alpha}_i D_j \Qone 
+ e^{\alpha}_j D_i \Qone  \right )
- 2 \Qone \normal_{\mu} \Sper^{\mu}_{\phantom{\mu}\alpha \beta} e^{\beta}_l  
\left ( e^{\alpha}_i \kappa^{l}_j +  e^{\alpha}_j \kappa^{l}_i \right ) \nonumber \\
&  + \lie_{
- \frac{1}{2} \Kpernornor \grad (\Qone) + 2 \sigma \Qone \kappa ( \grad(\Qone) )} \,  h_{ij}  + \frac{1}{2} \left ( D_{i} \Qone D_j \Kpernornor 
+ D_{j} \Qone D_i \Kpernornor \right ) \nonumber \\
& -  \lie_{\vecTone} \lie_{\vecTone} \kappa_{ij} 
-  \lie_{2 \Qone \kappa (\vecTone ) +  D_{\vecTone} \vecTone } \, \kappa_{ij} 
- 2 \sigma \Qone \lie_{\grad (\Qone) } \kappa_{ij}
\nonumber \\
& - \Qone^2 \left ( 
\normal^{\mu} \normal^{\nu} \normal^{\delta} ( \nabla_{\delta} R_{\alpha\mu\beta\nu} ) e^{\alpha}_i e^{\beta}_j
+ 2 \normal^{\mu} \normal^{\nu} R_{\delta\mu\alpha\nu} e^{\delta}_l e^{\alpha}_j \kappa^{l}_i
+ 2 \normal^{\mu} \normal^{\nu} R_{\delta\mu\alpha\nu} e^{\delta}_l e^{\alpha}_i \kappa^{l}_j \right ),
\label{SecondPert.2}
\end{align}
where $\Sperper$ is given in (\ref{calSperper}),
$\Kperpernornor := \Kperper( \normal, \normal)$, 
$ \Kpernortan$  is the tangent vector defined by
$h(\Kpernortan, e_i ) = \Kper( \normal,  e_i)$,
$D_{\mu} Q$ is defined
by $\{ \normal^{\mu} D_{\mu} \Qone =0$, $e^{\mu}_i D_{\mu} \Qone = D_i \Qone\}$  
and, for any tangent vector ${V}$,
$\kappa( {V})$ is the tangent vector with components $\kappa^{i}_{\, j} V^j$.

The perturbed matching conditions at first order \cite{Battye01,Mukohyama00,
Mars2005}
demand the existence
of $\Qone^{\pm}$ and $\vecTone^{\pm}$ such that 
$[\hp]=[\kp]=0$. At second  order \cite{Mars2005}
the perturbed matching conditions hold iff
there exist $\hatQtwo^{\pm}$ and $\vecTtwo^{\pm}$ such that
$[\hpp]=[\kpp]=0$. 

We may now apply the  perturbed matching theory to our specific setting.
As in the main text, for any pair of quantities $F^{\pm}$
(we use $+$ and $-$ as super or subindexes indistinctly)
on $\Supfamp$ satisfying $[F] =0$ we simply write $F^{+} = F^- =:F$.
\begin{proposition}
\label{teo:first_order_matching} 
Let $(\mmm,\gback)$ be a static and spherically symmetric
spacetime with two regions as in Definition \ref{SpherRegion}.
Consider the metric perturbation tensors $\fpt^\pm$ of the form 
\begin{equation}
\Kper = -2 \omega(r,\theta) \Q^2(r) \sin^2 \theta dt d \phi
\label{teo:Kper}
\end{equation}
at either side $M^\pm$.
Let us assume that
\begin{align}
  (i) \quad   \normal(\Q)  |_{\Supfamp} \neq 0, \quad 
  (ii) \quad  \left .  \left ( \frac{1}{2}  \normal(\nu) 
  - \frac{  \normal(\Q)}{\Q} + \frac{1}{\Q  \normal(\Q)} 
\right )
\right |_{\Supfamp} \neq 0,
  \quad \
  (iii) \quad  \normal(\nu)|_{\Supfamp} \neq 0,
  \label{conditions}
\end{align}
where $ \normal := - e^{-\lambda/2} \partial_r$.
The perturbations $\fpt^\pm$ 
satisfy the first order matching conditions if and only if there
exists a constant $b_1$ such that
\begin{align}
\left[\omega \right]= b_1 \in \mathbb{R}, \qquad
\left[ \normal(\omega)\right]=0. \label{teo:omegap_matching_R}
\end{align}
Moreover, introducing the quantities
\begin{equation}
\B = \frac{1}{2} e^{\nu} \left (  \normal(   \normal(\nu ))
+ \frac{1}{2} ( \normal(\nu ))^2 \right ),
\qquad
\A = - \Q  \normal(  \normal(\Q )),
\label{DefUpsilon}
\end{equation}
the deformation vectors 
$Z_1^{\pm} = \Qone^{\pm} \normal^{\pm} + \vecTone^{\pm}$
must satisfy
\begin{align}
[\vecTone]= 
b_1 \tau  \eta  +  \zeta, \qquad
\left[Q_1\right]=0,\qquad \left[\B\right]Q_1=0,
\qquad
\left[\A\right]Q_1=0, \label{teo:Q1_matching_R} 
\end{align}
where $ \eta = \partial_{\varphi}$ and
$ \zeta$ is an arbitrary background Killing vector.
\end{proposition}

\begin{remark}
  This proposition holds in full generality, i.e. no a priori
  restriction (such as e.g. axial symmetry) is assumed on
  how the matching hypersurface is deformed to first order.
\end{remark}
\begin{remark}
  Conditions (\ref{conditions}) are well-defined because the expresions they involve agree when computed from  either side of the matching hypersurface. This is a consequence of the  background matching conditions (\ref{background_matching_R}). 
\end{remark}

\begin{proof}
Let $\gsph$ be the standard round sphere on $\mathbb{S}^2$ and denote by 
$\D$ its associated derivative. We use coordinates $\{\tau,\vartheta,\varphi\} :=
\{ \tau, x^A \}$ on $\Supfamp$. For any tangent vector $ V =
V^{\tau} \partial_{\tau} + V^A \partial_{x^A}$, and any symmetric  tensor
$\S = a_1  d\tau^2 + a_2 \gsph$, with $a_1$, $a_2$ constants, the
Lie derivative $\lie_{ V} \S$ can be expressed as 
\begin{align}
\lie_{ V}  \S = 2 a_1 \dot{V}^{\tau} d \tau^2 
+ 2 \left ( a_1 \D_A V^{\tau} + a_2 \partial_{\tau} V_A \right )
d \tau dx^A + a_2 \left ( \D_{A} V_B + \D_B V_A \right ) 
d x^A d x^B
\label{EasyIden}
\end{align}
where the dot denotes derivative with respect to $\tau$ and
Latin indices $A,B,\cdots$ are raised and lowered with $\gsph$.
By (\ref{h0ij})-(\ref{k0ij}),  
the tensors $h_{ij}$, $\kappa_{ij}$  are both of this form.
For notational convenience we write them as
\begin{align}
& h = \alpha_1 d\tau^2 + \alpha_2 \gsph,
\quad \quad
\kappa = \beta_1 d\tau^2 + \beta_2 \gsph,
\label{hkappa}  \\
& \alpha_1 := - e^{\nu} |_{\Supfamp},  \quad \quad
\alpha_2 := \Q^2 |_{\Supfamp} , \quad \quad
\beta_1 := - \frac{1}{2} e^{\nu }  \normal(\nu) |_{\Supfamp}, \quad \quad
\beta_2 := \Q  \normal(\Q) |_{\Supfamp}.
\label{defsalphabeta}
\end{align}
Note that  $\alpha_1, \alpha_2$ are both non-zero.
The first set of matching conditions (\ref{FirstPert.1}) are
\begin{align}
\lie_{[\vecTone]} h + 2 [\Qone] \kappa + [\Phi^{\star} (\Kper)] =0,
\label{FirstFirst}
\end{align}
where $\Phi^{\star}$ is the pull-back to $\Supfamp$. Note that 
the last term in (\ref{FirstFirst}) has components only in $d\tau d\varphi$.
Applying (\ref{EasyIden}), the $\{ A,B\} $ components of (\ref{FirstFirst}) read
\begin{align}
\alpha_2 \left ( \D_A [\Tone]_B + \D_B [\Tone]_A \right ) + 2 \beta_2 [\Qone] 
\gsph_{AB} =0. \label{ConfSph}
\end{align}
Thus $[\Tone]_A(\tau) $ is a conformal Killing vector of 
$\mathbb{S}^2$. Let $Y^a$ ($a=1,2,3$) 
be the spherical harmonics with $\ell=1$ on the sphere. More specifically,
$Y^a$ is defined 
as the restriction of the Cartesian
coordinate $x^a$ to the unit sphere, and the labels
are chosen so that the rotation generated by
$\eta$ has axis along $x^3$. The spherical harmonics $Y^a$  satisfy
$\D_A \D_B Y^a = - Y^a \gsph{}_{AB}$ and the six dimensional algebra
of conformal Killing vectors on $\mathbb{S}^2$ is spanned by
$\{\D_A Y^a\} $ (proper conformal Killings) and 
$\{ \epsilon_{AB} \D{}^B Y^a\} $ (Killing vectors) 
where $\epsilon_{AB}$ is the volume form of $(\mathbb{S}^2,\gsph)$ 
with $\{ \partial_{\vartheta}, \partial_{\varphi}\}$ positively oriented.
The axial Killing vector $\axial$ is tangent to
the foliation of  $\Supfamp$ by spheres, so in particular it defines an axial
Killing vector on the unit sphere and we can write
$\axial=\axial^A\partial_A$. By definition we have
 $\eta_A :=  \gsph{}_{AB}\axial^B = \epsilon_{AB} \D^B Y^3$. 
In expressions without indexes, we will use $\axialfSph := \gsph(\eta, \cdot)$
to distinguish $\axial_A$ from $\axial_\alpha$. Note that 
$\bm{\eta} = \alpha_2 \axialfSph$, where 
$\bm{\eta}$ is defined by lowering indices with the induced metric
on $\Supfamp$.

Consequently, (\ref{ConfSph}) is equivalent to the existence of   functions $f_a(\tau), q_a(\tau)$ such that
$[\Tone]_A = f_a (\tau) \D_A Y^a + q_a (\tau) \epsilon_{AB}
\D{}^B Y^a$ and
\begin{align}
[\Qone] =  \frac{\alpha_2}{\beta_2} f_a(\tau) Y^a,
\label{Qone1}
\end{align}
where we have used  assumption (i) in
(\ref{conditions}), i.e.
$\beta_2 \neq 0$. Now, the tensor 
$[ \Phi^{\star} (\Kper)]$ 
is $[ \Phi^{\star} (\Kper)] = -2 [\omega] \alpha_2 \eta_A d \tau d x^A$
and the $\{ \tau,A \}$
component of (\ref{FirstFirst}) becomes
\begin{align}
&\alpha_1 \D_A [\Tone]^{\tau} + \alpha_2 \partial_{\tau} [\Tone]_A 
- [\omega] \alpha_2 \eta_A =0 
\quad \quad \Longleftrightarrow \nonumber \\
& \D_A 
\left ( \alpha_1 [\Tone]^{\tau} + \alpha_2 \dot{f}_a Y^a \right ) 
+ \epsilon_{AB} \D{}^B \left ( \alpha_2 \dot{q}_a Y^a
\right ) -[\omega] \alpha_2  \eta_A = 0.
\label{tfi}
\end{align}
Taking $\D^A$ of (\ref{tfi})  and using that $\D^A ([\omega] \eta_A ) =0$
it follows
\begin{align*}
\Delta_{\gsph} 
\left ( \alpha_1 [\Tone]^{\tau} + \alpha_2 \dot{f}_a Y^a \right )  =0,
\end{align*}
hence the term in parenthesis depends only on $\tau$, i.e. 
\begin{align}
[\Tone]^{\tau} = - \frac{\alpha_2}{\alpha_1} \dot{f}_a Y^a + 
\Czo(\tau).
\label{Tone1}
\end{align}
Substituting back into (\ref{tfi}) yields
\begin{align}
\epsilon_{AB} \D^B \left ( \alpha_2 \dot{q}_a Y^a
\right ) =  [\omega]\alpha_2 \eta_A.
\label{tfi2}
\end{align}
At each value of $\tau$, the left hand side is a Killing vector of
$\mathbb{S}^2$. Since $\eta_A$ is also a Killing vector of the sphere, this imples that 
$[\omega]$ can at most depend on $\tau$. However, since $\omega(r,\theta)$
we conclude that $[\omega]$ is constant, and we write
$[\omega] = b_1$. Recalling that  $\eta_A = \epsilon_{AB} \D^B Y^3$, equation (\ref{tfi2}) can be written as
\begin{align*}
\alpha_2 \epsilon_{AB} \D^B \left ( \dot{q}_a Y^a - b_1 Y^3 \right ) =0,
\end{align*}
from which it imediatelly follows that $\dot{q}_1 = \dot{q}_2 =0$ and
$\dot{q}_3 = b_1 \Longleftrightarrow q_3 = b_1 \tau + c_3$ with
$c_3$ constant.  Finally, the $\{\tau, \tau\}$ component of (\ref{FirstFirst})
is, using (\ref{EasyIden}),
\begin{align*}
\alpha_1 \frac{d [\Tone]^{\tau}}{d \tau} + \beta_1 [\Qone] =0
\quad \quad 
\Longleftrightarrow
\quad \quad
\alpha_1 \dot{C}^{(1)}_0 - \alpha_2 \ddot{f}_a Y^a + \frac{\beta_1 \alpha_2}{\beta_2} f_a Y^a 
=0
\end{align*}
where in the last equality we inserted (\ref{Qone1}) and (\ref{Tone1}). This implies that $\Czo$ is constant and that $\ddot{f_a} = \frac{\beta_1}{\beta_2} f_a$.
Sumarizing, the linearized matching conditions $[\qper_{ij}]=0$
are fullfilled iff
\begin{align}
[\vecTone] & = \left ( - \frac{\alpha_2}{\alpha_1} \dot{f}_a Y^a + \Czo \right )
\partial_{\tau} + f_a (\tau) \D^A Y^a \partial_A + b_1 \tau \eta + \zeta_0,
\quad 
[\Qone] = \frac{\alpha_2}{\beta_2} f_a(\tau) Y^a, \label{sumFirst.1} \\
[\omega] & = b_1, \quad \quad \ddot{f}_a
= \frac{\beta_1}{\beta_2} f_a
\label{sumFirst.2}
\end{align}
where $\zeta_0$ is any Killing vector on the sphere. In particular, we have 
established the first in (\ref{teo:omegap_matching_R}).

We next impose the second set of linearized matching conditions.
The last term in (\ref{FirstPert.2}) (we drop the $\pm$ indexes here) is, using 
Remark \ref{nSperee} with $\y = \omega e^{-\nu}$,
\begin{align}
\Phi^{\star} \left ( \normal_{\mu}\Sper^{\mu}_{\phantom{\mu}\alpha\beta} \right ) & = 
\left . - \frac{1}{2} \left  (  \normal \left (\omega e^{-\nu} \right ) 
+ 
\frac{ \normal \left ( \la \xi,\xi \ra \la \eta, \eta \ra \right )}{
\la \xi, \xi \ra \la \eta,\eta \ra} \omega e^{-\nu} \right ) 
\right |_{\Supfamp}
\Phi^{\star} (\Sig_{\alpha\beta} ) \nonumber \\
& = \left .  
 \frac{1}{2 \alpha_1}   \left ( 
\normal (\omega) 
+ \frac{2 \omega}{\Q} \normal (\Q) \right ) 
\right |_{\Supfamp} 
\Phi^{\star} (\Sig_{\alpha\beta} )  \nonumber \\
& = 
\left (  \alpha_2 \normal (\omega) |_{\Supfamp} 
+ 2 \beta_2 \omega |_{\Supfamp}  \right ) d \tau \otimes_s \axialfSph
\label{Spereeexp}
\end{align}
after using $\la \xi, \xi \ra = - e^{\nu} \stackrel{\Supfamp}{=} \alpha_1$, $\la \eta, \eta 
\ra = \Q^2 \sin^2 \theta$ and the fact that
\begin{align}
\label{PhistarSig}
\Phi^{\star} (\Sig) =  2  \alpha_1 \alpha_2 d\tau \otimes_s \axialfSph,
\end{align}
where 
$\otimes_s$ stands for symmetrized tensor product, $\bm{\alpha} \otimes_s \bm{\beta} =
\frac{1}{2} ( \bm{\alpha} \otimes \bm{\beta} +
\bm{\beta} \otimes \bm{\alpha} )$. The Hessian $D_i D_j Q$ of any function $Q$
on $\Supfamp$ has the following components
\begin{align}
D_\tau D_{\tau} Q = \partial_{\tau} \partial_{\tau} Q,
\quad \quad
D_{\tau} D_A Q = \partial_{\tau} \partial_A Q,
\quad \quad 
D_{A} D_B Q = \D_A \D_B Q.
\label{HessExp}
\end{align}
Equations (\ref{FirstPert.2}) are therefore
\begin{align}
  \lie_{[\vecTone]} \kappa_{ij} 
  - D_i D_j [\Qone] - [ \Qone
  \normal^{\mu} \normal^{\nu} R_{\alpha\mu\beta\nu}
  e^{\alpha}_i e^{\beta}_j ] 
  + [\Qone] \kappa_{il} \kappa^{l}_j
  - \left (  \alpha_2 [\normal (\omega)] + 2 \beta_2 [\omega] \right )  (d \tau \otimes_s \axialfSph)_{ij} =0,
  \label{B5bis}
\end{align}
where we have used $\Kpernornor =0$ and have inserted 
(\ref{Spereeexp}).
We first consider the $\{ \tau, A\}$ component.
The third and fourth terms are spherically symmetric, hence
their $\{ \tau, A\}$ component vanishes. The first two
are computed using  (\ref{EasyIden}) and
(\ref{HessExp})  as well as  (\ref{sumFirst.1}). The result is
\begin{align}
  \left ( - \frac{\beta_1  \alpha_2}{\alpha_1} + \beta_2 - \frac{\alpha_2}{\beta_2}   \right ) \dot{f}_a \D_A Y^a
 - \frac{1}{2} \alpha_2 [\normal (\omega)] \eta_A =0.
\label{tauA}
\end{align}
The first factor in parenthesis 
is (ii) in (\ref{conditions}), hence non-zero by assumption.
Since the vector fields $\D_A Y^a$, $\axial_A$ are linearly independent
  and $\alpha_2 \neq 0$, \eqref{tauA} is equivalent to
  $\dot{f}_a=0$ and $[\normal (\omega)]=0$. The former combined with 
(\ref{sumFirst.2}) and $\beta_1=0$ forces $f_a=0$ and 
(\ref{sumFirst.1}) simplifies to
 \begin{align*}
   [\vecTone ]  = b_1 \tau \eta +  \Czo \partial_{\tau} + \zeta_0, \quad \quad 
   [\Qone] = 0.
 \end{align*}
 This proves the first two in (\ref{teo:Q1_matching_R})
with $\zeta\defi  \Czo \partial_{\tau} + \zeta_0$
any Killing vector on $\Supfamp$. 
Equations \eqref{B5bis} have been reduced
to $  \Qone
      [\normal^{\mu} \normal^{\nu} R_{\alpha\mu\beta\nu}
    e^{\alpha}_i e^{\beta}_j ] =0$.
It is straightforward to check that
\begin{align}
\normal^{\mu} \normal^{\nu} R_{\alpha\mu\beta\nu}  e^{\alpha}_i e^{\beta}_j 
= \B \delta_i^{\tau} \delta_j^{\tau}
+ \A \gsph_{AB} \delta_i^A \delta_j^B,
\label{nnRee}
\end{align}
which proves the last two statements of the Proposition. \fin
\end{proof}

\begin{remark}
\label{remark:gauge_iso_1}
The presence of a Killing vector $\zeta$ in (\ref{teo:Q1_matching_R})
is a consequence of the isometries present in the background
configuration, and can never be determined \cite{Mars2007}. 
\end{remark}
\begin{remark}
  \label{B04}
  Whenever $\Q(r)=r$ we have $\A=r e^{-\lambda}\lambda'/2$ and
$2e^{\lambda-\nu}\B=\nu''-\nu'\lambda'/4+\nu'^2/2$, and 
this lemma extends to the general case (without axial
  symmetry on $\Qone$ and $\vecTone$) the consequences of Proposition 1 in \cite{ReinaVera2015},
  and in particular $\Qone[\lambda']=\Qone[\nu'']=0$ from (\ref{teo:Q1_matching_R}).   We note that the condition $\nu'\neq 0$
  is (wrongly) missing in Proposition 1 in \cite{ReinaVera2015}. 
\end{remark}

Before going into the second order matching problem we state and prove a lemma
that will simplify the computations.
\begin{lemma}
\label{b1=0}
Let $(\mmm,\gback)$ be a static and spherically symmetric
spacetime with two regions as in Definition \ref{SpherRegion}.
Assume that the hypotheses
in Proposition \ref{teo:first_order_matching} hold and that
the corresponding first order matching conditions are satisfied.

Consider second order metric perturbation tensors $\spt^\pm$ of the form
\begin{align}
\spt =& \left(-4 e^{\nu(r)} h(r, \theta) + 2{\omega}^2(r, \theta)\Q^2(r) \sin ^2 \theta \right)dt^2 -2  \Wtwo (r,\theta) \Q^2(r) \sin^2\theta dt d\phi \nonumber\\
& + 4 e^{\lambda(r)} m(r, \theta) dr^2 +4  k(r, \theta) \Q^2(r) (d\theta^2+\sin ^2 \theta  d\phi^2)
+ 4e^{\lambda(r)}\partial_\theta f(r,\theta)\Q(r)dr d\theta.
\label{teo:sopert_tensor_W2_a}
\end{align}
Apply first a hypersurface gauge defined by $ U_1 = - \vecTone^--b_1\tau\partial_\varphi$,
$U_2 =0$ and then a spacetime gauge on each side defined by
$ \sper^- =  -b_1 t \partial_{\phi}$, $ \sper^+ =0$ and  $\sperper^{\pm}=0$.
Using superstript ${}^{g}$  to denote spacetime quantities in the final
gauge and ${}^{hg}$ to denote hypersurface quantities in the final
hypersurface and spacetime gauges, the following
identitites hold
\begin{align}
& \omega_-^g = \omega_-+b_1 ,\qquad \omega_+^g = \omega_+, \qquad [ \normal(\omega^{g})]=[ \normal(\omega)]=0,\quad [ \normal( \normal(\omega^{g}))]=[ \normal( \normal(\omega))]\nonumber\\
& f_+^{g} = f_++\beta_+(r_+),\quad f_-^{g} = f_-+\beta_-(r_-)\quad \Longrightarrow \quad  [f^g] = [ f] +[\beta], \quad [\beta] \in \mathbb{R} \label{transg1}\\  
&  [\omega^g] = 0, \quad
[h^g] = [h], \quad 
[k^g] = [k], \quad 
[m^g] = [m], \quad 
[\Wtwo^g] = [\Wtwo], \\
&[ \normal(h^g)] = [ \normal(h)], \quad [ \normal(k^g)] = [ \normal(k)],\quad
[ \normal(\Wtwo^g)] = [ \normal(\Wtwo)], \label{transg2} \\
&  [\vecTone^{hg}] =  \zeta, \quad \quad 
\Qone^{hg}{}^{\pm} =  \Qone^\pm \quad  (\Longrightarrow  \quad
[\Qone^{hg}] = [\Qone] =0) , 
 \label{transg3} \\
& [\vecTtwo^{hg}] = [\vecTtwo]-2b_1\tau D_{\vecTone^-} \axial-2 D_{\vecTone^-} \zeta
     -b_1^2\tau^2 D_{ \axial} \axial-2b_1\tau D_{\axial}\zeta-2b_1\Qone\tau \kappa(\axial),  \label{transg4}\\
&
\hatQtwo^{hg+}=\hatQtwo^{+},\qquad \hatQtwo^{hg-}=\hatQtwo^--2\b_1 \tau\axial(\Qone) \quad \Longrightarrow \quad
[\hatQtwo^{hg}] = [\hatQtwo] + 2 b_1 \tau  \eta(\Qone^-).  \label{transg5}
\end{align}
\end{lemma}
\begin{proof} By  Proposition \ref{prop:full_class_gauges}
with $C = - b_1$ (resp.  $C=0$) it follows
$\omega_-^g = \omega_-+b_1$ (resp. $\omega_+^g = \omega_+$), so that, 
in particular, $ \normal(\omega_\pm^{g})= \normal(\omega_{\pm })$,
$ \normal( \normal(\omega_{\pm}^g))= \normal( \normal(\omega_{\pm }))$. As a result,
\[
[\omega^g] = [\omega] - b_1  
\stackrel{(\ref{teo:omegap_matching_R})}{=} 0,\quad \quad [ \normal(\omega^g)] = [ \normal(\omega)]\stackrel{(\ref{teo:omegap_matching_R})}{=} 0.
\]
The same proposition with $A=B= \Y = \alpha=0$ (the second order gauge
vector $\sperper$ vanishes on both sides)
 gives  (\ref{transg1})-(\ref{transg2}).
Concerning the deformation vectors,
we apply  the hypersurface gauge transformation law
(\ref{eq:Z_hyp_gauges}),
followed by the spacetime gauge transformation  (\ref{eq:Z_spt_gauges})
and  insert $\sperper= U_2=0$. The result is 
\begin{align}
{Z}^{hg}_1&={Z}^h_1- \sper |_{\Supfamp}={Z}_1+ U_1- \sper |_{\Supfamp},\label{eq:Z1_hg}\\
{Z}^{hg}_2&={Z}^h_2-\nabla_{ \sper}  \sper |_{\Supfamp}-2\nabla_{{Z}^h_1} \sper+2\nabla_{ \sper} \sper |_{\Supfamp}\nonumber\\
  &={Z}_2+2\nabla_{{U}_1} Z_1- \kappa( U_1, U_1)  n
-2\nabla_{{Z}^h_1} \sper+\nabla_{ \sper} \sper |_{\Supfamp}\nonumber\\
&={Z}_2+2\nabla_{{U}_1} Z_1- \kappa( U_1, U_1)  n
-2\nabla_{ Z_1^{hg}} \sper|_{\Supfamp}-\nabla_{ \sper} \sper |_{\Supfamp}.
\label{eq:Z2_hg}
\end{align}
Inserting in \eqref{eq:Z1_hg} the explicit forms of $U_1$ and $\sper$ in the Lemma and using (\ref{teo:Q1_matching_R}) yields
\begin{align}
&  \left . \begin{array}{l}
\vecTone^{hg}{}^+ = \vecTone^+-\vecTone^--b_1\tau \axial=[\vecTone]-b_1\tau \axial
  = \zeta\\
                 \vecTone^{hg}{}^- = \vecTone^--\vecTone^--b_1 \tau \axial+b_1\tau \axial=0      \end{array}
  \right \} \qquad \Longrightarrow \qquad 
[\vecTone^{hg}]= \zeta,\label{eq:T1hg}\\
&\Qone^{hg}{}^+ = \Qone^{+}, \quad  
\Qone^{hg}{}^- = \Qone^{-}.
\end{align}
This proves (\ref{transg3}).
For $Z_2^{\pm}$ we first compute
\begin{align*}
\nabla_{ U_1}  Z_1^{\pm} = \nabla_{ U_1} (\vecTone^{\pm} +  \Qone^{\pm}  \normal) 
=D_{ U_1} \vecTone^{\pm}  - \kappa( U_1, \vecTone^{\pm})  \normal
+  U_1 (\Qone^{\pm} )  \normal
+ \Qone^{\pm} \kappa( U_1). 
\end{align*}
Inserting into (\ref{eq:Z2_hg}), together with $[\Qone]=0$, $ U_1=-\vecTone^-$, $\sper^-=-b_1t \axial$,
$ \sper^+=0$ and the first order $hg$ quantity
$ Z_1^{hg}{}^{-} = Q_1 \normal  + b_1 \tau \eta$,
 leads to
\begin{align*}
 Z_2^{hg}{}^+=&\Qtwo^+  \normal +\vecTtwo^+  +2D_{ U_1} \vecTone^+  -2 \kappa( U_1,  \vecTone^{+})  \normal
+2  U_1 (\Qone)  \normal 
+2 \Qone\kappa( U_1)- \kappa( U_1, U_1)  \normal,\\
 Z_2^{hg}{}^-=& Z_2^- +2D_{ U_1} \vecTone^-  -2 \kappa( U_1, \vecTone^{-})  \normal
+2  U_1 (\Qone)  \normal 
+2 \Qone \kappa( U_1)- \kappa( U_1, U_1)  \normal\\
                  &+2\nabla_{\Qone  \normal}(b_1 t \axial)|_{\Supfamp}-\nabla_{b_1 t \axial}(b_1 t \axial)|_{\Supfamp}\\
  =&\Qtwo^- \normal+\vecTtwo^- +2D_{ U_1} \vecTone^-  -2 \kappa( U_1, \vecTone^{-})  \normal
+2  U_1 (\Qone)  \normal 
+2 \Qone \kappa( U_1)- \kappa( U_1, U_1)  \normal\\
&
+2b_1\Qone\tau \kappa(\axial)-b_1^2\tau^2
\left ( D_{ \eta}  \eta - \kappa( \eta,
 \eta)  \normal \right ),
\end{align*}
To do this computation it is useful to introduce the spacelike unit
  vector field $ \normal := - e^{-\lambda/2} \partial_r$. This field restricts to
  $\Supfamp$ as the unit normal before and commutes with $ \eta$, so that
$\nabla_{\normal}  \eta =
[  \normal,  \eta] + \nabla_{ \eta}  \normal \stackrel{\Supfamp}{=}
\kappa( \eta)$. It must stressed that the extension of the normal
  $\normal$ does not
  change the outcome of the computation.
Extracting the tangential to $\Supfamp$
and using $ U_1=-\vecTone^--b_1\tau\axial$ we obtain the diference
\begin{align*}
[\vecTtwo^{hg}]=&\vecTtwo^+ +2 D_{ U_1}\vecTone^+ +2\Qone\kappa( U_1)
-\left(\vecTtwo^-+2D_{ U_1}\vecTone^- +2\Qone\kappa( U_1)
+2b_1\Qone\tau \kappa(\axial)-b_1^2\tau^2 D_{ \eta}  \eta\right)\\
=&[\vecTtwo]-2D_{(\vecTone^-+b_1\tau \axial)}[\vecTone]-2b_1\Qone\tau \kappa( \axial)+b_1^2\tau^2 D_{ \eta}  \eta\\
  =&[\vecTtwo]-2D_{\vecTone^-}\left(b_1\tau\axial+ \zeta\right)
  -2b_1\tau D_{ \axial}\left(b_1\tau \axial+\zeta\right)-2b_1\Qone\tau \kappa(\axial)+b_1^2\tau^2 D_{ \eta}
  \eta\\
  =&[\vecTtwo]-2b_1\tau D_{\vecTone^-} \axial-2 D_{T_1^-} \zeta
     -b_1^2\tau^2 D_{ \axial}\axial-2b_1\tau D_{\axial}\zeta-2b_1\Qone\tau \kappa(\axial).
\end{align*}
This proves (\ref{transg4}). Regarding the normal parts, we have
\begin{align*}
&\Qtwo^{hg}{}^+=\Qtwo^+ -2 \kappa( U_1,\vecTone^+)+2  U_1(\Qone)-\kappa( U_1, U_1),\\
&\Qtwo^{hg}{}^-=\Qtwo^- -2\kappa( U_1,\vecTone^-)+2  U_1(\Qone)- \kappa( U_1, U_1)+ b_1^2\tau^2\kappa(\axial,\axial),
\end{align*}
and the corresponding gauge invariant quantities  (\ref{def:Q2hat}) are
\begin{align*}
\hatQtwo^{hg}{}^+\stackrel{(\ref{def:Q2hat})}{=}&
\Qtwo^{hg+}+\kappa(\vecTone^{hg+},\vecTone^{hg+})-2\vecTone^{hg+}(\Qone)\\
  \stackrel{}{=} &
                   \Qtwo^{+}-2 \kappa( U_1,\vecTone^+)+2  U_1(\Qone)
                   -\kappa( U_1, U_1)
                   +\kappa(\vecTone^++U_1,\vecTone^++U_1)-2(\vecTone^++U_1)(\Qone)\\
  =&\Qtwo^++\kappa(\vecTone^{+},\vecTone^{+})-2\vecTone^{+}(\Qone)\stackrel{(\ref{def:Q2hat})}{=}\hatQtwo^+,\\
  \hatQtwo^{hg-}\stackrel{(\ref{def:Q2hat})}{=}&
  \Qtwo^{hg-}+\kappa(\vecTone^{hg-},\vecTone^{hg-})
  -2\vecTone^{hg-}(\Qone)
                                                 \stackrel{(\ref{eq:T1hg})}{=}
                                                 \Qtwo^{-hg}\\
  =&\Qtwo^- -2\kappa( U_1,\vecTone^-)+2 U_1(\Qone)- \kappa( U_1, U_1)+ b_1^2\tau^2\kappa(\axial,\axial)\\
  =&\Qtwo^- +2\kappa(\vecTone^-,\vecTone^-)+2b_1\tau\kappa(\axial,\vecTone^-)
     -2\vecTone^-(\Qone)-2b_1\tau \axial(\Qone)-\kappa(\vecTone,\vecTone )\\
     &  -b_1^2\tau^2\kappa( \axial,\axial)
     -2b_1\tau\kappa(\axial,\vecTone^-)+ b_1^2\tau^2\kappa(\axial,\axial)=\hatQtwo^--2b_1\tau  \axial(\Qone),
\end{align*}
which establishes  (\ref{transg5}).\fin
\end{proof}

\vspace{3mm}
We can now solve the second order matching problem.
\begin{proposition} \label{teo:second_order_matching}
Let $(\mmm,\gback)$ be a static and spherically symmetric
spacetime with two regions as in Definition \ref{SpherRegion}.
Assume that the hypotheses
in Proposition \ref{teo:first_order_matching} hold and that
the corresponding first order matching conditions are satisfied.

Consider second order metric perturbation 
tensors $\spt^\pm$ of the form
\begin{align}
\spt =& \left(-4 e^{\nu(r)} h(r, \theta) + 2{\omega}^2(r, \theta) \Q^2(r) \sin ^2 \theta \right)dt^2 -2 \Wtwo (r,\theta) \Q^2(r) \sin^2\theta dt d\phi \nonumber\\
& + 4 e^{\lambda(r)} m(r, \theta) dr^2 +4  k(r, \theta) \Q^2(r) (d\theta^2+\sin ^2 \theta  d\phi^2)
+ 4e^{\lambda(r)}\partial_\theta f(r,\theta)\Q(r)dr d\theta.
\label{teo:sopert_tensor_W2}
\end{align}
Then the second order matching conditions are satisfied if and only if there exist functions $\hatQtwo^{\pm}$
on $\Supfamp$  such that, in terms of
\begin{align}
\defor^{\pm} := \hatQtwo^{\pm} - 2 \Q (e^{\lambda/2} f) |_{\Sigma^{\pm}} + 2 \vecTone^{\pm} (\Qone)
\label{defor}
\end{align}
the following expressions hold
\begin{align}
  [ \defor]   = & \left . - \frac{\Q}{\normal(\Q)} \right |_{\Supfamp}
  \left (
  2 c_0 +
  (2 c_1 + H_1) \cos \vartheta  \right ),\label{jumpdefor} \\
 [ \Wtwo ]   = &    D_3 - 2  \zeta_0 (\omega^+ |_{\Supfamp}), \label{PropWtwo} \\
  [ \normal(\Wtwo)]  = & -2 \zeta_0 ( \normal(\omega) |_{\Supfamp})-2\Qone
                      [ \normal( \normal(\omega))],\label{PropnWtwo} \\
 [k]  = &  - \left . \normal(\Q) \right |_{\Supfamp} [ e^{\lambda/2}  f] +c_0 + c_1  \cos \vartheta, \label{Propk} \\
 [h]   = &   \frac{1}{2 } H_0 +  \left . \frac{\Q \normal(\nu)}{4 \normal(\Q)} \right |_{\Supfamp} \left (  2 [k] 
+ H_1 \cos \vartheta \right ) ,    \label{Proph} \\
[m] = &  2 [k]
+ \left . \frac{\Q}{\normal(\Q)} \right |_{\Supfamp} [  \normal(k)] 
+ \left ( H_1 - \frac{(2 c_1 + H_1)}{2 \normal(\Q)^2 |_{\Supfamp} }    \right ) \cos \vartheta \nonumber \\
&+ \frac{1}{2} \left [ ( \defor  + 2 \Q e^{\lambda/2} f ) 
\left ( 
- \frac{1}{\Q \normal(\Q)} \A 
  + \frac{\normal(\Q)}{\Q} \right ) \right ]
- \left . \frac{1}{2 \Q \normal(\Q)}  \right |_{\Supfamp}
(\Qone)^2  [  \normal(\A)],
\label{Propm} \\
[\normal(h) ]  = &  \left . \frac{\Q \normal(\nu)}{2 \normal(\Q)} \right |_{\Supfamp} [\normal(k)]
+ \left . \frac{\normal(\nu)}{2} \left ( 1 - \frac{\Q \normal(\nu)}{2 \normal(\Q)} 
\right ) \right |_{\Supfamp} \left ( 2 [k] + H_1 \cos \vartheta \right )
- \left . \frac{\normal(\nu)}{4 \normal(\Q)^2}  \right |_{\Supfamp} (H_1+ 2 c_1)\cos \vartheta \nonumber \\
& + \frac{\normal(\nu)}{4} \left [ 
\left ( \defor + 2 \Q e^{\lambda/2} f \right ) \left (
-\frac{1}{\Q \normal(\Q)} \A - \frac{2 }{e^{\nu} \normal(\nu)} \B + \frac{\normal(\Q)}{\Q} - \frac{\normal(\nu)}{2} 
\right) \right ]
 \nonumber \\
& - \frac{1}{4} (\Qone)^2 \left [ \frac{\normal(\nu)}{\Q \normal(\Q)}\normal(\A) + \frac{2}{e^{\nu}} \normal(\B)  \right ] , \label{Propnh} 
\end{align}
where $D_3$, $H_1$, $H_0$, $c_0$, $c_1$
are arbitrary constants.
\end{proposition}
\begin{remark}
  As in Remark \ref{B04}, setting $\Q(r)=r$ the results of this proposition
    extend to the general case (without assuming axial symmetry on $\Qone$, $\vecTone$, $\Qtwo$ and $\vecTtwo$) the outcome of Proposition 2 in \cite{ReinaVera2015}.
\end{remark}

\begin{proof}
  We exploit Lemma \ref{b1=0} to simplify the proof. Specifically, we solve the problem in the
  gauge $g$ and $hg$ and  then we translate into the original  gauge.
For the sake
of notational simplicity we drop the superindexes $g$ and $hg$  along the
proof  and we only restore them at the end.

In  the gauge of Lemma \ref{b1=0} we have $\vecTone^{-}=0$ (by \eqref{eq:T1hg}) and $[\omega] =0$ (so that we may simply write $\omega$). Thus, Proposition
\ref{teo:first_order_matching} 
gives $\vecTone^{+} =  \zeta
  = \Czo \partial_{\tau} +
   \zeta_0$ with $\zeta_0$ a Killing vector on the sphere.
   It is useful to introduce  the tangent vector to $\Supfamp$ given by
  \begin{align*}
      \vecVtwo \defi [\vecTtwo]
      - 2 \Qone \kappa( \zeta) - D_{ \zeta\,} {\zeta}.
  \end{align*}
  The second order matching conditions $[\qperper_{ij}] =0$ obtained from
  \eqref{SecondPert.1}  become, after using Proposition \ref{teo:first_order_matching}
  \begin{align}
    0 & = [\qperper_{ij}] =
    \lie_{\vecVtwo} h_{ij}
    + 2 [\hatQtwo] \kappa_{ij} + [\Phi^{\star}(\Kperper)]_{ij}
    + 2 \lie_{ \zeta\, } \qper_{ij}
    - 4 \Qone [ \normal_{\nu}\Sper^{\mu}_{\phantom{\mu}\alpha\beta} e^{\alpha}_i
      e^{\beta}_j] \nonumber \\
    & =
\lie_{\vecVtwo} h_{ij}
    + 2 [\hatQtwo] \kappa_{ij} + [\Phi^{\star}(\Kperper)]_{ij}
    + 4 {\zeta} (\Qone) \kappa_{ij}
    + 2 \lie_{ \zeta} \, \Phi^{\star} (\Kper)_{ij},
\label{SecondOrderMatchGen}
      \end{align}
      where in the last equality we
      inserted the explicit expression of $\qper$  
 from (\ref{FirstPert.1}) and used the facts that 
$\lie_{ \zeta} \, h_{ij} =
  \lie_{ \zeta} \, \kappa_{ij} =0$ and, in the present gauge, also
  $[ \normal_{\nu}\Sper^{\mu}_{\phantom{\mu}\alpha\beta} e^{\alpha}_i e^{\beta}_j] =0$.
  To ellaborate 
(\ref{SecondOrderMatchGen})
 further, we
use  $\Phi^{\star} (\Kper) =
- 2 \alpha_2 \omega |_{\Supfamp} d\tau \otimes_s \axialfSph$,
(\ref{PhistarSig}),
so that
\begin{align*}
  \lie_{ \zeta} \, \Phi^{\star} (\Kper)
  =
  - 2 \alpha_2 d\tau \otimes_s \left (
{\zeta}_0 (  \omega |_{\Supfamp}) \axialfSph
 + \omega|_{\Supfamp} \gsph([ \zeta_0,  \eta\, ],
\cdot)
  \right ).
\end{align*}
Inserting this and the pull-back of (\ref{teo:sopert_tensor_W2})
  on $\Supfamp$ transforms (\ref{SecondOrderMatchGen}) into
\begin{align}
 &\lie_{\vecVtwo} h 
 + 2 [\hatQtwo] \kappa   + 4 \alpha_1 [h] d \tau^2
 - 2 \alpha_2
 [ \Wtwo ]
 d\tau \otimes_s \axialfSph
+ 4 \alpha_2 [k] \gsph + 4 {\zeta} (\Qone) \kappa \nonumber \\
&     
    -4 \alpha_2 d\tau \otimes_s \left ( {\zeta}_0 (
    \omega |_{\Supfamp}) \axialfSph
+ \omega|_{\Supfamp} \gsph([ \zeta_0,  \eta \, ], \cdot)
  \right )=0.
  \label{genqperper}
\end{align}
We consider first the  $A, B$ components of this expression.
Decomposing
$\vecVtwo = \Vtwo^{\tau} \partial_{\tau} + \Vtwo^A \partial_{A}$
and computing $\lie_{\vecVtwo} h$ with the general
  identity (\ref{EasyIden}), we  find that these components give
\begin{align*}
  \alpha_2 ( \D_A {\Vtwo}_B + \D_B {\Vtwo}_A ) + 2 \beta_2
        [\hatQtwo] \gsph_{AB}
        +4 \alpha_2 [k] \gsph_{AB} + 4
        \beta_2 {\zeta} (\Qone) \gsph_{AB} =0.
  \end{align*}
As in the proof of Proposition \ref{teo:first_order_matching}, this is equivalent to
the
existence of six functions $\ftwo_a(\tau)$, $\qtwo_a(\tau)$ such that
\begin{align}
  \Vtwo{}_A & = \ftwo_a(\tau) \D_A Y^a + \qtwo_a(\tau) \epsilon_{AB} \D^B Y^a,
\label{ExpV2} \\
 [\hatQtwo] + 2 {\zeta} (\Qone) & =
  \frac{\alpha_2}{\beta_2} \left (
  \ftwo_a(\tau) Y^a - 2 [k] \right ).
  \label{JumphatQtwo}
\end{align}
We next consider the $\{ \tau, A\}$ component of
(\ref{genqperper}). Another application of 
(\ref{EasyIden}) gives
\begin{align}
  & \D_A \left (
  \alpha_1 \Vtwo^{\tau} + \alpha_2
  \dot{f}^{(2)}_a Y^a \right )
  + \alpha_2 \dot{q}^{(2)}_a \epsilon_{AB} \D^B Y^a  -  \alpha_2 \left (
  [\Wtwo] \eta_A 
  + 2  \zeta_0 (\omega|_{\Supfamp})  \eta_A
  + 2 \omega |_{\Supfamp} [ \zeta_0,  \eta\, ]_A \right ) =0.
 \label{tauA.2}
\end{align}
The divergence $\D^A$ of the second term is identically zero.
The divergence of
the last term is also zero because 
${\eta} (\Wtwo)=0$ and  $\eta_A$, $[ \zeta_0, \eta]_A$ are Killing vectors (hence
divergence-free) and, in addition,
\begin{align*}
  \eta^A \D_A \left ( \zeta_0^B \D_B (\omega|_{\Sigma_{0}} )\right )
  + [ \zeta_0,  \eta\, ]^A \D_A ( \omega|_{\Supfamp})
  & = \eta^A \zeta_0^B \D_A \D_B (\omega|_{\Supfamp})
  + (\zeta_0^B \D_B \eta^A) \D_A (\omega |_{\Supfamp}) \\
  & =
  \zeta_0^B \D_B \left ( \eta^A \D_A (\omega |_{\Supfamp}) \right ) =0,
\end{align*}
where in the first equality we expanded the Lie bracket
and in the last equality we used $\eta(\omega|_{\Supfamp})=0$.
Thus, the divergence of
(\ref{tauA.2}) is equivalent to 
\begin{align}
  \Delta_{\gsph} \left (
 \alpha_1 \Vtwo^{\tau} + \alpha_2
 \dot{f}^{(2)}_a Y^a \right )  =0
 \quad \quad \Longleftrightarrow
 \quad \quad
  \Vtwo^{\tau} = C_0^{(2)}(\tau)
  -\frac{\alpha_2}{\alpha_1} \dot{f}^{(2)}_a Y^a,
 \label{Vtwotau}
 \end{align}
with $C_0^{(2)}(\tau)$ an integration function depending only on $\tau$.
Substituting back into (\ref{tauA.2})  yields
\begin{align}
\dot{q}^{(2)}_a \epsilon_{AB} \D^B Y^a  
-  [\Wtwo] \eta_A   -2   \zeta_0 (\omega|_{\Supfamp}) \eta_A
-2  \omega |_{\Supfamp} [ \zeta_0,  \eta\, ]_A  =0.
\label{tauA.3}
\end{align}
We now decompose
$\zeta_0{}_{A} = \zeta_0{}_a \epsilon_{AB} \D^B Y^a$,
$ \zeta_0{}_a \in \mathbb{R}$
and define $\eta^a_A := \epsilon_{AB} \D^B Y^a$ so that in particular
$ \eta =  \eta{}^{\, 3}$. The commutation relations are, $a,b,\cdots =1,2,3$,
\begin{align*}
  [  \eta^{\,a} ,  \eta^{\,b} ] = - \epsilon^{ab}{}_{c} \, {\eta}^{\,c}
  \end{align*}
  with  $\epsilon^{abc}$ the Levi-Civita totally antisymmetric
  symbol, so (\ref{tauA.3}) takes the form
  \begin{align*}
    \dot{q}^{(2)}_a \eta^a_{A} - [\Wtwo] \eta_A - 2 {\zeta}_0
      (\omega|_{\Supfamp} ) \eta_A + 2 \omega|_{\Supfamp}
      \zeta_0{}_b \epsilon^{b3}{}_a \eta^a_A =0.
  \end{align*}
  By linear independence of $\{  \eta^{\,a}\, \}$, this is equivalent to
    \begin{align}
  \dot{q}^{(2)}_3 &= [\Wtwo] + 2 {\zeta}_0 (\omega|_{\Supfamp}), 
\label{dotqm3} \\
    \dot{q}^{(2)}_{a} &= -2 \omega|_{\Supfamp} \zeta_0{}_b \epsilon^{b3}{}_a,
    \quad a=1,2.
\label{dotqm12}
\end{align}
Since $[\Wtwo]$ and $\omega|_{\Supfamp}$ 
are constant along $\tau$ it follows that $\ddot{q}^{(2)}_a =0$, i.e.
there exist six constants $b^{(2)}_a$ and $d^{(2)}_a$ such that
\begin{align}
  q^{(2)}_{a} = b^{(2)}_a \tau + d^{(2)}_a, \quad \quad
  [\Wtwo] = b^{(2)}_3 - 2 {\zeta}_0 (\omega|_{\Supfamp}), \qquad
  b^{(2)}_a  + 2 \omega|_{\Supfamp} \zeta_0{}_b \epsilon^{b3}{}_a=0, \quad a =1,2.
\label{qtwos}
\end{align}
It only remains to impose the $\{\tau,\tau\}$ component of
(\ref{genqperper}), which is
\begin{align*}
   \alpha_1 \dot{\Vtwo}^{\tau} +  \beta_1 [\hatQtwo]  + 2\alpha_1
  [h] + 2 \beta_1 \zeta(\Qone) =0.
\end{align*}
Upon inserting (\ref{JumphatQtwo}) and (\ref{Vtwotau}) this is equivalent to
\begin{align*}
  \alpha_1 \dot{C}^{(2)}_0 
  + 2 \alpha_1 [h] + 
\alpha_2 \left [ \left ( - \ddot{f}^{(2)}_a
+  
\frac{\beta_1}{\beta_2} f^{(2)}_a \right ) Y^a - \frac{2\beta_1}{\beta_2}
[k] \right ]=0.
\end{align*}
The fact that $[h]$ and $[k]$ are $\tau$-independent 
and $\varphi$-independent imposes that 
$\dot{C}_0^{(2)}$ is constant  and that the term in parentheses is constant
for $a=3$ and zero for $a=1,2$. In other words, there exist $c^{(2)}_0, c^{(2)}_1,
f_{0}^{(2)} 
\in \mathbb{R}$ such that
\begin{align}
C^{(2)}_0(\tau) = c^{(2)}_0 + c^{(2)}_1 \tau, \quad 
\ddot{f}^{(2)}_a
        - \frac{\beta_1}{\beta_2} f^{(2)}_a = f_{0}^{(2)} \delta_a^3, \quad 
        [h] = \frac{\beta_1 \alpha_2}{\alpha_1\beta_2} [k]
        - \frac{1}{2 } c^{(2)}_1
        + \frac{\alpha_2 f_{0}^{(2)} }{2 \alpha_1} Y^3.
\label{jumph}
\end{align}
Summarizing, the second order matching conditions $[\qperper_{ij}]=0$
are equivalent to (\ref{ExpV2}), (\ref{JumphatQtwo}),
(\ref{Vtwotau}), (\ref{qtwos}) and (\ref{jumph}).
We next deal with  $[\kappaperper_{ij}]=0$.
We note the following facts:
\begin{align*}
\mbox{(a)} \quad \quad & \normal_{\mu} \normal^{\nu} \Sper^{\mu}_{\phantom{\mu}\nu\alpha} =0, \quad \quad 
  & & \mbox{(see }(\ref{calSperP2})) 
  \\
\mbox{(b)} \quad \quad & [ \Sper_{\mu\alpha\beta} e^{\mu}_{l} e^{\alpha}_{i} e^{\beta}_{j}]  =0, 
  & & \mbox{as a consequence of } [\omega] =0,  \\
\mbox{(c)} \quad \quad & [ \normal_{\mu} \Sper^{\mu}_{\phantom{\mu}\nu\alpha} e^{\alpha}_i e^{\beta}_j] =0,
                    & & \mbox{as consequence of } [\omega] =0, [ \normal(\omega)] =0 \mbox{ (see (\ref{Spereeexp}))}.
\end{align*}
Additional facts that we will use are
\begin{align*}
\mbox{(d)} \quad \quad & \lie_{[\vecTone]} \kappa_{ij} = \lie_{ \zeta\, } \kappa_{ij} =0, \\
\mbox{(e)} \quad \quad & 
\lie_{[\vecTtwo]} \kappa_{ij} -  \lie_{2 \Qone \kappa ([\vecTone] ) +  [D_{\vecTone} \vecTone] } \, \kappa_{ij} 
= \lie_{\vecVtwo} \kappa_{ij}, \\
\mbox{(f)} \quad \quad & [\lie_{\vecTone} \kappaper_{ij}] = \lie_{[\vecTone]} \kappaper_{ij} = \lie_{\vec\zeta\, } \kappaper_{ij} \\
& \quad \quad
= \lie_{ \zeta\, } \left (  -  D_i D_j \Qone + \Qone \left . \big ( 
- \normal^{\mu} \normal^{\nu} R_{\alpha\mu\beta\nu}
e^{\alpha}_i e^{\beta}_j + \kappa_{il} \kappa^{\,\,l}_{j} \big ) \right |_{\Supfamp^-} 
- \left.\normal_{\mu}\Sper^{\mu}_{\phantom{\mu}\alpha\beta}e^{\alpha}_i e^{\beta}_j\right|_{\Supfamp} \right ) \\
  & \quad \quad = - D_i D_j \lieQone 
+ \lieQone \left . \big (  -\normal^{\mu} \normal^{\nu} R_{\alpha\mu\beta\nu}
e^{\alpha}_i e^{\beta}_j + \kappa_{il} \kappa^{\,\,l}_{j} \big ) 
\right |_{\Supfamp^-} \\
& \quad \quad \quad  
 - \lie_{ \zeta\,} \left ( \left (  \alpha_2 \normal (\omega) |_{\Supfamp} 
+ 2 \beta_2 \omega |_{\Supfamp}  \right ) d \tau \otimes_s \axialfSph \right )_{ij},
\end{align*}
where   in (f) we used  that $ \zeta$ is the restriction
to $\Supfamp$ of an ambient Killing vector tangential to 
$\Supfamp$, which  has the consequence that $\lie_{ \zeta}$ 
commutes  with  the Hessian of $(\Supfamp,h)$ 
and that it anhilates
$(- \normal^{\mu} \normal^{\nu} R_{\alpha\mu\beta\nu}
e^{\alpha}_i e^{\beta}_j + \kappa_{il} \kappa^{\,\,l}_{j}) |_{\Supfamp^-}$.
Note that $\normal^{\mu} \normal^{\nu} R_{\alpha\mu\beta\nu}
e^{\alpha}_i e^{\beta}_j$ may be discontinuous on $\Supfamp$, but since it is 
multiplied by
$\lieQone$, it does not matter (by  \eqref{nnRee} and \eqref{teo:Q1_matching_R})
whether we evaluate it on  $\Supfamp^-$ (as we have chosen), or 
on $\Supfamp^+$. We have also inserted
(\ref{Spereeexp}) in the last equality.

Using (a)-(f) in (\ref{SecondPert.2}) together with $\Kpernornor=0$, $\Kpernortan=0$ and $[Q_1] =0$, the equations $[\kappaperper_{ij}]=0$ become
\begin{align}
0  = [ \kappaperper_{ij}]   = \nonumber & 
\lie_{\vecVtwo} \kappa_{ij} 
- D_i D_j \left ( [\hatQtwo]  + 2 \lieQone \right )
+  \left [ \hatQtwo \left ( - \normal^{\mu} \normal^{\nu} R_{\alpha\mu\beta\nu} e^{\alpha}_i e^{\beta}_j
+ \kappa_{il} \kappa^{l}_{j} \right ) \right ]  \nonumber \\
&  + 2 \lieQone 
\left . \big ( - \normal^{\mu} \normal^{\nu} R_{\alpha\mu\beta\nu} e^{\alpha}_i e^{\beta}_j
+ \kappa_{il} \kappa^{l}_{j} \big ) \right |_{\Supfamp^-}  
- 2\lie_{ \zeta\,} \left ( \left (  \alpha_2 \normal (\omega) |_{\Supfamp} 
+ 2 \beta_2 \omega |_{\Supfamp}  \right ) d \tau \otimes_s \axialfSph\right )_{ij} \nonumber \\
& - \left [\normal_{\mu} \Sperper^{\mu}_{\phantom{\mu}\alpha\beta} e^{\alpha}_i e^{\beta}_j \right ]
+ \frac{1}{2} [\Kperpernornor]  \kappa_{ij} 
- 2 \Qone 
\left [\normal_{\mu} \normal^{\nu} ( \nabla_{\nu}\Sper^{\mu}_{\phantom{\mu}\alpha\beta} ) e^{\alpha}_i e^{\beta}_j \right ] 
\nonumber \\
&  - \Qone^2 \left [ 
                  \normal^{\mu} \normal^{\nu} \normal^{\delta} ( \nabla_{\delta} R_{\alpha\mu\beta\nu} ) e^{\alpha}_i e^{\beta}_j
                  + 2 \normal^{\mu} \normal^{\nu} R_{\delta\mu\alpha\nu} e^{\delta}_l e^{\alpha}_j \kappa^{l}_i
                  + 2 \normal^{\mu} \normal^{\nu} R_{\delta\mu\alpha\nu} e^{\delta}_l e^{\alpha}_i \kappa^{l}_j \right ].
\label{GenMatchkappa}
\end{align}
We proceed with the first and third terms in the third line. 
Consider, as before,  the extension
$\normal := -e^{- \lambda/2} \partial_r$  of the normal vector  
off the matching hypersurface (the result is independent
of how we extend). Directly from the definition of $\Sperper^{\mu}_{\phantom{\mu}\alpha\beta}$ (\ref{calSperper}) we get
\begin{align}
\normal_{\mu} \Sperper^{\mu}_{\phantom{\mu}\alpha\beta}
 = \frac{1}{2} \left ( \nabla_{\alpha} \left ( {\Kperper}_{\mu\beta} \normal^{\mu} \right )
+ \nabla_{\beta} \left ( {\Kperper}_{\mu\alpha} \normal^{\mu} \right )
- \lie_{\normal} {\Kperper}_{\alpha\beta} \right ).
\label{nSperperalphabeta}
\end{align}
For any one-form $P_{\alpha}$ one has the following identity, easy to prove,
\begin{align*}
\Phi^{\star} ( \nabla P )_{ij} =
D_i P_j + \kappa_{ij} (P_{\alpha} \normal^{\alpha} |_{\Supfamp})
\end{align*}
where $P_i = \Phi^{\star} (P)_{i}$. Applying this to (\ref{nSperperalphabeta}) it follows
\begin{align}
\normal_{\mu} \Sperper^{\mu}_{\phantom{\mu}\alpha\beta} e^{\alpha}_i e^{\beta}_j =
\frac{1}{2} \left ( D_i \Kperpernortan_j + D_j \Kperpernortan_i - \Phi^{\star}(\lie_{\normal} \Kperper)_{ij} \right )
  +\Kperpernornor \kappa_{ij}
  \label{nSperpera}
\end{align}
where we have defined $\Kperpernortan_i := {\Kperper}_{\alpha\beta} \normal^{\alpha} |_{\Supfamp}  e^{\beta}_i$
and recall that $\Kperpernornor\defi {\Kperper}_{\alpha\beta} \normal^{\alpha}\normal^\beta|_{\Supfamp}$.
For the $\nabla \Sper$-term, we first observe that by the background symmetries (or by direct computation)  
the vector field $\normal$  is geodesic, i.e. $\nabla_{ \normal }  \normal=0$. Thus,   (\ref{calSperP2}) (c.f. (\ref{Spereeexp})) gives
\begin{align*}
\normal_{\mu} & \normal^{\nu} ( \nabla_{\nu}\Sper^{\mu}_{\phantom{\mu}\alpha\beta} ) 
= \normal^{\nu} \nabla_{\nu}  \left ( \normal_{\mu}\Sper^{\mu}_{\phantom{\mu}\alpha\beta} \right )
=   - \frac{1}{2} \normal^{\nu} \nabla_{\nu} \left  [ e^{-\nu} \left (  \normal(\omega) + \frac{2 \omega  \normal(\Q)}{\Q} \right ) \Sig_{\alpha\beta} 
\right ] \\
& = \normal^{\nu} \nabla_{\nu} \left  [ 
\sin^2 \theta \left ( 
\Q^2  \normal(\omega) + 2 \omega \Q  \normal(\Q)  \right ) 
(dt \otimes_s d \phi)_{\alpha\beta}  \right ]
\\
& = \sin^2 \theta \left [ 
\normal \left ( \Q^2  \normal(\omega) + 2 \omega \Q  \normal(\Q)  \right ) 
- \left ( 
\Q^2  \normal(\omega) + 2 \omega \Q  \normal(\Q)  \right ) \left ( \frac{1}{2} \normal (\nu) + \frac{\normal (\Q)}{\Q} \right ) \right ]
(dt \otimes_s d \phi)_{\alpha\beta}
\end{align*}
where we replaced  $\y = \omega e^{-\nu}$ and in the last equality we inserted  
\begin{align*}
\normal^{\nu} \nabla_{\nu} \nabla_{\alpha} t = - \frac{1}{2}  \normal(\nu) dt_{\alpha}, \quad \quad
\normal^{\nu} \nabla_{\nu} \nabla_{\alpha} \phi = - \frac{1}{\Q}  \normal(\Q) d\phi_{\alpha}, \quad \quad
\end{align*}
which follows by a simple computation. Consequently,
\begin{align*}
\left [
\normal_{\mu}  \normal^{\nu} ( \nabla_{\nu}\Sper^{\mu}_{\phantom{\mu}\alpha\beta} ) 
\right ]
& = 
[\normal \left ( \Q^2  \normal(\omega) + 2 \omega \Q  \normal(\Q)  \right )]
\left ( d \tau \otimes_s \axialfSph \right )_{ij} \nonumber \\
& = \alpha_2 [ \normal( \normal(\omega))] 
+  \omega|_{\Supfamp} [  \normal( \normal(\Q^2 ))] \left ( d \tau \otimes_s \axialfSph \right )_{ij}.
\end{align*}
With this and (\ref{nSperpera}), equation (\ref{GenMatchkappa}) is rewritten as
\begin{align}
& \lie_{\vecVtwo} \kappa_{ij} 
- D_i D_j \left ( [\hatQtwo]  + 2 \lieQone \right )
+  \left [ \hatQtwo \left ( - \normal^{\mu} \normal^{\nu} R_{\alpha\mu\beta\nu} e^{\alpha}_i e^{\beta}_j
+ \kappa_{il} \kappa^{l}_{j} \right ) \right ]  \nonumber \\
&  + 2 \lieQone 
\left . \big ( - \normal^{\mu} \normal^{\nu} R_{\alpha\mu\beta\nu} e^{\alpha}_i e^{\beta}_j
+ \kappa_{il} \kappa^{l}_{j} \big )   \right |_{\Supfamp^-}
- 2\lie_{ \zeta\,} \left ( \left (  \alpha_2 \normal (\omega) |_{\Supfamp} 
+ 2 \beta_2 \omega |_{\Supfamp}  \right ) d \tau \otimes_s \axialfSph\right )_{ij} \nonumber \\
& 
- \frac{1}{2} \left ( D_i [\Kperpernortan_j]  + D_j [\Kperpernortan_i] \right ) 
+ \frac{1}{2} \left [ (\lie_{\normal} \Kperper)_{\alpha\beta} e^{\alpha}_i
e^{\beta}_j \right ] - \frac{1}{2} [ \Kperpernornor] \kappa_{ij}  
 - 2 \Qone \alpha_2 [  \normal( \normal(\omega))] 
\left ( d \tau \otimes_s \axialfSph \right )_{ij}
\nonumber \\
&  - \Qone^2 \left [ 
                  \normal^{\mu} \normal^{\nu} \normal^{\delta} ( \nabla_{\delta} R_{\alpha\mu\beta\nu} ) e^{\alpha}_i e^{\beta}_j
                  + 2 \normal^{\mu} \normal^{\nu} R_{\delta\mu\alpha\nu} e^{\delta}_l e^{\alpha}_j \kappa^{l}_i
                  + 2 \normal^{\mu} \normal^{\nu} R_{\delta\mu\alpha\nu} e^{\delta}_l e^{\alpha}_i \kappa^{l}_j \right ] =0.
\label{GenMatchkappa.2}
\end{align}
where we have also used the
first order matching condition $\Qone [ \normal( \normal(\Q^2 ))] =0$.
So far we imposed no restriction on $\Kperper$. We now use (\ref{teo:sopert_tensor_W2}), which implies
\begin{align*}
  [ \Kperpernornor] = 4 [m], \quad \quad
[\Kperpernortan_i] = - 2  
D_i [ \Q e^{\lambda/2} f] 
\end{align*}
so that, in particular,
\begin{align}
- \frac{1}{2} \left ( D_i [\Kperpernortan_j]  + D_j [\Kperpernortan_i] \right ) 
= 2  D_i D_j [ \Q
e^{\lambda/2} f].
\label{D_itau_j}
\end{align}
We start by analyzing the $\{A, B\}$ component of
(\ref{GenMatchkappa.2}).
The backgroung spherical symmetry and the fact that 
$(\Phi^{\star}(\Kperper))_{AB}$ is proportional to $\gsph_{AB}$ implies 
\begin{align}
\D_A \D_B \left ( 
2 
[ \Q e^{\lambda/2} f] -   [\hatQtwo ]  - 2 \lieQone 
\right )  +  \L \gsph_{AB} =0
\label{HessEq}
\end{align}
for some function $\L$ that will be determined later.  This equation 
states that 
$\D_A ( 2 \Q
[e^{\lambda/2} f] -   [\hatQtwo ]  - 2 \lieQone)$
is a conformal Killing vector on the sphere.
The most general 
conformal Killing vector which, in addition, is a gradient
is a linear combination (with coefficients
that may depend of $\tau$) of $\D_A Y^a$. Hence, there exist three functions 
$s_a(\tau)$ such that
\begin{align*}
\D_A \left ( 2  \Q [e^{\lambda/2}  f] -   [\hatQtwo ]  - 2 \lieQone 
\right ) &
= s_a(\tau) \D_A Y^a \quad \quad
\Longleftrightarrow \\
& 2  \Q [e^{\lambda/2}  f]
 -   [\hatQtwo ]  - 2 \lieQone  =
s_a(\tau) Y^a + s_0(\tau),
\end{align*}
where $s_0(\tau)$ is a further integration ``constant''. Inserting (\ref{JumphatQtwo}) we finally arrive at
\begin{align}
\Q [e^{\lambda/2}  f] =  - \frac{\alpha_2}{\beta_2} [k] +
\frac{1}{2}
s_0(\tau) 
+ \frac{1}{2} \left ( s_a(\tau) + \frac{\alpha_2}{\beta_2} f_a^{(2)}(\tau)\right ) Y^a.
\label{jumpf}
\end{align}
This equation already provides relations between $s_a$ and $f_a^{(2)}$, but we will come to that later.
With this information, (\ref{HessEq}) reduces to
\begin{align}
\L =  s_a(\tau) Y^a.
\label{Lsm}
\end{align}
To find the explicit form of $\L$ (as well as for the rest of equations) we  
need $[(\lie_{\normal} \Kperper)_{\alpha\beta} e^{\alpha}_i e^{\beta}_j]$. 
It is convenient
to extend also $e^{\alpha}_i$ to a spacetime neighbourhood of $\Supfamp$ (the result being again independent of how the extension is made).
We make the natural
choice $e^{\alpha}_i \partial_{\alpha} = \partial_{x^i}$. 
The structure $\normal^{\mu} = \normal^r(r) \delta^{\mu}_r$  implies 
\begin{align*}
(\lie_{\normal} \Kperper)_{\alpha\beta} e^{\alpha}_i e^{\beta}_j
=  \normal \left ( {\Kperper}_{ij} \right ).
\end{align*}
Note that  ${\Kperper}_{ij}$ are spacetime scalars so 
their directional derivative is well-defined.
We  can now analyze the $\{ \tau,A\}$ component of
(\ref{GenMatchkappa.2}). The $i= \tau, j=A$ component of (\ref{D_itau_j}) is zero
because $\Q e^{\lambda/2} f |_{\Supfamp}$ does not depend on $\tau$.
Inserting the general identity (\ref{EasyIden}) and using the forms
of $\vecVtwo$, $[\Qtwo] + 2 \lieQone$ 
and $(\Kperper)_{\tau A}$, one finds
\begin{align*}
\left ( \beta_2 - \frac{\beta_1 \alpha_2}{\alpha_1} - \frac{\alpha_2}{\beta_2}
\right ) 
\dot{f}_a^{(2)} \D_A Y^a
+ \beta_2 \dot{q}_a^{(2)} \epsilon_{AB} \D^B Y^a 
+ \left (  \alpha_2 \normal (\omega) |_{\Supfamp} 
+ 2 \beta_2 \omega |_{\Supfamp}  \right ) 
 \zeta_0{}_b \epsilon^{b3}{}_a
\epsilon_{AB} \D^B Y^a 
& \\
- \left ( \lie_{ \zeta_0\,}  
\left (  \alpha_2 \normal (\omega) |_{\Supfamp} 
+ 2 \beta_2 \omega |_{\Supfamp}   \right ) +
\Big ( \beta_2 [\Wtwo] + \frac{1}{2} \alpha_2 [ \normal(\Wtwo)] \Big ) 
+ \Qone \alpha_2 
[ \normal (  \normal(\omega) ) |_{\Supfamp} ] \right ) \eta_A
 =0 &,
\end{align*}
where we used $\lie_{ \zeta_0} \bm{\eta}_A = -\zeta_0{}_b \epsilon^{b3}{}_a\epsilon_{AB} \D^B Y^a$.
By the hypotheses of the Proposition, the first factor in parenthesis is
non-zero, so linear independence of $\D_A Y^a$ and 
$\epsilon_{AB} \D^B Y^a$ implies firstly that
$\dot{f}^{(2)}_{a} =0$, which combined with the second in (\ref{jumph}) gives
\begin{align*}
f^{(2)}_a = 0 \quad (a=1,2), \quad \quad \quad
f_0^{(2)}= - \frac{\beta_1}{\beta_2} f_3^{(2)}, \quad \quad
f_3^{(2)} \mbox{constant},
\end{align*}
and secondly that, for $a=1,2$,
\begin{align*}
\beta_2 \dot{q}^{(2)}_a + 
\alpha_2 \normal(\omega) |_{\Supfamp} 
\zeta_0{}_b \epsilon^{b3}{}_a
+ 2 \beta_2 \omega |_{\Supfamp}
\zeta_0{}_b \epsilon^{b3}{}_a =0 \quad \quad 
\stackrel{(\ref{dotqm12})}{\Longleftrightarrow}
\quad \quad
\normal(\omega) |_{\Supfamp} 
\zeta_0{}_b \epsilon^{b3}{}_a =0,
\end{align*}
and for $a=3$, after inserting (\ref{dotqm3}),
\begin{align*}
[  \normal(\Wtwo)] = -2 \lie_{ \zeta_0} ( \normal(\omega)|_{\Supfamp}) -2 
\Qone [ \normal (  \normal(\omega) ) ].
\end{align*}
Observe that  the constancy of
$f^{(2)}_3$ and the vanishing of 
$f^{(2)}_1$, 
$f^{(2)}_2$, together  with  the fact that $f$ and $[k]$ are 
$\tau$- and $\varphi$-independent imply, via (\ref{jumpf}), that 
$s_a =0$ for $a=1,2$, and $s_0, s_3$ are both constant. With these
  restrictions, $\vecVtwo$ (from \eqref{ExpV2} and
    \eqref{Vtwotau}),  $[\hatQtwo]$ (from \eqref{JumphatQtwo}) and (\ref{jumpf}) become
\begin{align}
\vecVtwo & = 
\left ( c_0^{(2)} + c_1^{(2)} \tau  \right ) \partial_\tau + 
\left ( f_3^{(2)} \D^A Y^3 
+ ( b^{(2)}_a \tau + d^{(2)}_a) \epsilon^{AB} \D_B Y^a 
\right )
\partial_A, \label{JumpvecWtwo} \\
[ \hatQtwo] & = 2 \Q [ e^{\frac{\lambda}{2}} f] - 2 \lieQone
              - s_3 Y^3 - s_0,  \label{JumphatQtwobis} \\
 [k] & = - \frac{\beta_2}{\alpha_2} \Q [ e^{\lambda/2} f] 
+ \frac{\beta_2}{2 \alpha_2} s_0
+ \frac{1}{2} \left ( \frac{\beta_2}{\alpha_2} s_3 
+ f_3^{(2)} \right ) Y^3.
\label{jumpk}
\end{align}
The remaining equations involve the
curvature terms in the last line of 
(\ref{GenMatchkappa.2}).  With our extension $e_i = \partial_{x^i}$, it holds $[ n ,  e_i ] =0$, and then  \begin{align*}
\normal^{\alpha} \nabla_{\alpha} e^{\beta}_i |_{\Supfamp} = 
\kappa_i^{j} e^{\beta}_j.
\end{align*}
Using that the normal field $\normal$ is geodesic, as well as (\ref{nnRee}) and (\ref{hkappa}), we compute
\begin{align*}
                  \normal^{\mu} \normal^{\nu} \normal^{\delta} ( \nabla_{\delta} 
R_{\alpha\mu\beta\nu} ) e^{\alpha}_i e^{\beta}_j
& + 2 \normal^{\mu} \normal^{\nu} R_{\delta\mu\alpha\nu} e^{\delta}_l e^{\alpha}_j \kappa^{l}_i
 + 2 \normal^{\mu} \normal^{\nu} R_{\delta\mu\alpha\nu} e^{\delta}_l e^{\alpha}_i \kappa^{l}_j 
 \\  
& = \normal^{\delta} \nabla_{\delta}
\left ( 
\normal^{\mu} \normal^{\nu} R_{\alpha\mu\beta\nu}  e^{\alpha}_i e^{\beta}_j \right )
+  \normal^{\mu} \normal^{\nu} R_{\delta\mu\alpha\nu} e^{\delta}_l e^{\alpha}_j \kappa^{l}_i
         +  \normal^{\mu} \normal^{\nu} R_{\delta\mu\alpha\nu} e^{\delta}_l e^{\alpha}_i \kappa^{l}_j
  \\
 & = \left ( \normal (\B) +  
\frac{2\beta_1}{\alpha_1} \B \right )
\delta_i^{\tau} \delta_j^{\tau} + 
\left ( \normal (\A) + \frac{2 \beta_2}{\alpha_2} \A \right )
\gsph_{AB} \delta_i^{A} \delta_j^B,
\end{align*}
and therefore, given that $\Qone [\A ] =0$ and $\Qone [ \B] =0$, the last line
in (\ref{GenMatchkappa.2}) simplifies to
\begin{align*}
 \Qone^2 \left [ 
\normal^{\mu} \normal^{\nu} \normal^{\delta} ( \nabla_{\delta} 
R_{\alpha\mu\beta\nu} ) e^{\alpha}_i e^{\beta}_j
+ 2 \normal^{\mu} \normal^{\nu} R_{\delta\mu\alpha\nu} e^{\delta}_l e^{\alpha}_j \kappa^{l}_i
\right . & \left . 
+ 2 \normal^{\mu} \normal^{\nu} R_{\delta\mu\alpha\nu} e^{\delta}_l e^{\alpha}_i \kappa^{l}_j  \right ]  =  \\
& \Qone^2 [ \normal(\B)] \delta_{i}^{\tau} \delta_j^{\tau}
+  \Qone^2 [  \normal(\A)] \delta_i^{A} \delta_j^B \gsph_{AB}.
\end{align*}
We also need  $\kappa_{il} \kappa^{l}{}_j$, namely
\begin{align*}
\kappa_{il} \kappa^{l}{}_j=
\frac{\beta_1^2}{\alpha_1} \delta_i^{\tau} \delta_j^\tau
+ \frac{\beta_2^2}{\alpha_2} 
\delta_i^{A} \delta_j^B \gsph_{AB},
\end{align*}
and we can finally obtain the explicit form of $\L$ by collecting the appropriate ${A,B}$ terms
in (\ref{GenMatchkappa.2}) (all except for the two Hessians):
\begin{align*}
\L =&  
- 2 \beta_2 f_a^{(2)} Y^a
+ \left [ \hatQtwo \left ( - \A + \frac{\beta_2^2}{\alpha_2} \right ) \right ]
+ 2 \lieQone
\left ( - \A^- + \frac{\beta_2^2}{\alpha_2} \right ) \\
& + 2 \alpha_2 [  \normal( k)]  + 4 [k] \beta_2
- 2 [m] \beta_2
- \Qone^2 [  \normal(\A)],
\end{align*}
where for any quantity $a$ we set  $a^-\defi a|_{\Supfamp^-}$.
Hence (\ref{Lsm}), and the properties we have found for $f_a^{(2)}$
and $s_a$ yield
\begin{align}
[m] = &  2 [k] - \frac{1}{2 \beta_2} \Qone^2 
[  \normal(\A)] + \frac{\alpha_2}{\beta_2} [  \normal(k)] \nonumber \\
& - \left (f_3^{(2)} + \frac{s_3}{2\beta_2} \right ) Y^3
+ \frac{1}{2} \left [ \hatQtwo \left ( - \frac{\A}{\beta_2} 
+ \frac{\beta_2}{\alpha_2} \right ) \right ]
+ \lieQone
\left ( - \frac{\A^-}{\beta_2} + \frac{\beta_2}{\alpha_2} \right ) , \label{jumpm}
\end{align}
The last step is to impose the $\{\tau, \tau\}$ component 
of   (\ref{GenMatchkappa.2}). Using
$\Vtwo^{\tau} = c_0^{(2)} 
+ c_1^{(2)} \tau$ (see \eqref{JumpvecWtwo})
and the fact that $[\hatQtwo] + 2  \lieQone$ is
$\tau$-independent (see (\ref{JumphatQtwobis})), this $\{\tau, \tau\}$
component is
\begin{align*}
2 \beta_1 c_1^{(2)} 
+ \left [ \hatQtwo \left ( - \B + \frac{\beta_1^2}{\alpha_1} \right ) \right ]
& + 2 (\lie_{ \zeta\,} \Qone)
\left ( - \B^- + \frac{\beta_1^2}{\alpha_1} \right ) 
\\
& + 2 \alpha_1 [ \normal(h)] 
+ 4 \beta_1 [h]
-2 [m] \beta_1 - \Qone^2 [  \normal(\B)] =0,
 \end{align*}
where we also used $[\normal ( \Q^2 \omega^2)] =0$. Solving
for $[ \normal(h)]$ and inserting $[m]$ from (\ref{jumpm}) and
$[h]$ from (\ref{jumph}) one finds 
\begin{align*}
[\normal(h) ]  = & \frac{\beta_1\alpha_2}{\alpha_1 \beta_2}[\normal(k)]
+ \frac{\beta_1}{\alpha_1} \left ( 1 - \frac{\beta_1 \alpha_2}{\alpha_1 \beta_2 } 
\right ) \left ( 2 [k] - f_3^{(2)} Y^3 \right )
- \frac{\beta_1}{2 \alpha_1 \beta_2} s_3 Y^3 \\
& + \frac{\beta_1}{2 \alpha_1} \left ( \left [ 
\hatQtwo \left (
-\frac{\A}{\beta_2}+\frac{\B}{\beta_1}+\frac{\beta_2}{\alpha_2}-\frac{\beta_1}{\alpha_1}\right) \right ]
+ 2\, \zeta(\Qone) \left(-\frac{\A^-}{\beta_2}+\frac{\B^-}{\beta_1}+\frac{\beta_2}{\alpha_2}-\frac{\beta_1}{\alpha_1}\right)
\right .
 \\
& + \left . \Qone^2 \left [ - \frac{\normal(\A)}{\beta_2} + \frac{\normal(\B)}{\beta_1}  \right ] 
\right ). 
\end{align*}
This concludes the process of solving the second order matching 
conditions in the $g$-gauge. We put together the results 
and restore the $g's$ and $hg's$:
\begin{align}
  [ \Wtwo^g ]   = &    b^{(2)}_3 - 2  \zeta_0 (\omega^g |_{\Supfamp}), \nonumber \\
  [ \normal(\Wtwo)^g]  = & -2 \zeta_0 ( \normal(\omega^g) |_{\Supfamp})-2\Qone^{hg}
                      [ \normal( \normal(\omega^g))],\nonumber\\
 [k^g]  = &  - \frac{\beta_2}{\alpha_2} \Q[ e^{\lambda/2}  f^g] +
\frac{\beta_2}{2 \alpha_2} s_0 
+ \frac{1}{2} \left ( \frac{\beta_2}{\alpha_2} s_3 + 
f_3^{(2)} \right ) Y^3, \label{jumpkg} \\
 [h^g]   = &  - \frac{1}{2 } c^{(2)}_1+  \frac{\beta_1 \alpha_2}{2 \alpha_1\beta_2} \left (  2 [k^g] 
- f_3^{(2)} Y^3 \right ) ,    \nonumber\\
[m^g] = &  2 [k^g] - \frac{1}{2 \beta_2} (\Qone^{hg})^2 
[  \normal(\A)] + \frac{\alpha_2}{\beta_2} [  \normal(k)] \nonumber \\
& - \left (f_3^{(2)} + \frac{s_3}{2\beta_2} \right ) Y^3
+ \frac{1}{2} \left [ \hatQtwo^{hg} \left ( - \frac{\A}{\beta_2} 
+ \frac{\beta_2}{\alpha_2} \right ) \right ]
+ \zeta(\Qone^{hg})
 \left ( - \frac{\A^-}{\beta_2} + \frac{\beta_2}{\alpha_2} \right )  \label{jumpmg} \\
[\normal(h^g) ]  = &  \frac{\beta_1\alpha_2}{\alpha_1 \beta_2}[\normal(k^g)]
+ \frac{\beta_1}{\alpha_1} \left ( 1 - \frac{\beta_1 \alpha_2}{\alpha_1 \beta_2 } 
\right ) \left ( 2 [k^g] - f_3^{(2)} Y^3 \right )
- \frac{\beta_1}{2 \alpha_1 \beta_2} s_3 Y^3 \nonumber \\
& + \frac{\beta_1}{2 \alpha_1} \left ( \left [ 
\hatQtwo^{hg} \left (
-\frac{\A}{\beta_2}+\frac{\B}{\beta_1}+\frac{\beta_2}{\alpha_2}-\frac{\beta_1}{\alpha_1}\right) \right ]
+ 2 \zeta(\Qone^{hg}) \left(-\frac{\A^-}{\beta_2}+\frac{\B^-}{\beta_1}+\frac{\beta_2}{\alpha_2}-\frac{\beta_1}{\alpha_1}\right)
\right .
 \nonumber \\
& + \left . (\Qone^{hg})^2 \left [ - \frac{\normal(\A)}{\beta_2} + \frac{\normal(\B)}{\beta_1}  \right ] 
\right ), \label{jumpnhg} \\
[\vecTtwo^{hg}] = &
\left ( c_0^{(2)} + c_1^{(2)} \tau  \right ) \partial_\tau + 
\left ( f_3^{(2)} \D^A Y^3 
+ ( b^{(2)}_m \tau + d^{(2)}_m) \epsilon^{AB} \D_B Y^m 
\right )
\partial_A 
+ 2 \Qone^{hg} \kappa( \zeta) + D_{ \zeta\, } {\zeta}  \label{jumpTtwo} \\
[ \hatQtwo^{hg}] = & 2 \Q [ e^{\lambda/2} f^g] - 2 \zeta (\Qone^{hg} )
- s_3 Y^3 - s_0,  \label{jumpQhat2} 
\end{align}
where $f_3^{(2)}, c_0^{(2)}, c_1^{(2)}, b^{(2)}_1, b^{(2)}_2, b^{(2)}_3, d^{(2)}_1,
d^{(2)}_2, d^{(2)}_3, s_0, s_3$ are
constants, the spherical Killing $\zeta_0$
decomposes as $\zeta_0 = \zeta_0{}_a  \eta^{\,a}$  and $\{\zeta_0{}_a\}$ and
$\{ b^{(2)}_a\}$ are related by 
\begin{align*}
b^{(2)}_a + 2 \omega^g |_{\Supfamp} \zeta_0{}_b \epsilon^{b3}{}_a =0 \quad \quad (a=1,2).
\end{align*}
Observe that the constant $s_0$ appears only in (\ref{jumpkg}) and
(\ref{jumpQhat2}), accompanying the term $[e^{\lambda/2} f^g]$.
This reflects the fact that $f$ is defined up to an arbitrary 
additive function that can be different on both sides of $\Supfamp$.
The arbitrary difference $[e^{\lambda/2} (f^g - f) ]$ can thus be combined with
$s_0$ to produce a single constant. This will be used in the redefinition of $s_0$ below.
The constants $c_0^{(2)}$, $d_a^{(2)}$ appear
only in $[\vecTtwo^{hg}]$ and state that $[\vecTtwo^{hg}]$ is defined
up to an additive Killing vector $\zeta^{(2)} := c_0^{(2)} \partial_{\tau}
+ d_a^{(2)} \epsilon^{AB} \D_B Y^a \partial_A$ of
$(\Supfamp,h_{ij})$.

We can now apply the gauge relations described in Lemma \ref{b1=0}
to rewrite these conditions in the original gauge. We introduce the redefinitions
of constants (see Remark \ref{rede} below)
\begin{align}
  &f_3^{(2)}\to -H_1,\quad b^{(2)}_3 \to  D_3,\quad
c_1^{(2)}\to -H_0,\nonumber\\
&s_0\to 2\frac{\alpha_2}{\beta_2}c_0+2\Q[e^{\lambda/2} (f^g - f) ],\quad
s_3 \to  \frac{\alpha_2}{\beta_2}(2c_1+H_1) \label{redefinition}
\end{align}
and use the explicit expression \eqref{defsalphabeta} 
for $\alpha_1, \alpha_2, \beta_1, \beta_2$
which imply
\begin{align*}
\frac{\beta_1 \alpha_2}{\alpha_1 \beta_2 } = \frac{\Q \normal(\nu)}{2 \normal(\Q)}, 
\quad \quad \quad
\frac{\beta_2}{\alpha_2}  =  \frac{\normal(\Q)}{\Q},
\quad \quad \quad
\frac{\beta_1}{\alpha_1} = \frac{1}{2} \normal(\nu),
\end{align*}
and the first four equations yield (\ref{PropWtwo})-(\ref{Proph}) immediately.
From Lemma \ref{b1=0} we have 
$\Qone^{hg}{}^{\pm} = \Qone^{\pm} = \Qone$, $\hatQtwo^{hg}{}^+ = \hatQtwo^+$ and
$\hatQtwo^{hg}{}^- = \hatQtwo^- - 2 b_1 \tau \eta(\Qone)$, while
Proposition \ref{teo:first_order_matching} states
$[\vecTone] = 
  b_1 \tau \eta + \zeta$. Recalling the definition (\ref{defor})
it is immediate that (\ref{jumpQhat2}) is equivalent to (\ref{jumpdefor}). 
Next, we use the identity $[ab] = a^+ [b] + [a] b^-$ 
(valid for any $a,b$) to compute, for an arbitrary
quantity $\P$,
\begin{align*}
[\hatQtwo^{hg} \P] + 2 \zeta(\Qone^{hg}) \P^- & =
[\hatQtwo \P] + 2 b_1 \tau \eta(\Qone) \P^-
+ 2 \zeta(\Qone) \P^- \\
& = [ ( \defor + 2 \Q  e^{\lambda/2} f )\, \P]
- 2 [ \vecTone(Q_1) \, \P]
+ 2 ( \zeta(\Qone) + b_1 \tau \eta(\Qone) ) \P^-  \\
& = [ ( \defor + 2 \Q  e^{\lambda/2} f )\, \P]
- 2 \vecTone^+(Q_1) [\P].
\end{align*}
This identity applied respectively to
\begin{align*}
\P = - \frac{\A}{\beta_2} + \frac{\beta_2}{\alpha_2}
\quad \quad \quad \quad \mbox{and} \quad \quad \quad \quad
\P = - \frac{\A}{\beta_2} + \frac{\B}{\beta_1}+\frac{\beta_2}{\alpha_2}
- \frac{\beta_1}{\alpha_1}
\end{align*}
transforms (\ref{jumpmg}) into
into (\ref{Propm})
and (\ref{jumpnhg}) into
(\ref{Propnh}),
after using that  $\vecTone^+(\Qone) [\P] = \vecTone^+ (\Qone [\P]) =0$,
which follows form the constancy of $[\P]$  on $\Supfamp$ and
(\ref{teo:Q1_matching_R}).
To conclude the proof, note that
(\ref{jumpTtwo}) simply determines $\vecTtwo^+$ in terms
of $\vecTtwo^-$. Neither term appears in the rest of
expressions, so this condition poses no additional
restriction to the matching.
\fin
\end{proof}
\begin{remark}
\label{rede}
    The redefinition of constants (\ref{redefinition}) at the end of the proof
    has been done so that the result matches the expressions found and used in 
    \cite{ReinaVera2015}.
\end{remark}

\section{Basic  analytic lemmas}
\label{App:Lemmas}
We use the notation, conventions and definitions of elliptic operators in \cite{Gilbarg}. Specifically, $\domain$ denotes a 
domain of $\mathbb{R}^n$ (i.e. a connected open subset). As usual $\partial \domain$
denotes its topological boundary and $\overline{\domain}$ its closure.
A second order operator
 $L = a^{ij}(x)\frac{\partial^2  }{\partial x_ix_j} + b^i(x) \frac{\partial  }{\partial x_i}  + c(x)$, $a^{ij}(x) = a^{ji}(x)$,
defined on $\domain$ is
uniformly elliptic if the lowest eigenvalue $\lambda(x)$ and largest eigenvalue $\Lambda(x)$ satisfy that $\lambda$ is positive
and
$\Lambda/\lambda$ is bounded  on $\domain$. 
At points $x \in \partial \domain$ where the outer normal exists, this will be denoted by $\partial_{\nu}$. 

We need the following version of the boundary point lemma and maximum principle. 

\begin{lemma}[\bf Boundary point lemma]
		Suppose that $L$ is uniformly elliptic, $u \in C^2(\domain)$ and $Lu \geq 0$ in $\domain$. Let $x_0 \in \partial \domain$ be such that
\begin{enumerate}
	\item $u$ is continuous at $x_0$ and $u(x_0) \geq 0$.
	\item $u(x_0) > u(x)$ for all $x \in \domain$,
	\item $\partial \domain$ satisfies an interior sphere condition at $x_0$ (i.e. there exists a ball $B \subset \domain$ with $x_0 \in \partial B$).
\item $c \leq 0$ and $|c|/\lambda$, $|b^i|/\lambda$ are bounded in $B$.
\end{enumerate}
Then the outer normal derivative of $u$ at $x_0$, if it exists, satisfies the strict inequality 
\begin{equation}
	  \partial_\nu u (x_0) > 0.
\end{equation}
\finn
\end{lemma}

Although not stated in this form in \cite{Gilbarg}, the proof
of Lemma 3.4 in \cite{Gilbarg} also establishes this version.
Concerning the next result, its validity
is explicitly stated in a remark after Theorem 3.5 in \cite{Gilbarg}. 

\begin{theorem}[\bf Strong maximum principle]
Let $L$ be uniformly elliptic on a  domain $\domain$ and $u \in C^2(\domain)$ satisfy  $Lu \geq 0 \,\, (\leq 0)$. Assume $c \leq 0$ and $|c|/\lambda$, $|b|/\lambda$
are locally bounded in $\domain$. Then $u$ cannot achieve a non-negative maximum (non-positive minimum) in the interior of $\domain$ unless it is constant.\finn
\end{theorem}

We shall use these results in a  very simple context, namely for second order ODE operators. We consider two types of intervals
$I^+= (0,\valA )$ ($\valA  >0$) and $I^- = (\valA ,\infty)$. In both cases the  interior sphere condition is obviously satisfied. We use $\r$
to denote the real coordinate of $I^{+}$ and $I^{-}$. The outer normal derivative at $\r=\valA $
is obviously $\partial_{\r}$ for $I^+$ and $\partial_{\r}$ for $I^-$.

 The first result we need is the following (the proof
is an essentially trivial consequence of the previous results, but we include it for completeness)

\begin{lemma}
\label{LemmaIp}
On $I^+=(0,\valA )$, let  $L^+$ be
\begin{equation}
	L^+:= \frac{d^2}{d\r^2} + b^+ (\r) \frac{d}{d\r} + c^+ (\r) ,
\end{equation}
where $|b^+(r)|$ and $|c^+(r)|$ are locally bounded in $(0,\valA ]$. 
Let $f \in C^2(I^+) \cap C^0(\overline{I^+}) \cap C^1((0,\valA ])$ satisfy $L^+ f = 0$.
Assume $c^{+}(\r) \leq 0$. Then,
\begin{itemize}
	\item[(i)] $f(\valA )> f(0) \geq 0 \quad \Longrightarrow \quad \partial_\r f (\valA ) >0$,
	\item[(ii)]  $f(\valA )<f(0) \leq 0 \quad \Longrightarrow \quad \partial_\r f (\valA ) <0$,
	\item[(iii)] $f(\valA )= f(0)=0 \quad \Longrightarrow \quad f(\r) = 0 \quad \forall \r \in  \overline{I^+}$.
\end{itemize}
\end{lemma}

\begin{proof} $L^+$ is obviously uniformly elliptic with $\lambda =1$. Note, in particular that $|c^+|/\lambda$ and $|b^+|/\lambda$ are locally bounded in
$I^+$.
Consider first the case $f(\valA )>f(0) \geq 0$. Since $f$ is non-constant, the strong maximum principle implies that 
the supremum of $f$ in $\overline{I^+}$ is $f(\valA )$ and it is achieved only at $\r =\valA $. By
the boundedness of $c$ and $|b^i|$ on $(0,\valA ]$ we may apply the boundary point lemma at $\r = \valA $ 
to conclude $\partial_{\r} f(\valA ) > 0$.
The case (ii) follows from (i) when applied to $-f$. Finally, when $f(\valA )=f(0)=0$ the strong maximum principle implies $f(\r)=0$
immediately.
\fin
\end{proof}

\begin{lemma}
\label{LemmaRot}
In the setting of Lemma \ref{LemmaIp} assume further that
$c^+(r)$ is not identically zero, $f(0) \geq 0$ and $\partial_rf(0)=0$. Then
$f(\valA)>f(0)$ and $\partial_r f(\valA)>0$.
\end{lemma}
\begin{proof}
  The supremum of $f$ in $\overline{I^{+}}$ is clearly non-negative and $f$ cannot be constant
    because $c^+(r)$ is not identically zero. Thus, 
    the strong maximum principle implies that the supremum can only be achieved at the boundary. This supremum cannot be $f(0)$ because $\partial_r f(0)=0$ would contradict the boundary point lemma. Thus, the supremum is at
    $f(\valA)> f(0)$ and  the boundary point lemma implies  $\partial_r f(\valA)>0$, as claimed.\fin
\end{proof}

An analogous result to Lemma \ref{LemmaIp} holds for the unbounded domain $I^-$.

\begin{lemma}
\label{LemmaIm}
On $I^-= (\valA ,\infty)$, let $L^-$ be
\begin{equation}
	L^-:= \frac{d^2}{d\r^2} + b^- (\r) \frac{d}{d\r} + c^- (\r) ,
\end{equation}
where $|b^-(\r)|$ and $|c^-(\r)|$ are  bounded in $\overline{I^+}$. 
Let $f \in C^2(I^-) \cap C^1(\overline{I^-})$  satisfy $L^- f = 0$ and
\begin{equation}
	\lim_{\r \to \infty} f (\r) = \finfty < \infty    \label{limitfm}
\end{equation}
Assume that  $c^- (\r) \leq  0$ in $I^-$. Then,
\begin{itemize}
	\item[(i)] $f(\valA )> \finfty \geq 0 \quad \Longrightarrow \quad \partial_\r f (\valA ) < 0$,
	\item[(ii)] $f(\valA )< \finfty \leq 0 \quad \Longrightarrow \quad \partial_\r f (\valA ) > 0$,
	\item[(iii)] $f(\valA )= \finfty =0 \quad \Longrightarrow \quad f(\r) = 0 \quad \forall \r \in  \overline{I^-}$.
\end{itemize}
\end{lemma}

\begin{proof} The first two statements are immediate consequences of 
the strong maximum principle and boundary point lemma. For the third one,  assume by contradiction that there is $\r_0 > \valA $ with 
$f(\r_0) \neq 0$. By replacing $f \rightarrow -f$ we may assume without loss of generality that $f(\r_0) >0$. By the limit assumption 
(\ref{limitfm}) with $\finfty =0$ there exists $\r_1$ sufficiently large (in particular satisfying $\r_1 > \r_0$) with $f(\r_1) < f(\r_0)$. The strong maximum
principle applied to $(0,\r_1)$ gives a contradition, because the function is not constant but 
its supremum (which is at least $f(\r_0)$ and hence positive) is achieved necessarily at an interior point. Thus, it must be that $f(\r)=0$ as claimed.
\fin
\end{proof}

\section{Existence and uniqueness of bounded global solutions of a class of ODE}
  
\label{App:bocher-like}

We use the following result (Corollary 6.2 in \cite{Kigurazde}).
We will use
$f\in C^0((a,\infty])$ to indicate $f\in C^0(a,\infty)$ and that the limit of $f(s)$ as   $s\to \infty$
exists and is finite.
\begin{lemma}
  \label{Asymp}
  Consider the second order homonegous ODE
  \begin{align}
    \ddot{z} + \alpha(s) \dot{z} + \beta(s) z =0  \label{ODEz}
  \end{align}
  defined on the interval $(s_0, \infty)$. Assume that
    $\alpha,\beta \in C^{0}((s_0,\infty])$
  and  let
  $\alpha_0 := \lim_{s\rightarrow \infty} \alpha(s)$, $\beta_0
  := \lim_{s\rightarrow \infty} \beta(s)$.  Assume further that
  \begin{align}
    \int_{s_0}^\infty \left | \alpha(s) - \alpha_0 \right | ds < \infty, \quad \quad
    \int_{s_0}^\infty \left | \beta(s) - \beta_0 \right | ds < \infty 
    \label{integral}
  \end{align}
    Define
  \begin{align*}
    \mu_{\pm} := \frac{-\alpha_0 \pm \sqrt{\alpha_0^2 - 4 \beta_0}}{2}.
  \end{align*}
  If $\mu_{\pm}$ are real and distinct, then (\ref{ODEz}) admits two lineary independent real solutions $z_{\pm}(s)$ satisfying the following asymptotic
  behaviour at $s \rightarrow \infty$
  \begin{align}
    z_{\pm}(s) = e^{\mu_{\pm} s} \left ( 1 + o(1) \right ),
    \quad \quad \dot{z}_\pm(s) = e^{\mu_{\pm} s} \left ( \mu_{\pm} + o(1) \right ). \label{behav}
  \end{align}\fin
\end{lemma}

We want to apply this result to analyze the behaviour of solutions
to ODE with certain type of singularities at $t=0$. Specifically, in the main text we need the following lemma.

\begin{lemma}
  \label{Lemma:behaviour}
  Consider the second order homogeneous ODE
  \begin{align}
    t^2 x'' + t \Af(t) x'  +\Bf(t) x= 0 \label{ODEx}
  \end{align}
  defined in the interval $(0,t_0)$. Assume that
  $\Af(t),\Bf(t) \in C^1([0,t_0))$ and let $a_0 := \Af(0)$ and $b_0 := \Bf(0)$. Define
  \begin{align*}
    \lambda_{\pm} := \frac{a_0 - 1 \pm \sqrt{(a_0-1)^2 - 4 b_0}}{2}.
  \end{align*}
\item[(i)] If $4 b_0 < (a_0-1)^2  $ there exist two real linearly independent
  solutions $x_{\pm}(t)$ of (\ref{ODEx}),
    and have the following behaviour near $t=0$:
  \begin{align}
    x_{\pm}(t) = t^{-\lambda_{\pm}} \left ( 1 + o(1) \right ), \quad
    \quad x_{\pm}'(t)= -t^{-(1+\lambda_{\pm})} \left ( \lambda_{\pm} + o(1) \right ).
    \label{claimx} 
  \end{align}
  \item[(ii)] If either $b_0 <0$ or $(b_0 =0, a_0 >1)$ then there exists a unique up to scaling solution $x(t)$ of (\ref{ODEx}) that stays bounded in $(0,t_0)$. $x(t)$ extends continuously at $t=0$ 
with $x(0)=0$ if  $b_0 <0$
and $x(0)\neq 0$ if $(b_0 =0, a_0 >1)$.
\item[(iii)] When $(b_0 =0,a_0 >1)$ assume further that
  $\Bf(t) = t^2 \Qf(t)$ where
  $\Qf (t) \in C^1([0,t_0))$ and satisfying $\Qf(0) \neq 0$. Then, the bounded
  solution $x(t)$ in item (ii)  extends to a $C^2$ function
  in $[0,t_0)$ satisfiying $x'(0)=0$.
\end{lemma}

\begin{proof}
  Consider the change of variables $t(s) = e^{-s}$ which sends $(0,t_0)$ to
  $(s_0:=- \ln t_0,\infty)$. Define $z(s):=x(t(s))$, $\alpha(s):=1 - \Af(t(s))$,
  $\beta(s):=\Bf(t(s))$. The ODE (\ref{ODEx}) takes the form (\ref{ODEz}). For any function $\gamma(s)$ we have the equality
  \begin{align*}
    \int_{a}^{\infty} \gamma(s) ds = \int_{0}^{e^{-a}} \frac{\gamma(s(t))}{t} dt.
  \end{align*}
  Since the functions $\Af(t)$ and $\Bf(t)$ are $C^1$ up to $t=0$, 
  $|\Af(t)-a_0|/t$ and $|\Bf(t)-b_0|/t$ are bounded,
  so the  hypotheses of Lemma \ref{Asymp} are satisfied. In addition
  $\alpha_0 = 1-  a_0$, $\beta_0=b_0$ so that, in particular
  $\mu_{\pm} = \lambda_{\pm}$. When $4 b_0 < (a_0-1)^2$ we have
  that $\lambda_{+}$ and $\lambda_{-}$ are real and distinct.
  The linearly independent
  solutions $z_{\pm}(s)$ whose existence is guaranteed by Lemma
  \ref{Asymp} show the existence of two solutions $x_{\pm}(t)$ with
  the behaviour claimed in (\ref{claimx}). This proves item (i).
  
  For item (ii), in either case $b_0<0$ or $(b_0=0, a_0 >1)$ we have
  $\lambda_+ >0$ and $\lambda_-\leq 0$. The solution $x_{+}(t)$
  of item (i) is unbounded near zero, while $x_{-}(t)$ is bounded. Since the general solution is a linear combination of both, the first statement follows. The continuous extension at $t=0$ is direct from (\ref{claimx}) given that
    $\lambda_{-}<0$ when $b_0 <0$ and
    $\lambda_{-}=0$ when $(b_0 =0, a_0 >1)$.

  Finally, for item (iii) we already know  by item (ii) that
    $x_{-}(t)$  (which is the only one up to scaling that remains bounded) admits a continuous extension to $t=0$. 
    Furthermore the corresponding
  $z_{-}(s)$ satisfies
  \begin{align}
    \lim_{s \rightarrow \infty}  \dot z_{-}(s)=0 \label{limzminus}
  \end{align}
  as a consequence of (\ref{behav}) and $\mu_{-} = \lambda_{-} =0$. We next prove that $x_{-}'(t)$ satisties $\lim_{t \rightarrow 0^+} x_{-}'(t) = 0$. First observe that $\Qf(0) \neq 0$ implies that there is $t_1>0$ sufficiently small such that
  $\Bf(t) = t^2 \Qf(t)$ has a constant sign in $(0,t_1)$. We restrict to this domain and to the equivalent in the $s$-variable $(s_1:=- \ln t_1, \infty)$, where $\beta(s)$ is guaranteed to vanish nowhere. Since both
  $\alpha(s)$ and $\beta(s)$ are $C^1((s_1, \infty))$ we may take a derivative of (\ref{ODEz}) and replace $z(s)$ obtained algebraically from (\ref{ODEz}) itself.
  The result is the following ODE for $\q(s):=\dot{z}(s)$
  \begin{align}
    \ddot{\q} + \left ( \alpha(s) - \frac{\dot{\beta}(s)}{\beta(s)}
    \right ) \dot{\q} + \left ( \beta(s) + \dot{\alpha}(s)
    - \frac{\alpha(s) \dot{\beta}(s)}{\beta(s)} \right ) \q
    =0. \label{ODE3}
  \end{align}
  It is immediate to compute
  \begin{align*}
    \dot{\alpha} (s) =  \left . t \Af'(t)  \right |_{t= e^{-s}},
    \quad \quad
    \frac{\dot{\beta}(s)}{\beta(s)}= \left .
    -2 - t \frac{\Qf'(t)}{\Qf(t)} \right |_{t= e^{-s}}
  \end{align*}
  so the coefficients of the ODE (\ref{ODE3}) are continuous in $(s_1, \infty]$. One checks easily that the integral conditions (\ref{integral}) are also satisfied.
    The corresponding constants of Lemma \ref{Asymp} are
  \begin{align*}
    \mu_{\pm} = \frac{a_0 -3 \pm | a_0 +1 |}{2}.
  \end{align*}
  Using now $a_0 >1$ we find $\mu_{+} = a_0-1$ and $\mu_{-} =-2$, so by Lemma \ref{Asymp}, the function $\q$ must have the asymptotic behaviour
  \begin{align*}
    \q(s)= e^{-2s} \left (\aone + o(1) \right ) + e^{(a_0-1)s} \left (\atwo + o(1) \right ) 
  \end{align*}
  with constants $\aone, \atwo$. But $\q(s)= \dot{z}(s)$ is forced to 
  approach zero at infinity (see (\ref{limzminus})), so  $\atwo =0$. We conclude that
  \begin{align}
    x'(t) = \left . - (e^{s} \q(s))  \right |_{s= -\ln t} =
    -t \left ( \aone + o(1) \right ) \label{xprime}
  \end{align}
 and we have shown that $x'(t)$ extends continuously to $t=0$ with the value zero. It only remains to show that $x''(t)$ also extends
  continuously to $t=0$, but this follows at once from the ODE itself
  \begin{align*}
    x''(t) + t^{-1} \Af(t) x'(t) + \Qf(t) x(t)=0
  \end{align*}
  since we already know that the second and third terms extend continuously to $t=0$ (the second term by (\ref{xprime})). \fin
\end{proof}

The following theorem is the main result of the appendix.
\begin{theorem}
\label{theorem:global_r}
  Let $a$ be a positive constant. Assume that
  $\Af^{+}, \Bf^{+} : [0,a] \rightarrow \mathbb{R}$ and
  $\Af^{-}, \Bf^{-} : [a,\infty) \rightarrow \mathbb{R}$ are $C^1$
  on their
  respective domains and  that the limits
  \begin{align}
    \lim_{r \rightarrow + \infty} \Af^-(r), \quad
    \lim_{r \rightarrow + \infty} \Bf^-(r), \quad
    \lim_{r \rightarrow + \infty} r^2 \frac{d \Af^-(r)}{dr}, \quad
    \lim_{r \rightarrow + \infty} r^2 \frac{d \Bf^-(r)}{dr}
    \label{conditions2}
  \end{align}
  exist and are finite. Define the constants
  \begin{align*}
    a_0 = \Af^+(0), \quad b_0 = \Bf^+(0), \quad
    a_{\infty} = \lim_{r \rightarrow + \infty} \Af^{-}(r), \quad
    b_{\infty} = \lim_{r \rightarrow + \infty} \Bf^{-}(r)
  \end{align*}
  and assume that $b_0, b_{\infty} < 0$. Let 
  $\inhomo^+ \in C^0([0,a])$, $\inhomo^- \in C^0([a,\infty))$, 
$\czero, \cone \in \mathbb{R}$  and
  consider the ODE problem $(\star)$ defined by
  \begin{align}
    &                                          r^2 \frac{d^2 u^+(r)}{dr^2} + r \Af^+(r) \frac{du^+(r)}{dr}
      + \Bf^+(r) u^+(r) = \inhomo^+(r) \quad \quad \mbox{ on } (0,a],
      \label{Eqminus}\\
    &                      r^2 \frac{d^2 u^-(r)}{dr^2} + r \Af^-(r) \frac{du^-(r)}{dr}
      + \Bf^-(r) u^-(r) = \inhomo^-(r) \quad \quad \mbox{ on } [a,\infty), \label{Eqplus}\\
    &                        u^+(a) - u^-(a) = \czero, \qquad\quad
      \left . \frac{d u^+}{dr} \right |_{r=a} - 
      \left . \frac{d u^-}{dr} \right |_{r=a} = \cone.
      \label{trans}
  \end{align}
  Assume that the inhomogeneous terms satisfy
  $\inhomo^+ (r) = r^{\alpha_0} ( \sigma_+ + o(1))$ 
  near $r=0$ and $\inhomo^- (r) = r^{\alpha_{\infty}} ( \sigma_- + o(1))$
  near infinity, with constants $\alpha_0, \alpha_{\infty}$ and $\sigma_{\pm}$.
  Define
  \begin{align*}
    \lambda^0_{\pm} := \frac{a_0 - 1 \pm \sqrt{(a_0-1)^2 -4 b_0}}{2},
    \quad \quad 
    \lambda^\infty_{\pm} := \frac{1 - a_{\infty} \pm \sqrt{(a_{\infty}-1)^2 -4 b_\infty}}{2}.
  \end{align*}
  If the constants satisfy
  \begin{align}
    & \alpha_0 \geq 0, \quad
      \alpha_0 +\lambda_{-}^0 \neq 0, \quad \alpha_{\infty}  \leq 0,
      \quad \alpha_{\infty} - \lambda^{\infty}_- \neq 0,
      \label{Assump}
  \end{align}
  then $(\star)$ has a unique bounded in $(0,\infty)$ solution $\{u^+(r),u^-(r)\}$  and moreover
    \begin{itemize}
    \item  $u^+(r)$ can be extended
        as a $C^0([0,a])$ function
        of order $O(r^{\min\{-\lambda^0_-,\alpha_0\}})$, and if $1+\lambda^0_-\leq 0$ and
        $\alpha_0-1\geq 0$, then $u^+(r)$ can be also extended
        as a $C^1([0,a])$ function and $u^+{}'(r)$ is $O(r^{\min\{-(1+\lambda^0_-),\alpha_0-1\}})$.
    \item $u^-(r)$ is of order $O(r^{-\min\{|\lambda^\infty_-|,|\alpha_\infty|\}})$
      and $u^-{}'(r)$ is $O(r^{-\min\{|\lambda^\infty_--1|,|\alpha_\infty-1|\}})$ near $r=\infty$.
    \end{itemize}
\end{theorem}

\begin{remark}
  If $\inhomo^\pm=0$ and $\czero=\cone=0$ the unique solution is the trivial $u(r)=0$,
    which also extends to the origin.
\end{remark}

\begin{proof}
  We first analyse the homogeneous problem.
  By Lemma \ref{Lemma:behaviour}, item (i), the homogeneous equation
  (\ref{Eqminus}) with $\inhomo^+=0$
  admits two linearly independent solutions $u^+_{+}(r)$ and
  $u^+_-(r)$,  both of class $C^2((0,a])$ (we may include $r=a$ because
  $\Af^+,\Bf^+$ are $C^1$ up to  this boundary),
  with behaviour near $r=0$ given by
  \begin{align}
    u^+_{\pm}(r) = r^{-\lambda^0_{\pm}} \left ( 1 + o(1) \right ), \quad
    \frac{d u^+_{\pm}(r)}{dr} = - r^{-(1+\lambda^0_{\pm})} \left ( \lambda^0_{\pm} + o(1) \right ).
    \label{eq:u_interior}
  \end{align}
  Since $\lambda^0_-< 0$, because $b_0<0$ by assumption,
  $u^+_-(r)$ extends to a $C^0([0,a])$ function  with $u^+_-(0)=0$.
  In the domain $[a,\infty)$ we consider the change of coordinate
  $r = t^{-1}$ which transforms the homogeneous ODE \eqref{Eqplus} with
  $\inhomo^-=0$ into the form
  \begin{align*}
    t^2 \frac{d^2 \widehat{u}^{-}(t)}{dt^2} + t (2 - \widehat{\Af}(t))
    \frac{d \widehat{u}^{-}(t)}{dt} + \widehat{\Bf}(t) \widehat{u}^{-} (t)=0
  \end{align*}
  where for any function $f(r)$ we denote by $\widehat{f}(t) :=
  f(t^{-1}):(0,a^{-1}]\to \mathbb{R}$. Conditions \eqref{conditions2} imply inmediately that
  $\widehat{\Af}, \widehat{\Bf}$ extend to $t=0$
  as $C^1([0,a^{-1}])$ functions, 
  and we may apply item (i) in Lemma \ref{Lemma:behaviour} to conclude that
  there exist two independent solutions $u^{-}_{\pm}(t) \in 
  C^2([a,\infty))$ satisfying 
  \begin{align*}
    \widehat{u}^-_{\pm}(t) = t^{-\lambda^{\infty}_{\pm}} \left ( 1 + o(1) \right ), \quad
    \frac{d \widehat{u}^-_{\pm}(t)}{dt} = - t^{-(1+\lambda^{\infty}_{\pm})} \left ( \lambda^{\infty}_{\pm} + o(1) \right ).
  \end{align*}
  In terms of the original function, this behaviour translates onto
  \begin{align}
    u^-_{\pm}(r) = r^{\lambda^{\infty}_{\pm}} \left ( 1 + o(1) \right ), \quad
    \frac{d u^-_{\pm}(r)}{dr} =  r^{\lambda^{\infty}_{\pm}-1} \left ( \lambda^{\infty}_{\pm} + o(1) \right ).\label{eq:u_exterior}
  \end{align}
  Since  $\lambda^\infty_-< 0$, because $b_\infty<0$ by assumption,
  $u^-_-(r)$ vanishes at $r\to \infty$.
  We let $W^{\pm}(r)$ be the Wronskian of the functions
  $u^{\pm}_{-}(r)$, $u^{\pm}_{+}(r)$, i.e.
  \begin{align*}
    W^{\pm} (r) := u^{\pm}_{-}(r) \frac{d u^\pm_+(r)}{dr}
    - u^{\pm}_{+}(r) \frac{d u^\pm_-(r)}{dr}.
  \end{align*}
  It is inmediate from the previous considerations that
  \begin{align*}
    & W^{+}(r) = r^{-a_0} \left (-\sqrt{(a_0-1)^2 - 4 b_0} 
      + o(1) \right ) \quad \quad \quad && \mbox{near } r= 0, \\
    & W^{-}(r) = r^{-a_{\infty}} \left (\sqrt{(a_{\infty}-1)^2 - 4 b_\infty}
      + o(1) \right ) \quad \quad && \mbox{near } r = + \infty.
  \end{align*}
  We may now include the inhomogeneus term. The general solution of the inhomogeneous problem on each domain in given by the general formula
  \begin{align}
    u^{\pm}(r) = & C^{\pm}_{+} u^{\pm}_+(r) + C^{\pm}_{-} u^{\pm}_-(r) + u^{\pm}_P(r) , \label{general}
  \end{align}
  where  $C^{\pm}_+$, $C^{\pm}_-$ are arbitrary
  constants and a particular solution
  $u^{\pm}_P(r)$ on each domain is given by
  \begin{align}
    u^{\pm}_P(r)=  u^{\pm}_+(r) \int_{r_{\pm}}^r \frac{u^{\pm}_{-}(s) \inhomo^{\pm}(s)}{s^2 W^{\pm}(s)} ds -
    u^{\pm}_{-}(r) \int_{r_{\pm}}^r \frac{u^{\pm}_{+}(s) \inhomo^{\pm}(s)}{s^2 W^{\pm}(s)} ds,
    \label{eq:uparticular_wronskian}
  \end{align}
  where $r_{\pm}$ are arbitrary values subject to
  $r_- \in (0,a]$ and $r_+ \in [a,\infty)$. The behaviour of the integrands  near zero and near infinity are, respectively,
  \begin{align*}
    \frac{u^+_{\mp}(s) \inhomo^+(s)}{s^2 W^+(s)} = s^{\alpha_0 +
    \lambda^0_{\pm} -1} ( \mu^{+} + o(1) ), \quad
    \frac{u^-_{\mp}(s) \inhomo^-(s)}{s^2 W^-(s)} =
    s^{\alpha_{\infty} -
    \lambda^{\infty}_{\pm} -1} ( \mu^{-} + o(1) ), \quad \quad
  \end{align*} 
  for suitable constants $\mu^{+}$, $\mu^{-}$. Since  $\lambda^0_{+}, \lambda^{\infty}_+$ are positive, assumption (\ref{Assump}) implies $\alpha_0 + \lambda^0_{\pm} \neq 0$, $\alpha_{\infty} - \lambda^{\infty}_{\pm} \neq 0$. It is then straighforward to check, using l'H\^{o}pital's rule, that
  \begin{align*}
    & \int_{r^+}^r \frac{u^+_{\mp}(s) \inhomo^+(s)}{s^2 W^+(s)} ds =
      r^{\alpha_0 + \lambda^0_{\pm}} \left ( \frac{\mu^+}{\alpha_0 + \lambda^0_{\pm}} + o(1) \right ) + \Qf^+_{\pm}  \quad && \mbox{near} \quad r=0, \\
    & \int_{r^-}^r \frac{u^-_{\mp}(s) \inhomo^-(s)}{s^2 W^-(s)} ds =
      r^{\alpha_{\infty} - \lambda^\infty_{\pm}} \left ( \frac{\mu^-}{\alpha_{\infty} - \lambda^{\infty}_{\pm}} + o(1) \right ) + \Qf^-_{\pm}
      \quad && \mbox{near} \quad r = \infty, 
  \end{align*}
  where $\Qf^{+}_{\pm}$ and $\Qf^{-}_{\pm}$ are constants.
  Consequently
  \begin{align}
    U^+_P(r) := & u^+_P(r) - \Qf^+_{+} u^+_{+}(r) + \Qf^+_{-} u^+_{-}(r) \label{Upr_def} \\
                &  = r^{\alpha_0}  \left (\frac{ \mu^+( \lambda^0_{-}-\lambda^0_{+})}{(\alpha_0 + \lambda^0_{+})(\alpha_0 + \lambda^0_{-})} + o (1) \right )
                  \quad\quad &&\mbox{near }  r=0,   \label{Upr0}\\
    U^-_P(r) := & u^-_P(r) -\Qf^-_{+} u^-_{+} (r) + \Qf^-_{-} u^-_{-} (r) \label{Umr_def} \\ 
                &=r^{\alpha_{\infty}} \left (\frac{\mu^-(\lambda_+^\infty-\lambda_-^\infty)}{(\alpha_{\infty} - \lambda^{\infty}_{+})(\alpha_{\infty} - \lambda^{\infty}_{-})} + o (1) \right )
                  \quad && \mbox{near }  r=+ \infty. \nonumber 
  \end{align}
  Absorbing the constants $\Qf^{\pm}_{\pm}$ into $C^{\pm}_{\pm}$, the general solution (\ref{general}) has the form
  \begin{align*}
    u^{\pm}(r) = & C^{\pm}_{+} u^{\pm}_+(r) + C^{\pm}_{-} u^{\pm}_-(r) + U^{\pm}_P(r).
  \end{align*}
  Note that  $U^+_P(r)$ is bounded near $r=0$ while
  $U^-_P(r)$ is bounded at infinity. We now impose  that the solution $\{ u^+(r), u^-(r) \}$
  is bounded everywhere. Since $\lambda^0_-, \lambda^{\infty}_- <0$ and
  $\lambda^0_+, \lambda^{\infty}_+ >0$
  this is equivalent to setting 
  $C^+_+= C^-_{+} =0$ and we are left with two constants to determine.
  Thus, the general solution (\ref{general})  reads
  \begin{equation}
    u^{\pm}(r) =  C^{\pm}_{-} u^{\pm}_-(r) + U^{\pm}_P(r).\label{eq:u_sol}
  \end{equation}

  Imposing the matching conditions (\ref{trans}) yields  a system of two equations of the form
  \begin{equation}
    \label{compat}
    \left (
      \begin{array}{cc}
        u^+_-(a) & - u^-_{-}(a) \\
        \frac{du^+_-}{dr} \big |_{r=a} &
                                         - \frac{du^-_-}{dr} \big |_{r=a}
      \end{array}
    \right )
    \left (
      \begin{array}{c}
        C^+_- \\
        C^-_-
      \end{array} \right )
    = \left (
      \begin{array}{c}
        \czero+ U^-_P(a)-U^+_P(a)\\
        \cone+ \frac{dU^-_P}{dr}\big|_{r=a}-\frac{dU^+_P}{dr}\big|_{r=a}
      \end{array} \right ).
  \end{equation}
  We apply now
  Lemma \ref{LemmaIp} to $u^+_-(r)$ and Lemma \ref{LemmaIm}
  to $u^-_-(r)$  to conclude that $u^{+}_-(r)$ and
  $u^-_-(r)$ and their derivatives are all non-zero at $a$ and,
  moreover, $u^+_-(a)$ has the same sign as it derivative at $a$, while
  $u^-_-(a)$ has opposite sign than its derivative. It follows that the
  $2 \times 2$ matrix in \eqref{compat} is invertible, and hence there
  exits a unique pair of constants $\{C^+_-, C^-_{-}\}$ satifying the
  transition conditions \eqref{trans}. This concludes the proof of existence and uniqueness of a
  bounded solution of problem $(\star)$.
  
  We conclude with the 
  the behaviour of the first derivative of the solutions \eqref{eq:u_sol}, i.e.
  \[\frac{d u^{\pm}(r)}{dr} =  C^\pm_{-} \frac{du^{\pm}_-(r)}{dr} + \frac{dU^{\pm}_P(r)}{dr},\]
  at the origin and at infinity correspondingly.
  Firstly, the terms $dU^{\pm}_P/dr$ are obtained by direct differentiation of their
  definitions \eqref{Upr_def} and \eqref{Umr_def}
  and introducing \eqref{eq:u_interior}-\eqref{eq:u_exterior} together with
  the differentiation of \eqref{eq:uparticular_wronskian}, which provides
  \[
    \frac{du^{\pm}_P(r)}{dr}=  \frac{du^{\pm}_+(r)}{dr} \int_{r_{\pm}}^r \frac{u^{\pm}_{-}(s) \inhomo^{\pm}(s)}{s^2 W^{\pm}(s)} ds -
    \frac{du^{\pm}_{-}(r)}{dr} \int_{r_{\pm}}^r \frac{u^{\pm}_{+}(s) \inhomo^{\pm}(s)}{s^2 W^{\pm}(s)} ds.
  \]
  The results are
  \begin{align}
    \frac{dU^+_P(r)}{dr} 
    &=r^{\alpha_0-1} \left (\frac{\alpha_0\mu^+(\lambda^0_--\lambda^0_+)}{(\alpha_0 + \lambda^0_{+})(\alpha_0 + \lambda^0_{-})}+ o (1) \right )
      \quad \quad  &&\mbox{near }  r=0,   \label{Upr0_diff}\\
    \frac{dU^-_P(r)}{dr} 
    &=r^{\alpha_\infty-1} \left (\frac{\alpha_\infty\mu^-(\lambda^\infty_+-\lambda^\infty_-)}{(\alpha_\infty - \lambda^\infty_{+})(\alpha_\infty - \lambda^\infty_{-})}+ o (1) \right )
      \quad \quad  &&\mbox{near }  r=\infty,   \label{Uprinf_diff}
  \end{align}
  Regarding $u^+(r)$, the assumption $1+\lambda^0_-\leq 0$ ensures,
  c.f. \eqref{eq:u_interior}, that
  $du^{+}_-/dr$  has a limit at $r\to 0$ and is, in fact, $O(r^{-(1+\lambda_-^0)})$.
  The expression \eqref{Upr0_diff} implies that
  if $\alpha_0-1  \geq  0$ then $dU^+_P/dr$ also
  has a limit at $r\to 0$ and is $O(r^{(\alpha_0-1)})$.
  The claim for  $u^+(r)$ follows.
  As for $u^-(r)$, since $\alpha_\infty\leq 0$ and
  $\lambda^\infty_-<0$ by assumption, the claim follows analogously
  from \eqref{eq:u_exterior} and \eqref{Uprinf_diff}.
  \fin
\end{proof}

\begin{remark}
\label{remark:diff_origin_particular_D1}
The behaviour of the first derivative of the particular solution $U_p^+(r)$ near the origin
is given by \eqref{Upr0_diff}, and therefore, if $\alpha_0-1\geq 0$ then $U_P^+(r)$
extends to a $C^1([0,a])$ function as
$U{}_P^+(r)=r^{\alpha_0}(U^0_P+ o(1))$ and
$U_P^+{}'(r)=r^{\alpha_0-1}(\alpha_0U^0_P+ o(1))$ with $U^0_P\in\mathbb{R}$.
\end{remark}

\bibliography{references}{}
\bibliographystyle{review_bib}

\end{document}